\documentclass[12pt]{article}
\usepackage[affil-it]{authblk}
\usepackage{fullpage}

\RequirePackage{amsmath,amssymb}
\RequirePackage{natbib}
\RequirePackage{graphicx}
\usepackage{subcaption}
\usepackage[export]{adjustbox}

\usepackage[normalem]{ulem}
\RequirePackage[OT1]{fontenc}
\usepackage{latexsym}
\usepackage{etoolbox}
\usepackage{calc}
\usepackage{fancyhdr}
\usepackage{ifthen}
\usepackage{tabularx}
\RequirePackage[colorlinks,citecolor=blue,urlcolor=blue]{hyperref}
\usepackage{epsfig,xcolor}
\usepackage{comment}
\usepackage{algorithm}
\usepackage{algpseudocode}

\newcommand{\mI}{\ensuremath{\mathbb I}}
\newcommand{\mE}{\ensuremath{\mathbb E}}
\newcommand{\mP}{\ensuremath{\mathbb P}}
\newcommand{\mR}{\ensuremath{\mathbb R}}

\newtheorem{remark}{Remark}[section]

\numberwithin{equation}{section}
\newtheorem{prop}{Proposition}[section]
\newtheorem{thm}{Theorem}[section]
\newtheorem{lem}{Lemma}[section]
\newtheorem{coro}[prop]{Corollary}
\newenvironment{proof}{\paragraph{Proof:}}{\hfill$\square$}
\newtheorem{ass}{Assumption}[section]

\newcommand{\infoOSP}{\texttt{OGInfoSP}}
\newcommand{\hell}{\hat\ell}
\newcommand{\infoOSPCal}{\texttt{OGInfoSP-Cal}}
\newcommand{\thone}{\ensuremath{\mu_{\alpha}}}
\newcommand{\thtwo}{\ensuremath{\mu_{\alpha}^{cal}}}
\newcommand{\mFCR}{\mathrm{mFCR}}
\newcommand{\FCR}{\mathrm{FCR}}

\begin{document}

\title{Selecting Informative Conformal Prediction Sets with an Optimized FCR-Controlled Approach}

\author[1]{Israela Solomon}
\author[2]{Etienne Roquain}
\author[1]{Saharon Rosset}
\author[1]{Ruth Heller}

\affil[1]{Department of Statistics and Operations Research, Tel-Aviv University}
\affil[2]{Laboratoire de Probabilités, Statistique et Modélisation (LPSM), Sorbonne Université}

\date{} 
\maketitle

\begin{abstract}
Conformal methods provide prediction sets for outcomes with confidence guarantees. We study their use in a selective inference setting, where inference is performed only when the prediction set is informative. The analyst may consider as informative, for example, cases with prediction sets that are sufficiently small, exclude null values, or satisfy other appropriate monotone constraints.  
Because inference is typically restricted to informative cases in practical applications, accounting for the resulting selection bias is crucial to maintaining false coverage rate (FCR) control. 
A general framework for constructing such informative conformal prediction sets while controlling the FCR on the selected sample was suggested in \cite{GazinHellerMarandonRoquain2025}.  In this work we focus on  oracle-guided procedures. We derive the optimal decision policy under a suitable power objective in the oracle setting where the probability of belonging to each prediction set can be computed. In practice, of course, only estimated probabilities are available. We therefore introduce calibration procedures that adjust the oracle policy to maintain finite-sample FCR control. We show that this approach can achieve substantially higher power than available alternatives. We demonstrate the effectiveness of our new methods for classification outcomes on both real and simulated data. 
\end{abstract}

\section{Introduction}
Conformal inference provides  uncertainty guarantees in prediction problems by using  labeled data to construct a prediction set for each new example, with the guarantee that the probability that the outcome lies in the prediction set is at least $1-\alpha$. When many new examples are available, however, this guarantee may be insufficient: while we expect a fraction $1-\alpha$ of  prediction sets to contain the outcome overall, among the ``interesting" prediction sets (i.e., those that are informative for the analyst), the fraction  that contain the outcome may be substantially lower. 

For example, consider a classification outcome with possible classes $[K]= \{1,\ldots,K \}$. When constructing $1-\alpha$ level prediction sets, some examples yield the full set $[K]$, which  is a trivial and clearly  uninformative prediction set. The fraction of prediction sets  that contain the outcome among those with non-trivial prediction sets is necessarily lower than the corresponding fraction among all test samples, since all non-covering prediction sets are necessarily non-trivial. 

The analyst's goal is to guarantee coverage only for prediction sets that are informative. Importantly, this does not mean that examples with uninformative prediction sets are ignored. Rather, such examples may trigger a different downstream action. For instance, in a medical application, only prediction sets that clearly support a clinical decision may be reported. If a prediction set is too ambiguous (e.g., it includes both the category ``healthy" and ``sick"), the patient may be referred for additional testing. By contrast, when the prediction set excludes one of these categories, the patient is classified as either healthy or sick. The analyst then seeks a meaningful uncertainty guarantee for this selected subset of cases: among all patients who do not undergo further examination, the fraction whose prediction sets fail to contain the true outcome should be at most $\alpha$ in expectation.

The classic split conformal set-up we consider is as follows. The variables $\{(X_i,Y_i), i\in[n+m]\}$ are independent and identically distributed (iid) pairs from the joint distribution $P_{XY}$ on $\mathcal X\times \mathcal Y$ (e.g., $\mR^d\times [K]$ for classification, or $\mR^d\times \mR$ for regression).  The calibration sample is  $\{(X_j,Y_j),j\in [n] \}$;  the test sample is $\{(X_{n+i},Y_{n+i}), i \in [m]\}$.
While all measurements are observed in the calibration sample, only the $X_{n+i}$'s are observed in the test sample and the aim is to provide  prediction sets for  the unobserved  $(Y_{n+i})_{i\in [m]}$.  The prediction sets are subsets of $\mathcal Y= [K]$ for  classification, or subsets of $\mathcal Y=\mR$ for regression.

We are concerned with the setting in which the analyst has a clear notion that only a subset of the $\mathcal Y$ space is informative \citep{GazinHellerMarandonRoquain2025}.  We denote the set of informative prediction sets by $\mathcal I$.
For example, in classification, if the analyst considers as informative only prediction sets of size at most 2, then $\mathcal{I}$ contains all singletons and pairs of classes. 

It is well known that selection can invalidate error guarantees \citep{benjamini2005false}. Therefore, although we expect only $\alpha$ of the examples to fail to cover the true outcome when each prediction set is constructed to guarantee at least $1-\alpha$ coverage, the expected fraction of non-covering examples among the selected subset may be substantially higher. This issue was pointed out  for selected confidence intervals by \cite{benjamini2005false}, and later for selected conformal prediction sets by \cite{bao2024selective}.

In \cite{GazinHellerMarandonRoquain2025}, two procedures were suggested to address the selection of informative examples: \texttt{InfoSP}, that selects informative prediction sets by applying the Benjamini-Hochberg (BH, \citealt{BH1995}) procedure on 'adjusted levels' (where the 'adjusted level' of the example is the minimal $\alpha$ for which the prediction set will be informative), and then constructs the prediction sets for the selected $\mathcal S\subseteq [m]$ examples at the reduced level $\alpha|\mathcal S|/m$; \texttt{InfoSCOP}, that adds a pre-processing step that uses half the calibration samples for selecting a subset of the remaining calibration and test samples, so that the fraction of examples for which informative prediction sets can be constructed following pre-processing is large.  \texttt{InfoSCOP} was suggested to address the potential conservativeness of \texttt{InfoSP}  when the fraction of informative examples  in the test sample  is small (so $|\mathcal S|/m$ will be small). However,  how to choose the most favorable tuning parameters for the pre-processing step in \texttt{InfoSCOP} remains unclear.  The main limitation of both \texttt{InfoSP} and \texttt{InfoSCOP} is that they are driven primarily by the intuition of reducing the original $\alpha$ level for which the prediction sets are constructed (thus widening the prediction sets constructed),  rather than by firm optimality guidelines. This limitation motivated us to suggest an  oracle-guided informative selective prediction sets procedure, which we shall refer to as  \infoOSP. This new procedure  does not suffer from the potential conservatism of \texttt{InfoSP} and is clearly guided by optimality theory.

This work adds to the increasing body of work on conformal selective inference. Specifically, the selection problem when constructing multiple prediction sets simultaneously was studied in \cite{bao2024selective, jin2025confidence, GazinHellerMarandonRoquain2025, NEURIPS2024_6b99cfe8}.  
These works build on the  multiple comparisons literature:  \cite{benjamini2005false}  introduced the false coverage rate (FCR) as an error criterion to control when  constructing multiple  confidence intervals on selected parameters; the problem of selecting confidence intervals  satisfying specific notions of informativeness was studied in \cite{weinstein20,weinstein20online}.  
In concurrent work, \citet{li2026conformalizedlargescaleselectiveinference} assign each example a candidate informative set and a trust score, then apply BH to generalized conformal p-values to select which candidates to report. In contrast, we start from the informativeness family itself and jointly determine whether to report and which informative set to construct.

The structure of the paper and our main contributions are as follows. 
\begin{itemize}
    \item In \S~\ref{sec-optimal} we formalize the power objective and constraint, and solve the oracle optimization problem, which   assumes that conditional on $X$, the probability of belonging to a prediction set can be computed exactly. We solve the oracle problem by varying a single Lagrange multiplier $\mu\geq 0$, and using the fact that the maximization for a given $x\in \mathcal X$ has a simple geometric interpretation as the upper envelope of a finite collection of lines. 
    \item In \S~\ref{sec-proc} we provide an oracle-guided data-driven procedure, \infoOSP, that uses the calibration sample in order to guarantee control of the FCR. We also provide a variant \infoOSPCal\ that is useful for online testing. 
    \item In \S~\ref{subsec-finite-sample-theory} we provide the finite-sample FCR guarantee for \infoOSP\ (Theorem \ref{thm-FCR-control}) and a finite-sample per-example guarantee for \infoOSPCal. Asymptotic optimality for both procedures and additional guarantees are given in  \S~\ref{app-asymptotic-oracle}-\ref{app-rates-and-asymptotics}. 
   \item In \S~\ref{sec-exisitingprocedures}, \ref{sec-simulations}, and \S~\ref{app-simulations}, we apply the procedures to prevalent settings where informative selection takes place.  Numerical experiments demonstrate the high power and the FCR control of our novel procedures, as well as the FCR inflation of naive methods that do not adjust for informative selection.
   \item In \S~\ref{sec-discussion} we discuss extensions of interest. 
\end{itemize}

\section{The oracle decision policy }\label{sec-optimal}

Let $D: \mathcal X \rightarrow \{0,1\}$ denote a binary decision function for an example $(X,Y)\sim P_{XY}$. 
The value $D(X)=1$ indicates that the analyst constructs a prediction set for $Y$, while $D(X)=0$ indicates that no prediction set is constructed. 
When $D(X)=1$, the constructed prediction set is given by the set-valued function $C: \mathcal X \rightarrow \mathcal I$, where $\mathcal I$ denotes the collection of informative prediction sets.

Let $w:\mathcal I\to (0,B)$ be a bounded positive weight function that quantifies the utility of constructing a prediction set $C$.
Higher values of $w(C)$ correspond to more informative or more concentrated prediction sets.
Our goal is to maximize the weighted power
\begin{align}
\Pi(D,C)
&=
\mE\left[
\frac{1}{m}\sum_{i\in[m]}
w\!\left(C(X_{n+i})\right)
\mI\{Y_{n+i}\in C(X_{n+i})\}
D(X_{n+i})
\right]\nonumber\\
&=
\mE\left[
w\!\left(C(X)\right)
\mI\{Y\in C(X)\}
D(X)
\right].
\label{equPower}
\end{align}

\begin{remark}\label{rem-gain function}
The weight function $w(C)$ allows the analyst to encode preferences over prediction sets.
In many settings, higher utility corresponds to higher resolution, meaning that the prediction set is more concentrated and therefore more informative about the value of $Y$.
For example, for classification outcomes one may take $w(C)=1/|C|$, which penalizes large label sets, and  can be interpreted as measuring the average probability mass per element of $C$. 
Other weight functions may be preferable in applications where the informativeness of prediction sets with the same cardinality varies.
In our numerical examples we adopt the choice $w(C)=1/|C|$.
\end{remark}

We further require that the inflation of non-covering prediction sets among the selected examples be small. 
In selective conformal prediction, the FCR has typically been used for this purpose \citep{bao2024selective, GazinHellerMarandonRoquain2025}:
\begin{equation}\label{eq-FCR}
FCR(D,C)
=
\mE\left[
\frac{
\sum_{i\in[m]}\mI\{Y_{n+i}\notin C(X_{n+i})\}D(X_{n+i})
}{
1\lor \sum_{i\in[m]}D(X_{n+i})
}
\right].
\end{equation}

A closely related error measure is the marginal FCR, defined as
\begin{equation}\label{eq-mFCR}
mFCR(D,C)
=
\frac{
\mE\left[
\sum_{i\in[m]}\mI\{Y_{n+i}\notin C(X_{n+i})\}D(X_{n+i})
\right]
}{
\mE\left[
\sum_{i\in[m]}D(X_{n+i})
\right]
}
=
\frac{
\mE\left[\mI\{Y\notin C(X)\}D(X)\right]
}{
\mE\left[D(X)\right]
}.
\end{equation}
We use the convention that if the denominator is zero, the error rate is defined to be zero. 
We shall impose the constraint $mFCR(D,C)\leq \alpha$, which is simpler to handle technically. Note, however, that in this setting $mFCR(D,C)$ and $FCR(D,C)$ are tightly connected, as formalized in the following proposition. 

\begin{prop}\label{prop-FCR-oracle}
    Suppose $(X_i,Y_i)_{i\in [n+m]}$ is an iid sample from $P_{XY}$.  Then for all $D$ and~$C$, $$FCR(D,C) = mFCR(D,C)\:\mP\Big(\sum_{i\in [m]}D(X_{n+i})>0\Big). $$ 
\end{prop}

Since $\mP\left(\sum_{i\in [m]} D(X_{n+i}) > 0\right)\in [0,1]$, we have $FCR(D,C)\leq mFCR(D,C), $ so level $\alpha$ control of the  $mFCR(D,C)$ entails  control of the $FCR(D,C)$:  $$FCR(D,C)\leq mFCR(D,C)\leq \alpha.$$

All proofs for this section are provided in
\S~\ref{app-proofsOracle}. 

\subsection{Oracle problem formulation}\label{subsec-oracle-problem-formulation}
We consider the oracle setting in which the nonconformity scores perfectly recover the conditional probability of the outcome belonging to a candidate prediction set.
Specifically, for each $(X,Y)\sim P_{XY}$ and each $C\in\mathcal I$, the conditional probability $\mP(Y\in C\mid X)$ is known. For classification outcomes, this reduces to recovering precisely the  probability of belonging to a class $k$, $\mP(Y=k\mid X)$,  for all $k\in[K]$. 

We seek functions $D$ and $C$ that solve the following optimization problem:
\begin{eqnarray}
&& \max_{D:\mathcal X\to\{0,1\},\; C:\mathcal X\to\mathcal I} \ \Pi(D,C) \nonumber\\
&& \text{subject to } \quad mFCR(D,C)\leq \alpha.
\label{eq-OMT-problem}
\end{eqnarray}

In this oracle setting, both the weighted power and the $mFCR$ constraint can be written in terms of the known conditional probabilities. 
Specifically,
\begin{equation}\label{equPower-2}
\Pi(D,C)
:=
\mE\left[
w\!\left(C(X)\right)
\mP(Y\in C(X)\mid X)
D(X)
\right].
\end{equation}
The $mFCR$ constraint can be equivalently expressed using
\begin{equation}\label{def-G(D,C)}
G(D,C)
:=
\mE\left[
\left(1-\mP(Y\in C(X)\mid X)-\alpha\right)
D(X)
\right].
\end{equation}
After rearrangement, the constraint $mFCR(D,C)\leq\alpha$ is equivalent to requiring $G(D,C)\leq 0$, which is linear in the decision function $D$.
Our optimization problem  \eqref{eq-OMT-problem} is therefore expressed as follows:
\begin{eqnarray}
&& \max_{D:\mathcal X\to\{0,1\},\; C:\mathcal X\to\mathcal I} \ \Pi(D,C) \nonumber\\
&& \text{subject to } \quad G(D,C)\leq 0.
\label{eq-OMT-problem-Gconstraint}
\end{eqnarray}

We assume $\mathcal I$ is a finite collection of candidate prediction sets. 
This assumption is naturally satisfied for classification outcomes, where  $\mathcal I$ may consist of a finite collection of label subsets, for example all subsets up to a prescribed cardinality or a curated family of interpretable label groups.
In regression settings, $\mathcal I$ may consist of a finite grid of candidate prediction intervals, such as sign-restricted intervals with endpoints taken from a fixed finite grid.

Next, we describe the sequence of results leading to the solution of the optimization problem.

\subsection{Solving the Lagrangian with upper-envelope geometry}\label{subsec-optimal-constructing}
 
Consider the Lagrangian for the optimization problem \eqref{eq-OMT-problem-Gconstraint}
\begin{align*}
\mathcal{L}(D,C,\mu)
&:=
\Pi(D,C)-\mu G(D,C)\\
&=
\mE\Big[
\big(
w(C(X))\mP(Y\in C(X)\mid X)
-\mu(1-\mP(Y\in C(X)\mid X)-\alpha)
\big)
D(X)
\Big],
\end{align*}
where $\mu\geq 0$ and $\Pi(D,C)$ and $G(D,C)$ are defined in \eqref{equPower-2} and \eqref{def-G(D,C)}, respectively. 

We first solve the reduced problem
\begin{equation}
 \max_{D:\mathcal X\to\{0,1\},\; C:\mathcal X\to\mathcal I} \ \mathcal{L}(D,C,\mu) 
\label{eq-OMT-reduced}
\end{equation}
For fixed $\mu\geq 0$, the Lagrangian is separable across $x\in \mathcal X$. Thus the maximization can be carried out pointwise. For $x\in \mathcal X$, each candidate $C\in \mathcal I$ defines a line in $\mu$:
$$\ell_{x,C}(\mu) = w(C) \mP(Y \in C \mid X = x) + \mu \left(\mP(Y \in C \mid X = x) - (1 - \alpha)\right).$$ The intercept $w(C) \mP(Y \in C \mid X = x)$ is the weighted oracle power contribution of $C$ at~$x$. The slope $\mP(Y \in C \mid X = x) - (1 - \alpha)$ is the conditional coverage difference of $C$ from $1-\alpha.$ Thus, as $\mu$ increases, the optimization gives increasing preference to sets with larger conditional coverage probability. 

Define the upper envelope $$ \mathcal{U}_x(\mu) = \max_{C \in \mathcal{I}} \ell_{x,C}(\mu).$$

The candidate informative prediction set selected at $x\in \mathcal X$ for multiplier $\mu$ is the set that attains the upper envelope $\mathcal{U}_x(\mu)$, defined as:
\begin{equation}\label{eq-C-mu}
C^{\mu}(x) \in \arg\max_{C\in\mathcal I} \ell_{x,C}(\mu).
\end{equation}
To ensure maximum power (which is  the expectation of $w(C)\mP(Y\in C\mid X)D(X)$), we break ties by selecting the set with largest weight; see Lemma~\ref{lem:tie-breaking} and the preceding discussion in \S~\ref{app-proofsOracle} for the explanation.  
If  more than one prediction set corresponds to the largest weight, we further break ties in conditional  coverage probability by selecting  the one with the smallest index in a lexicographically ordered $\mathcal I$.\footnote{Assumption \ref{ass-probratio} ensures that such events occur with probability zero by requiring that any nontrivial linear combination of the conditional coverage probabilities has a non-atomic distribution.} While the choice of tie-breaking  among sets with equal weight is flexible,  it must be systematic so that if $C_1\neq C_2$ but $\ell_{x,C_1}(\cdot) = \ell_{x,C_2}(\cdot)$,  exactly one is uniquely chosen  for the line segment forming the upper envelope.

The corresponding decision rule is \begin{equation}\label{eq-D-mu}
D^{\mu}(x) = \mI\left\{ \mathcal{U}_x(\mu) \geq 0 \right\}.
\end{equation}

The pointwise maximizers of the Lagrangian, $C^{\mu}$ and $D^{\mu}$, are the solution to the optimization problem in \eqref{eq-OMT-reduced}. Since $\mathcal{U}_x(\mu)$ is the upper envelope of lines, it is piecewise linear. Therefore $C^{\mu}(x)$ and  $D^{\mu}(x)$ are piecewise constant. $C^{\mu}(x)$ changes only at {\it breakpoints}  where the envelope switches from one line to another; $D^{\mu}(x)$ changes only at the {\it zero crossing} where $\mathcal U_x(\mu)$ crosses zero. See illustration in Figure \ref{fig-envelope}.

We denote the largest conditional coverage probability among the informative prediction sets by 
\begin{equation}\label{eq-T(X)}
T(x) = \max_{C\in\mathcal{I}} \mP(Y\in C\mid X = x).
\end{equation}
It plays a critical role in determining the behavior of the upper envelope $\mathcal U_x(\mu)$ and the subsequent selection decisions. 

\begin{lem} \label{lem-U-by-Tx}
If $T(x) \ge 1 - \alpha$, then $\mathcal{U}_x(\mu) \ge 0$  $\forall \mu \ge 0$. Otherwise, $\mathcal{U}_x(\mu)$ is decreasing in $\mu$.
\end{lem}

In the following proposition, we establish monotonicity properties for key terms involving $C^{\mu}$ and $D^{\mu}$. 

\begin{prop}\label{prop-monotoneD}
For each fixed $x\in \mathcal X$, as $\mu\geq 0$ is increasing: 
1.~$\mP(Y\in C^{\mu}(x) \mid X=x)$ is non-decreasing; 
2.~$w(C^{\mu}(x))\mP(Y\in C^{\mu}(x) \mid X=x)$ is non-increasing;
3.~$D^{\mu}(x)$ is non-increasing;
4.~$w(C^{\mu}(x))\mP(Y\in C^{\mu}(x) \mid X=x)D^{\mu}(x)$ is non-increasing; 
5.~$(1 - \mP(Y \in C^{\mu}(x) \mid X = x) -\alpha)D^{\mu}(x)$ is non-increasing.
\end{prop}

An immediate consequence is that the expectation of items 4 and 5 of Proposition \ref{prop-monotoneD}, which correspond to the power objective and constraint of the optimization problem in \eqref{eq-OMT-problem-Gconstraint}, are both non-increasing in $\mu$.
\begin{coro}\label{coro-prop-monotoneD}
$\Pi(D^{\mu}, C^{\mu})$ and $G(D^{\mu}, C^{\mu})$ are non-increasing in $\mu.$
\end{coro}

\begin{figure}[htbp]
    \centering
    \begin{tabular}{ccc}
    	\hspace{-0.03\textwidth}
        \includegraphics[width=0.44\textwidth,valign=t]{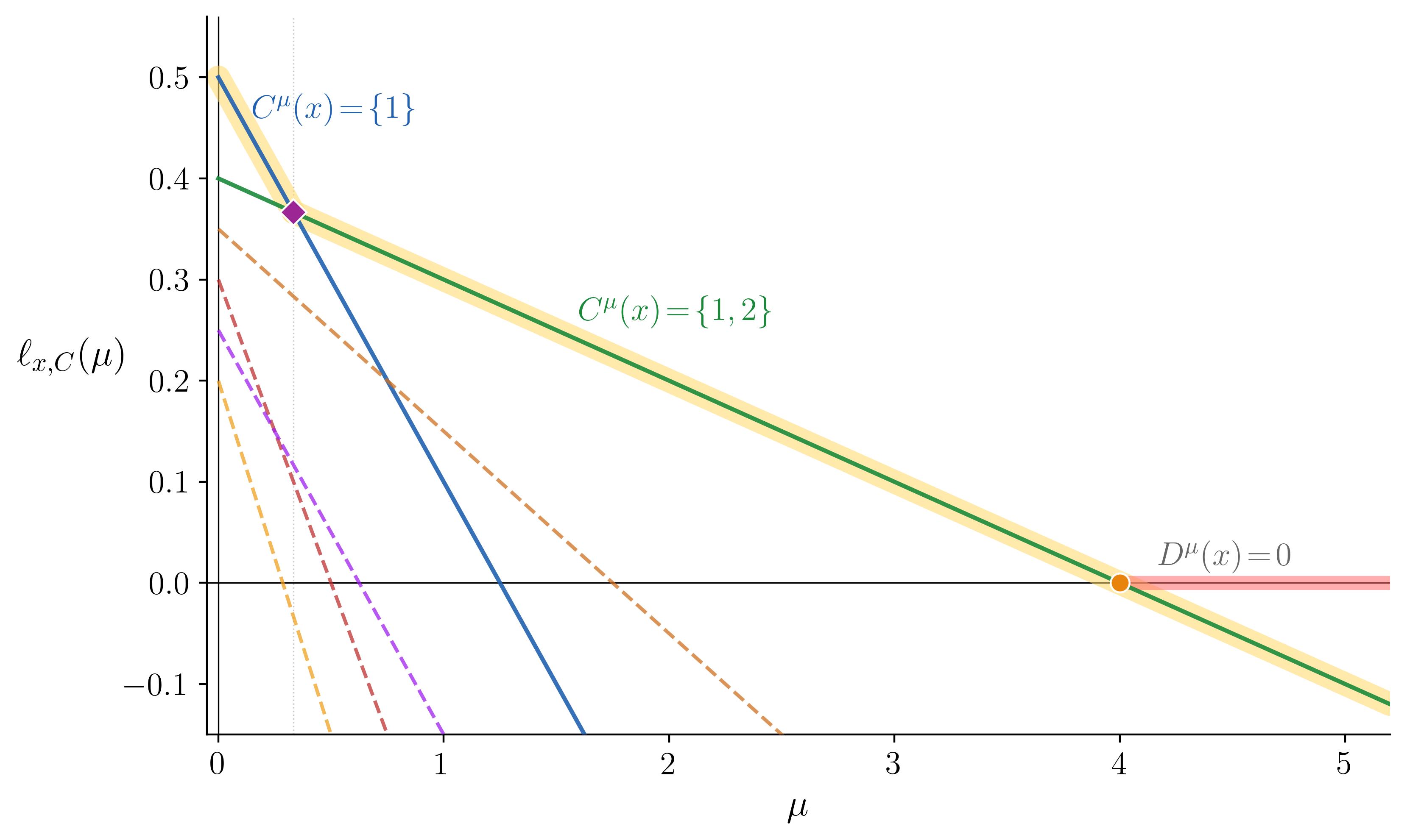}
        \includegraphics[width=0.44\textwidth,valign=t]{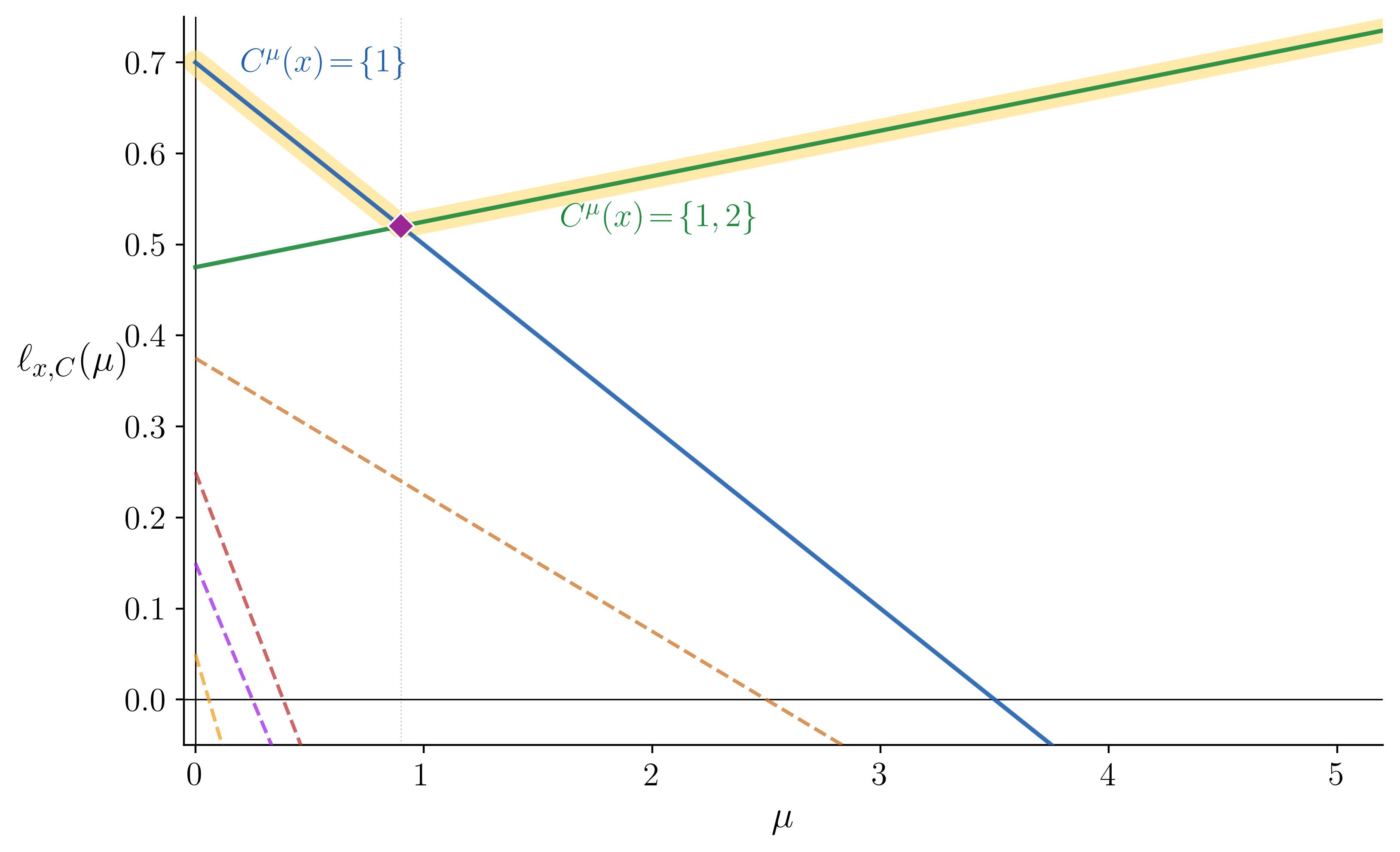}
        \hspace{0.005\textwidth}
        \includegraphics[width=0.12\textwidth,valign=t]{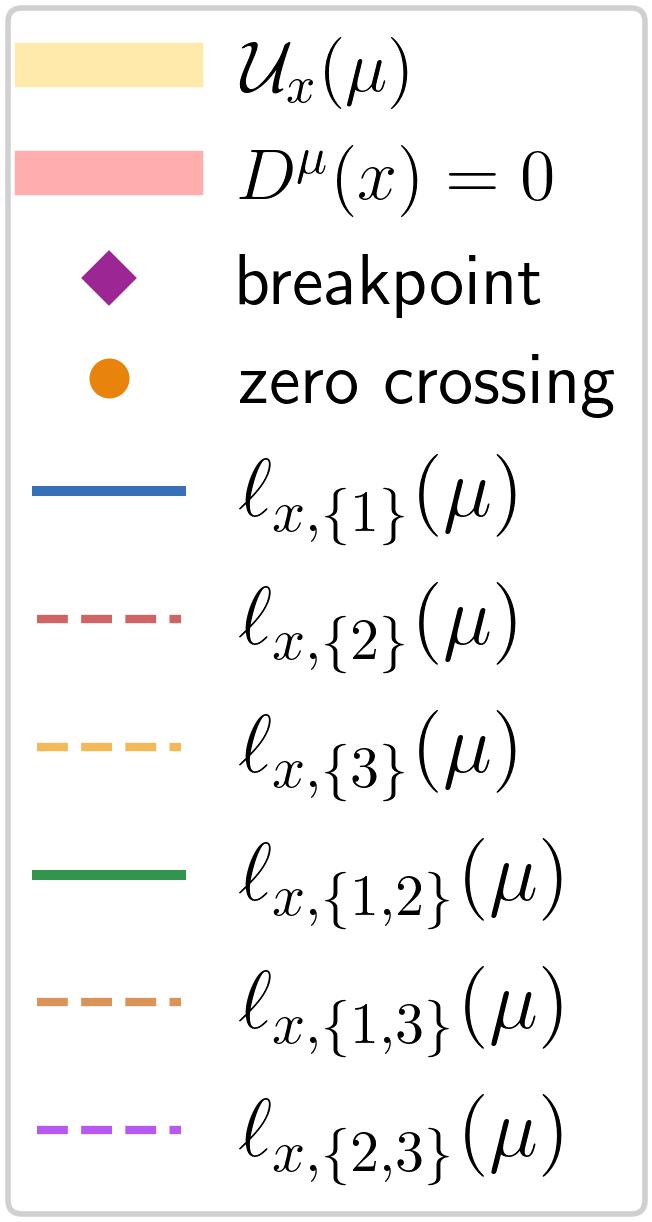}
    \end{tabular}
    \caption{An example for $K = 3$, where $\mathcal{I}$ is all sets with $1$ or $2$ classes, $w(C) = \frac{1}{|C|}$, and $\alpha = 0.1$. In the left panel, the conditional class probabilities are $0.5$, $0.3$, and $0.2$. The upper envelope $\mathcal{U}_x(\mu)$ consists of $\ell_{x,\{1\}}(\mu)$ for $\mu \in [0, \frac{1}{3}]$ and $\ell_{x,\{1,2\}}(\mu)$ for $\mu > \frac{1}{3}$, with a zero crossing at $\mu = 4$. Therefore $C^{\mu}(x)$ is $\{1\}$ for $\mu \in [0, \frac{1}{3}]$ and $\{1,2\}$ for $\mu > \frac{1}{3}$, and $D^{\mu}(x) = \mI\{\mu \le 4\}$. In the right panel, the conditional class probabilities are $0.7$, $0.25$, and $0.05$. The upper envelope $\mathcal{U}_x(\mu)$ consists of $\ell_{x,\{1\}}(\mu)$ for $\mu \in [0, \frac{9}{10}]$ and $\ell_{x,\{1,2\}}(\mu)$ for $\mu > \frac{9}{10}$. There is no zero crossing since the slope of $\ell_{x,\{1,2\}}(\mu)$ is positive, as the set probability is above $1 - \alpha$. Therefore $D^{\mu}(x) = 1$ for all $\mu \ge 0$. $C^{\mu}(x)$ is $\{1\}$ for $\mu \in [0, \frac{9}{10}]$ and $\{1,2\}$ for $\mu > \frac{9}{10}$.}
    \label{fig-envelope}
\end{figure}

The behavior for large $\mu$ is also clear from the envelope geometry. For large enough $\mu$, the slope dominates the intercept and the line segment in the upper envelope is the one with the largest  possible slope, i.e.,   $T(X)-(1-\alpha)$ for $T(x)$ defined in \eqref{eq-T(X)}.

\begin{prop}\label{lem:oracle_mularge}
    For all $x\in \mathcal{X}$, there exists $\mu_0(x)\geq 0$ such that for all $\mu\geq \mu_0(x)$,  $D^{\mu}(x) = \mI\{T(x)\geq 1-\alpha\}$ and $C^{\mu}(x) \in \arg\max_{C\in\mathcal{I}} \mP(Y\in C\mid X=x)$. If multiple sets share the same maximum coverage probability, the set with the largest weight $w(C)$ is selected.
\end{prop}

\subsection{The oracle procedure}

Denote the optimal solution for the optimization problem~\eqref{eq-OMT-problem-Gconstraint} by
\[f^*=\max_{\substack{D:\mathcal X\to\{0,1\},\; C:\mathcal X\to\mathcal I, \\ G(D,C)\leq 0}} \Pi(D,C).\]
For any feasible $(D,C)$ (i.e., satisfying $G(D,C) \leq 0$), 
the Lagrangian $\mathcal{L}(D,C,\mu)$ serves as an upper bound for the objective  $\Pi(D,C)$ for any $\mu \geq 0$:  $$\mathcal{L}(D,C,\mu) = \Pi(D,C) - \mu G(D,C) \geq \Pi(D,C).$$ Consequently, for any $\mu\geq 0$, the optimal value $f^*$ is bounded above by the Lagrangian evaluated at $(D^{\mu}, C^{\mu})$:
\begin{align}
f^* &= \max_{\substack{D:\mathcal X\to\{0,1\},\; C:\mathcal X\to\mathcal I, \\ G(D,C)\leq 0}} \Pi(D,C)
\leq \max_{\substack{D:\mathcal X\to\{0,1\},\; C:\mathcal X\to\mathcal I, \\ G(D,C)\leq 0}} \mathcal{L}(D,C,\mu) \nonumber\\
&\leq \max_{D:\mathcal X\to\{0,1\},\; C:\mathcal X\to\mathcal I} \mathcal{L}(D,C,\mu) = \mathcal{L}(D^{\mu}, C^{\mu},\mu), \label{prop-f-leq-L}
\end{align}
 where $C^{\mu}$ and $D^{\mu}$ are defined by \eqref{eq-C-mu} and \eqref{eq-D-mu}, respectively.

As $\mu$ increases, both the power and the constraints are non-increasing (see Corollary \ref{coro-prop-monotoneD}). Thus the optimal solution should use the smallest $\mu$ for which the constraint is satisfied. 

If $\mP\!\left(T(X)\geq 1-\alpha\right)=0$ 
, i.e., the largest conditional coverage probability among informative prediction sets is almost surely below $1-\alpha$, 
then  \eqref{def-G(D,C)} implies that the constraint $G(D,C)\le 0$ can only be satisfied by  $D(X)=0$ almost surely, making this trivial solution  optimal. 
When $\mP\!\left(T(X)> 1-\alpha\right)=0$,  this case is formalized in Proposition~\ref{prop-D-G}. 
We henceforth focus on the more interesting and practically relevant setting in which 
$$\mP\!\left(T(X)> 1-\alpha\right)>0,$$
so that the oracle procedure can achieve nontrivial power. 
In this setting, we have the following result.

\begin{prop}\label{prop-positive-selection}
	  If $\mP(T(X) \geq 1-\alpha) > 0$, then $\mE\left[ D^{\mu}(X)\right]>0$ for any $\mu \geq 0$.
\end{prop}

Next, we establish that in well-powered settings the constraint functional   $G(D^{\mu}, C^{\mu})$   has a limit that is strictly negative. 
\begin{prop}\label{prop-decreasingG}
If $\mP(T(X) > 1-\alpha) > 0$, then $G(D^{\mu}, C^{\mu})<0$ for all sufficiently large  $\mu$.
\end{prop}

Next, we show that the oracle optimization problem \eqref{eq-OMT-problem-Gconstraint} has a solution, assuming the following sufficient condition. See Theorem~\ref{thm-random-policy}  
for a randomized counterpart that does not require this assumption.

\begin{ass}\label{ass-probratio}
	The law of $X$ is non-atomic and for every $C \in \mathcal{I}$, the random variable $\mP(Y\in C\mid X)$ has a continuous, non-atomic distribution. Moreover, for every distinct $C_1, C_2 \in \mathcal{I}$, 
    any nontrivial linear combination of $\mP(Y\in C_1\mid X)$ and $\mP(Y\in C_2\mid X)$ is non-zero with probability one.
\end{ass}
For a fixed $\mu$, this assumption ensures that almost surely in $X$,   the upper envelope $\mathcal U_X(\mu)$ is non-zero and the candidate set $C^{\mu}(X)$ is determined without tie-breaking. This means  $\mathcal{U}_X$ is almost never at a breakpoint or a zero crossing for a given $\mu$. As a consequence, $G(D^{\mu}, C^{\mu})$ is a continuous  and non-increasing function of $\mu$. If $\lim_{\mu \rightarrow \infty} G(D^{\mu}, C^{\mu})< 0$, then the set $\{\mu: G(D^{\mu}, C^{\mu})\leq 0\}$ has a unique minimum $\mu^*$, which is the multiplier at which the policy is implemented. This is formalized in the following Theorem.

\begin{thm}\label{thm-optimal-nonrandom}
     Suppose Assumption \ref{ass-probratio} holds.  If $\mP(T(X) > 1-\alpha) > 0$, then $\mu^* = \min\{\mu\geq 0: G(D^{\mu}, C^{\mu})\leq 0\}$ exists and 
     $(D^{\mu^*},C^{\mu^*})$ is an optimal solution to the oracle problem  \eqref{eq-OMT-problem-Gconstraint}. If $\mP(T(X) > 1-\alpha) = 0$, then the optimal decision is not to construct any prediction set. 
\end{thm}

\section{Practical procedures with FCR control}\label{sec-proc}

We consider a prediction algorithm, built from training data that is independent of the calibration and test data. We assume that for each $x \in \mathcal{X}$, the estimated conditional probability of belonging to each informative prediction set is available. For classification outcomes, this simplifies to the common assumption that the estimated probability of belonging to a class $k$, $\hat{\mP}(Y = k \mid X = x)$, is available for all $k\in [K]$. Note that the estimated conditional probabilities are implicitly evaluated conditionally on the training data, so for a fixed training set  $\mathcal D_{train}$, we write $\hat{\mP}(Y \in C \mid X = x)$ as shorthand for $\hat{\mP}(Y \in C \mid X = x, \mathcal D_{train})$.  

We propose the following approach, that is motivated by the oracle construction of \S~\ref{sec-optimal}, 
now with the oracle probability $\mP(Y \in C \mid X = x)$ replaced by the fitted probability $\hat{\mP}(Y \in C \mid X = x)$ for each fixed $x\in \mathcal X$ and $C\in \mathcal I$. We  use the same derivations as in \S~\ref{sec-optimal} in order to determine the decision policy and which prediction set is constructed at a given $\mu$. 
A key difference from \S~\ref{sec-optimal} is the fact that the derived procedure uses the calibration sample in order to have a finite-sample FCR control guarantee. Specifics follow.

For any $x \in \mathcal{X}$, each candidate $C \in \mathcal{I}$ defines a line in $\mu$:
$$\hell_{x,C}(\mu) = w(C) \hat{\mP}(Y \in C \mid X = x) + \mu \left(\hat{\mP}(Y \in C \mid X = x) - (1 - \alpha)\right).$$
The upper envelope $ \max_{C \in \mathcal{I}} \hell_{x,C}(\mu)$ determines $\hat C^{\mu}(x)$ and $\hat D^{\mu}(x)$:
\begin{equation} \label{eq-C-mu-procedure}
\hat C^{\mu}(x) \in \arg\max_{C \in \mathcal{I}}\hell_{x,C}(\mu)  
\end{equation} 
and 
\begin{equation} \label{eq-D-mu-procedure}
\hat D^{\mu}(x) = \mI\left\lbrace \max_{C \in \mathcal{I}} \hell_{x,C}(\mu) > 0 \right\rbrace.
\end{equation}
In \eqref{eq-C-mu-procedure} ties are broken in favor of smaller weight. 
The tie-breaking rule, as well as the strict inequality in \eqref{eq-D-mu-procedure},  are driven by the constraint to guarantee $FCR\leq \alpha$, see proof of our main theorem \ref{thm-FCR-control}. This is  in contrast to the choice of tie-breaking in the oracle setting (where the ties in \eqref{eq-C-mu} are broken in favor of larger weights,  
and the $\geq$ in \eqref{eq-D-mu} is driven by the requirement to maximize power).
Considering again the examples in Figure~\ref{fig-envelope} but this time with estimated probabilities: in the left panel,  $\hat C^{\mu}(x)$ is $\{1\}$ for $\mu \in [0, \frac{1}{3})$ and $\{1,2\}$ for $\mu \ge \frac{1}{3}$, and $\hat D^{\mu}(x) = \mI\{\mu < 4\}$; in the right panel,  $\hat C^{\mu}(x)$ is $\{1\}$ for $\mu \in [0, \frac{9}{10})$ and $\{1,2\}$ for $\mu \ge \frac{9}{10}$. 

In the event that multiple prediction sets  share the same minimal weight in the tie-breaking rule above, we resolve the additional ties in the estimated conditional coverage probability as in the oracle setting via a secondary rule.  The only requirement for this secondary tie-breaking  rule, specifically for prediction sets that have the same estimated conditional coverage probability and weight,  is that it  must be systematic, so that if $C_1\neq C_2$ but $\hell_{x,C_1}(\cdot) = \hell_{x,C_2}(\cdot)$,  exactly one is uniquely selected  for the line segment forming the upper envelope.  For concreteness, this secondary rule selects the prediction set with the smallest index in a lexicographically ordered $\mathcal I$.

The monotonic behavior in Proposition \ref{prop-monotoneD} carries over to the practical setting where probabilities are estimated. 
 See Proposition \ref{prop-monotoneD-procedure} for a formal statement.  So the decision indicator $\hat D^\mu(X_i)$ is non-increasing in $\mu$. In addition, items 4 and 5 in Proposition \ref{prop-monotoneD-procedure} state that $w(\hat C^{\mu}(X_i))\hat{\mP}(Y\in \hat C^{\mu}(X_i) \mid X=X_i)\hat D^{\mu}(X_i)$  and  
  $(1 - \hat{\mP}(Y \in \hat C^{\mu}(X_i) \mid X = X_i) -\alpha)\hat D^{\mu}(X_i)$ are non-increasing in $\mu$. But 
since only approximate probabilities are available,  finite-sample FCR control cannot be guaranteed by relying solely on the approximated oracle rule. Instead, we utilize the calibration sample to achieve finite-sample FCR control.  

We start by introducing our first algorithm, \infoOSP, in which  the  false coverage proportion (FCP) for a given $\mu$  is estimated using both the calibration and test samples:
\begin{equation} \label{eq-FCP-mu-procedure}
\widehat{FCP}^{\mu} = \frac{\frac{1}{1+n} \left(1 + \sum_{i=1}^n \mI\{Y_i\notin \hat C^{\mu}(X_i)\}\hat D^{\mu}(X_i)\right)}{\frac{1}{m} \left(1\lor\sum_{i=1}^m \hat D^{\mu}(X_{n+i})\right)}.
\end{equation}
The next step is to select the appropriate $\mu$: 
\begin{equation}\label{eq-thone}
	\mu_\alpha = \min\left\lbrace \mu\geq 0: \widehat{FCP}^{\mu} \leq \alpha \right\rbrace,
\end{equation}
and we set $\mu_\alpha = \infty$  if the set is empty. Note that for a non-empty set, the minimum necessarily exists since $\widehat{FCP}^{\mu}$ is right continuous.  For each $i \in [m]$, we then report $\hat C^{\thone}(X_{n+i})$ whenever $\hat D^{\thone}(X_{n+i}) = 1$ if $\thone < \infty$ (if $\thone = \infty$, no prediction sets are produced).

This procedure has the property that  $\hat D^{\mu}(x)=1$ for all $\mu\geq 0$ if there exists an informative prediction set with estimated conditional coverage probability at least $1-\alpha$. 
Using the same reasoning as in \S~\ref{sec-optimal}, this occurs because the existence of such a set ensures that the upper envelope remains positive for all $\mu\geq 0,$ which in turn keeps the decision indicator at one. 
Therefore,  
for 
$\mu_\alpha$ such that $\widehat{\mathrm{FCP}}^{\mu_\alpha}\le\alpha$,  prediction sets will necessarily be constructed at least for examples $i\in [m]$  satisfying   
    $\max_{C \in \mathcal{I}} \hat{\mP}(Y_{n+i} \in C | X = X_{n+i})\geq 1-\alpha$.

Henceforth we denote the lines and upper envelope for $i \in [n+m]$ by  $$\ell_{i,C}(\mu) = \hell_{X_i,C}(\mu) \quad \textrm{and} \quad \mathcal{U}_i(\mu) = \max_{C \in \mathcal{I}} \ell_{i,C}(\mu).$$ 
This procedure relies on finding a finite set of candidates for $\thone$ and allows reducing the collection of informative sets considered for each example. 
The implementation relies on two consequences of the upper-envelope geometry. First, since $\mathcal{U}_i$ is piecewise linear, $\hat C^{\mu}(X_i)$ and $\hat D^{\mu}(X_i)$ are piecewise constant in $\mu$, changing only at \textit{breakpoints} where the envelope switches from one line to another, or \textit{zero crossings} where $\mathcal{U}_i$ crosses $0$.  Therefore, it suffices to search for $\thone$ over a finite set of candidate $\mu$ values that contains all breakpoints and zero crossings across all $i \in [n+m]$. Second, for each example $i \in [n+m]$ we may  restrict the collection of candidate informative sets to a subset $\mathcal{I}_i \subseteq \mathcal{I}$ that contains all sets appearing on the upper envelope for $X_i$.

The \infoOSP\ procedure is presented in Algorithm~\ref{algo:OGinfoSP}, with additional practical details provided in Section~\ref{subsec-implementationDetails}. We add Assumption \ref{ass-martingale} which enables us to express this procedure in a more conventional form as a threshold selection algorithm in  \S~\ref{sec:threshold},  as well as prove that it controls the finite-sample FCR, see  \S~\ref{subsec-finite-sample-theory} for details. See  \S~\ref{app-rates-and-asymptotics} for a proof that with  additional regularity  assumptions, asymptotically (in training, calibration, and test sample sizes), the optimal power is achieved.

\newlength{\AlgLabelWidth}
\newlength{\AlgSubIndent}
\newlength{\AlgIOWidth}

\setlength{\AlgLabelWidth}{1.4em}
\setlength{\AlgSubIndent}{1.4em}
\setlength{\AlgIOWidth}{4.5em}

\newcommand{\AlgStep}[3][0.65em]{%
  \Statex
  \makebox[\AlgLabelWidth][l]{\textbf{#2}}%
  \parbox[t]{\dimexpr\linewidth-\AlgLabelWidth\relax}{%
    \setlength{\parindent}{0pt}%
    \setlength{\parskip}{0pt}%
    \linespread{1.1}\selectfont
    #3%
  }%
  \vspace{#1}%
}

\newcommand{\AlgSubStep}[2]{%
  \Statex
  \hspace{\AlgSubIndent}%
  \makebox[\AlgLabelWidth][l]{\textbf{#1}}%
  \parbox[t]{\dimexpr\linewidth-\AlgSubIndent-\AlgLabelWidth\relax}{%
    \setlength{\parindent}{0pt}%
    \setlength{\parskip}{0pt}%
    \linespread{1.1}\selectfont
    #2%
  }%
  \vspace{0.35em}%
}

\newcommand{\AlgIO}[2]{%
  \Statex
  \makebox[\AlgIOWidth][l]{\textbf{#1:}}%
  \parbox[t]{\dimexpr\linewidth-\AlgIOWidth\relax}{%
    \setlength{\parindent}{0pt}%
    \setlength{\parskip}{0pt}%
    \linespread{0.98}\selectfont
    #2%
  }%
  \vspace{0.45em}%
}

\newcommand{\AlgOutput}[1]{%
  \Statex
  \makebox[\AlgIOWidth][l]{\textbf{Output:}}%
  \parbox[t]{\dimexpr\linewidth-\AlgIOWidth\relax}{%
    \setlength{\parindent}{0pt}%
    \setlength{\parskip}{0pt}%
    \linespread{0.98}\selectfont
    #1\par
    \vspace{0.05em}%
    \noindent
    \hspace*{-\AlgIOWidth}%
    \rule{\dimexpr\linewidth+\AlgIOWidth\relax}{0.5pt}%
  }%
  \vspace{0.45em}%
}

\newcommand{\AlgKw}[1]{\textbf{#1}}

\newcommand{\AlgFullRule}{%
  \noindent
  \hspace*{-\leftmargin}%
  \rule{\dimexpr\linewidth+\leftmargin\relax}{0.5pt}%
}

\newcommand{\AlgCalibrationHeader}[3]{%
  \Statex
  \begingroup
  \linespread{1}\selectfont
  \parbox[t]{\linewidth}{%
    \setlength{\parindent}{0pt}%
    \setlength{\parskip}{0pt}%

    \noindent\textbf{#1}\par
    \vspace{0.05em}%

    \noindent
    \makebox[\AlgIOWidth][l]{\textbf{Input:}}%
    \parbox[t]{\dimexpr\linewidth-\AlgIOWidth\relax}{%
      #2%
    }%
    \par
    \vspace{-0.05em}%

    \noindent
    \makebox[\AlgIOWidth][l]{\textbf{Output:}}%
    \parbox[t]{\dimexpr\linewidth-\AlgIOWidth\relax}{%
      #3%
    }%
    \par
    \vspace{0.10em}%

    \AlgFullRule
  }%
  \endgroup
  \vspace{0.25em}%
}

\newcommand{\AlgPhase}[3]{%
  \Statex
  \vspace{-0.25em}%
  \begingroup
  \linespread{1}\selectfont
  \parbox[t]{\linewidth}{%
    \setlength{\parindent}{0pt}%
    \setlength{\parskip}{0pt}%

    \AlgFullRule\par
    \vspace{0.20em}%

    \noindent\textbf{#1}\par
    \vspace{0.12em}%

    \noindent
    \hspace*{\AlgLabelWidth}%
    \parbox[t]{\dimexpr\linewidth-\AlgLabelWidth\relax}{%
      \setlength{\parindent}{0pt}%
      \setlength{\parskip}{0pt}%
      #2\par
      \vspace{0.05em}%
      #3%
    }%
  }%
  \endgroup
}

\begin{algorithm}[t]
\caption{Oracle-guided optimized selection procedure (\infoOSP)}
\label{algo:OGinfoSP}
\small
\begin{algorithmic}

\AlgIO{Input}{Calibration data $\{(X_i,Y_i)\}_{i=1}^n$, test covariates $\{X_{n+i}\}_{i=1}^m$, a collection of informative sets $\mathcal{I}$, and a level $\alpha$.}
\AlgOutput{Informative prediction sets for the selected test points.}

\AlgStep{1.}{\AlgKw{For each} $i \in [n+m]$, choose a subset $\mathcal{I}_i \subseteq \mathcal{I}$ of candidates to be the prediction set for $X_i$. See Section~\ref{subsec-computing-Ii} for details.}

\AlgStep{2.}{Compute a set $M$ of candidate $\mu$ values where the value of $\widehat{FCP}^{\mu}$ as defined in Step~3c can decrease. See Section~\ref{subsec-computing-M} for details.}

\AlgStep[0.3em]{3.}{\AlgKw{For each} $\mu \in M \cup \{0\}$, sorted in increasing order, \AlgKw{do}}

\AlgSubStep{a.}{\AlgKw{For each} $i \in [n+m]$, compute the best prediction set $\hat C^{\mu}(X_i)$ using \eqref{eq-C-mu-procedure},\\ i.e., $\hat C^{\mu}(X_i)\in \arg\max_{C\in \mathcal I_i}\ell_{i,C}(\mu)$, with ties broken in favor of smaller $w(C)$.}

\AlgSubStep{b.}{\AlgKw{For each} $i \in [n+m]$, compute the decision indicator $\hat D^{\mu}(X_i)$ using \eqref{eq-D-mu-procedure},\\ i.e., $\hat D^{\mu}(X_i) = \mI\{\ell_{i, \hat C^{\mu}(X_i)}(\mu)>0\}$.}

\AlgSubStep{c.}{Let $\widehat{FCP}^{\mu} = \frac{\frac{1}{1+n} \left(1 + \sum_{i=1}^n \mI\{Y_i\notin \hat C^{\mu}(X_i)\}\hat D^{\mu}(X_i)\right)}{\frac{1}{m} \left(1\lor\sum_{i=1}^m \hat D^{\mu}(X_{n+i})\right)}$ (as in \ref{eq-FCP-mu-procedure}).}

\AlgSubStep{d.}{%
  \AlgKw{If} $\widehat{FCP}^{\mu} \leq \alpha$, \AlgKw{then}\par
  \hspace*{\AlgSubIndent}%
  \parbox[t]{\dimexpr\linewidth-\AlgSubIndent\relax}{%
    denote $\mu_{\alpha}:=\mu$ (this is the value in \eqref{eq-thone}),\\
    report $\left\{ \hat C^{\mu_\alpha}(X_{n+i}) : i \in [m], \hat D^{\mu_\alpha}(X_{n+i}) = 1 \right\}$, and stop.%
  }%
}

\AlgStep{4.}{\AlgKw{If} $\widehat{FCP}^{\mu} > \alpha$ for all $\mu \in M$, \AlgKw{then} no prediction sets are constructed.}

\end{algorithmic}
\end{algorithm}

\begin{algorithm}[t]
\caption{Oracle-guided optimized selection procedure (\infoOSPCal)}
\label{algo:OGinfoSP-Cal}
\small
\begin{algorithmic}

\AlgCalibrationHeader{Calibration:}{Calibration data $\{(X_i,Y_i)\}_{i=1}^n$, a collection of informative sets $\mathcal{I}$, and a level $\alpha$.}{A threshold $\thtwo$.}

\AlgStep{1.}{\AlgKw{For each} $i \in [n]$, choose a subset $\mathcal{I}_i \subseteq \mathcal{I}$ of candidates to be the prediction set for $X_i$. See Section~\ref{subsec-computing-Ii} for details.}

\AlgStep{2.}{Compute a set $M$ of candidate $\mu$ values where the value of $\widehat{FCP}_{cal}^{\mu}$ as defined in Step~3c can decrease. See Section~\ref{subsec-computing-M} for details.}

\AlgStep[0.3em]{3.}{\AlgKw{For each} $\mu \in M \cup \{0\}$, sorted in increasing order, \AlgKw{do}} 

\AlgSubStep{a.}{\AlgKw{For each} $i \in [n]$, compute the best prediction set $\hat C^{\mu}(X_i)$ using \eqref{eq-C-mu-procedure},\\ i.e., $\hat C^{\mu}(X_i)\in \arg\max_{C\in \mathcal I_i}\ell_{i,C}(\mu)$, with ties broken in favor of smaller $w(C)$.}

\AlgSubStep{b.}{\AlgKw{For each} $i \in [n]$, compute the decision indicator $\hat D^{\mu}(X_i)$ using \eqref{eq-D-mu-procedure},\\ i.e., $\hat D^{\mu}(X_i) = \mI\{\ell_{i, \hat C^{\mu}(X_i)}(\mu)>0\}$.}

\AlgSubStep{c.}{Let $\widehat{FCP}_{cal}^{\mu} = \frac{1 + \sum_{i=1}^n \mI\{Y_i\notin \hat C^{\mu}(X_i)\}\hat D^{\mu}(X_i)}{1 + \sum_{i=1}^n \hat D^{\mu}(X_{i})}$.}

\AlgSubStep{d.}{\AlgKw{If} $\widehat{FCP}_{cal}^{\mu} \leq \alpha$, \AlgKw{then} report $\thtwo:=\mu$, and stop.}

\AlgStep[-0.25em]{4.}{\AlgKw{If} $\widehat{FCP}_{cal}^{\mu} > \alpha$ for all $\mu \in M$, \AlgKw{then} report $\thtwo := \infty$.}

\AlgPhase{Prediction for a test point:}{Given a test point $X_{n+j}$, compute $\hat{C}^{\thtwo}(X_{n+j})$ using \eqref{eq-C-mu-procedure} and $\hat{D}^{\thtwo}(X_{n+j})$ using \eqref{eq-D-mu-procedure}.}{Report $\hat{C}^{\thtwo}(X_{n+j})$ if $\hat{D}^{\thtwo}(X_{n+j}) = 1$.}

\end{algorithmic}
\end{algorithm}

Next, we introduce our second algorithm  \infoOSPCal\ (see Algorithm~\ref{algo:OGinfoSP-Cal}), in which  both the numerator and the denominator of the estimated FCP at a given $\mu$ are evaluated only with the calibration sample. 
The appropriate $\mu$ is therefore: 
\begin{equation}\label{eq-thtwo}
\thtwo = \min\left\lbrace\mu\geq 0: \frac{1+\sum_{i=1}^n \mI\{Y_i\notin \hat C^{\mu}(X_i)\}\hat D^{\mu}(X_i)}{1+\sum_{i=1}^n \hat D^{\mu}(X_i)}\leq \alpha \right\rbrace,
\end{equation}
and we set $\thtwo=\infty$ if the set is empty. For a non-empty set the minimum necessarily exists by right continuity of the estimated FCP. 
If $\thtwo<\infty$, then for every new example $X$ we construct $\hat C^{\thtwo}(X)$ only if $\hat D^{\thtwo}(X)=1.$
\infoOSPCal\ is specifically designed for online applications where test examples arrive sequentially and require immediate decisions: the threshold $\thtwo$ is determined using only the calibration sample;  this threshold is then  applied as a fixed rule to every future test example upon arrival.  
In our numerical experiments, the performance of \infoOSPCal\ is very similar to that of \infoOSP. In \S~\ref{subsec-finite-sample-theory} we  prove that it provides a finite-sample FCR guarantee for each future example separately; in \S~\ref{app-rates-and-asymptotics} we prove  that with  additional regularity  assumptions, asymptotically (in the training and calibration size) the  FCR and mFCR are controlled and the power is optimal, for any number of future examples.

\subsection{Implementation details}\label{subsec-implementationDetails}
Herein we provide details for carrying out \infoOSP. For carrying out \infoOSPCal, the implementation details are  similar and therefore omitted. In \S~\ref{subsec-computational-complexity}-\ref{subsec-computational-complexity-cal} we analyze the computational complexity of both algorithms. 

\subsubsection{Computing $\mathcal{I}_i$ (Step 1)} \label{subsec-computing-Ii}
Without additional assumption, $\mathcal{I}_i = \mathcal{I}$ is always a valid choice. When $\mathcal{I}$ is large, structural assumptions on $\mathcal{I}$ and $w$ may allow reducing $\mathcal{I}_i$.

For example, consider the setting with classification outcomes where  $w(C) = f(|C|)$ is a non-increasing function of the set cardinality, and that $\mathcal I$ is  defined by excluding certain classes and/or certain cardinalities. We denote this specific setting  as the {\it cardinality-based case}. Note that in this setting  we address  notions of  informativeness that are very natural, since interest often  lies only in a subset of classes, or only in prediction sets that are not too large.  Consider a non-excluded cardinality $j$ and any $i \in [n+m]$. All sets of size $j$ have the same weight $w(j)$, and their corresponding lines are $\ell_{i,C}(\mu) = \left(w(j) + \mu\right) \hat{\mP}(Y \in C \mid X = X_i) - \mu (1 - \alpha)$. Hence, out of the sets of size $j$, only the set $C$ that maximizes $\hat{\mP}(Y \in C \mid X = X_i)$ can appear on $\mathcal{U}_i(\mu)$. This set must consist of the top $j$ non-excluded classes ranked by $\hat{\mP}(Y = k \mid X = X_i)$.
Therefore we can construct $\mathcal{I}_i$ as follows. For any $i \in [n+m]$, sort the non-excluded classes by $\hat{\mP}(Y = k | X = X_i)$. For any non-excluded cardinality $j$, add to $\mathcal{I}_i$ the set with the top $j$ classes by this order. The resulting collection $\mathcal{I}_i$ has size at most~$K$.

Under other assumptions on $\mathcal{I}$ and $w$, constructing small $\mathcal{I}_i$ may require different techniques, such as dynamic programming. In general, without any structural constraints on $\mathcal{I}$ or $w$, a reduction may not be possible, as all $|\mathcal I|$ candidates may appear on the upper envelope.  For $K$ classification outcomes, $|\mathcal{I}|$ may be exponential in $K$, since $|\mathcal I|\leq 2^K-1$.

\subsubsection{Computing $M$ and iterating over candidates (Steps 2-3)} \label{subsec-computing-M}
We propose two methods for defining and computing $M$ and for iterating over its candidates. As shown in Section~\ref{subsec-computational-complexity}, Method~2 has a better asymptotic running time. However, Method~1 is more straightforward and easier to implement in code.

\paragraph*{Method 1 - computing all intersections:} Let $M$ be the set of all intersections among the lines $\{\ell_{i,C}(\mu)\}_{i \in [n], C\in \mathcal{I}_i}$ as well as their zero crossings (see illustration in Figure~\ref{fig-envelope} of all points where the lines intersect or cross zero). It is sufficient to include only the lines of the calibration points $i \in [n]$, since $\widehat{FCP}^{\mu}$ can decrease only at these points. To see this, note that the numerator of $\widehat{FCP}^{\mu}$ depends only on the calibration set, while the denominator is non-increasing in $\mu$ (as $\hat{D}^{\mu}(X_{n+i})$ is non-increasing in $\mu$). The set $M$ is then
\begin{align*}
M &= \left\{\frac{w(C)\hat{\mP}(Y \in C \mid X = X_i)}{1 - \alpha - \hat{\mP}(Y \in C \mid X = X_i)} : i \in [n],\, C \in \mathcal{I}_i, \hat{\mP}(Y \in C_1 \mid X = X_i) < 1 - \alpha\right\} \\
&\cup \Bigg\{\frac{w(C_2) \hat{\mP}(Y \in C_2 \mid X = X_i) - w(C_1) \hat{\mP}(Y \in C_1 \mid X = X_i)}{\hat{\mP}(Y \in C_1 \mid X = X_i) - \hat{\mP}(Y \in C_2 \mid X = X_i)} \\
&\hspace{0.7cm}: i \in [n],\, C_1,C_2 \in \mathcal{I}_i,\, w(C_1) < w(C_2),\, \hat{\mP}(Y \in C_1 \mid X = X_i) > \hat{\mP}(Y \in C_2 \mid X = X_i)\Bigg\}. 
\end{align*}
Then, for any $\mu \in M$, compute $\hat C^{\mu}(X_i)$ and $\hat D^{\mu}(X_i)$ by iterating over the sets in $\mathcal{I}_i$ for each $i \in [n+m]$.

\paragraph*{Method 2 - traversing the upper envelopes:} We can explicitly compute the upper envelope $\mathcal{U}_i$ for each $i \in [n+m]$ using a standard divide and conquer algorithm  (see illustration in Figure \ref{fig-envelope}; a description of the algorithm is provided, for example, in Algorithm 6.2.1  of \citet{SharirAgarwal1995}).
We refer to the breakpoints and zero crossings of the upper envelopes collectively as \emph{event points}, and let $M$ be the set of all event points of all these upper envelopes. We can efficiently implement step (3) by traversing the upper envelopes and updating $\hat C^{\mu}(X_i)$, $\hat D^{\mu}(X_i)$, and $\widehat{FCP}^{\mu}$ incrementally. To do so, we maintain a minimum heap with the leftmost event point that was not handled yet from each envelope. We initialize $\hat C^0(X_i)$, $\hat D^0(X_i)$ using the leftmost line segment in $\mathcal{U}_i$ for each $i \in [n+m]$, compute $\widehat{FCP}^0$, and insert the first event point of each envelope into the heap. We iteratively pick the next leftmost event point. We update $\hat C^{\mu}(X_i)$ and $\hat D^{\mu}(X_i)$ for the index $i$ that contributed this event point. The update is carried out using  the adjacent line segment in $\mathcal{U}_i$. We incrementally update $\widehat{FCP}^{\mu}$ using the difference for index $i$. After handling an event point, we remove it from the minimum heap and insert the following event point from the same envelope.

\subsubsection{The computational complexity of \infoOSP}\label{subsec-computational-complexity}

We first analyze each step and then the overall complexity. Let $S = \max_{i \in [n+m]} |\mathcal{I}_i|$. 

\textbf{Step 1.}
If no reduction is performed and we take $\mathcal{I}_i = \mathcal{I}$, no computation is needed and $S = |\mathcal{I}|$. In the cardinality-based case, the reduction to at most K candidate sets per example described in  \S~\ref{subsec-computing-Ii} yields $S \leq K$. It requires sorting $K$ classes for each $i \in [n+m]$, therefore it takes $O\left(K\log(K)(n+m)\right)$ time.

\textbf{Steps 2-3 via Method 1.}
For each $i \in [n]$, computing the pairwise intersections and zero crossings of the $|\mathcal{I}_i|$ lines takes $O(|\mathcal{I}_i|^2)$ time. Over all calibration points, this gives $|M| = O(S^2 n)$ candidates, computed in $O(S^2 n)$ time. Sorting $M$ takes $O\left(S^2 n \log(S^2 n)\right)$. For each candidate $\mu \in M$, computing $\hat C^{\mu}(X_i)$ for all $i \in [n+m]$ by iterating over $\mathcal{I}_i$ takes $O(S(n+m))$, and computing $\hat D^{\mu}(X_i)$ and $\widehat{FCP}^{\mu}$ takes $O(n+m)$. The total running time of Steps~2-3 is therefore $O(S^3 n(n+m))$.

\textbf{Steps 2-3 via Method 2.}
The upper envelope of $r$ lines can be computed in $O(r\log(r))$ time using a standard divide and conquer algorithm, and consists of at most $r$ segments (see, e.g., Theorem 6.1 for Algorithm 6.2.1 in \citet{SharirAgarwal1995}). Thus computing all $n+m$ upper envelopes takes $O\left(S\log(S)(n+m)\right)$ time, and $|M| = O(S(n+m))$. At each iteration, the incremental update of $\hat C^{\mu}(X_i)$, $\hat D^{\mu}(X_i)$, and $\widehat{FCP}^{\mu}$ requires $O(1)$ operations, and updating the minimum heap requires $O(\log(n+m))$ operations. The total running time of Steps~2-3 is therefore $O\left(S(n+m)\left(\log(S) + \log(n+m)\right)\right)$.

\textbf{Overall complexity.}
With $\mathcal{I}_i = \mathcal{I}$ for all $i \in [n+m]$, the running time is $O\left(|\mathcal{I}|^3 n(n+m)\right)$ with Method~1, or $O\left(|\mathcal{I}|(n+m)(\log(|\mathcal{I}|\right) + \log(n+m)))$ with Method~2.
In the cardinality-based case, with the reduction to at most K candidate sets per example described in \S~\ref{subsec-computing-Ii}, the total running time is $O\left(K^3 n(n+m)\right)$ with Method~1, or $O\left(K(n+m)(\log(K) + \log(n+m))\right)$ with Method~2.

\subsubsection{The computational complexity of \infoOSPCal}\label{subsec-computational-complexity-cal}

\hspace*{\parindent}\textbf{Calibration phase complexity.}
The key difference from \infoOSP\ is that steps 1-4 of the calibration phase use only the calibration sample, with the FCP estimated by $\widehat{FCP}_{cal}^{\mu}$. Therefore, a similar analysis shows that, with $\mathcal{I}_i = \mathcal{I}$ for all $i \in [n]$, the running time is $O\left(|\mathcal{I}|^3 n^2\right)$ with Method~1, or $O\left(|\mathcal{I}|n(\log(|\mathcal{I}|) + \log(n))\right)$ with Method~2. In the cardinality-based case, with the reduction to at most K candidate sets per example described in \S~\ref{subsec-computing-Ii}, the total running time is $O\left(K^3 n^2\right)$ with Method~1, or $O\left(Kn(\log(K) + \log(n))\right)$ with Method~2.

\textbf{Prediction phase complexity.}
The computational complexity for predicting at a test point $X_{n+i}$ is $O(|\mathcal{I}_i|)$, which is $O(K)$ in the cardinality-based case and $O(|\mathcal{I}|)$ otherwise.

\subsection{ \infoOSP\ and \infoOSPCal\ as  threshold selection algorithms}\label{sec:threshold}

The threshold respresentations below rely on a 
 nestedness property: $\hat C^{\mu_1}(X_i)\subseteq \hat C^{\mu_2}(X_i)$ for all $\mu_1<\mu_2$.  Consequently,  if $Y_i\notin \hat C^{\mu_2}(X_i)$ then $Y_i\notin \hat C^{\mu_1}(X_i)$ for all $\mu_1<\mu_2$. 
This assumption, formalized next, 
 is sufficient for the theoretical guarantees of both algorithms, established in \S~\ref{subsec-finite-sample-theory}. It is also natural in view of
the threshold-selection step: under nestedness, choosing the smallest $\mu$ for which the estimated FCP is at most $\alpha$, as in Step 3(d)
of Algorithms~\ref{algo:OGinfoSP} and~\ref{algo:OGinfoSP-Cal}, reflects
the aim of performing inference using prediction sets that are as narrow as possible while keeping the estimated proportion of miscoverages among the reported sets at most $\alpha$.

\begin{ass}\label{ass-martingale}
For any $x \in \mathcal{X}$ and any $\mu_1 < \mu_2$, $\hat C^{\mu_1}(x) \subseteq \hat C^{\mu_2}(x)$.
\end{ass}

See \S~\ref{app-ass-martingale-details} for details on settings of interest (e.g., when interest lies in prediction sets that do not exceed  a certain cardinality or that exclude a specific range of outcomes) where Assumption \ref{ass-martingale} is known to hold. 

Under Assumption \ref{ass-martingale},  we can describe the two algorithms using the following key statistics that  are well-defined for $(X,Y)$ from the joint distribution $P_{XY}$ on $\mathcal X\times \mathcal Y$, 
with the convention that the minimum of an empty set is $\infty$:
\begin{eqnarray}
    && \lambda(X,Y) : = \min \{\mu: Y\in \hat C^{\mu}(X) \textrm{ or } \hat D^{\mu}(X)=0\},  \nonumber \\
    && \Lambda(X):=\min \{\mu:  \hat D^{\mu}(X)=0\}. \label{eq-keystatistics}
\end{eqnarray}
The minima necessarily exist when the sets are non-empty by right continuity\footnote{Specifically, the tie breaking rule in \eqref{eq-C-mu-procedure} and the strict inequality in \eqref{eq-D-mu-procedure} ensure that the mappings $\mu\mapsto\widehat C^\mu(x)$ and
$\mu\mapsto\widehat D^\mu(x)$ are right-continuous, as formalized in Proposition \ref{prop-right-continuous}; so whenever a defining set is non-empty, its infimum is attained.} of the events considered as a function of $\mu$.

Computing the statistics above is the main algorithmic challenge. With these statistics  we can describe the method as a familiar threshold-selection rule. Specifically, note that 
 $\lambda(X,Y)>\mu$ if and only if $Y\notin \hat C^{\mu}(X)$ and $\hat D^{\mu}(X)=1$, and that $\Lambda(X)>\mu$ if and only if $\hat D^{\mu}(X)=1$. Hence,  we can rewrite in terms of the key statistics Step 3(c) of  Algorithm \ref{algo:OGinfoSP}  as  
 \begin{equation}\widehat{FCP}^{\mu}  = \frac{\frac{1}{n+1} \left(1 + \sum_{i=1}^n \mI\{\lambda(X_i,Y_i)>\mu\}\right)}{\frac{1}{m} \left(1\lor\sum_{i=1}^m \mI\{\Lambda(X_{n+i})>\mu\}\right)}.\label{eq-forappendix} \end{equation}
 and Step 3(c) of  Algorithm \ref{algo:OGinfoSP-Cal}  as  
 \begin{equation}\widehat{FCP}^{\mu}_{cal}  = \frac{ 1 + \sum_{i=1}^n \mI\{\lambda(X_i,Y_i)>\mu\}}{1+ \sum_{i=1}^n \mI\{\Lambda(X_{i})>\mu\}}.\label{eq-forappendix}\end{equation}
Sorting $(\lambda(X_i,Y_i))_{i\in[n]}$ in increasing order, we denote by $(\lambda_{(i)})_{i\in[\tilde n]}$ the sorted values that are not infinite, where $\tilde n<n$ if $n-\tilde n$ of $(\lambda(X_i,Y_i))_{i\in[n]}$ are infinite. Then 
 Algorithm \ref{algo:OGinfoSP} provides the same output as the following threshold-selection algorithm: 
 \begin{enumerate}
\item  Let $i_{\alpha} = \min \left\lbrace i\in [\tilde n]: \frac{\frac{1}{n+1} \left(1 + \sum_{j=1}^n \mI\{\lambda(X_j, Y_j)>\lambda_{(i)}\}\right)}{\frac{1}{m} \left(1\lor\sum_{j=1}^m \mI\{\Lambda(X_{n+j})>\lambda_{(i)}\}\right)}\leq \alpha \right \rbrace$ if the set is non-empty, and $\tilde n+1$ otherwise; 
\item 
 If $i_{\alpha}<\tilde n+1$, denote $\thone:= \lambda_{(i_{\alpha})} $ (this is the value in \eqref{eq-thone}) and construct prediction sets $\left\{ \hat C^{\thone}(X_{n+i}) : i \in [m], \Lambda(X_{n+i})>\thone \right\}$; otherwise, construct no prediction sets.  
 \end{enumerate}
 Similarly, Algorithm \ref{algo:OGinfoSP-Cal} provides the same output as the following threshold-selection algorithm for  test point $X_{n+j}$: 
 \begin{enumerate}
\item  Let $i^{cal}_{\alpha} = \min \left\lbrace i\in [\tilde n]: \frac{1 + \sum_{j=1}^n \mI\{\lambda(X_j, Y_j)>\lambda_{(i)}\}}{1+ \sum_{j=1}^n \mI\{\Lambda(X_{j})>\lambda_{(i)}\}}\leq \alpha \right \rbrace$ if the set is non-empty, and $\tilde n+1$ otherwise; 
\item 
 If $i^{cal}_{\alpha} < \tilde n + 1$, denote $\thtwo:= \lambda_{(i^{cal}_{\alpha})} $ (this is the value in \eqref{eq-thtwo}) and construct prediction set $ \hat C^{\thtwo}(X_{n+j})$ if $\Lambda(X_{n+j})>\thtwo $; otherwise, construct no prediction set.  
 \end{enumerate}
See \S~\ref{prop-equivalent procedures-proof} for a formal proof that the threshold-selection algorithms above are indeed equivalent to the original Algorithms \ref{algo:OGinfoSP} and \ref{algo:OGinfoSP-Cal}.

\section{Finite-sample theoretical guarantees}\label{subsec-finite-sample-theory}

Our main theoretical result in Theorem \ref{thm-FCR-control} relies on the following lemma, which we present as a standalone statement since it may be of independent interest. Unless given in the main text, proofs for this section are provided in
\S~\ref{app-theoretical-results}.
\begin{lem}\label{lemma-martingale}
Let $x\in \mathcal X$, $y\in \mathcal Y$, and $\mu \in [0, \infty)$.  Let $f_{\mu}(x,y)\in \{0,1\}$ and $g_{\mu}(x)\in \{0,1\}$ be  indicator functions satisfying the following conditions  for every $(x,y)\in \mathcal X \times \mathcal Y$: 
\begin{enumerate}
    \item $f_{\mu}(x,y)$ is non-increasing in $\mu$ and bounded from above by $g_{\mu}(x)$ on $[0,\infty)$.
    \item $\mu \mapsto f_{\mu}(x,y)$ and $\mu \mapsto g_{\mu}(x)$ are right-continuous on $[0, \infty)$.
    \item $f_{\infty}(x,y) = g_{\infty}(x) = 0$.
\end{enumerate} 
Let $\{(X_i,Y_i)\}_{i=1}^{n+m}$ be iid random variables with common distribution $P_{XY}$. Then
	$$err: = \mE\left[\frac{\sum_{i=1}^m f_{\mu_\alpha}(X_{n+i},Y_{n+i})}{1 \lor \sum_{i=1}^m g_{\mu_\alpha}(X_{n+i})}\right] \leq \alpha$$
	where
	$$\mu_\alpha = \min\left\lbrace \mu\geq 0 : \frac{\frac{1}{n+1}\left(\sum_{i=1}^n f_{\mu}(X_i,Y_i) + 1\right)}{\frac{1}{m}\left(1 \lor \sum_{i=1}^m g_{\mu}(X_{n+i})\right)} \leq \alpha \right\rbrace,$$
	or $\mu_\alpha = \infty$ if no such $\mu$ exists.
\end{lem}

Taking $f_{t}(x,y) = \mI\{h_{xy}(x,y)>t\}$, $g_{t}(x) = \mI\{h_x(x)>t\}$ in Lemma \ref{lemma-martingale} yields the following  score-threshold formulation. This corollary will be useful in illustrative applications in \S~\ref{sec-exisitingprocedures}. 
\begin{coro}\label{coro-mainLemma}
   Let $h_{xy}:\mathcal X\times\mathcal Y \rightarrow \mR$ and $h_x: X\rightarrow  \mR$ be two (measurable) score functions satisfying $h_{xy}(x,y)\leq h_x(x)$ for all $(x,y)\in \mathcal X\times\mathcal Y$.  Let $\{(X_i,Y_i)\}_{i=1}^{n+m}$ be iid random variables.  Then 
	$$err: = \mE\left[\frac{\sum_{i=1}^m \mI\{h_{xy}(X_{n+i},Y_{n+i})> t_\alpha\}}{1 \lor \sum_{i=1}^m \mI\{h_x(X_{n+i})> t_\alpha\}}\right] \leq \alpha$$
	where
	$$t_\alpha = \min\left\lbrace t\geq 0 : \frac{\frac{1}{n+1}\left(\sum_{i=1}^n \mI\{h_{xy}(X_i,Y_i)> t\} + 1\right)}{\frac{1}{m}\left(1 \lor \sum_{i=1}^m \mI\{h_x(X_{n+i})> t\}\right)} \leq \alpha \right\rbrace,$$
	or $t_\alpha = \infty$ if no such $t$ exists.
\end{coro}

In addition to the lemma,  we need the next two propositions for proving Theorem \ref{thm-FCR-control}.

\begin{prop} \label{prop-monotoneD-procedure}
For any $x \in \mathcal{X}$, let $\hat C^{\mu}(x)$ and $\hat D^{\mu}(x)$ be defined as in \eqref{eq-C-mu-procedure} and \eqref{eq-D-mu-procedure}, respectively.   As $\mu\geq 0$ is increasing: 
1.~$\hat{\mP}(Y\in \hat C^{\mu}(x) \mid X=x)$ is non-decreasing;
2.~$w(\hat C^{\mu}(x))\hat{\mP}(Y\in \hat C^{\mu}(x) \mid X=x)$ is non-increasing;
3.~$\hat D^{\mu}(x)$ is non-increasing; 
4.~$w(\hat C^{\mu}(x))\hat{\mP}(Y\in \hat C^{\mu}(x) \mid X=x)\hat D^{\mu}(x)$  is non-increasing; 
5.~$(1 - \hat{\mP}(Y \in \hat C^{\mu}(x) \mid X = x) -\alpha)\hat D^{\mu}(x)$ is non-increasing. 
\end{prop}

\begin{prop} \label{prop-right-continuous}
 $\forall \ x \in \mathcal{X}$, $\mu \mapsto \hat C^{\mu}(x)$ and $\mu \mapsto \hat D^{\mu}(x)$ are right-continuous on $[0,\infty)$.
\end{prop}

The proof of Proposition \ref{prop-monotoneD-procedure}  is the same as that of Proposition \ref{prop-monotoneD}, and is therefore omitted. 
The only differences are that here ties in $\hat C^{\mu}(x)$ are broken in favor of smaller weights (whereas in \S~\ref{sec-optimal} they are broken in favor of larger weights), and the decision rule $\hat D^{\mu}(x)$ selects when the upper envelope is strictly positive (whereas in \S~\ref{sec-optimal} it selects when the upper envelope is nonnegative).  These differences do not affect the monotonicity results.

Our main theorem, showing  finite-sample FCR control of procedure \infoOSP\, is as follows.

\begin{thm}\label{thm-FCR-control}
	If Assumption~\ref{ass-martingale} is satisfied, then the \infoOSP\ procedure satisfies
	$$FCR:=\mE\left[\frac{\sum_{i=1}^m\mI\left\lbrace Y_{n+i}\notin \hat C^{\mu_\alpha}(X_{n+i})\right\rbrace \hat D^{\mu_\alpha}(X_{n+i})}{1\lor\sum_{i=1}^m \hat D^{\mu_\alpha}(X_{n+i})}\right] \leq \alpha,$$
	where $$\mu_\alpha = \min\left\lbrace \mu\geq 0:  \frac{\frac{1}{n+1}\left(\sum_{i=1}^n \mI\left\lbrace Y_i\notin \hat C^{\mu}(X_i)\right\rbrace \hat D^{\mu}(X_i)+1\right)}{\frac{1}{m}\left(1\lor \sum_{i=1}^m \hat D^{\mu}(X_{n+i})\right)} \leq \alpha \right\rbrace$$
	and we set $\mu_\alpha = \infty$ with $\hat D^{\infty} \equiv 0$ if no such $\mu$ exists.
\end{thm}

\begin{proof} 
Let $f_{\mu}(x,y) = \mI\{y \notin \hat C^{\mu}(x)\} \hat D^{\mu}(x)$ and $g_{\mu}(x) = \hat D^{\mu}(x)$ for any $\mu \in [0, \infty)$, and $f_{\infty}(x,y) = g_{\infty}(x) = 0$.
Let $x,y \in \mathcal{X} \times \mathcal{Y}$. The conditions of Lemma~\ref{lemma-martingale} are satisfied:

\begin{enumerate}
\item The inequality $f_{\mu}(x,y) \le g_{\mu}(x)$ is immediate. Let $\mu_1 < \mu_2$. By Assumption~\ref{ass-martingale}, $\hat C^{\mu_1}(x) \subseteq \hat C^{\mu_2}(x)$, which implies that $\mI\{y \notin \hat C^{\mu_1}(x)\} \geq \mI\{y \notin \hat C^{\mu_2}(x)\}$. Thus, $\mI\{y \notin \hat C^{\mu}(x)\}$ is non-increasing in $\mu$. From Proposition~\ref{prop-monotoneD-procedure}, $\hat D^{\mu}(x)$ is also non-increasing in~$\mu$. Therefore, $f_{\mu}(x,y)$ is non-increasing in $\mu$.

\item By Proposition~\ref{prop-right-continuous}, $\mu \mapsto \hat C^{\mu}(x)$ and $\mu \mapsto \hat D^{\mu}(x)$ are right-continuous on $[0,\infty)$. It follows that $\mu \mapsto \mI\{y \notin \hat C^{\mu}(x)\}$ is also right-continuous, and consequently $\mu \mapsto f_{\mu}(x,y)$ and $\mu \mapsto g_{\mu}(x)$ are right-continuous on $[0,\infty)$.

\item $f_{\infty}(x,y) = g_{\infty}(x) = 0$ by definition.
\end{enumerate}
The result now follows by applying Lemma~\ref{lemma-martingale}.
\end{proof}

We conclude with a finite-sample guarantee for \infoOSPCal\ when applied to a single future test point. 

\begin{prop} \label{prop-fcr-single-example}
Under Assumption~\ref{ass-martingale}, the \infoOSPCal\ procedure satisfies 
$$\mE\left[ \frac{\mI\{Y_{n+j} \notin \hat C^{\thtwo}(X_{n+j})\} \hat D^{\thtwo}(X_{n+j})}{1\lor \hat D^{\thtwo}(X_{n+j})}\right] = \mP\left(Y_{n+j} \notin \hat C^{\thtwo}(X_{n+j}), \hat D^{\thtwo}(X_{n+j}) = 1\right)\leq \alpha.$$
\end{prop}

\section{Application to classification with class exclusion} \label{sec-exisitingprocedures}

Fix a nonempty set $\mathcal K_1\subseteq [K]$ of non-excluded classes.  The analyst is interested in reporting the class of an example only if it is in $\mathcal K_1$.  Thus, 
$$\mathcal I = \{\{k\}: k\in \mathcal K_1 \}, \ w(\{k\})=1 \ \forall k \  \in \mathcal K_1.$$
If an example's true  class  is outside $\mathcal K_1$, any reported class is necessarily incorrect since classes in $[K]\setminus \mathcal K_1$ are never reported. 
  
When $\mathcal K_1 = [K]$, no classes are excluded, and the problem reduces to ordinary classification with abstention. The resulting procedure below coincides  with that of \cite{FSRzhao2023}. When $|\mathcal K_1|=1$, the procedure becomes  a one-versus-rest detection rule that can be interpreted as  novelty (or out-of-distribution) detection.  
 
Let $\hat y(x) \in \arg\max_{k\in \mathcal K_1}\hat{\mP}(Y = k \mid X = x)$ and $s(x):= \hat{\mP}(Y = \hat y(x) \mid X = x)$, where ties are resolved using a fixed rule.
The oracle-guided decision rule using \eqref{eq-C-mu-procedure} and \eqref{eq-D-mu-procedure} reduces to
\begin{align*}
\hat C^{\mu}(X_i) &= \arg\max_{k \in \mathcal K_1} \left(1 + \mu\right) \hat{\mP}(Y =k \mid X = X_i) =  \{\hat y(X_i)\} \\
\hat D^{\mu}(X_i) &=  \mI\left\lbrace (1+\mu) s( X_i) - \mu(1-\alpha) > 0\right\rbrace = \mI\left\lbrace  s(X_i) > t(\mu)\right\rbrace,
\end{align*}
where $t(\mu) := \frac{\mu(1-\alpha)}{1+\mu} \in [0,1-\alpha)$ is strictly increasing in $\mu$. Therefore, thresholding $s(X_i)$ at a specific $t<1-\alpha$ corresponds to  \infoOSP\ for the classification with class exclusion or abstention setting.

Since $\hat C^{\mu}(X_i)=\{\hat y(X_i)\}$ does not depend on $\mu$,  Corollary \ref{coro-mainLemma} can be applied directly with $h_{xy}(X,Y) = \mI\{Y\neq \hat y(X)\}s(X)$ and $h_x(X) = s(X)\in [0,1].$ The resulting procedure searches for the threshold in the full range $t\in [0,1]$, and  thresholds $t\geq 1-\alpha$ constitute a valid extension of  \infoOSP\ for this setting. 

\begin{coro}\label{coro-informativeClassification}
  For $t\in [0,1]$, define  $\widehat{FCP}(t) = \frac{\frac{1}{n+1} \left(1+\sum_{i=1}^n \mI\left\{Y_i \neq \hat{y}(X_i), s(X_i) > t\right\} \right)}{\frac{1}{m} \left(1\lor \sum_{i=1}^m \mI\left\{s( X_{n+i}) > t\right\}\right)}$. 
 Let $t_\alpha = \min\left\{t \in [0,1] : \widehat{FCP}(t) \leq \alpha\right\}$, or $t_{\alpha} = \infty$ if no such $t$ exists.
For each $i\in [m]$, report $\{\hat{y}(X_{n+i})\}$ if $s(X_{n+i}) > t_\alpha$, and otherwise abstain. Then $$FCR: = \mE\left[ \frac{\sum_{i=1}^m \mI\left\{Y_{n+i} \neq \hat{y}(X_{n+i}), s(X_{n+i}) > t_{\alpha}\right\}}{1\lor \sum_{i=1}^m \mI\left\{s( X_{n+i}) > t_{\alpha}\right\}}\right] \leq \alpha.$$
\end{coro}

Because every reported prediction set is a singleton, this FCR is also the FDR of the reported classifications. The procedure can  be viewed through the lens of multiple testing using statistics that resemble conformal $p$-values.  For each test point $i\in[m]$, consider the (random) null hypothesis $H_i:\; Y_{n+i}\neq \hat y(X_{n+i})$.  Using the score $S(x,y) = \mI\{y\neq \hat y(x)\}s(x)$,  define
$p_i \;:=\; \frac{1}{n+1}\Bigl(1+\sum_{j=1}^n \mI\{S(X_j,Y_j)\ge s(X_{n+i})\}\Bigr).$
It can be shown that applying the level-$\alpha$ BH procedure  to $p_1,\ldots,p_m$ provides the same inference as Corollary \ref{coro-informativeClassification}: 
\[
\{i\in[m]: s(X_{n+i})>t_\alpha\}=\{i\in[m]: p_i\le \hat k\alpha/m\},
\]
where $\hat k$ is the number of rejections by BH. Thus, BH reports the same labels and inherits the FDR/FCR guarantee of Corollary \ref{coro-informativeClassification}. 

When $|\mathcal K_1|=1$ and the calibration sample consists only of observations from the complementary non-novel classes, the procedure in Corollary \ref{coro-informativeClassification} coincides with  the conformal novelty detection rule of \cite{mary2022semisupervised}. Their proof of  FDR control was based on a semi-supervised exchangeability argument \citep{bates2023testing, marandon2024adaptive} rather than on the iid argument used in  Corollary \ref{coro-informativeClassification}.  

For  online  classification with class exclusion and abstention,  Algorithm \ref{algo:OGinfoSP-Cal} reduces to computing the threshold  
$$t_\alpha^{cal} = \min\left\{t \in [0,1] : \frac{1+\sum_{i=1}^n \mI\left\{Y_i \neq \hat{y}(X_i), s(X_i) > t\right\}}{1+  \sum_{i=1}^n \mI\left\{s( X_{i}) > t\right\}} \leq \alpha\right\},$$ with $t_{\alpha}^{cal}=1$ if the set is empty,  and   reporting for examples that arrive online  $\hat{y}(X_{n+j})$ if $s(X_{n+j}) > t_\alpha^{cal}$  and otherwise abstain.

Our framework is more general for classification outcomes, since it   handles more general informative families $\mathcal I$ and weights $w(C), $ with proven finite-sample FCR control in Theorem \ref{thm-FCR-control}. In \S~\ref{sec-simulations} and \S~\ref{app-simulations} we consider such more general settings; examples of the setting in this section are shown in \S~\ref{sim-CIFAR}, \S~\ref{app-sim-CIFAR-K3} and \S~\ref{app-sim-CIFAR-K10}.

\section{Numerical experiments with classification outcomes}\label{sec-simulations}

We evaluate the empirical performance of \infoOSP\ and \infoOSPCal\ in classification outcomes with several notions of informativeness. We compare them with \texttt{InfoSP} and   \texttt{InfoSCOP} \citep{GazinHellerMarandonRoquain2025}, which adjust for informative selection and have  finite-sample FCR guarantees, and with the following unadjusted baselines. Classic conformal first constructs standard split conformal $1-\alpha$ prediction sets and then reports only those satisfying the informativeness criterion. In the CIFAR-10 experiments, we additionally consider the  method we call Top-Class, which reports the class with the largest estimated conditional probability whenever the corresponding singleton belongs to $\mathcal I$, and otherwise abstains. 
Classic conformal, \texttt{InfoSP}, and \texttt{InfoSCOP} use
the typical nonconformity score
$1-\widehat{\mathbb P}(Y=k\mid X)$ (see \S~\ref{app-sim-CIFAR-K3}-\ref{app-sim-CIFAR-K10} for use of an adaptive score with classic conformal).

We consider here the  setting with $K=10$ categorical outcomes, in which the prediction algorithm was trained on examples that come from the same distribution as  the calibration and test examples. In \S~\ref{sim-CIFAR},  we use part of the calibration sample to  calibrate the misspecified estimated conditional probabilities via 
vector scaling, as suggested in  \citet{guo2017calibration}. Additional experiments in  \S~\ref{app-simulations} consider:  Gaussian mixtures with $K=4$ classes, as well as with $K=10$ classes with label shift between the training and target distributions;   CIFAR-10 settings with $K=3$ classes, as well as  with $K=10$ classes that are  balanced or imbalanced. The label-shift correction is described in    \S~\ref{app-vector-scaling}.

The experiments assess both batch FCR control after informative selection, and resolution-adjusted power. 
For each replicate, let $D_i$ indicate whether a prediction set is reported for test example $i$, and let $C_i$ denote the corresponding
prediction set. We record the FCP, and the realized resolution-adjusted power corresponding to the choice $w(C)=1/|C|$:
$$ FCP = \frac{\sum_{i=1}^m \mI\{Y_{n+i}\notin C_i\}D_i}{1\vee\sum_{i=1}^m D_i}, \quad \widehat\Pi = \frac{1}{m} \sum_{i=1}^m \frac{\mI\{Y_{n+i}\in C_i\}D_i}{|C_i|}.$$
 The FCR and
resolution-adjusted power are estimated by averaging these quantities
over replicates. In the Gaussian-mixture experiments, we additionally report the mean number of reported prediction sets that contain the true outcome,
$ \sum_{i=1}^m \mI\{Y_{n+i}\in C_i\}D_i. $

Interestingly, in all our simulations the FCR and power of \infoOSP\ and \infoOSPCal\ are remarkably close. \infoOSPCal\ determines its threshold using only the calibration sample and can therefore  be deployed when the test examples arrive online. 
Compared to the unadjusted baselines, the advantage of \infoOSP\ and \infoOSPCal\ is clear: their  FCR  can be much higher than the nominal level $\alpha$. Importantly, we show that  \infoOSP\ and \infoOSPCal\ can have much better power than the \texttt{InfoSPP} and \texttt{InfoSPCOP}. 

Reproducible code is available at \url{https://github.com/israela/OGInfoSP}. The repository also includes an interactive viewer, where the user can choose which methods and configurations to compare for a rich set of  scenarios. The viewer displays the following metrics:  estimated resolution adjusted power,  average number selected,  average number of true selections, and  FCR.

\subsection{Experiments with bivariate normal mixtures }\label{sim-Gaussian}

We consider a $K=10$ mixture of bivariate normal distributions. One component is centered at the origin and the means of the remaining  nine components are equally spaced on a circle of radius SNR. All components have covariance matrix $I_2$. The class probabilities, common to the training, calibration, and test distributions, are 0.45 for the central class, 0.15 for one peripheral class, and 0.05 for each of the remaining eight classes. Instances of this data generation are depicted in \S~\ref{app-sim-Gaussian-K10-label-shift}, Figure \ref{app-sim-Gaussian-K10-datagen}. 

For each SNR $ \in \{0.1,0.2, \ldots, 5 \}$, we train a logistic regression model on 10,000 independent observations and hold it fixed across simulation iterations and methods. Informativeness is defined either by exclusion of the central class or by prediction-set size at most three. 

Figure \ref{fig:gaussian-K10-uneven} shows the results. Key observations regarding the FCR  are as follows.  The naive approach based on classic conformal fails to control the FCR at the nominal level across a broad range  of  signals.    \infoOSP\ and \infoOSPCal\ have FCR at most the nominal level $\alpha = 0.05$, with value approximately $\alpha$ for a wide range of SNR values, indicating little conservativeness. In contrast, the FCR of the procedures in \cite{GazinHellerMarandonRoquain2025} tends to be lower, especially in the setting of excluding a null class.

Key observations regarding the (resolution-adjusted) power are as follows. The procedures of \cite{GazinHellerMarandonRoquain2025} can be substantially  less powerful, especially when informativeness is defined by class exclusion.  Interestingly, classic conformal is also (slightly) less powerful than \infoOSP\ over a wide range of SNR values for which its FCR is below the nominal level, in the setting where  informativeness is the  exclusion of the central class. Unlike \infoOSP, the competing procedures do not jointly optimize the selection decision and the informative set to report. Instead, their suboptimal selection rules are induced by conformal sets constructed from standard nonconformity scores at a fixed or adjusted level, with an additional preprocessing step for \texttt{InfoSCOP}. This suboptimality is reflected at every SNR in a lower expected fraction of examples receiving a covering prediction set than under \infoOSP\ and \infoOSPCal.

Additional Gaussian-mixture results are presented in \S~\ref{app-simulations}. In \S~\ref{app-sim-Gaussian-K4}, the $K=4$ simulations yield the same qualitative conclusions. 
A further contribution is the integration of a label-shift correction, from the training distribution to the calibration and test distributions,  into our procedure. In \S~\ref{app-sim-Gaussian-K10-label-shift}, we evaluate this pipeline for $K=10$, showing that correcting the fitted class probabilities using a fraction of the calibration sample yields substantial power gains for \infoOSP\ and \infoOSPCal, whereas its effect on competing procedures is inconsistent across settings. 

\begin{figure}[!t]
    \centering

    \includegraphics[width=0.94\linewidth]
    {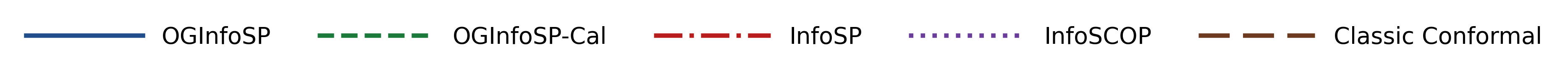}
    \par\vspace{-0.5em}

    \includegraphics[width=0.32\linewidth]
    {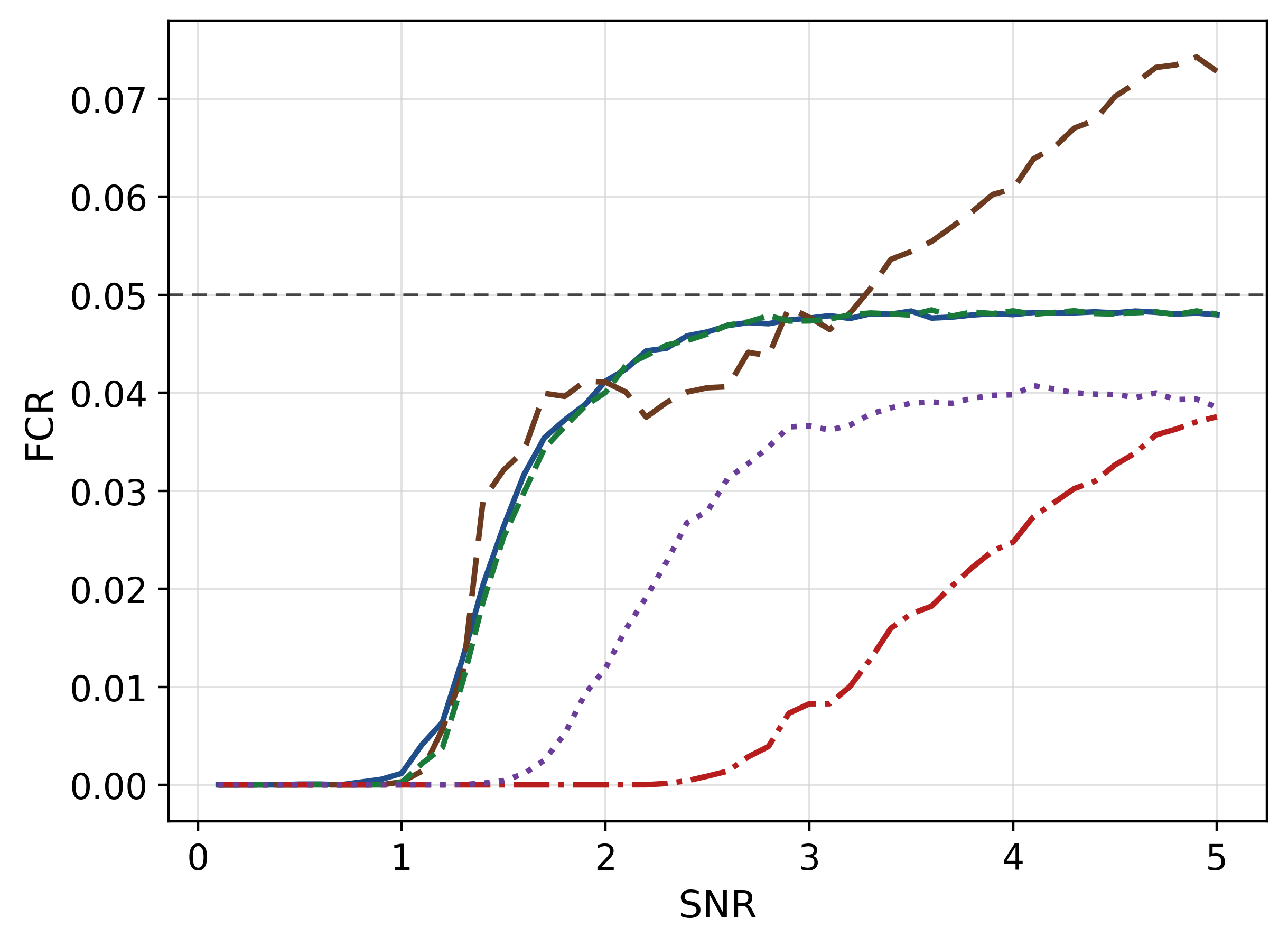}
    \hfill
    \includegraphics[width=0.32\linewidth]
    {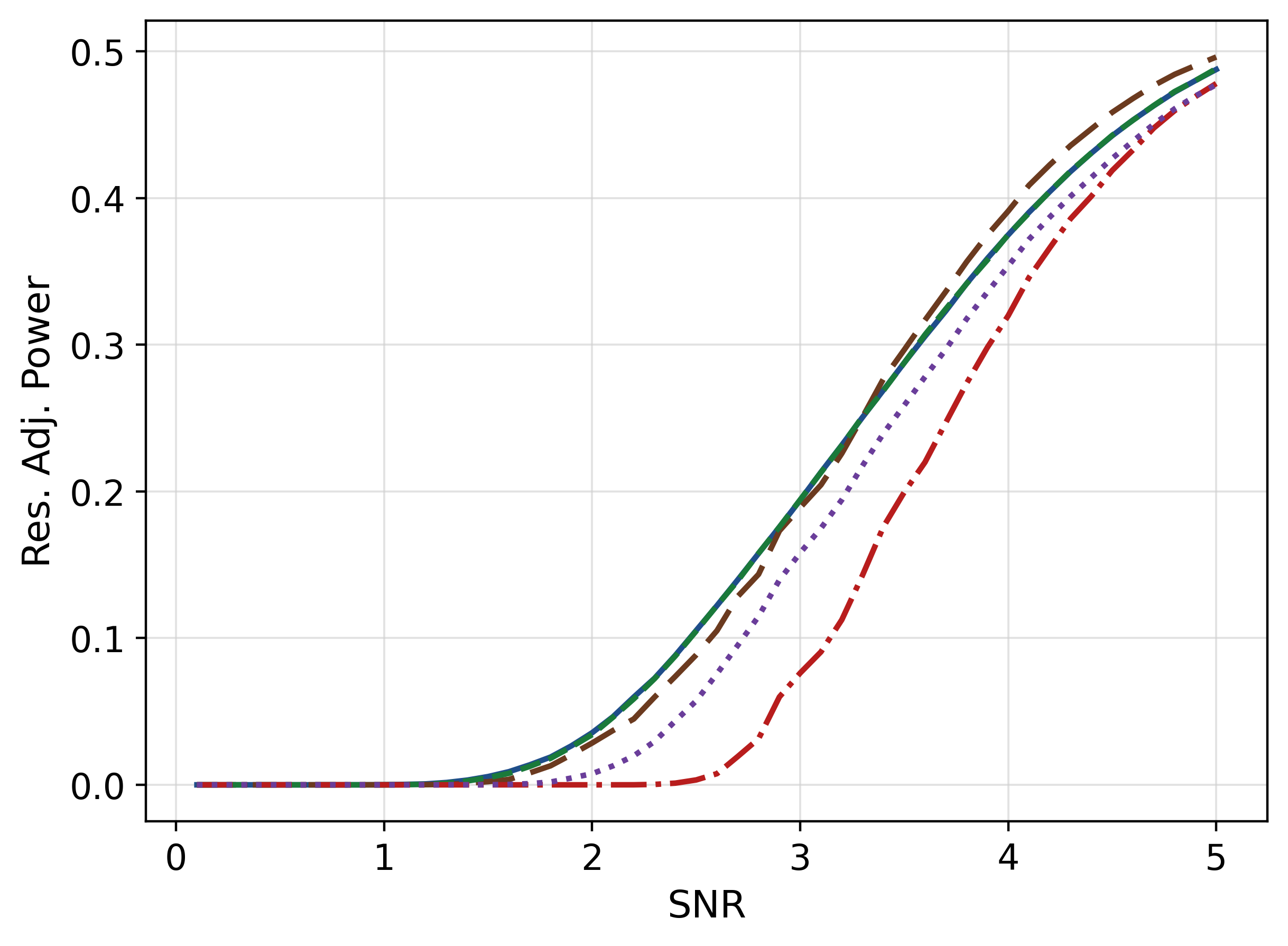}
    \hfill
    \includegraphics[width=0.32\linewidth]
    {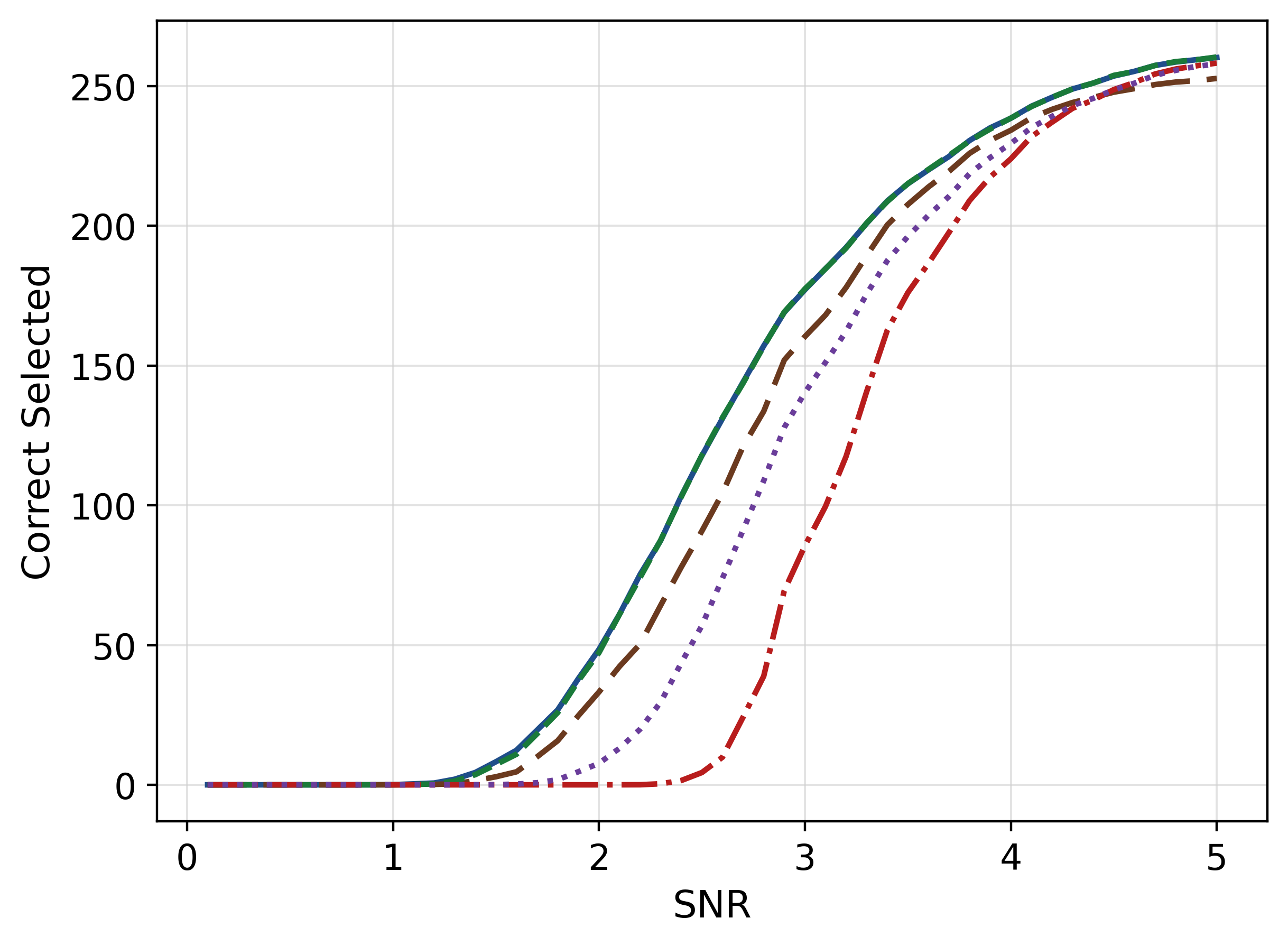}

    \par\vspace{-0.3em}
    {\small\textbf{(a) Exclusion of the central class}}
    \par\vspace{0.5em}

    \includegraphics[width=0.32\linewidth]
    {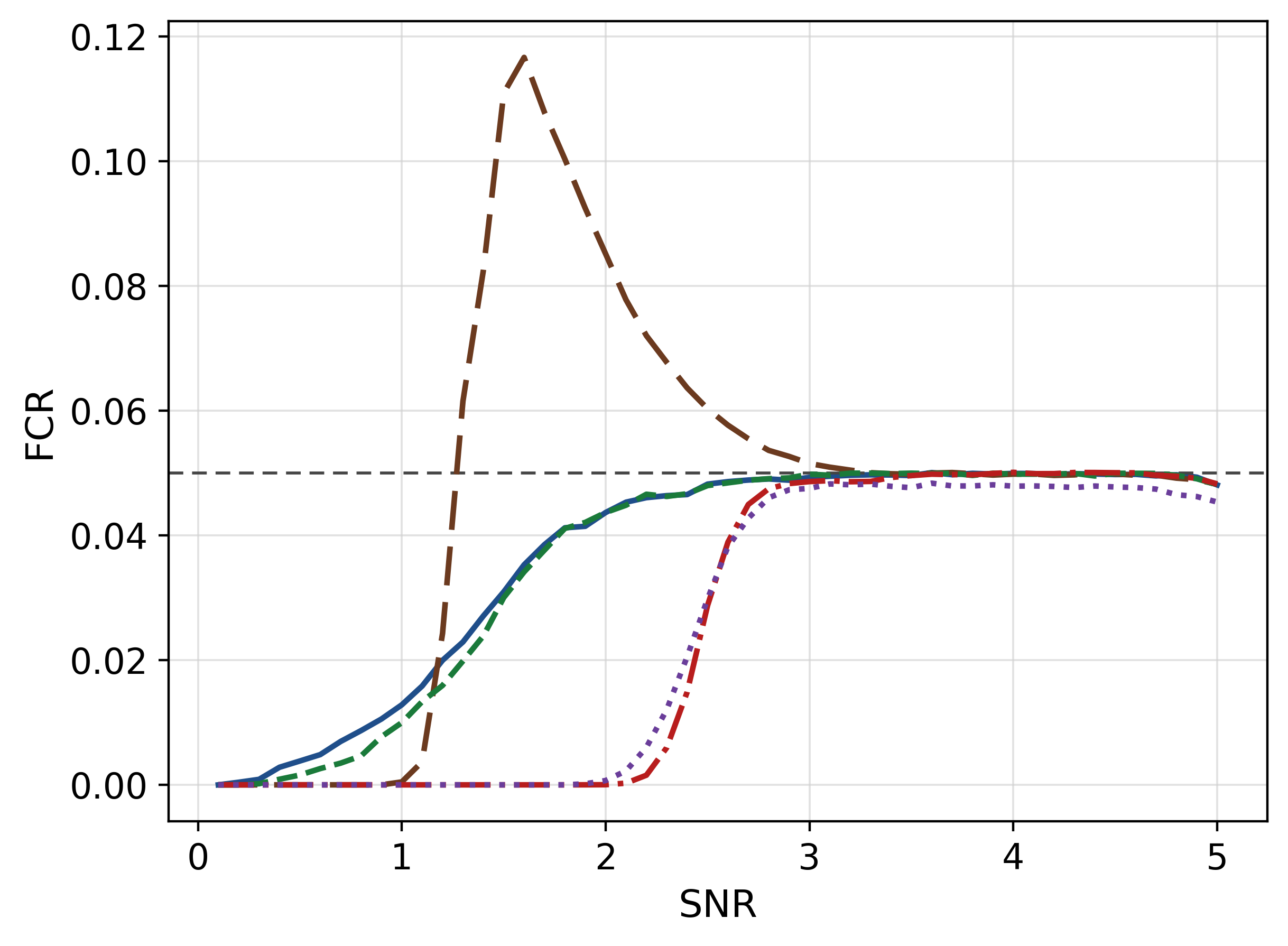}
    \hfill
    \includegraphics[width=0.32\linewidth]
    {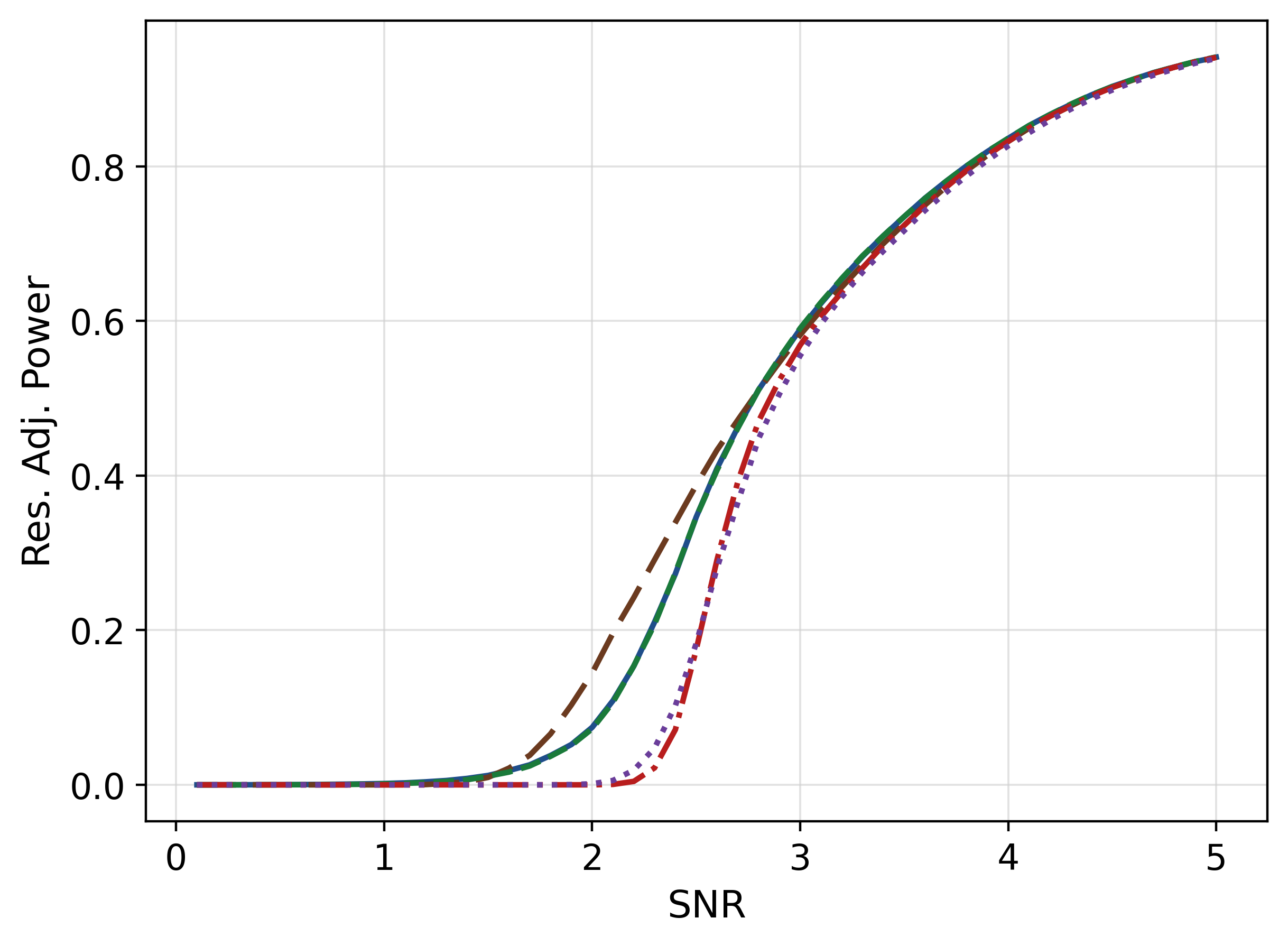}
    \hfill
    \includegraphics[width=0.32\linewidth]
    {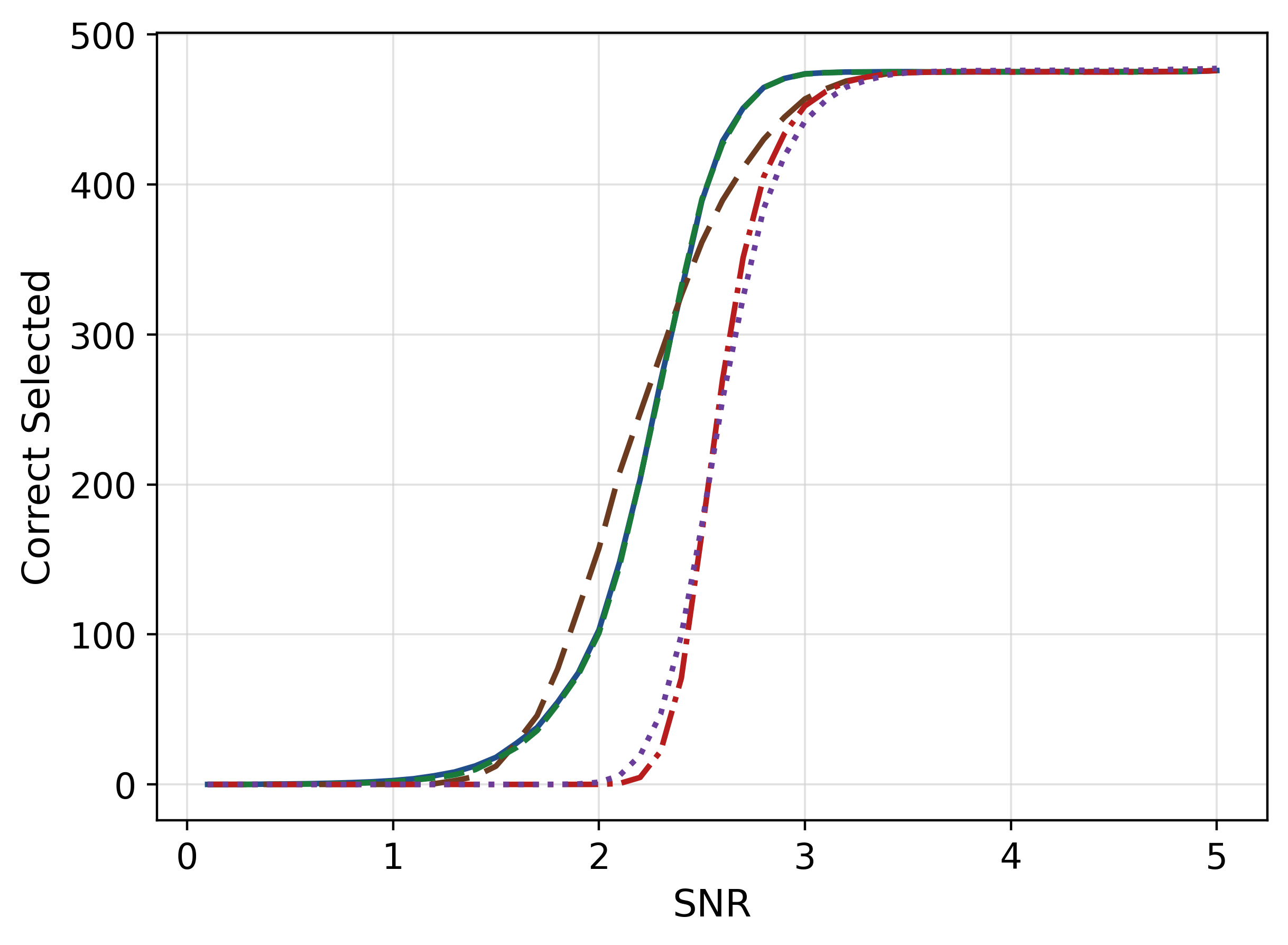}

    \par\vspace{-0.3em}
    {\small\textbf{(b) Prediction sets of size at most three}}

    \caption{Results for the  $K=10$ bivariate normal
    mixture detailed in \ref{sim-Gaussian}. The columns show the FCR, resolution-adjusted power, and
    mean number of selected prediction sets containing the true outcome.
    Informativeness is defined by exclusion of the central class in (a) and
    prediction-set size at most three in (b). Here $n=m=500$,
    $\alpha=0.05$, and results are averaged over $10{,}000$ iterations; the
    horizontal line marks the nominal FCR level.}
    \label{fig:gaussian-K10-uneven}
\end{figure}

\subsection{Image-classification experiments  on CIFAR-10}\label{sim-CIFAR}

As in \cite{GazinHellerMarandonRoquain2025},  we illustrate the performance of our methods on the  CIFAR-10 image dataset \citep{krizhevsky2009learning}. We provide analyses for 3 and 10  classes. 

We train a convolutional neural network (CNN) on the $K$ classes  from the CIFAR-10 dataset, using all  available training images (5,000 per class). The network consists of 6 convolutional layers (organized in 3 blocks of 2), batch normalization, max-pooling, global average pooling, dropout, and a fully connected output layer, trained using the Adam optimizer with a learning rate of 0.001 and cross-entropy loss with label smoothing of 0.05. The softmax outputs are used as class probability estimates. In each  iteration, calibration and test samples are drawn without replacement from available test images (1,000 per class). 

The class distribution is even for the training sample, and either even or uneven for the calibration and test samples. Due to space limitations, we illustrate here only the case $K=10$ with an even class distribution for the training, calibration, and test samples;  see   \S~\ref{app-sim-CIFAR-K10} for additional results. Results for $K=3$ are  in  \S~\ref{app-sim-CIFAR-K3}. 

Since the estimated conditional probabilities from the CNN are known to be miscalibrated, we apply the  vector scaling correction proposed in \cite{guo2017calibration} as a preprocessing step
using half the calibration sample, leaving only half the calibration sample for the application of the different procedures. See  \S~\ref{app-sim-CIFAR-K10} for a comparison of the different procedures  with and without this preprocessing step. 

Figure~\ref{fig:CIFAR10-even} shows the results for two informativeness notions, when the class distribution is even for the training sample as well as for the calibration and test sample, for $K=10$.  The power improvement of \infoOSP\ and \infoOSPCal\ over the methods that have an FCR control guarantee,    \texttt{InfoSCOP} and \texttt{InfoSP},  is large. A non-negligible power improvement is manifest also  over classic conformal, for which  the average FCP happens to be below the nominal level. The method Top Class, that reports the class with largest estimated conditional probability if it is informative, has an average FCP far above the nominal level.

\begin{figure}[!t]
    \centering

    \includegraphics[width=0.47\linewidth]
    {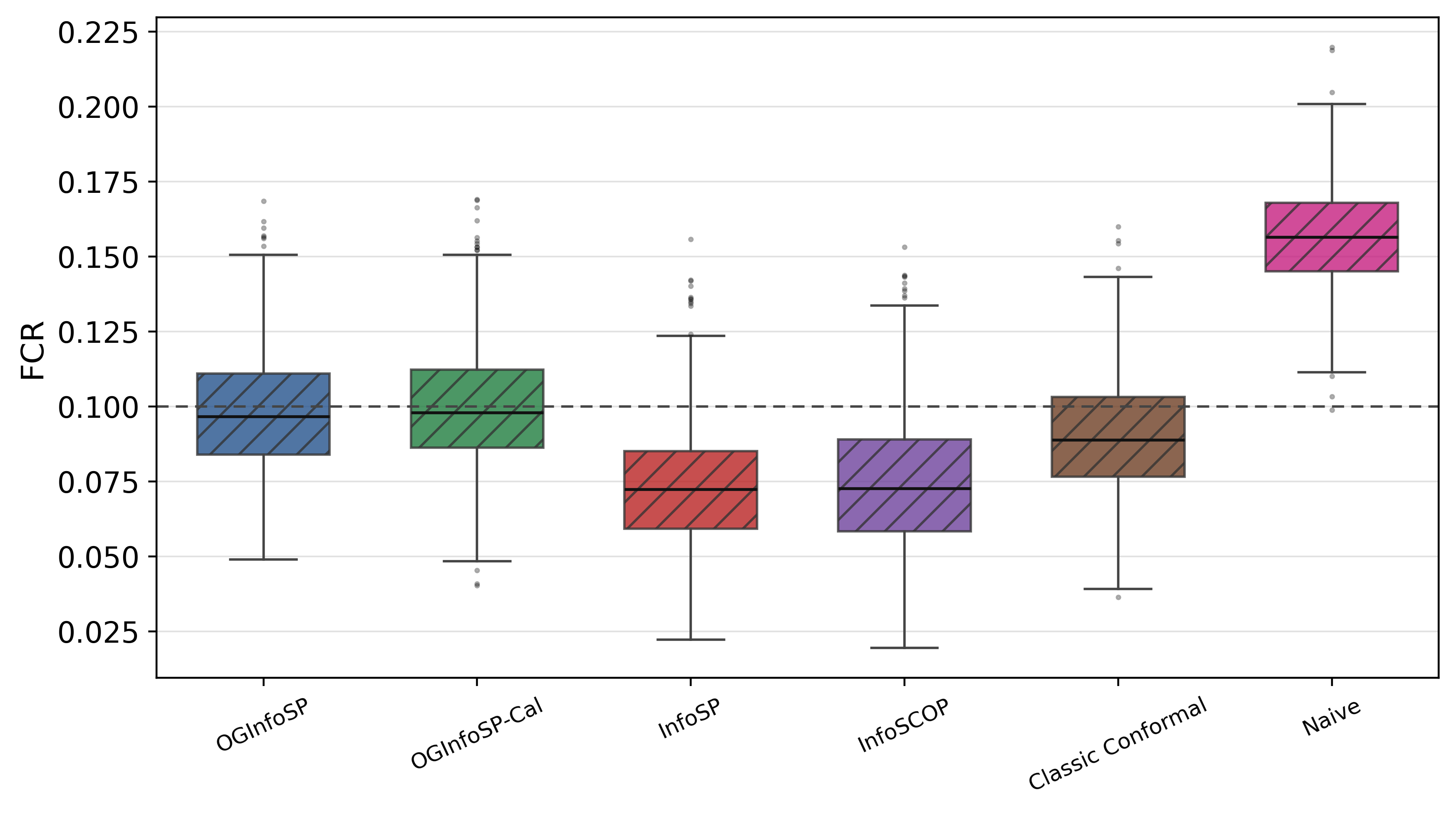}
    \hfill
    \includegraphics[width=0.47\linewidth]
    {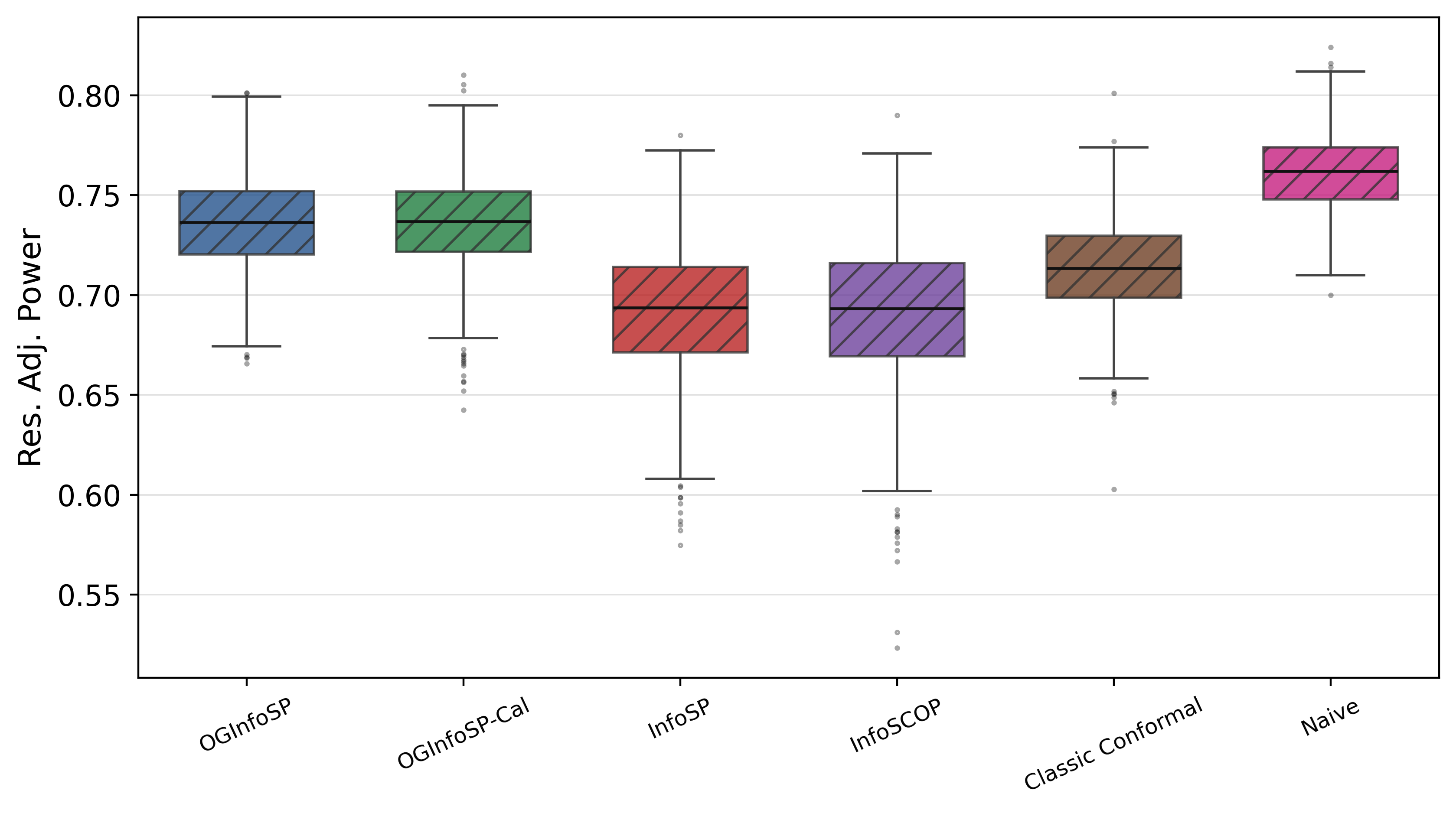}

    \par\vspace{-0.3em}
    {\small\textbf{(a) Exclusion of cats}}
    \par\vspace{0.5em}

    \includegraphics[width=0.47\linewidth]
    {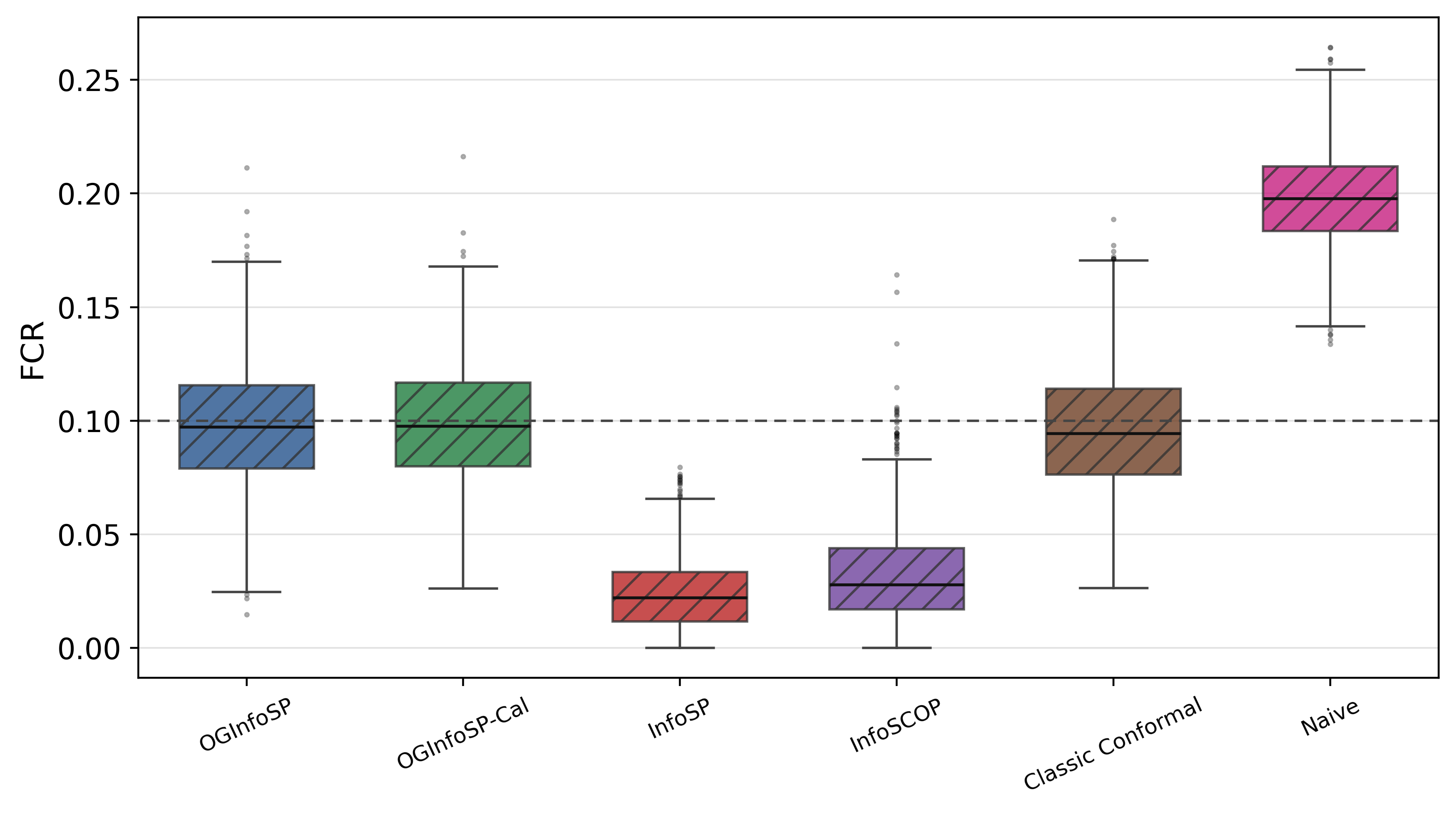}
    \hfill
    \includegraphics[width=0.47\linewidth]
    {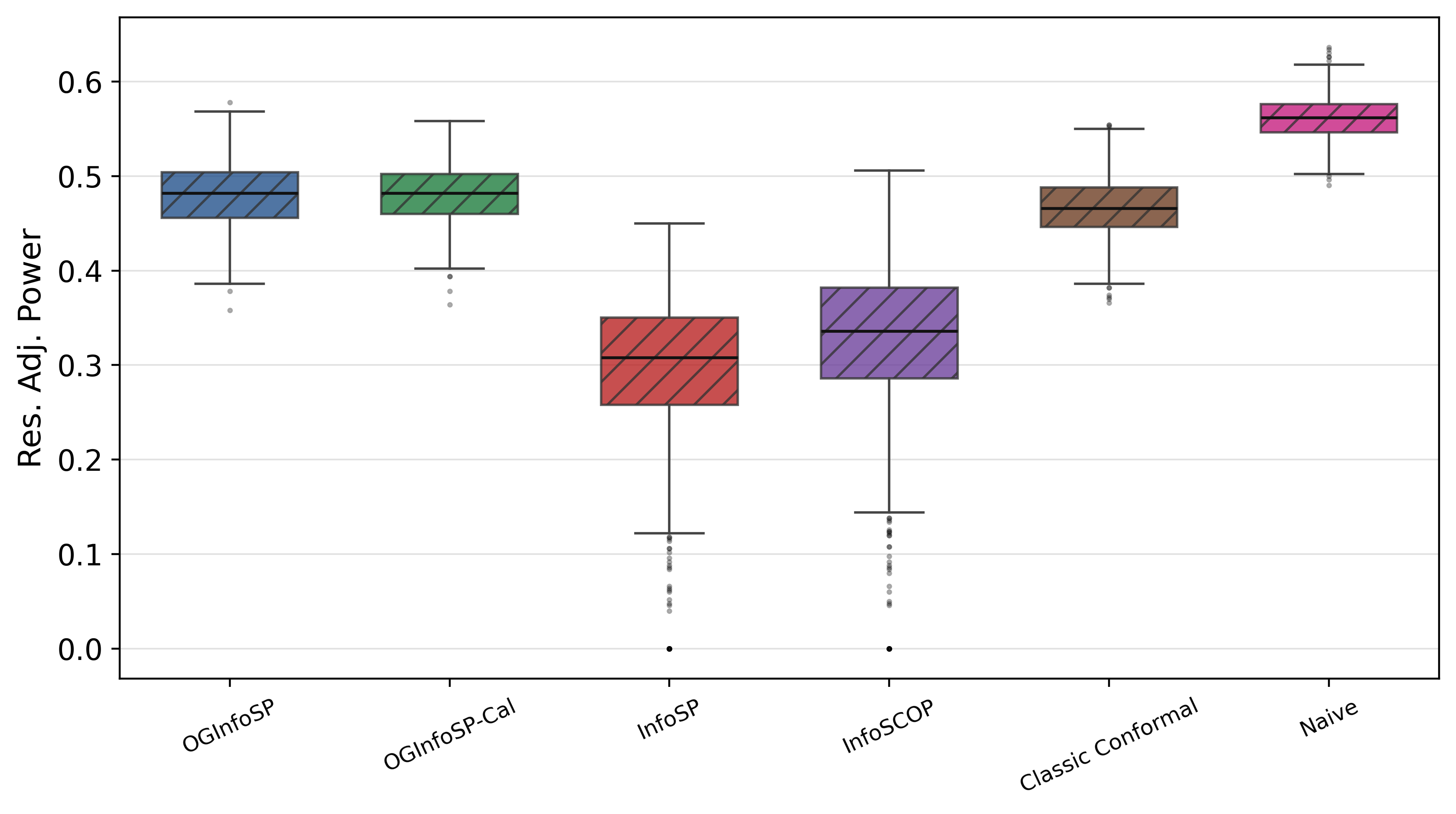}

    \par\vspace{-0.3em}
    {\small\textbf{(b) Predict single animal class, while  excluding  non-animal classes.}}

    \caption{For the CIFAR-10 data described in \S~\ref{sim-CIFAR} with $K=10$ classes evenly distributed, the FCR (column 1) and resolution-adjusted power (column 2) for prediction sets that:  (a) exclude the class of cats;  (b) are of size one and exclude the three non-animal classes. The sample sizes are $m= 500, n=1000$, $\alpha = 0.1$. Half the calibration samples are used for estimated the conditional probabilities using vector scaling, and the remaining half are used by the methods. The results are based on $1{,}000$ iterations.}
    \label{fig:CIFAR10-even}
\end{figure}

\section{Discussion and Future work}\label{sec-discussion}

We introduced an oracle-guided framework for selecting informative conformal prediction sets while controlling the false coverage rate on the selected sample. By viewing the selection and prediction-set construction problem through the geometry of an upper envelope of affine functions, we found  a simple oracle policy. It led  to two  practical  data-driven procedures \infoOSP\ (for the batch setting) and \infoOSPCal (for the online setting), which perform very similarly empirically.  These procedures have excellent theoretical properties: finite-sample error control (\S~\ref{subsec-finite-sample-theory}) as well as asymptotic optimality (\S~\ref{app-asymptotic-oracle}, \S~\ref{app-rates-and-asymptotics}).  Specifically, in \S~\ref{app-asymptotic-oracle} we make the asymptotic connection between the oracle policy and practical procedures in the idealized setting that the estimated conditional probabilities are the oracle conditional probabilities, and in \S~\ref{app-rates-and-asymptotics} we show that under mild regularity conditions, both procedures converge to the oracle policy,  as  sample sizes grow to infinity.

Corollary \ref{coro-mainLemma}  provides a connection to the application of the BH procedure on statistics that resemble conformal $p$-values which may be of independent interest. Unlike previous results \citep{bates2023testing, Jin2023selection}, which utilize exchangeability under a null hypothesis, this formulation relies on the  dominance assumption  $h_{xy}(X,Y)\leq h_x(X)$ for all $(X,Y)\in \mathcal X\times\mathcal Y$. Specifically, using the equivalence between thresholding and the BH procedure (pointed out in \cite{mary2022semisupervised} for semi-supervised multiple testing), the examples that satisfy $h_x(X_{n+i})> t_\alpha$, $i\in [m]$,  are those rejected by the BH procedure applied with  $p_i= (1+n)^{-1} \big(1+\sum_{j=1}^n\mI\{h_{xy}(X_j,Y_j)\geq h_x(X_{n+i})\}\big)$, $i\in [m]$. As far as we know, the FDR control guarantee of this BH procedure  is novel.

Our finite-sample guarantee is established under a restriction on the type of informative predictions that can be considered, formalized in Assumption \ref{ass-martingale}.  Broadly, the restriction is satisfied  when the notion of informativeness treats all outcomes symmetrically,  but it may be restrictive in more complex applications, see \S~\ref{app-ass-martingale-details} for a counterexample. It would be valuable to assess the potential robustness of the procedure to violations of Assumption \ref{ass-martingale}, as well as to provide  broader classes of informativeness under which the FCR guarantee is still maintained.

There are many additional interesting  directions for future work, including the following.  First, our numerical experiments focused on natural but relatively simple notions of informativeness. In many scientific applications, the informative sets may be constrained by additional structure. For example, in cell-type annotation, prediction sets may be required to respect a biological hierarchy or graph over cell types. It would be interesting  to study versions of \infoOSP\ in which \(\mathcal I\) is restricted to sets obeying graph-structured constraints, as in \cite{corbetta2025conformal}. 
Second, our framework currently assumes that the collection of candidate prediction sets \(\mathcal I\) is finite. This assumption is natural for classification, but less so for regression, where prediction sets are often intervals or unions of intervals indexed by continuously varying endpoints. Extending the framework to infinite collections of prediction sets is therefore an important next step.

\section*{Acknowledgement} 
The authors thank Royi Jacobovic and Yosef Rinnot for fruitful discussions regarding Lemma \ref{lemma-martingale}. Israela Solomon and Ruth Heller acknowledge support from the Israel Science Foundation (ISF) grant 406/24.  
Etienne Roquain acknowledges grants ANR-21-CE23-0035 (ASCAI) and ANR-23-CE40-0018-01 (BACKUP) of the French National Research Agency ANR and the French National program SUN (project TENET).
Saharon Rosset acknowledges support from the Israel Science Foundation (ISF) grant 3250/24. 

\bibliographystyle{plainnat}
\bibliography{bibliography}

@misc{li2026conformalizedlargescaleselectiveinference,
      title={Conformalized Large-Scale Selective Inference with Informative and Trustworthy Prediction Sets}, 
      author={Wangcheng Li and Guanlan Zhao and Xu Guo and Wenguang Sun},
      year={2026},
      eprint={2605.27012},
      archivePrefix={arXiv},
      primaryClass={math.ST}
}

@article{Mass1990,
	author = {Massart, P.},
	coden = {APBYAE},
	fjournal = {The Annals of Probability},
	issn = {0091-1798},
	journal = {Ann. Probab.},
	mrclass = {60E15 (60G50 62G30)},
	mrnumber = {MR1062069 (91i:60052)},
	mrreviewer = {Gutti J. Babu},
	number = {3},
	pages = {1269--1283},
	title = {The tight constant in the {D}voretzky-{K}iefer-{W}olfowitz inequality},
	volume = {18},
	year = {1990}}

@article{tsybakov2007fast,
  title={Fast learning rates for plug-in classifiers},
  author={Tsybakov, AB and Audibert, J-Y},
  journal={Annals of Statistics},
  volume={35},
  number={2},
  pages={608--633},
  year={2007}
}

@article{bos2022convergence,
  title={Convergence rates of deep ReLU networks for multiclass classification},
  author={Bos, Thijs and Schmidt-Hieber, Johannes},
  journal={Electronic Journal of Statistics},
  volume={16},
  number={1},
  pages={2724--2773},
  year={2022},
  publisher={The Institute of Mathematical Statistics and the Bernoulli Society}
}

@InProceedings{weinstein20online,  title =  {Online Control of the False Coverage Rate and False Sign Rate},  author =       {Weinstein, Asaf and Ramdas, Aaditya},  booktitle =  {Proceedings of the 37th International Conference on Machine Learning},  pages =  {10193--10202},  year =  {2020},  editor =  {III, Hal Daumé and Singh, Aarti},  volume =  {119},  series =  {Proceedings of Machine Learning Research},  month =  {13--18 Jul},  publisher =    {PMLR}
}

@article{weinstein20,
  title={{Selective sign-determining multiple confidence intervals with {FCR} control}},
  author={Weinstein, A. and Yekutieli, D.},
  journal={Statistica Sinica},
  volume={30},
  pages={531--555},
  year={2020}
}

@inproceedings{NEURIPS2024_6b99cfe8,
 author = {Wang, Xiaoning and Huo, Yuyang and Peng, Liuhua and Zou, Changliang},
 booktitle = {Advances in Neural Information Processing Systems},
 doi = {10.52202/079017-1867},
 editor = {A. Globerson and L. Mackey and D. Belgrave and A. Fan and U. Paquet and J. Tomczak and C. Zhang},
 pages = {58574--58609},
 publisher = {Curran Associates, Inc.},
 title = {Conformalized Multiple Testing after Data-dependent Selection},
 volume = {37},
 year = {2024}
}

@book{SharirAgarwal1995,
  author    = {Sharir, Micha and Agarwal, Pankaj K.},
  title     = {Davenport-Schinzel Sequences and Their Geometric Applications},
  publisher = {Cambridge University Press},
  address   = {Cambridge},
  year      = {1995},
  isbn      = {9780521470254}
}

@techreport{krizhevsky2009learning,
  title={Learning multiple layers of features from tiny images},
  author={Krizhevsky, Alex and Hinton, Geoffrey and others},
  year={2009},
  publisher={Toronto, ON, Canada}
}

@article{byrd1995limited,
  title={A limited memory algorithm for bound constrained optimization},
  author={Byrd, Richard H and Lu, Peihuang and Nocedal, Jorge and Zhu, Ciyou},
  journal={SIAM Journal on Scientific Computing},
  volume={16},
  number={5},
  pages={1190--1208},
  year={1995}
}

@inproceedings{guo2017calibration,
  title={On calibration of modern neural networks},
  author={Guo, Chuan and Pleiss, Geoff and Sun, Yu and Weinberger, Kilian Q},
  booktitle={International Conference on Machine Learning},
  pages={1321--1330},
  year={2017},
  organization={PMLR}
}

@book{KaratzasShreve1991,
  author    = {Karatzas, Ioannis and Shreve, Steven E.},
  title     = {Brownian Motion and Stochastic Calculus},
  publisher = {Springer-Verlag},
  year      = {1991},
  volume    = {113},
  series    = {Graduate Texts in Mathematics},
  address   = {New York},
  edition   = {2nd},
  doi       = {10.1007/978-1-4612-0949-2}
}

@InProceedings{corbetta2025conformal,
author="Corbetta, Daniela
and Finos, Livio
and Risso, Davide",
editor="Pollice, Alessio
and Mariani, Paolo",
title="Conformal Inference for Cell Type Prediction with Graph-Structured Constraints",
booktitle="Methodological and Applied Statistics and Demography II",
year="2025",
publisher="Springer Nature Switzerland",
address="Cham",
pages="549--554"
}

@article{jin2025confidence,
  title={Confidence on the focal: conformal prediction with selection-conditional coverage},
  author={Jin, Ying and Ren, Zhimei},
  journal={Journal of the Royal Statistical Society Series B: Statistical Methodology},
  volume={87},
  number={4},
  pages={1239--1259},
  year={2025},
  publisher={Oxford University Press},
  doi={10.1093/jrsssb/qkaf016}
}

@article{bates2023testing,
  title={Testing for outliers with conformal p-values},
  author={Bates, Stephen and Cand{\`e}s, Emmanuel and Lei, Lihua and Romano, Yaniv and Sesia, Matteo},
  journal={Ann. Statist.},
  volume={51},
  number={1},
  pages={149--178},
  year={2023},
  publisher={Institute of Mathematical Statistics}
}

@article{mary2022semisupervised,
  author  = {David Mary and Etienne Roquain},
  title   = {Semi-supervised multiple testing},
  journal = {Electronic Journal of Statistics},
  volume  = {16},
  number  = {2},
  pages   = {4240--4282},
  year    = {2022},
  doi     = {10.1214/22-EJS2050}
}

@article{marandon2024adaptive,
  title={Adaptive novelty detection with false discovery rate guarantee},
  author={Marandon, Ariane and Lei, Lihua and Mary, David and Roquain, Etienne},
  journal={The Annals of Statistics},
  volume={52},
  number={1},
  pages={157--183},
  year={2024},
  publisher={Institute of Mathematical Statistics}
}

@inproceedings{Romano2020,
 author = {Romano, Yaniv and Sesia, Matteo and Cand{\`e}s, Emmanuel},
 booktitle = {Advances in Neural Information Processing Systems},
 editor = {H. Larochelle and M. Ranzato and R. Hadsell and M.F. Balcan and H. Lin},
 pages = {26369--26382},
 publisher = {Curran Associates, Inc.},
 title = {Classification with Valid and Adaptive Coverage},
 volume = {33},
 year = {2020}
}

@article{Jin2023selection,
  author  = {Ying Jin and Emmanuel J. Cand{\`e}s},
  title   = {Selection by Prediction with Conformal p-values},
  journal = {Journal of Machine Learning Research},
  year    = {2023},
  volume  = {24},
  number  = {244},
  pages   = {1--41}
}

@article{bao2024selective,
  title={Selective conformal inference with false coverage-statement rate control},
  author={Bao, Yajie and Huo, Yuyang and Ren, Haojie and Zou, Changliang},
  journal={Biometrika},
  volume={111},
  number={3},
  pages={727--742},
  year={2024},
  doi={10.1093/biomet/asae010}
}

@article{BH1995,
	author = {Benjamini, Yoav and Hochberg, Yosef},
	coden = {JSTBAJ},
	fjournal = {Journal of the Royal Statistical Society. Series B. Methodological},
	issn = {0035-9246},
	journal = {J. Roy. Statist. Soc. Ser. B},
	mrclass = {62J15},
	mrnumber = {MR1325392 (96d:62143)},
	number = {1},
	pages = {289--300},
	title = {Controlling the false discovery rate: a practical and powerful approach to multiple testing},
	volume = {57},
	year = {1995}}

@article{benjamini2005false,
  title={False discovery rate--adjusted multiple confidence intervals for selected parameters},
  author={Benjamini, Yoav and Yekutieli, Daniel},
  journal={Journal of the American Statistical Association},
  volume={100},
  number={469},
  pages={71--81},
  year={2005},
  publisher={Taylor \& Francis}
}

@article{GazinHellerMarandonRoquain2025,
  title = {Selecting informative conformal prediction sets with false coverage rate control},
  author = {Gazin, Ulysse and Heller, Ruth and Marandon, Ariane and Roquain, Etienne},
  journal = {Journal of the Royal Statistical Society: Series B (Statistical Methodology)},
  volume = {87},
  number = {4},
  pages = {909--929},
  year = {2025},
  doi = {10.1093/jrsssb/qkae120}
}

@article{BinLemma,
    author = {Benjamini, Yoav and Krieger, Abba M. and Yekutieli, Daniel},
    title = {Adaptive linear step-up procedures that control the false discovery rate},
    journal = {Biometrika},
    volume = {93},
    number = {3},
    pages = {491-507},
    year = {2006},
    month = {09},
    issn = {0006-3444},
    doi = {10.1093/biomet/93.3.491}
}

@article{FSRzhao2023,
title={Controlling {FSR} in selective classification},
  author={Zhao, Guanlan and Su, Zhonggen},
  journal={arXiv preprint arXiv:2311.03811},
  year={2023}
}

\appendix

\section{Proofs of Section \ref{sec-optimal}}\label{app-proofsOracle}
\subsection*{Proof of Proposition \ref{prop-FCR-oracle}}\label{proof-prop-FCR-oracle}

\begin{eqnarray}
&& FCR(D,C) = \mE\left[
\frac{
\sum_{i\in[m]}\mI\{Y_{n+i}\notin C(X_{n+i})\}D(X_{n+i})
}{
1\lor \sum_{i\in[m]}D(X_{n+i})
} 
\right]  \nonumber\\ 
&& = \mE\left[\mE\left[
\frac{
\sum_{i\in[m]}\mI\{Y_{n+i}\notin C(X_{n+i})\}D(X_{n+i})
}{
1\lor \sum_{i\in[m]}D(X_{n+i})
} 
\mid (D(X_{n+j}))_{j\in [m]}\right]\right]  \nonumber\\ 
&&  = \mE\left[
\frac{
\sum_{i\in[m]}\mP\left(Y_{n+i}\notin C(X_{n+i})\mid (D(X_{n+j}))_{j\in [m]}\right)D(X_{n+i})
}{
1\lor \sum_{i\in[m]}D(X_{n+i})
} 
\right]  \nonumber\\
&&= \mE\left[
\frac{
\sum_{i\in[m]}\mP\left(Y_{n+i}\notin C(X_{n+i})\mid D(X_{n+i})=1\right)D(X_{n+i})
}{
1\lor \sum_{i\in[m]}D(X_{n+i})
} 
\right] \nonumber \\
&&= \mP(Y\notin C(X) \mid D(X)=1) \mE\left[
\frac{
\sum_{i\in[m]}D(X_{n+i})
}{
1\lor \sum_{i\in[m]}D(X_{n+i})
} 
\right] \nonumber \\
&& =\mP(Y\notin C(X) \mid D(X)=1)\mP\left(\sum_{i=1}^m D(X_{n+i})>0\right) = mFCR(D,C)\mP\left(\sum_{i=1}^m D(X_{n+i})>0\right). \nonumber
\end{eqnarray}
The second  equality uses the law of iterated expectations, by conditioning on the selection indicators $(D(X_{n+i}))_{i\in [m]}$. 
The fourth equality follows from the fact that the examples are independent.  The fifth equality follows from the fact that the examples are identically distributed.  The last equality follows since  $\mP(Y\notin C(X) \mid D(X)=1) = \frac{
\mE\left[\mI\{Y\notin C(X)\}D(X)\right]
}{
\mE\left[D(X)\right]},$ which is exactly the definition of the mFCR. 

\subsection*{Proof of Lemma \ref{lem-U-by-Tx}}
If $T(x) \ge 1 - \alpha$, then there exists $C' \in \mathcal{I}$ such that $\ell_{x,C'}(\mu)$ has a non-negative slope. Since its intercept is non-negative, $\ell_{x,C'}(\mu) \ge 0$ for all $\mu \ge 0$. Therefore $\mathcal{U}_x(\mu) \ge \ell_{x,C'}(\mu) \ge 0$ for all $\mu \ge 0$. Otherwise, assume that $T(x) < 1 - \alpha$. Then $\mP(Y \in C \mid X = x) < 1 - \alpha$ for all $C \in \mathcal{I}$. In this case, all the lines have negative slopes and thus $\mathcal{U}_x(\mu)$ is decreasing.

\subsection*{Proof that the tie breaking rule in \eqref{eq-C-mu} maximizes power}\label{app-oracle-tie-breaking}
The candidate's contribution to power  is the expectation of $w(C)\mP(Y\in C\mid X)D(X)$. This is the intercept of the line $\ell_{X,C}$.  Therefore,  to ensure maximum power,  ties have to be broken by selecting the set with largest intercept. 
The following lemma shows that if lines intersect, then the line with largest intercept corresponds to the largest weight.

\begin{lem}\label{lem:tie-breaking}
Suppose that the lines $\ell_{x,C_1}(\mu)$ and $\ell_{x,C_2}(\mu)$ intersect at some $\mu>0$. If $\ell_{x,C_1}$ has a larger intercept than $\ell_{x,C_2}$, then $w(C_1) > w(C_2)$.
\end{lem}
\begin{proof}
	Since $\ell_{x,C_1}$ has a larger intercept than $\ell_{x,C_2}$ and the two lines intersect at some $\mu > 0$, $\ell_{x,C_1}$ must have a smaller slope; otherwise, it would be strictly above $\ell_{x,C_2}$ for all $\mu \ge 0$. Thus $\mP(Y \in C_1 \mid X = x) < \mP(Y \in C_2 \mid X = x)$. Suppose, for contradiction, that $w(C_1) \le w(C_2)$. Then $w(C_1) \mP(Y \in C_1 \mid X = x) < w(C_2) \mP(Y \in C_2 \mid X = x)$, which contradicts the assumption that $\ell_{x,C_1}$ has the larger intercept. Therefore $w(C_1) > w(C_2)$.
\end{proof}

\subsection*{Proof of Proposition \ref{prop-monotoneD}}\label{prop-monotoneD-proof}
We use the following lemma. 
\begin{lem} \label{lem-envelope-properties}
Let $\mathcal{U}(\mu)$ be the upper envelope of lines $\ell_1(\mu), \dots, \ell_r(\mu)$. Denote by $\ell_{(1)}, \dots, \ell_{(s)}$ the lines whose segments appear on $\mathcal{U}(\mu)$ for $\mu \ge 0$, ordered from left to right. Then:
\begin{enumerate}
\item $\mathcal{U}$ is convex.
\item The sequence of slopes of $\ell_{(1)}, \dots, \ell_{(s)}$ is non-decreasing.
\item The sequence of intercepts of $\ell_{(1)}, \dots, \ell_{(s)}$ is non-increasing.
\item $\ell_{(s)}$ has the largest slope among $\ell_1, \dots, \ell_r$. If multiple lines have the same slope as $\ell_{(s)}$, then $\ell_{(s)}$ has the largest intercept among them.
\end{enumerate}
\end{lem}

\begin{proof}
\begin{enumerate}
\item The epigraph of $\mathcal{U}$ is the intersection of the half planes above $\ell_1, \dots, \ell_r$. The intersection of convex sets is convex, thus $\mathcal{U}$ is convex.

\item Since $\mathcal{U}$ is convex, its slope cannot decrease from left to right. Hence the sequence of slopes of $\ell_{(1)}, \dots, \ell_{(s)}$ is non-decreasing.

\item For any $i \in [s-1]$, $\ell_{(i)}$ and $\ell_{(i+1)}$ intersect at some $\mu > 0$, and $\ell_{(i+1)}$ has larger slope than $\ell_{(i)}$. If $\ell_{(i+1)}$ had larger intercept than $\ell_{(i)}$, then it would lie strictly above $\ell_{(i)}$ for all $\mu \ge 0$, contradicting that $\ell_{(i)}$ appears in the upper envelope for some $\mu \ge 0$. Hence the intercept of $\ell_{(i+1)}$ is smaller than that of $\ell_{(i)}$.

\item Let $\ell$ be any one of the lines. If its slope were larger than the slope of $\ell_{(s)}$, then for sufficiently large $\mu$ we would have $\ell(\mu) > \ell_{(s)}(\mu)$, contradicting the fact that the last segment of the upper envelope is supported by $\ell_{(s)}$. Now suppose that $\ell$ has the same slope as $\ell_{(s)}$. If $\ell$ had larger intercept, then $\ell$ would lie strictly above $\ell_{(s)}$, and $\ell_{(s)}$ could not appear on the upper envelope.
\end{enumerate}
\end{proof}

The proofs of the items in Proposition \ref{prop-monotoneD} are as follows. 
\begin{enumerate}
\item The slopes of the segments that form $\mathcal{U}_x$ are $\mP(Y \in C^{\mu}(x) \mid X = x) - (1 - \alpha)$. By Lemma~\ref{lem-envelope-properties}, the sequence of slopes is non-decreasing. Therefore $\mP(Y \in C^{\mu}(x) \mid X = x)$ is non-decreasing.

\item The intercepts of the lines that correspond to the segments that form $\mathcal{U}_x$ are $w(C^{\mu}(x))\mP(Y \in C^{\mu}(x) \mid X = x)$. By Lemma~\ref{lem-envelope-properties}, the sequence of intercepts is non-increasing. Therefore $w(C^{\mu}(x))\mP(Y \in C^{\mu}(x) \mid X = x)$ is non-increasing.

\item By Lemma~\ref{lem-U-by-Tx}, if $T(x) \ge 1 - \alpha$, then $D^{\mu}(x) = \mI\left\{ \mathcal{U}_x(\mu) \geq 0 \right\} = 1$ for all $\mu \ge 0$. Otherwise, $\mathcal{U}_x(\mu)$ is decreasing, and so $D^{\mu}(x)$ is non-increasing.

\item $w(C^{\mu}(x))\mP(Y\in C^{\mu}(x) \mid X=x)D^{\mu}(x)$ is non-increasing since $w(C^{\mu}(x))\mP(Y\in C^{\mu}(x) \mid X=x)$ and $D^{\mu}(x)$ are both non-increasing and non-negative.

\item Since $\mP(Y\in C^{\mu}(x) \mid X=x)$ is non-decreasing, $1 - \mP(Y \in C^{\mu}(x) \mid X = x) - \alpha$ is non-increasing. If $T(x) \ge 1 - \alpha$, then by Lemma~\ref{lem-U-by-Tx}, $D^{\mu}(x) = 1$ for all $\mu \ge 0$, and therefore $(1 - \mP(Y \in C^{\mu}(x) \mid X = x) -\alpha)D^{\mu}(x)$ is non-increasing. Otherwise, assume $T(x) < 1 - \alpha$. Then $\mP(Y \in C \mid X = x) < 1 - \alpha$ for all $C \in \mathcal{I}$. Hence, $1 - \mP(Y \in C^{\mu}(x) \mid X = x) - \alpha > 0$ for all $\mu \ge 0$. Therefore $(1 - \mP(Y \in C^{\mu}(x) \mid X = x) -\alpha)D^{\mu}(x)$ is non-increasing as the product of two non-increasing and non-negative functions.
\end{enumerate}

\subsection*{Proof of Proposition \ref{lem:oracle_mularge}}\label{lem:oracle_mularge-proof}
Let $x \in \mathcal{X}$. If $T(x) \ge 1 - \alpha$, then by Lemma~\ref{lem-U-by-Tx}, $D^{\mu}(x) = \mI\left\{ \mathcal{U}_x(\mu) \geq 0 \right\} = 1 \, \forall \mu \ge 0$. Otherwise, if $T(x) < 1 - \alpha$, then all the lines $\{\ell_{x,C}\}_{C \in \mathcal{I}}$ are strictly decreasing, and thus for all sufficiently large $\mu$, $\mathcal{U}_x(\mu) < 0$ and $D^{\mu}(x) = \mI\left\{ \mathcal{U}_x(\mu) \geq 0 \right\} = 0$. In both cases we have $D^{\mu}(x) = \mI(T(x)\geq 1-\alpha)$ for all sufficiently large $\mu$.

By Lemma~\ref{lem-envelope-properties}, for all sufficiently large $\mu$, $\mathcal{U}_x(\mu)$ is supported by a line with the largest slope, i.e. largest $\mP(Y \in C \mid X = x) - (1 - \alpha)$, and in case multiple lines have this slope, it has the largest intercept among them. Therefore, $C^{\mu}(x) \in \arg\max_{C \in \mathcal{I}} \mP(Y \in C \mid X = x)$, and if multiple sets have the same maximal conditional coverage probability, $\ell_{x,C^{\mu}(x)}$ has the largest intercept among the corresponding lines. Since the intercept for each set $C$ is $w(C)\mP(Y \in C \mid X = x)$ and all such sets share the same probability, $C^{\mu}(x)$ has the largest $w(C)$ among them.
\hfill$\square$

\subsection*{Proof of Proposition \ref{prop-positive-selection}}\label{prop-positive-selection-proof}
If $T(X)\geq 1-\alpha$, then $D^{\mu}(X)=1$ for all $\mu\geq 0$ by Lemma \ref{lem-U-by-Tx}. Therefore, $$\mP(T(X)\geq 1-\alpha)\leq \mP(D^{\mu}(X)=1) = \mE\left[ D^{\mu}(X)\right].$$ Since the left-hand side is strictly positive by assumption,  the result follows. 

\subsection*{Proof of Proposition \ref{prop-decreasingG}}\label{prop-decreasingG-proof}
    By Proposition \ref{lem:oracle_mularge}, for $\mu$ large enough $D^{\mu}(x) = \mI\{T(x)\geq 1-\alpha\}$ and $C^{\mu}(x) \in \arg\max_{C\in\mathcal{I}} \mP(Y\in C\mid X=x)$. Substituting these into the definition of $G(D^{\mu}, C^{\mu})$ yields $$G(D^{\mu}, C^{\mu})  = \mE\left[(1 -T(X)-\alpha)\mI\{T(X)\geq 1-\alpha\}\right] = \mE\left[(1-\alpha -T(X))\mI\{1-\alpha -T(X) < 0\}\right].$$
    Since $\mP(1-\alpha -T(X) <0) > 0$ by assumption, it follows that the last expression is negative.  
\hfill$\square$

\subsection*{Proof of Theorem \ref{thm-optimal-nonrandom}}\label{proof-thm-optimal-nonrandom}
For the case that $\mP(T(X)>1-\alpha)=0,$ see Proposition~\ref{prop-D-G}. Henceforth, we assume $\mP(T(X)>1-\alpha)>0$. 

If $mFCR(D^0, C^0) \leq \alpha$, then $\mu^*=0$, and $(D^0,C^0)$ is indeed an optimal solution, since it is a feasible solution that maximizes the resolution-adjusted power
$\Pi(D,C)=\mathcal{L}(D,C,0)$.

Henceforth suppose $mFCR(D^0, C^0)>\alpha$, which implies $G(D^0, C^0)>0$. 
From Proposition~\ref{prop-decreasingG}, $G(D^{\mu}, C^{\mu})$ is non-increasing  in $\mu$ and $G(D^{\mu}, C^{\mu})<0$ for sufficiently large $\mu$.
It thus follows that if $G(D^{\mu}, C^{\mu})$ is continuous in $\mu$, then there exists a finite $\mu^*>0$ such that
$G(D^{\mu^*},C^{\mu^*})=0$.

To show that $G(D^{\mu}, C^{\mu})$ is continuous, it is enough to show that for fixed $\mu\geq 0$ and any sequence $\mu_n\to\mu$,
$\mP(Y\in C^{\mu_n}(X)\mid X)\to \mP(Y\in C^{\mu}(X)\mid X)$ and
$D^{\mu_n}(X)\to D^{\mu}(X)$ almost surely.
Then, by the dominated convergence theorem,
\[
\lim_{n\to\infty} G(D^{\mu_n},C^{\mu_n}) = \mE\!\left[ \lim_{n\to\infty} (1-\mP(Y\in C^{\mu_n}(X)\mid X)-\alpha)D^{\mu_n}(X) \right] = G(D^{\mu},C^{\mu}),
\]
so $G(D^{\mu}, C^{\mu})$ is continuous.

For any $C\in\mathcal I$, by Assumption~\ref{ass-probratio} the random variable
$\mP(Y\in C\mid X)$ has a continuous, non-atomic distribution.
For fixed $\mu$, $C^{\mu}(X)$ can switch between $C_1$ and $C_2$ only if
\[
\left(w(C_1)+\mu\right)\mP(Y\in C_1\mid X) = \left(w(C_2)+\mu\right)\mP(Y\in C_2\mid X),
\]
where $w(C_1)+\mu$ and $w(C_2)+\mu$ are constants.
By Assumption~\ref{ass-probratio},
\[
\mP\bigg(\!\left(w(C_1)+\mu\right) \mP(Y\in C_1\mid X) = \left(w(C_2)+\mu\right) \mP(Y\in C_2\mid X)\bigg)=0.
\]
Therefore, with probability one there are no ties at $\mu$ among any two candidate lines. Let $$\mathcal N_{\mu} = \cup_{C_1,C_2\in \mathcal I} \{x\in \mathcal X: \ell_{x,C_1}(\mu) = \ell_{x,C_2}(\mu) \}.$$ For $x\notin \mathcal N_{\mu},$ the maximizer $C^{\mu}(x)$ is unique. 
Since $\mathcal I$ is finite, this uniqueness implies a strictly positive gap at $\mu$ between   $\ell_{x,C^{\mu}(x)}(\mu) $  and the lines of all other prediction sets $C\in \mathcal I\setminus\{ C^{\mu}(x)\}$.  Because each line $\ell_{x,C}(\cdot)$ is continuous, there exists an open neighborhood around $\mu$ for which the gap persists. Consequently, the maximizer remains constant for all $\mu'$ within this neighborhood, meaning $ C^{\mu'}(x)= C^{\mu}(x)$. Any sequence $\mu_{\ell}\rightarrow \mu$ will eventually be entirely within this neighborhood, implying that $C^{\mu_{\ell}}(x) = C^{\mu}(x)$ for sufficiently large $\ell$.  It follows that  
$\mP(Y\in C^{\mu_{\ell}}(x)\mid X=x)\rightarrow \mP(Y\in C^{\mu}(x)\mid X=x)$ as $\mu_{\ell}\rightarrow \mu$ for any $x\notin \mathcal N_{\mu}$. Since $X\notin \mathcal N_{\mu}$ with probability one, it follows that $\mP(Y\in C^{\mu_{\ell}}(X)\mid X)\rightarrow \mP(Y\in C^{\mu}(X)\mid X)$ almost surely as $\mu_{\ell}\rightarrow \mu$.  

By Proposition~\ref{prop-monotoneD}, $D^{\mu}(X)$ is monotone decreasing in $\mu$, so for any $X$ there is at most one value of $\mu$ at which $D^{\mu}(X)$ switches from $1$ to $0$. By continuity of the upper envelope, $\mathcal U_X(\mu_{\ell})\rightarrow \mathcal U_X(\mu)$. Therefore, the  convergence of  $D^{\mu_{\ell}}(X) = \mI(\mathcal U_X(\mu_{\ell})\geq 0)$ can only fail  on the event $\{  \mathcal U_X(\mu)=0 \}$. 
For  $ \mathcal U_X(\mu)=0$, there must exist $C\in\mathcal I$ such that
\[
\left(w(C)+\mu\right)\mP(Y\in C\mid X)-\mu(1-\alpha)=0,
\]
or equivalently
\[
\mP(Y\in C\mid X)=\frac{\mu(1-\alpha)}{w(C)+\mu}.
\]
By Assumption~\ref{ass-probratio}, this event has probability zero for each fixed $C$,
\[
\mP\!\left(\mP(Y\in C\mid X)=\frac{\mu(1-\alpha)}{w(C)+\mu}\right)=0.
\]
Since $\mathcal I$ is finite,   $\mP(\mathcal U_X(\mu)=0)\leq \sum_{C\in \mathcal I}\mP\!\left(\mP(Y\in C\mid X)=\frac{\mu(1-\alpha)}{w(C)+\mu}\right) =0$.  Let $\mathcal Z_{\mu} = \{ x\in \mathcal X: \mathcal U_x(\mu) =0\}$. For $x\notin \mathcal Z_{\mu}$, the continuity of  $\mathcal U_x(\mu)$ implies that its sign is locally constant. Thus, there exists a neighborhood of $\mu$ within which $D^{\mu'}(x) = D^{\mu}(x)$ for every $\mu'$ in the neighborhood. Any sequence $\mu_{\ell}\rightarrow \mu$ will eventually be entirely within this neighborhood, implying that $D^{\mu_{\ell}}(x)=D^{\mu}(x)$ for sufficiently large $\ell$ for any  $x\notin \mathcal Z_{\mu}.$ Since $X \notin \mathcal Z_{\mu}$ with probability one, it follows that $D^{\mu_{\ell}}(X)\rightarrow D^{\mu}(X)$
almost surely as $\mu_{\ell}\rightarrow \mu$. 

We conclude that $G(D^{\mu}, C^{\mu})$ is continuous and that there exists a finite $
\mu^*>0$ with $G(D^{\mu^*}, C^{\mu^*})=0$, and this $\mu^*$ equals
\[
\mu^*=\min\{\mu\ge 0:\; G(D^{\mu}, C^{\mu})\le 0\}.
\]

From Eq.~\eqref{prop-f-leq-L},
\[
f^* \le \mathcal L(D^{\mu^*},C^{\mu^*},\mu^*) = \Pi(D^{\mu^*},C^{\mu^*})-\mu^*G(D^{\mu^*},C^{\mu^*}) = \Pi(D^{\mu^*},C^{\mu^*}).
\]
By definition of $f^*$ as the maximal achievable value among feasible solutions,
it follows that
$\Pi(D^{\mu^*},C^{\mu^*})=f^*$.
\hfill$\square$

\section{Details for \S~\ref{sec:threshold}}
\subsection{Sufficient conditions for Assumption~\ref{ass-martingale} to hold}\label{app-ass-martingale-details}
A simple condition that guarantees Assumption~\ref{ass-martingale} is formalized in the next proposition.
\begin{prop}
  Consider the condition that for any $x \in \mathcal{X}$ and for any $C_1, C_2 \in \mathcal{I}$, if
$$C_1 \in \operatorname*{\arg\max}_{C \in \mathcal{I}, w(C) = w(C_1)} \hat{\mP}(Y \in C | X = x),$$
$$C_2 \in \operatorname*{\arg\max}_{C \in \mathcal{I}, w(C) = w(C_2)} \hat{\mP}(Y \in C | X = x),$$
and $w(C_1) > w(C_2)$, then $C_1 \subset C_2$. Then Assumption~\ref{ass-martingale} holds.
\end{prop}

\begin{proof}
Assume the condition holds. Let $x \in \mathcal{X}$ and $0 \le \mu_1 < \mu_2$. Recall that by definition,
$$\hat{C}^{\mu_1}(x) \in \arg\max_{C \in \mathcal{I}} w(C) \hat{\mP}(Y \in C \mid X = x) + \mu_1 \left(\hat{\mP}(Y \in C \mid X = x) - (1 - \alpha)\right)$$
and
$$\hat{C}^{\mu_2}(x) \in \arg\max_{C \in \mathcal{I}} w(C) \hat{\mP}(Y \in C \mid X = x) + \mu_2 \left(\hat{\mP}(Y \in C \mid X = x) - (1 - \alpha)\right).$$
Therefore,
$$\hat{C}^{\mu_1}(x) \in \operatorname*{\arg\max}_{C \in \mathcal{I}, w(C) = w(\hat{C}^{\mu_1}(x))} \hat{\mP}(Y \in C | X = x)$$
and
$$\hat{C}^{\mu_2}(x) \in \operatorname*{\arg\max}_{C \in \mathcal{I}, w(C) = w(\hat{C}^{\mu_2}(x))} \hat{\mP}(Y \in C | X = x).$$

It follows from Proposition~\ref{prop-monotoneD-procedure} that $w(\hat{C}^{\mu_1}(x))$ is non-increasing in $\mu$, and in particular $w(\hat{C}^{\mu_1}(x)) \ge w(\hat{C}^{\mu_2}(x))$. Therefore, the condition guarantees $\hat{C}^{\mu_1}(x) \subseteq \hat{C}^{\mu_2}(x)$.
\end{proof}

This sufficient condition holds in many settings of interest when  $w(C) = f(|C|)$ is a non-increasing function of the set cardinality. 
In particular, this structure is satisfied in all settings considered in \cite{GazinHellerMarandonRoquain2025} for classification outcomes. This includes cases where prediction sets are informative only if they do not exceed  a certain cardinality or if they exclude a specific range of outcomes. 
In these settings, the prediction sets that can determine the  upper envelope consists only of the classes with largest estimated probabilities for a given cardinality. The following proposition shows that this is enough to  ensure that Assumption \ref{ass-martingale} is satisfied.

\begin{prop}
Assume that for all $C\in \mathcal I$, $w(C)=f(|C|)$ for some positive non-increasing function $f$.
    Consider the informative sets $$ \mathcal{I}=\{C\subseteq [K]\backslash \mathcal{Y}_0\::\: |C| \in J\}, \quad J \subseteq [K]$$ for a (possibly empty) excluded value set $\mathcal{Y}_0\subset [K]$ with $\mathcal{Y}_0\neq [K]$, and a set $J$ of allowed cardinalities. Then Assumption~\ref{ass-martingale} holds. 
\end{prop}

\begin{proof}
The informative prediction set of size $j\in J$ with largest weight (i.e., the set selected by the tie-breaking rule among tied maximizers)  is composed of the $j$ classes with largest estimated conditional probabilities, out of the 
$[K]\backslash \mathcal{Y}_0$ classes that may compose informative prediction sets. This means that for $j_1<j_2$, the informative prediction set of size $j_1$ with largest weight is necessarily a subset of the  informative prediction set of size $j_2$ with largest weight. Therefore, Assumption \ref{ass-martingale} is satisfied.  
\end{proof}

Interestingly, subset closure of the candidate sets in $\mathcal I$  is not sufficient on its own to satisfy Assumption \ref{ass-martingale}, even when the weights depend solely on cardinality. 
This is the  basic assumption  made in  \cite{GazinHellerMarandonRoquain2025}, formally defined  as follows.
 \begin{ass}\label{ass-Imonotone}
For all $C\in \mathcal I$, we have that every nonempty  $C'\subset C$ implies $C'\in \mathcal I$. 
 \end{ass}

A counterexample illustrating that Assumption~\ref{ass-Imonotone} is not sufficient to guarantee Assumption~\ref{ass-martingale}, even when the weights depend solely on the cardinality via $w(C)=f(|C|)=1/|C|$, is presented in Figure \ref{fig-example-notsatisfyingLagrangeMultiplierAssumption}. In this counterexample, the informative set collection is $\mathcal{I}=\{\{1\}, \{2\}, \{3\}, \{2,3\}\}$, which satisfies Assumption~\ref{ass-Imonotone}. However, $C^{\mu}(x)$ is $\{1\}$ for $\mu \in [0, \frac{1}{2})$ and $\{2,3\}$ for $\mu \ge \frac{1}{2}$, meaning Assumption~\ref{ass-martingale} does not hold.

\begin{figure}[htbp]
    \centering
    \begin{tabular}{ccc}
    	\hspace{-0.03\textwidth}
        \includegraphics[width=0.64\textwidth,valign=t]{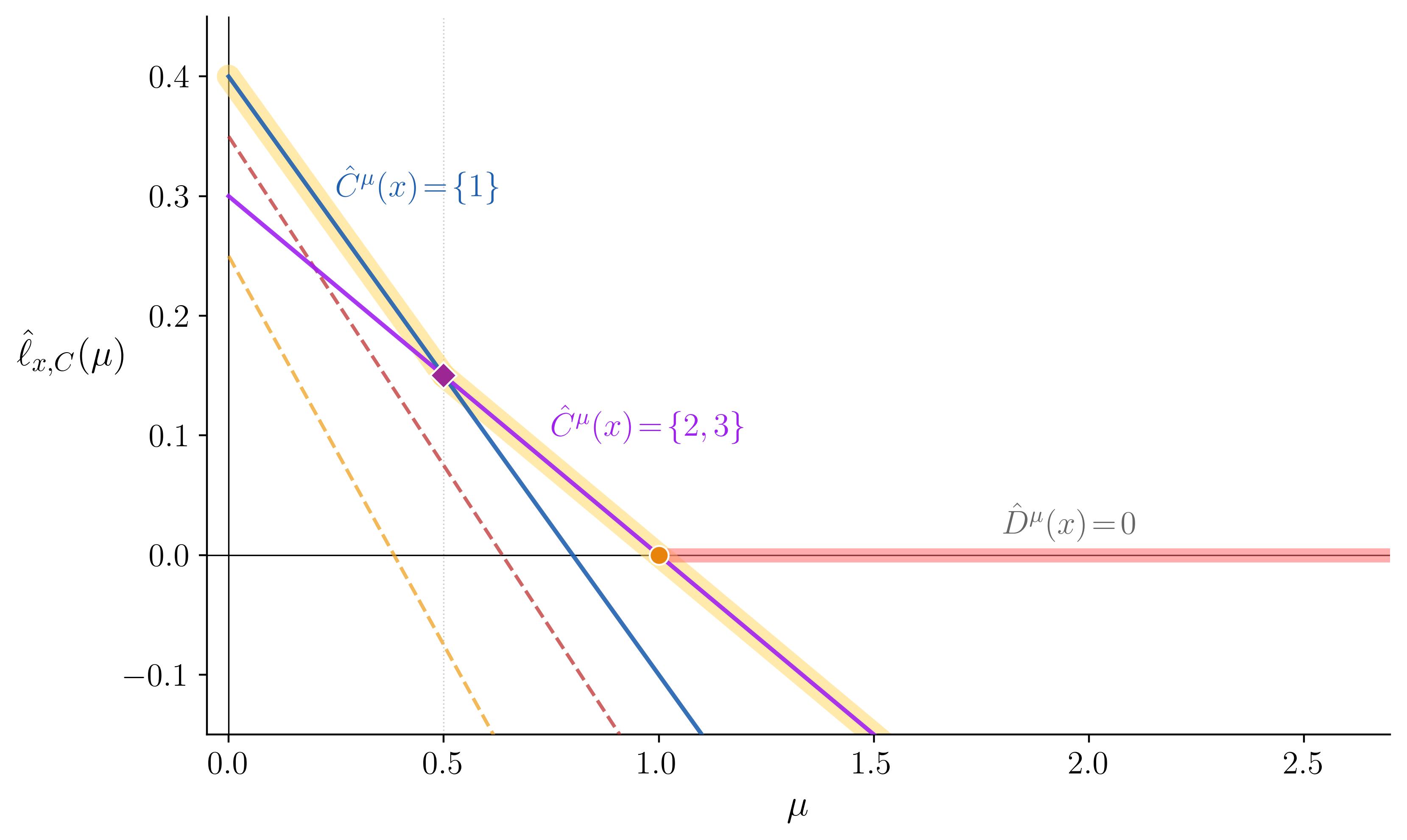}
        \hspace{0.005\textwidth}
        \includegraphics[width=0.12\textwidth,valign=t]{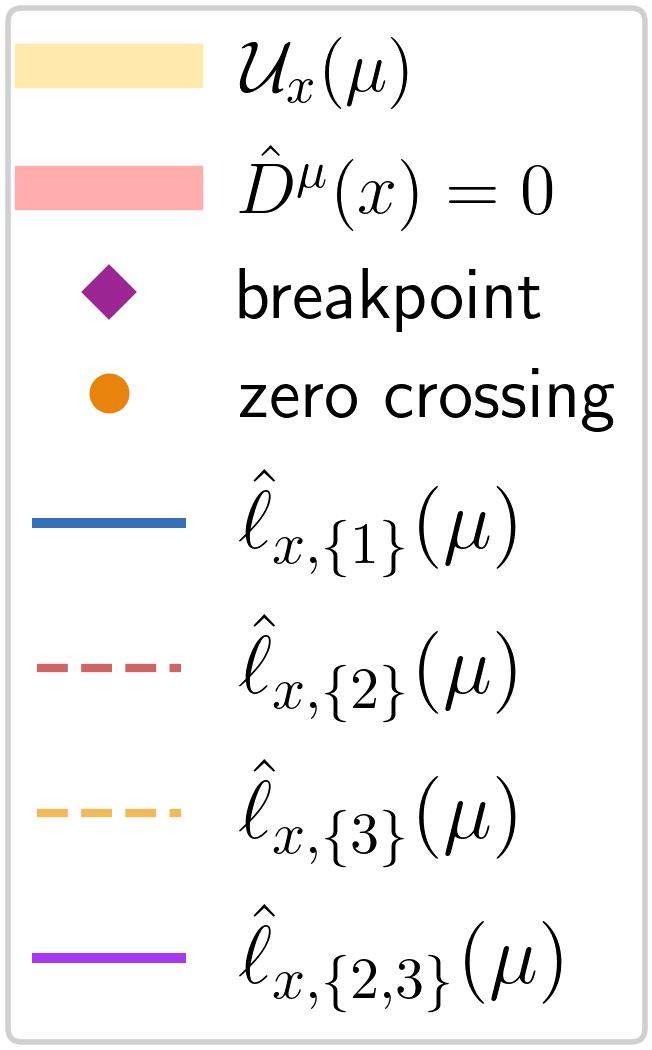}
    \end{tabular}
    \caption{Counterexample showing that Assumption \ref{ass-Imonotone} is insufficient to guarantee the nestedness required by Assumption \ref{ass-martingale}. Here, $K=3$,  $\mathcal{I}=\{\{1\}, \{2\}, \{3\}, \{2,3\}\}$, $w(C) = \frac 1{|C|}$, and $\alpha = 0.1$. The estimated conditional class probabilities are: 0.4 for class 1, 0.35 for class 2, and 0.25 for class 3. The upper envelope $\mathcal U_x(\mu)$ consists of 
    $\hell_{x,\{1\}}(\mu) = 0.4 +\mu(0.4-0.9) = 0.4-0.5\mu$ for $\mu<0.5$, and $\hell_{x,\{2,3\}}(\mu) = (1/2) (0.35+0.25) +\mu(0.35+0.25-0.9) = 0.3-0.3\mu$ for $\mu\geq 0.5$.  Therefore, $\hat{C}^{\mu}(x)$ is $\{1\}$ for $\mu \in [0, \frac{1}{2})$ and $\{2,3\}$ for $\mu \ge \frac{1}{2}$.}
    \label{fig-example-notsatisfyingLagrangeMultiplierAssumption}
\end{figure}

Conversely, Assumption~\ref{ass-martingale} may hold for $\mathcal I$ in which some subsets of informative sets are not themselves informative,  e.g., when $\mathcal I$ consists of all non-trivial prediction sets that contain a fixed class.

\subsection{Proof that the threshold-selection algorithms are equivalent to the original algorithms}\label{prop-equivalent procedures-proof}
\begin{prop}\label{prop-equivalent procedures}
Under Assumption \ref{ass-martingale}, the \infoOSP\ procedure described in Algorithm  \ref{algo:OGinfoSP} is equivalent to the threshold-selection rule defined by the following steps: 
\begin{enumerate}
    \item Sort $(\lambda(X_i,Y_i))_{i\in[n]}$ in increasing order; denote by $(\lambda_{(i)})_{i\in[\tilde n]}$ the sorted finite values. 
    \item Let $i_{\alpha} = \min \left\lbrace i\in [\tilde n]: \frac{\frac{1}{n+1} \left(1 + \sum_{j=1}^n \mI\{\lambda(X_j, Y_j)>\lambda_{(i)}\}\right)}{\frac{1}{m} \left(1\lor\sum_{j=1}^m \mI\{\Lambda(X_{n+j})>\lambda_{(i)}\}\right)}\leq \alpha \right \rbrace$ if the set is non-empty, and $\tilde n+1$ otherwise. 
    \item If $i_{\alpha}<\tilde n+1$, denote $\thone:= \lambda_{(i_{\alpha})} $ (this is the value in \eqref{eq-thone}) and construct prediction sets $\left\{ \hat C^{\thone}(X_{n+i}) : i \in [m], \Lambda(X_{n+i})>\thone \right\}$; otherwise, construct no prediction sets.  
\end{enumerate}
\end{prop}
\begin{proof}
Since $\hat D^\mu$ is non-increasing in $\mu$ (see
Proposition~\ref{prop-monotoneD-procedure}), the denominator of
$\widehat{\mathrm{FCP}}^\mu$ in Step~3(c) of
Algorithm~\ref{algo:OGinfoSP} is non-increasing in $\mu$.
By Assumption~\ref{ass-martingale},
$\mI\{Y_i\notin\hat C^\mu(X_i)\}$ is also non-increasing in $\mu$, and therefore so is $\mI\{Y_i\notin\hat C^\mu(X_i)\}\hat D^\mu(X_i).$

Hence, the numerator of $\widehat{\mathrm{FCP}}^\mu$ is
non-increasing in $\mu$. Moreover,
$\widehat{\mathrm{FCP}}^\mu$ can decrease only when its
numerator decreases, since whenever the numerator is constant,
the denominator is non-increasing and thus
$\widehat{\mathrm{FCP}}^\mu$ is non-decreasing.

By the definition of $\lambda$ and the right-continuity established
in Proposition~\ref{prop-right-continuous}, for every
$\mu\in[0,\infty)$,
\[
\mI\{Y_i\notin\hat C^\mu(X_i)\}\hat D^\mu(X_i)
=
\mI\{\lambda(X_i,Y_i)>\mu\}.
\]
Therefore, the numerator can decrease only at the values
$(\lambda(X_i,Y_i))_{i\in[\tilde n]}$.  Thus,  whenever the defining set for $\mu_\alpha$ is nonempty, its minimum is attained at one of the ordered values $(\lambda_{(i)})_{i\in[\tilde n]}$. Note that  if
$\widehat{\mathrm{FCP}}^0\leq\alpha$, then
$\lambda_{(1)}=0$: otherwise every calibration observation would contribute one to the numerator at $\mu=0$, in which case
$\widehat{\mathrm{FCP}}^0\geq 1>\alpha$.

Since Step~3(c) of
Algorithm~\ref{algo:OGinfoSP} can be written in terms of the key statistics as
$$\widehat{\mathrm{FCP}}^\mu=\frac{\frac{1}{n+1}\left(1+\sum_{j=1}^n\mI\{\lambda(X_j,Y_j)>\mu\}
\right)}{\frac{1}{m}\left(1\vee\sum_{j=1}^m
\mI\{\Lambda(X_{n+j})>\mu\}\right)},$$
the smallest index for which the inequality in Step~3(d) of
Algorithm~\ref{algo:OGinfoSP} is
satisfied is precisely $i_\alpha$. Consequently, if
$i_\alpha<\tilde n+1$, then $\mu_\alpha=\lambda_{(i_\alpha)}.$
For finite $\mu_\alpha$, the definition of $\Lambda$ gives
\[
\hat D^{\mu_\alpha}(X_{n+i})=1
\quad\Longleftrightarrow\quad
\Lambda(X_{n+i})>\mu_\alpha.
\]
Thus, Algorithm~\ref{algo:OGinfoSP} and the threshold-selection
rule construct exactly the same prediction sets. If
$\mu_\alpha=\infty$, both procedures construct no prediction
sets.
\end{proof}

The corresponding equivalence for  \infoOSPCal\ follows by the same argument, with the calibration-sample selection count replacing the test-sample selection count in the denominator, and is  therefore omitted. 

\section{Proofs of Section \ref{subsec-finite-sample-theory}}
\subsection*{Proof of Lemma \ref{lemma-martingale}}\label{app-proof-lemma-martingale}
We first note that $\mu_{\alpha}$ is well defined. As $f_{\mu}(x,y)$ and $g_{\mu}(x)$ are  right-continuous in $\mu$, $ \frac{\frac{1}{n+1}\left(\sum_{i=1}^n f_{\mu}(X_i,Y_i) + 1\right)}{\frac{1}{m}\left(1 \lor \sum_{i=1}^m g_{\mu}(X_{n+i})\right)}$ is right continuous in $\mu$. If the set $\left\lbrace \mu\geq 0 : \frac{\frac{1}{n+1}\left(\sum_{i=1}^n f_{\mu}(X_i,Y_i) + 1\right)}{\frac{1}{m}\left(1 \lor \sum_{i=1}^m g_{\mu}(X_{n+i})\right)} \leq \alpha \right\rbrace$ is not empty, its infimum is attained.

If $\mu_{\alpha} = \infty,$ then $err\leq \alpha.$ Otherwise,
\begin{align*}
err&= \mE\left[\frac{n+1}{m} \cdot \frac{\frac{1}{n+1}(\sum_{i=1}^n f_{\mu_\alpha}(X_i,Y_i) + 1)}{\frac{1}{m}\left(1 \lor \sum_{i=1}^m g_{\mu_\alpha}(X_{n+i})\right)} \cdot \frac{\sum_{i=1}^m f_{\mu_\alpha}(X_{n+i},Y_{n+i})}{\sum_{i=1}^n f_{\mu_\alpha}(X_i,Y_i) + 1}\right] \\
&\leq \frac{(n+1)\alpha}{m} \mE\left[\frac{\sum_{i=1}^m f_{\mu_\alpha}(X_{n+i},Y_{n+i}) }{\sum_{i=1}^n f_{\mu_\alpha}(X_i,Y_i) + 1}\right],
\end{align*}
where the inequality follows from the definition of $\mu_{\alpha}$. Therefore it is enough to show that $$\mE\left[\frac{\sum_{i=1}^m f_{\mu_\alpha}(X_{n+i},Y_{n+i})}{\sum_{i=1}^n f_{\mu_\alpha}(X_i,Y_i) + 1}\right] \leq \frac{m}{n+1}.$$

We shall begin by showing that $\frac{\sum_{i=1}^m f_{\mu}(X_{n+i},Y_{n+i})}{\sum_{i=1}^n f_{\mu}(X_i,Y_i) + 1}$ defines a supermartingale
with respect to the filtration $\mathbb{F}=(\mathcal{F_{\mu}})_{\mu\geq0}$, where
$$
\mathcal{F}_{\mu} = \sigma\bigg(\bigg\lbrace
\sum_{i=1}^m f_{\mu'}(X_{n+i},Y_{n+i}),\
\sum_{i=1}^m g_{\mu'}(X_{n+i}),\
\sum_{i=1}^n f_{\mu'}(X_i,Y_i) : 0\leq\mu'\leq\mu \bigg\rbrace\bigg).
$$

For any $0 \leq \mu_1 < \mu_2$,
\begin{align}
&\mE\left[\frac{\sum_{i=1}^m f_{\mu_2}(X_{n+i},Y_{n+i})}{\sum_{i=1}^n f_{\mu_2}(X_i,Y_i) + 1} \mid \mathcal{F}_{\mu_1}\right] \nonumber \\
&\hspace{1em}= \mE\left[\sum_{i=1}^m f_{\mu_2}(X_{n+i},Y_{n+i}) \mid \mathcal{F}_{\mu_1}\right] \mE\left[\frac{1}{\sum_{i=1}^n f_{\mu_2}(X_i,Y_i) + 1} \mid \mathcal{F}_{\mu_1}\right] \label{eq(1)} \\
&\hspace{1em}\leq \frac{\left(\sum_{i=1}^m f_{\mu_1}(X_{n+i},Y_{n+i})\right) \mP(f_{\mu_2}(X,Y) = 1 \mid f_{\mu_1}(X,Y) = 1)}{\left(\sum_{i=1}^n f_{\mu_1}(X_i,Y_i) + 1\right) \mP(f_{\mu_2}(X,Y) = 1 \mid f_{\mu_1}(X,Y) = 1)} \label{ineq(2)} \\
&\hspace{1em}= \frac{\sum_{i=1}^m f_{\mu_1}(X_{n+i},Y_{n+i})}{\sum_{i=1}^n f_{\mu_1}(X_i,Y_i) + 1}. \nonumber
\end{align}

Equality \eqref{eq(1)} is due to the conditional independence of $\{(X_i, Y_i)\}_{i=1}^n$ and $\{(X_{n+i}, Y_{n+i})\}_{i=1}^m$ given $\mathcal{F}_{\mu_1}$. To see this, note that $\mathcal{F}_{\mu_1}$ is generated separately by $\sigma((X_i,Y_i): i\in [n] )$ and $\sigma((X_{n+i}, Y_{n+i}): i \in [m])$ that are independent of each other. Since $\sum_{i=1}^m f_{\mu_2}(X_{n+i},Y_{n+i})$ and the  first two sums  in $\mathcal F_{\mu_1}$ are measurable wrt  $\sigma((X_{n+i}, Y_{n+i}): i \in [m])$, and $\frac{1}{\sum_{i=1}^n f_{\mu_2}(X_i,Y_i) + 1}$ and the last sum in $\mathcal F_{\mu_1}$ are measurable wrt $\sigma((X_i,Y_i): i\in [n] )$, conditional independence follows.

We shall now provide our  justification for  inequality \eqref{ineq(2)}. 
First, note that conditioning on $\sum_{i=1}^n f_{\mu_1}(X_i,Y_i) = s$, exactly $s$ indices have $f_{\mu_1}(X_i,Y_i) = 1$. For the other $n-s$ indices, $f_{\mu_2}(X_i,Y_i) = f_{\mu_1}(X_i,Y_i) = 0$ by monotonicity.

Second, note that the distribution for a group of indicators $(f_{\mu_2}(X_i,Y_i))_{i\in \mathcal S}$ at $\mu_2>\mu_1$ is unchanged when we further condition  on the filtration $\mathcal{F}_{\mu_1}$ if these indicators have value one at $\mu_1$, i.e., for $\mathcal S\subseteq[n+m]$ that satisfies $f_{\mu_1}(X_i,Y_i)=1 \ \forall i\in \mathcal S$,  $$[(f_{\mu_2}(X_i,Y_i))_{i\in \mathcal S} \mid \ (f_{\mu_1}(X_j,Y_j))_{j\in [n+m]}]\stackrel{d}= [(f_{\mu_2}(X_i,Y_i))_{i\in \mathcal S} \mid \ \mathcal{F}_{\mu_1}, (f_{\mu_1}(X_j,Y_j))_{j\in [n+m]}].$$ This follows since by  monotonicity $f_{\mu_1}(X_i,Y_i) = 1$ implies $g_{\mu'}(X_i) = f_{\mu'}(X_i,Y_i) = 1$ for any $\mu' \leq \mu_1$; so the added information in $\mathcal{F}_{\mu_1}$ when  $f_{\mu_1}(X_i,Y_i)=1 \ \forall i\in \mathcal S$ is only about sums of indicators from examples independent of $(X_i,Y_i))_{i\in \mathcal S}$; therefore, conditioning on this additional information does not change the distribution of $(f_{\mu_2}(X_i,Y_i))_{i\in \mathcal S}$.

Since $(X_i,Y_i)_{i \in [n+m]}$ are iid and 
the only indicators in $\sum_{i=1}^n f_{\mu_2}(X_i,Y_i)$ that can be one are those with $f_{\mu_1}(X_i, Y_i)=1$, $i\in[n]$, it follows that
{\small 
\begin{equation}\label{eq-binom-cal}
\sum_{i=1}^n f_{\mu_2}(X_i,Y_i) \mid \mathcal{F}_{\mu_1}, (f_{\mu_1}(X_j,Y_j))_{j\in [n+m]} \sim B\left(\sum_{i=1}^n f_{\mu_1}(X_i,Y_i),\ \mP(f_{\mu_2}(X,Y)=1 | f_{\mu_1}(X,Y) = 1)\right),
\end{equation}
}
where $B(k,p)$ denotes the binomial distribution with parameters $k$ and $p$. 
Similarly,
{\small 
\begin{equation}\label{eq-binom-test}
\sum_{i=1}^m f_{\mu_2}(X_{n+i},Y_{n+i}) \mid \mathcal{F}_{\mu_1}, (f_{\mu_1}(X_j,Y_j))_{j\in [n+m]} \sim B\left(\sum_{i=1}^m f_{\mu_1}(X_{n+i},Y_{n+i}),\ \mP(f_{\mu_2}(X,Y)=1 \mid f_{\mu_1}(X,Y) = 1)\right).
\end{equation}
}
Since the binomial distribution in \eqref{eq-binom-cal} is the same for every $(f_{\mu_1}(X_j,Y_j))_{j\in [n+m]}$ with a fixed sum $\sum_{i=1}^n f_{\mu_1}(X_i,Y_i)$,   it follows also that 
$$\sum_{i=1}^n f_{\mu_2}(X_i,Y_i) \mid \mathcal{F}_{\mu_1} \sim B\left(\sum_{i=1}^n f_{\mu_1}(X_i,Y_i),\ \mP(f_{\mu_2}(X,Y)=1 | f_{\mu_1}(X,Y) = 1)\right).$$ Similarly, since the binomial distribution in \eqref{eq-binom-test} is the same for every $(f_{\mu_1}(X_j,Y_j))_{j\in [n+m]}$ with a fixed sum $\sum_{i=1}^m f_{\mu_1}(X_{n+i},Y_{n+i})$), it follows also that $$\sum_{i=1}^m f_{\mu_2}(X_{n+i},Y_{n+i}) \mid \mathcal{F}_{\mu_1} \sim B\left(\sum_{i=1}^m f_{\mu_1}(X_{n+i},Y_{n+i}),\ \mP(f_{\mu_2}(X,Y)=1 \mid f_{\mu_1}(X,Y) = 1)\right).$$

We use these binomial distributions to evaluate the conditional expectations in \eqref{eq(1)}: 
{\small
\begin{align}
&& \mE\left[\sum_{i=1}^m f_{\mu_2}(X_{n+i},Y_{n+i}) \mid \mathcal{F}_{\mu_1}\right] = \left(\sum_{i=1}^m f_{\mu_1}(X_{n+i},Y_{n+i})\right) \mP(f_{\mu_2}(X,Y) =1\mid f_{\mu_1}(X,Y) = 1) \label{eq-num-expectation}  \\
&& 
\mE\left[\frac{1}{\sum_{i=1}^n f_{\mu_2}(X_i,Y_i) + 1} \mid \mathcal{F}_{\mu_1}\right] \leq \frac{1}{\left(\sum_{i=1}^n f_{\mu_1}(X_i,Y_i) + 1\right) \mP(f_{\mu_2}(X,Y)=1 \mid f_{\mu_1}(X,Y) = 1)} \label{eq-denom-expectation}
\end{align}
}
where the inequality follows from Lemma~1 of~\cite{BinLemma}, which states that if $Z \sim B(k-1, p)$, then $\mE[\frac{1}{Z+1}]=\frac 1{kp}\left(1-(1-p)^k\right) \leq \frac{1}{kp}$. Equality \eqref{eq-num-expectation} and inequality \eqref{eq-denom-expectation} together imply inequality \eqref{ineq(2)}.

We conclude that $\frac{\sum_{i=1}^m f_{\mu}(X_{n+i},Y_{n+i})}{\sum_{i=1}^n f_{\mu}(X_i,Y_i) + 1}$ is indeed a right-continuous supermartingale. Additionally, $\mu_\alpha$ is a stopping time since $\{\mu_\alpha \leq \mu\} = \left\{\exists 0 \leq \mu' \leq \mu \textrm{ s.t. } \frac{\frac{1}{n+1}\left(\sum_{i=1}^n f_{\mu'}(X_i,Y_i) + 1\right)}{\frac{1}{m}\left(1 \lor \sum_{i=1}^m g_{\mu'}(X_{n+i})\right)} \leq \alpha\right\} \in \mathcal{F}_\mu$ for all $\mu \geq 0$. Therefore, we can apply the Optional Stopping Theorem:

\begin{align}
&\hspace{-1em} \mE\left[\frac{\sum_{i=1}^m f_{\mu_\alpha}(X_{n+i},Y_{n+i})}{\sum_{i=1}^n f_{\mu_\alpha}(X_i,Y_i) + 1}\right] \nonumber \\
&\leq \mE\left[\frac{\sum_{i=1}^m f_0(X_{n+i},Y_{n+i})}{\sum_{i=1}^n f_0(X_i,Y_i) + 1}\right] \label{ineq(3)} \\
&= \mE\left[\sum_{i=1}^m f_0(X_{n+i},Y_{n+i})\right] \mE\left[\frac{1}{\sum_{i=1}^n f_0(X_i,Y_i) + 1}\right] \label{eq(4)} \\
&\leq \frac{m\mE\left[f_0(X,Y)\right]}{(n+1)\mE\left[f_0(X,Y)\right]} = \frac{m}{n+1} \label{ineq(5)}
\end{align}
Inequality~\eqref{ineq(3)} follows from the Optional Stopping Theorem (specifically, the application of Theorem 3.22 in \citet{KaratzasShreve1991}). Equality~\eqref{eq(4)} holds since $\{(X_i, Y_i)\}_{i=1}^n$ and $\{(X_{n+i}, Y_{n+i})\}_{i=1}^m$ are independent. For~\eqref{ineq(5)}, $\{f_0(X_i,Y_i)\}_{i=1}^{n+m}$ are iid, Bernoulli, so $\sum_{i=1}^m f_0(X_{n+i},Y_{n+i}) \sim B(m,\ \mE\left[f_0(X,Y)\right])$ and $\sum_{i=1}^n f_0(X_i,Y_i) \sim B(n,\ \mE\left[f_0(X,Y)\right])$. Then the inequality follows again from Lemma~1 of~\cite{BinLemma}.
\hfill$\square$

\subsection*{Proof of Proposition \ref{prop-right-continuous}}\label{prop-right-continuous-proof}
The mapping $\mu \mapsto \hat{C}^{\mu}(x)$ is piecewise-constant with a finite number of breakpoints, where ties are resolved by prioritizing the smallest weight and, if necessary, the smallest index in~$\mathcal{I}$. Consider a breakpoint $\mu'$ at which two lines $\hell_{x, C_1}(\mu)$ and $\hell_{x, C_2}(\mu)$ intersect, where $C_1$ is the maximizer for $\mu < \mu'$ and $C_2$ is the maximizer for $\mu > \mu'$. Therefore, the line $\hell_{x, C_2}(\mu)$ has a steeper slope, and $\hat{\mP}(Y \in C_2 | X = x) > \hat{\mP}(Y \in C_1 | X = x)$. Since the lines intersect at $\mu'>0$, it implies $w(C_2) < w(C_1)$. By the tie breaking rule of preferring the set with smaller weight, $\hat{C}^{\mu'}(x) = C_2$. This shows that $\mu \mapsto \hat{C}^{\mu}(x)$ is right-continuous.

Since $\mu \mapsto \hat{C}^{\mu}(x)$ is right-continuous and each of the lines $\{\hell_{x, C}(\mu)\}_{C \in \mathcal{I}}$ is continuous, it follows that $\mu \mapsto \hell_{x, \hat{C}^{\mu}(x)}(\mu)$ is right-continuous. By Proposition~\ref{prop-monotoneD-procedure}, $\hat{D}^{\mu}(x)$ is non-increasing in $\mu$. A non-increasing indicator that thresholds a right-continuous function is itself right-continuous, therefore, $\mu \mapsto \hat{D}^{\mu}(x) = \mI\{\hell_{x, \hat{C}^{\mu}(x)}(\mu) > 0\}$ is right-continuous.
$\square$

\subsection*{Proof of Proposition \ref{prop-fcr-single-example}}\label{prop-fcr-single-example-proof}
Since the test points are iid, it suffices to prove the result for $j=1$.

Let $f_{\mu}(x,y) = \mI\{y \notin \hat C^{\mu}(x)\} \hat D^{\mu}(x)$ and $g_{\mu}(x) = \hat D^{\mu}(x)$ for any $\mu \in [0, \infty)$, and $f_{\infty}(x,y) = g_{\infty}(x) = 0$.
For any $i \in [n+1]$, let $\mu^i_\alpha = \min\left\lbrace \mu\geq 0 : \frac{1 + \sum_{r \in [n+1] \setminus[i]} f_{\mu}(X_r,Y_r)}{1 + \sum_{r \in [n+1] \setminus[i]} g_{\mu}(X_r)} \le \alpha \right\rbrace$, or $\mu^i_\alpha = \infty$ if the set is empty.

Fix any realization of $\{(X_i,Y_i)\}_{i=1}^{n+1}$. If $f_{\mu^i_\alpha}(X_i,Y_i) = 0$ for all $i \in [n+1]$, then $\frac{1}{n+1} \sum_{i=1}^{n+1} f_{\mu^i_\alpha}(X_i,Y_i) = 0$. Otherwise, let $i^* = \arg\min_{i \in [n+1]} \{\mu^i_\alpha : f_{\mu^i_\alpha}(X_i,Y_i) = 1\}$. Due to Assumption~\ref{ass-martingale}, $f$ is non-increasing. Hence, for any $i \in [n+1]$, if $f_{\mu^i_\alpha}(X_i,Y_i) = 1$, then $f_{\mu^{i^*}_\alpha}(X_i,Y_i) = 1$. Therefore,
$$ \sum_{i=1}^{n+1} f_{\mu^i_\alpha}(X_i,Y_i) \le 1 + \sum_{i \neq i^*} f_{\mu^{i^*}_\alpha}(X_i,Y_i) = \frac{1 + \sum_{i \neq i^*} f_{\mu^{i^*}_\alpha}(X_i,Y_i)}{1 + \sum_{i \neq i^*} g_{\mu^{i^*}_\alpha}(X_i)} \left(1 + \sum_{i \neq i^*} g_{\mu^{i^*}_\alpha}(X_i)\right) \le \alpha (1+n),$$
where the last inequality follows from the definition of $\mu^{i^*}_\alpha$ and  $\sum_{i \neq i^*} g_{\mu^{i^*}_\alpha}(X_i)\leq n$.

We showed $\frac{1}{n+1} \sum_{i=1}^{n+1} f_{\mu^i_\alpha}(X_i,Y_i) \le \alpha$ for any realization of $\{(X_i,Y_i)\}_{i=1}^{n+1}$. Therefore,
$$\mE\left[ \frac{f_{\thtwo} (X_{n+1}, Y_{n+1})}{1 \lor g_{\thtwo} (X_{n+1})} \right] = \mE\left[ f_{\thtwo}(X_{n+1},Y_{n+1}) \right] = \mE\left[ \frac{1}{n+1} \sum_{i=1}^{n+1} f_{\mu^i_\alpha}(X_i,Y_i) \right] \le \alpha.$$
$\square$

\section{Additional theoretical results}\label{app-theoretical-results}

\begin{prop}\label{prop-D-G}
 
 Let $C_{max}(X) = \arg\max_{C\in\mathcal{I}} \mP(Y\in C\mid X)$ (in case of a tie, prefer $C$ with larger weight $w(C)$),  $D_G(X) = \mI\left\lbrace T(X) \geq 1-\alpha\right\rbrace$, and $C_G(X) = C_{max}(X)$ is a feasible policy that minimizes~$G(D,C)$. If $\mP(T(X) > 1-\alpha) = 0$, then $(D_G, C_G)$ is an optimal solution to \eqref{eq-OMT-problem}.
\end{prop}

\begin{proof}
Consider some $D$ with $\Pi(D,C)>\Pi(D_G, C_G)$.
\begin{align*}
\Pi(D_G, C_G) &< \Pi(D,C) = \mE\left[w(C(X)) \mP(Y \in C(X) \mid X) D(X)\right] \\
&= \mE\left[w(C(X)) \mP(Y \in C(X) \mid X) D(X) \mI\{T(X) \geq 1-\alpha\}\right] \\
&\hspace{1em}+ \mE\left[w(C(X)) \mP(Y \in C(X) \mid X) D(X) \mI\{T(X) < 1-\alpha\}\right] \\
&\leq \Pi(D_G, C_G) + B \mE\left[D(X) \mI\{T(X) < 1-\alpha\}\right].
\end{align*}
Thus $\mE\left[D(X) \mI\{T(X) < 1-\alpha\}\right] > 0$.

For any $X$, we have $\mP(Y \notin C(X) \mid X) > \alpha$ if $T(X) < 1-\alpha$, and $\mP(Y \notin C(X) \mid X) \geq \alpha$ if $T(X) = 1-\alpha$. Hence,
\begin{align*}
&\mE\left[\mP(Y \notin C(X) \mid X) D(X)\right]\\
&\hspace{1em}= \mE\left[\mP(Y \notin C(X) \mid X) D(X) \mI\{T(X) = 1-\alpha\}\right]
+ \mE\left[\mP(Y \notin C(X) \mid X) D(X) \mI\{T(X) < 1-\alpha\}\right] \\
&\hspace{1em}> \alpha \mE\left[D(X) \mI\{T(X) = 1-\alpha\}\right]
+ \alpha \mE\left[D(X) \mI\{T(X) < 1-\alpha\}\right] \\
&\hspace{1em}= \alpha \mE\left[D(X)\right].
\end{align*}
Then,
$$mFCR(D,C) = \frac{\mE\left[\mP(Y \notin C(X) \mid X) D(X)\right]}{\mE\left[D(X)\right]} > \frac{\alpha \mE\left[D(X)\right]}{\mE\left[D(X)\right]} = \alpha.$$
Therefore any $D$ with $\Pi(D,C)>\Pi(D_G, C_G)$ is not feasible.
\end{proof}

If we allow a randomized policy rather than the fixed policy $(D^{\mu^*},C^{\mu^*})$ that was introduced in \S~\ref{sec-optimal}, then the optimal solution does not need Assumption \ref{ass-probratio} to be satisfied. The following theorem characterizes the possibly randomized oracle procedure.

\begin{thm}\label{thm-random-policy}
	If $\mP(T(X) > 1-\alpha) > 0$ and $\mu^* = \inf\{\mu: G(D^{\mu}, C^{\mu})\leq 0\}$, then $\mathcal{L}(D^{\mu^*}, C^{\mu^*}, \mu^*) = f^*$ and there exists an optimal random policy that achieves this power.
\end{thm}

\begin{proof}
By Eq.~\eqref{prop-f-leq-L}, $\mathcal{L}(D^{\mu^*}, C^{\mu^*}, \mu^*) \ge f^*$. Let $\delta = \mathcal{L}(D^{\mu^*}, C^{\mu^*}, \mu^*) - f^*$ and assume $\delta > 0$.

Note that $\mathcal{L}(D^{\mu}, C^{\mu}, \mu)$ is a continuous convex function in $\mu$ as the upper envelope of the lines $\{\mathcal{L}(D^{\mu'}, C^{\mu'}, \mu) : \mu' \ge 0\}$. Moreover, the slope of $\mathcal{L}(D^{\mu'}, C^{\mu'}, \mu)$ is negative for any $\mu' < \mu^*$ and non-negative for any $\mu' > \mu^*$. Therefore, $\mathcal{L}(D^{\mu^*}, C^{\mu^*}, \mu^*) = \min_\mu \mathcal{L}(D^{\mu}, C^{\mu}, \mu)$, and $\mathcal{L}(D^{\mu}, C^{\mu}, \mu) \ge f^* + \delta$ for any $\mu \geq 0$.

We first assume that $\mu^* > 0$. Let $g(\mu):= G(D^{\mu}, C^{\mu})$ for notational convenience. We define a new policy $\left(D^{\mu^*,\epsilon}, C^{\mu^*,\epsilon}\right)$ that with probability $q$ follows $\left(D^{\mu^* - \epsilon}, C^{\mu^* - \epsilon}\right)$ and with probability $1-q$ follows $\left(D^{\mu^* + \epsilon}, C^{\mu^* + \epsilon}\right)$, where $q = \frac{-g(\mu^* + \epsilon)}{g(\mu^* - \epsilon) - g(\mu^* + \epsilon)}$.

\begin{align*}
g(\mu^*,\epsilon) &= q g(\mu^* - \epsilon) + (1-q) g(\mu^* + \epsilon) \\
&= \frac{-g(\mu^* + \epsilon) g(\mu^* - \epsilon)}{g(\mu^* - \epsilon) - g(\mu^* + \epsilon)} + \left(1 - \frac{-g(\mu^* + \epsilon)}{g(\mu^* - \epsilon) - g(\mu^* + \epsilon)}\right) g(\mu^* + \epsilon) = 0.
\end{align*}

Thus $mFCR(D^{\mu^*,\epsilon}, C^{\mu^*,\epsilon}) = \alpha$ and
$$\mathcal{L}(D^{\mu^*,\epsilon}, C^{\mu^*,\epsilon}, \mu^*) = \Pi(D^{\mu^*,\epsilon}, C^{\mu^*,\epsilon}) - \mu^* g(\mu^*,\epsilon) = \Pi(D^{\mu^*,\epsilon}, C^{\mu^*,\epsilon}) \le f^*.$$

\begin{align*}
f^* &\ge \mathcal{L}(D^{\mu^*,\epsilon}, C^{\mu^*,\epsilon}, \mu^*) \\
&= q \mathcal{L}(D^{\mu^* - \epsilon}, C^{\mu^* - \epsilon}, \mu^*) + (1-q) \mathcal{L}(D^{\mu^* + \epsilon}, C^{\mu^* + \epsilon}, \mu^*) \\
&= q \mathcal{L}(D^{\mu^* - \epsilon}, C^{\mu^* - \epsilon}, \mu^* - \epsilon) + (1-q) \mathcal{L}(D^{\mu^* + \epsilon}, C^{\mu^* + \epsilon}, \mu^* + \epsilon) \\
&\hspace{1em} + q \left(\mathcal{L}(D^{\mu^* - \epsilon}, C^{\mu^* - \epsilon}, \mu^*) - \mathcal{L}(D^{\mu^* - \epsilon}, C^{\mu^* - \epsilon}, \mu^* - \epsilon)\right) \\
&\hspace{1em} + (1-q) \left(\mathcal{L}(D^{\mu^* + \epsilon}, C^{\mu^* + \epsilon}, \mu^*) - \mathcal{L}(D^{\mu^* + \epsilon}, C^{\mu^* + \epsilon}, \mu^* + \epsilon)\right) \\
&= q \mathcal{L}(D^{\mu^* - \epsilon}, C^{\mu^* - \epsilon}, \mu^* - \epsilon) + (1-q) \mathcal{L}(D^{\mu^* + \epsilon}, C^{\mu^* + \epsilon}, \mu^* + \epsilon) \\
&\hspace{1em} + q \left((\Pi(D^{\mu^* - \epsilon}, C^{\mu^* - \epsilon}) - \mu^* g(\mu^* - \epsilon)) - (\Pi(D^{\mu^* - \epsilon}, C^{\mu^* - \epsilon}) - (\mu^* - \epsilon) g(\mu^* - \epsilon))\right) \\
&\hspace{1em} + (1-q) \left((\Pi(D^{\mu^* + \epsilon}, C^{\mu^* + \epsilon}) - \mu^* g(\mu^* + \epsilon)) - (\Pi(D^{\mu^* + \epsilon}, C^{\mu^* + \epsilon}) - (\mu^* + \epsilon) g(\mu^* + \epsilon))\right) \\
&= q \mathcal{L}(D^{\mu^* - \epsilon}, C^{\mu^* - \epsilon}, \mu^* - \epsilon) + (1-q) \mathcal{L}(D^{\mu^* + \epsilon}, C^{\mu^* + \epsilon}, \mu^* + \epsilon) \\
&\hspace{1em} - q \epsilon g(\mu^* - \epsilon) + (1-q) \epsilon g(\mu^* + \epsilon) \\
&\ge f^* + \delta - q \epsilon g(\mu^* - \epsilon) + (1-q) \epsilon g(\mu^* + \epsilon) \\
&= f^* + \delta + \epsilon \left(g(\mu^*,\epsilon) - 2q g(\mu^* - \epsilon)\right) \\
&= f^* + \delta - 2\epsilon q g(\mu^* - \epsilon) \\
&= f^* + \delta + 2\epsilon \frac{g(\mu^* - \epsilon)g(\mu^* + \epsilon)}{g(\mu^* - \epsilon) - g(\mu^* + \epsilon)}
\end{align*}

Therefore $\delta \le 2\epsilon \frac{g(\mu^* - \epsilon)(-g(\mu^* + \epsilon))}{g(\mu^* - \epsilon) + (- g(\mu^* + \epsilon))} \le 2\epsilon \min\{g(\mu^* - \epsilon), -g(\mu^* + \epsilon)\} \le 2 \epsilon$.
Choosing $\epsilon < \frac{\delta}{2}$ leads to a contradiction. We conclude $\delta = 0$ and $\mathcal{L}(D^{\mu^*}, C^{\mu^*}, \mu^*) = f^*$.

We also showed that there exists a policy $\left(D^{\mu^*,\epsilon}, C^{\mu^*,\epsilon}\right)$ with $mFCR(D^{\mu^*,\epsilon}, C^{\mu^*,\epsilon}) = \alpha$, $\Pi(D^{\mu^*,\epsilon}, C^{\mu^*,\epsilon}) = f^*$, and $\mathcal{L}(D^{\mu^*,\epsilon}, C^{\mu^*,\epsilon}, \mu^*) = \mathcal{L}(D^{\mu^*}, C^{\mu^*}, \mu^*)$. 

Finally, suppose $\mu^* = 0$. By definition, $(D^0, C^0)$ maximizes $\mathcal{L}(D, C, 0) = \Pi(D,C)$. Therefore, if $mFCR(D^0, C^0) \leq \alpha$, then $(D^0, C^0)$ is a feasible optimal policy and we have $\mathcal{L}(D^0, C^0, 0) = f^*$. Otherwise, if $mFCR(D^0, C^0) > \alpha$, combining $(D^0, C^0)$ and $(D^\epsilon, C^\epsilon)$ into a random policy gives the same result as above.
\end{proof}

\section{Addressing label shift from training to calibration plus test samples}\label{app-vector-scaling}

For classification outcomes, let  $(\hat{\mP}(Y=k\mid X))_{k\in [K]}$  be the estimated probabilities for class membership of an algorithm derived from unobserved training data $(\tilde{X}_i,\tilde{Y}_i)_{i\in [N]}$ that are iid from a certain joint distribution on $\mathcal X\times \mathcal Y$.  The analyst observes the calibration data  $(X_i,Y_i)_{i\in [n]}$ that arises from another joint distribution on $\mathcal X\times\mathcal Y$. It is assumed that the two  label-conditional  distributions are the same, i.e., $\forall x\in \mathcal X$,  $$\mP(X=x\mid Y=k) = \mP(\tilde X =x\mid \tilde Y=k) \ \forall \ k\in [K]. $$
However, the marginal distribution of the labels need not coincide, so $$(\mP(Y=k))_{k\in [K]}\neq (\mP(\tilde Y=k))_{k\in [K]}.$$

By Bayes' rule, it follows that $$\mP(Y=k\mid X) = \frac{\mP(X\mid Y=k)\mP(Y=k)}{\mP(X)}$$ and similarly
$$\mP(\tilde Y=k\mid \tilde X) = \frac{\mP(\tilde X\mid \tilde Y=k)\mP(\tilde Y=k)}{\mP(\tilde X)}.$$
Since the label-conditional distributions are the same, it follows that $$\mP(Y=k\mid X=X)  =  \mP(\tilde Y=k\mid \tilde X=X)\frac{\mP(Y=k)}{\mP(\tilde Y=k)}\frac{\mP(\tilde X=X)}{\mP(X=X)}\propto \mP(\tilde Y=k\mid \tilde X=X)\frac{\mP(Y=k)}{\mP(\tilde Y=k)}.$$
 Let $(b_k)_{k\in [K]}$ with $\sum_{k=1}^K b_k=0$,  be constants such that $$\mP(Y=k\mid X=x)= \frac{e^{b_k} \mP(\tilde Y=k\mid \tilde X=x)}{\sum_{\ell=1}^Ke^{b_{\ell}} \mP(\tilde Y=\ell\mid \tilde X=x)}. $$ The maximum likelihood estimates of $(b_k)_{k\in [K]}$ can be found by solving the minimization problem for the negative log likelihood (NLL)  $$\min_{\{(b_k)_{k\in [K]}: \sum_{k=1}^K b_k=0\}}-\sum_{i=1}^n\left( b_{y_i}+\log\left(\frac{\mP(\tilde Y=y_i\mid \tilde X = X_i)}{\sum_{\ell=1}^Ke^{b_{\ell}} \mP(\tilde Y=\ell\mid \tilde X = X_i)} \right) \right) $$ if the conditional distribution $(\mP(\tilde Y=k\mid \tilde X=x))_{k\in [K]}$ is known.
 However, since  what is available from the training model is only the estimated conditional probabilities that we denote by  $\tilde {\mP}( \tilde Y=k\mid \tilde X=x)$, the minimization problem above is solved using $\tilde \mP$ instead of $\mP$. 

Note that $\tilde {\mP}( \tilde Y=k\mid \tilde X=x)\propto e^{z_k(X)}, $ where $(z_{k}(X))_{k\in [K]}$ are the logits of the training model, thus what is actually carried out is the bias-only version of vector scaling in \cite{guo2017calibration}.

The numerical optimization is carried out  via the L-BFGS-B algorithm \citep{byrd1995limited}, a limited-memory quasi-Newton method designed for large-scale optimization under box constraints.

\begin{remark}
Since this is a  plug-in heuristic, it is not guaranteed to improve upon simpler alternatives such as  temperature scaling. In practice, it may be preferable to split the calibration sample into a ``training'' subset and a held-out ``calibration'' subset: use the training subset to learn improved predictive scores (e.g., via more flexible machine-learning methods that produce better softmax outputs for the joint distribution generating the calibration and test samples), and then apply \infoOSP\ on the calibration subset along with the test sample.
\end{remark}

\begin{remark}
    For CIFAR-10 numerical experiments, we carry out full vector scaling (i.e,  we estimate $2K-1$ parameters for shifting and rescaling the logarithm of the estimated conditional probabilities from the training model). This is due to the fact that even when there is no label shift, the estimated conditional probabilities are miscalibrated, so bias only correction is not enough. 
\end{remark}

\section{Additional numerical experiments with classification outcomes}\label{app-simulations}

\newcommand{\SMpair}[2]{%
  \includegraphics[width=0.49\linewidth]{#1}%
  \hfill
  \includegraphics[width=0.49\linewidth]{#2}%
}
\newcommand{\SMtriplet}[3]{%
  \includegraphics[width=0.32\linewidth]{#1}%
  \hfill
  \includegraphics[width=0.32\linewidth]{#2}%
  \hfill
  \includegraphics[width=0.32\linewidth]{#3}%
}
\newcommand{\SMlegend}[1]{%
  \includegraphics[width=0.96\linewidth]{#1}\par\smallskip
}

\subsection{Experiments with $K=4$ bivariate normal mixtures }\label{app-sim-Gaussian-K4}

 We consider a Gaussian mixture with $K=4$ components, where each component is bivariate normal. The means for the four components are $$(0,0), (SNR,0), (SNR,SNR), (0, SNR).$$ So the overlap between components decreases with increasing SNR. Let $\pi= (0.1, 0.7,0.1,0.1)$ denote the relative frequencies of the four components.

As in \S~\ref{sim-Gaussian}, the estimated conditional class probabilities are obtained by training a logistic regression model on $10{,}000$ independent training samples for each SNR. For each SNR, this fixed classifier is used for all iterations and across all procedures.

We consider two forms of informative prediction. First,  non-trivial prediction sets, for which informative prediction sets are those of size at most 3, i.e., $\mathcal I = \{C: |C|\leq  3\}$. Second, excluding one class, specifically  excluding the second class which has the largest prior probability. 

Figure \ref{app-fig-K4} shows the results. The key observations are the same as made for $K=10$, see \S~\ref{sim-Gaussian}.


\begin{figure}[!htbp]
  \centering
  \SMlegend{simulations_July2026/legends/gaussian_all_methods_no_vs_legend.png}
  \SMtriplet{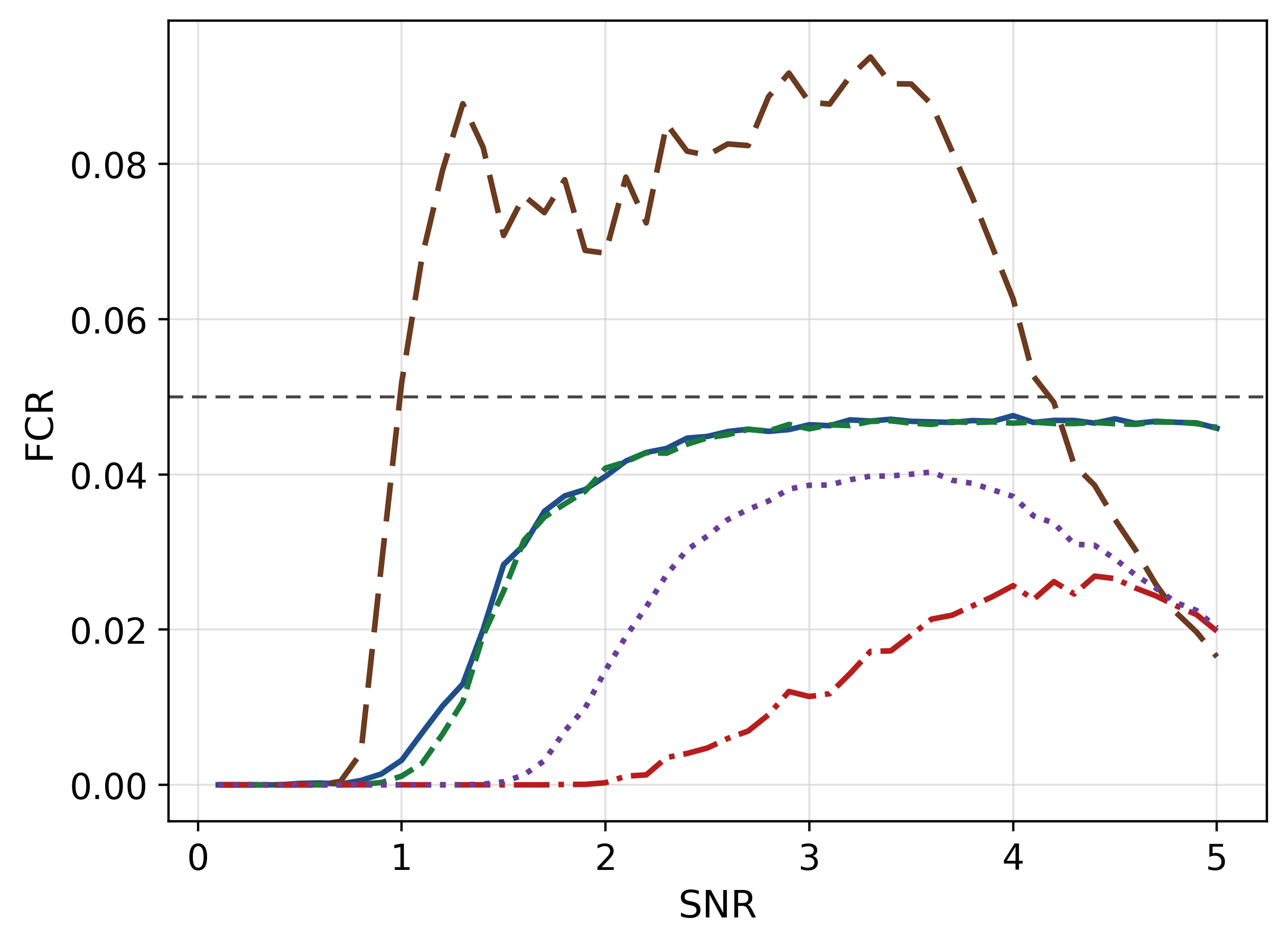}{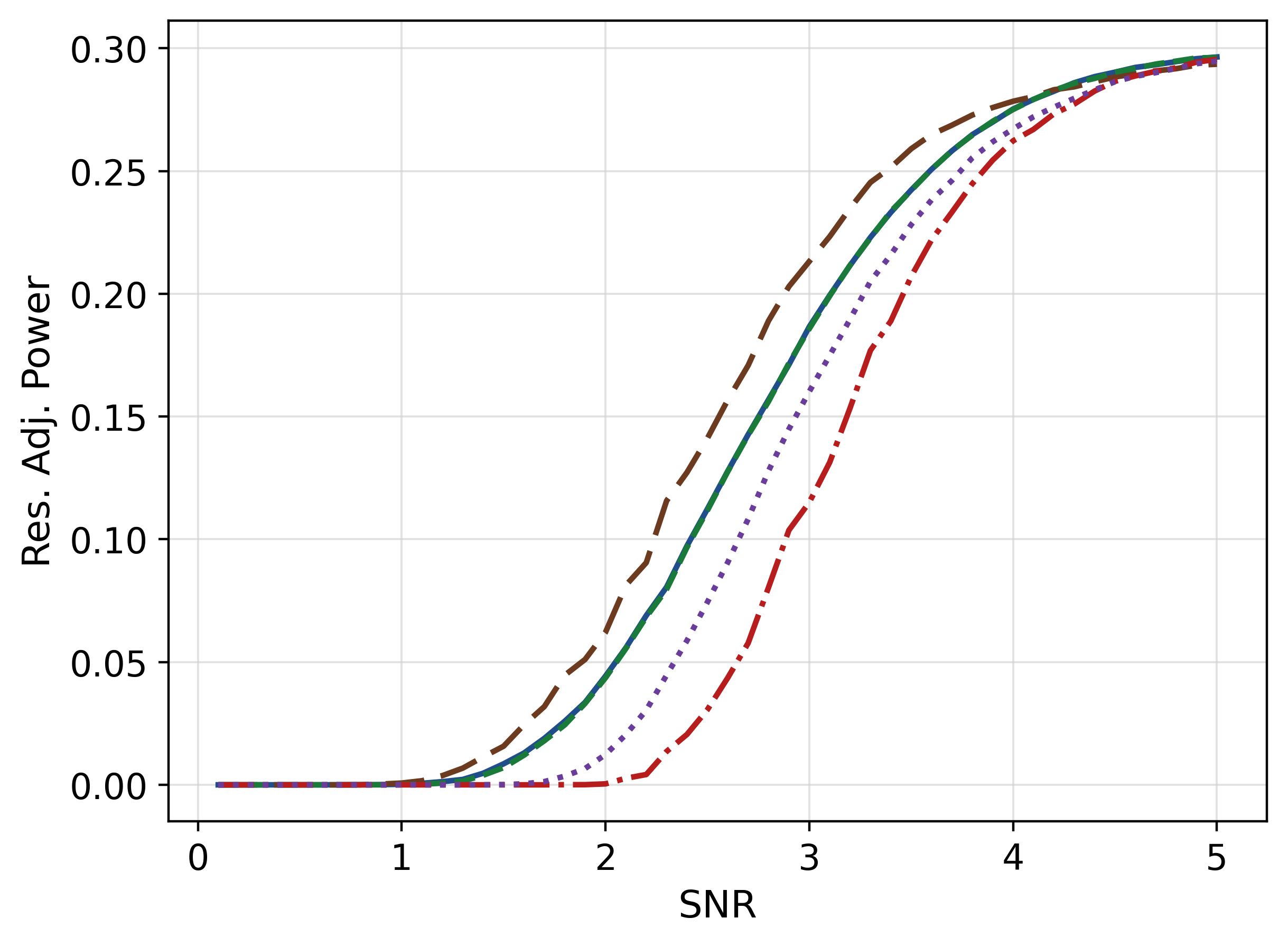}{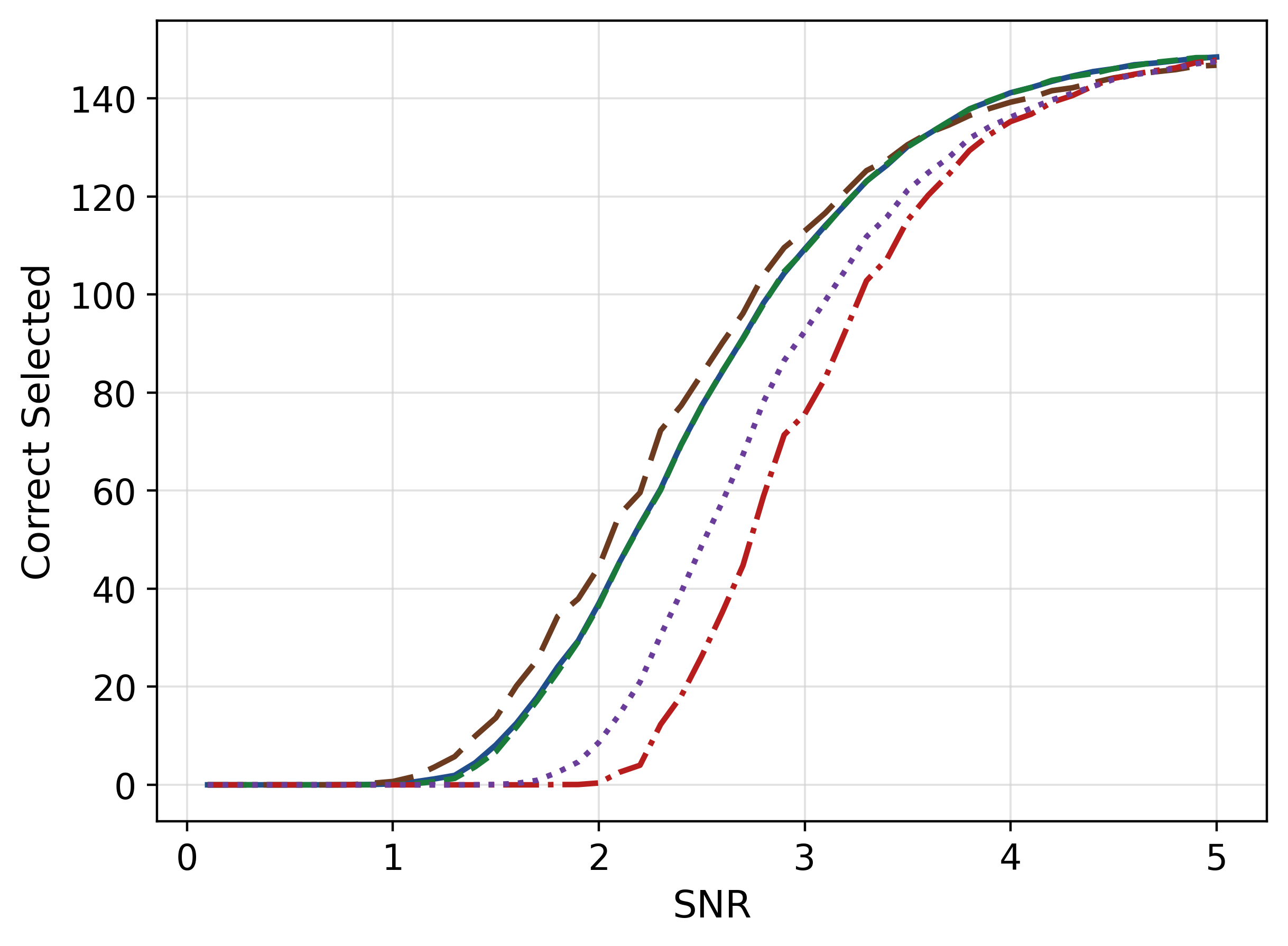}

  \SMlegend{simulations_July2026/legends/gaussian_all_methods_no_vs_legend.png}
  \SMtriplet{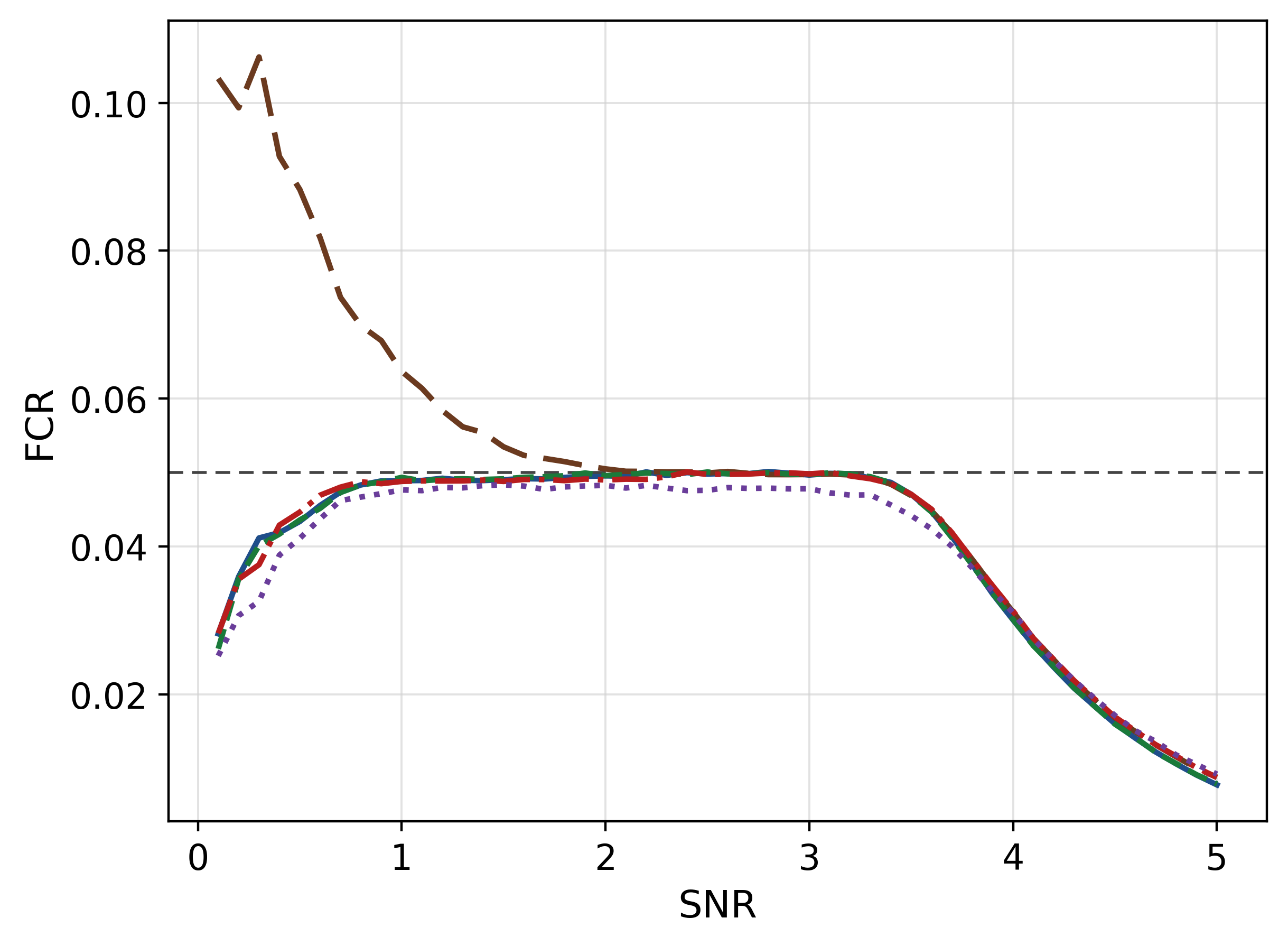}{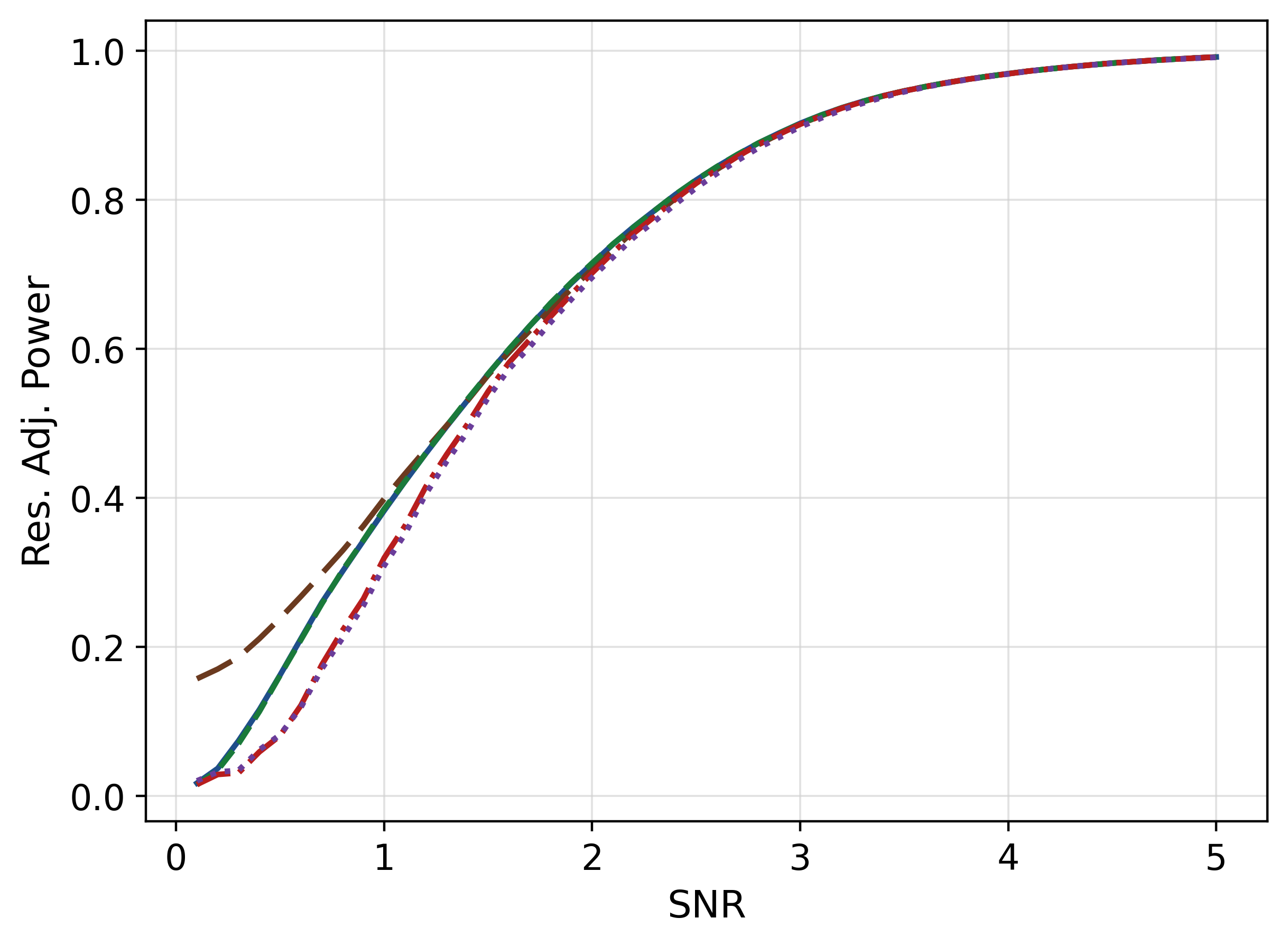}{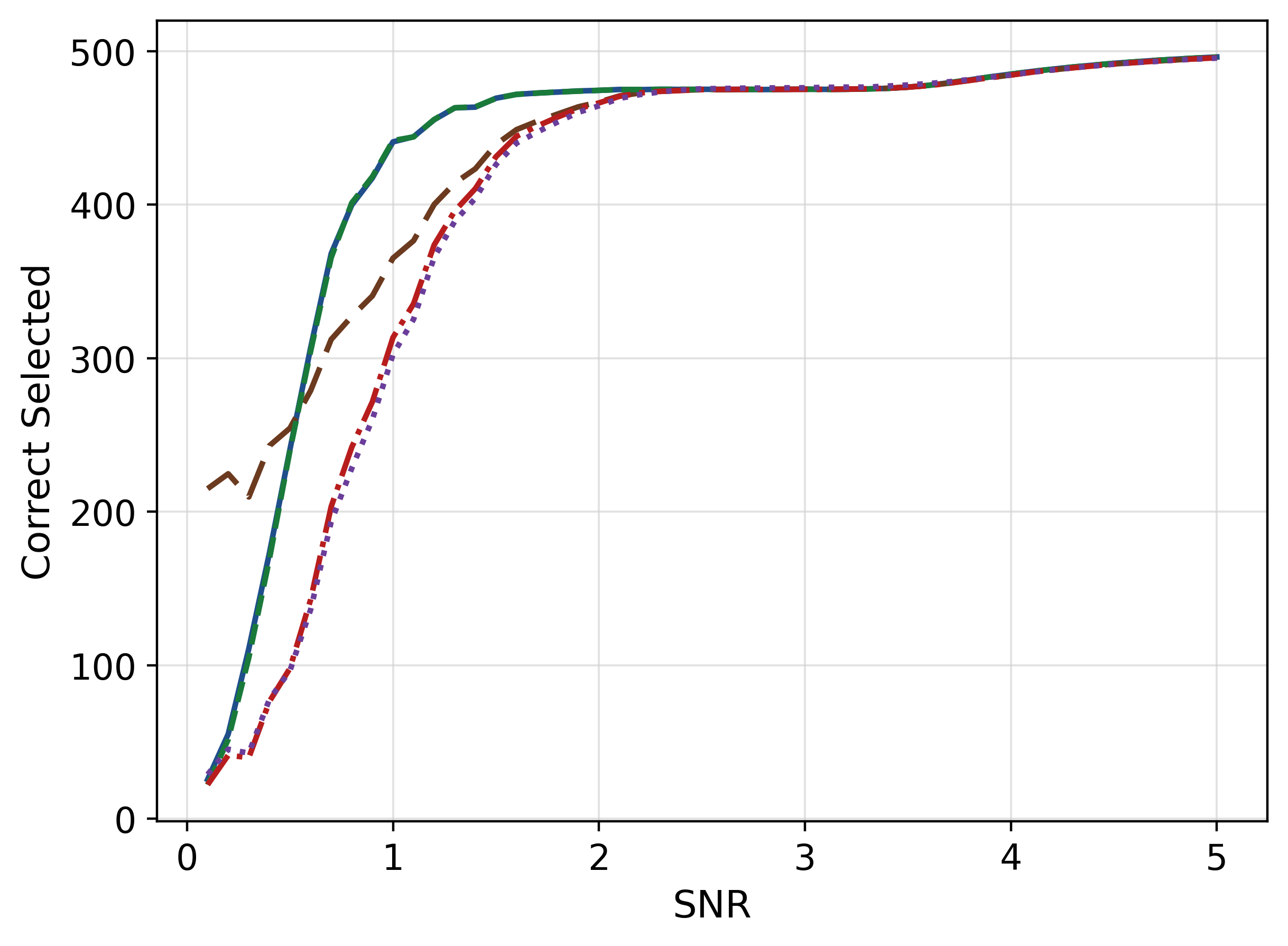}
  \caption{Results for the  $K=4$ bivariate normal
    mixture detailed in \ref{app-sim-Gaussian-K4}. The columns show the FCR, resolution-adjusted power, and
    mean number of selected prediction sets containing the true outcome.
    Informativeness requires exclusion of the largest class in the top row and non-trivial prediction-set in the bottom row. Here $n=m=500$,
    $\alpha=0.05$, and results are averaged over $10{,}000$ iterations; the
    horizontal line marks the nominal FCR level.}\label{app-fig-K4}
\end{figure}

\subsection{Experiments with label shift for $K=10$ bivariate normal mixtures }\label{app-sim-Gaussian-K10-label-shift}

The data from the $K=10$ bivariate normal mixture introduced in \S~\ref{sim-Gaussian} is depicted in Figure \ref{app-sim-Gaussian-K10-datagen} for SNR=3 and SNR=5.

\begin{figure}[!htbp]
  \centering 
  \includegraphics[width=0.48\linewidth]
  {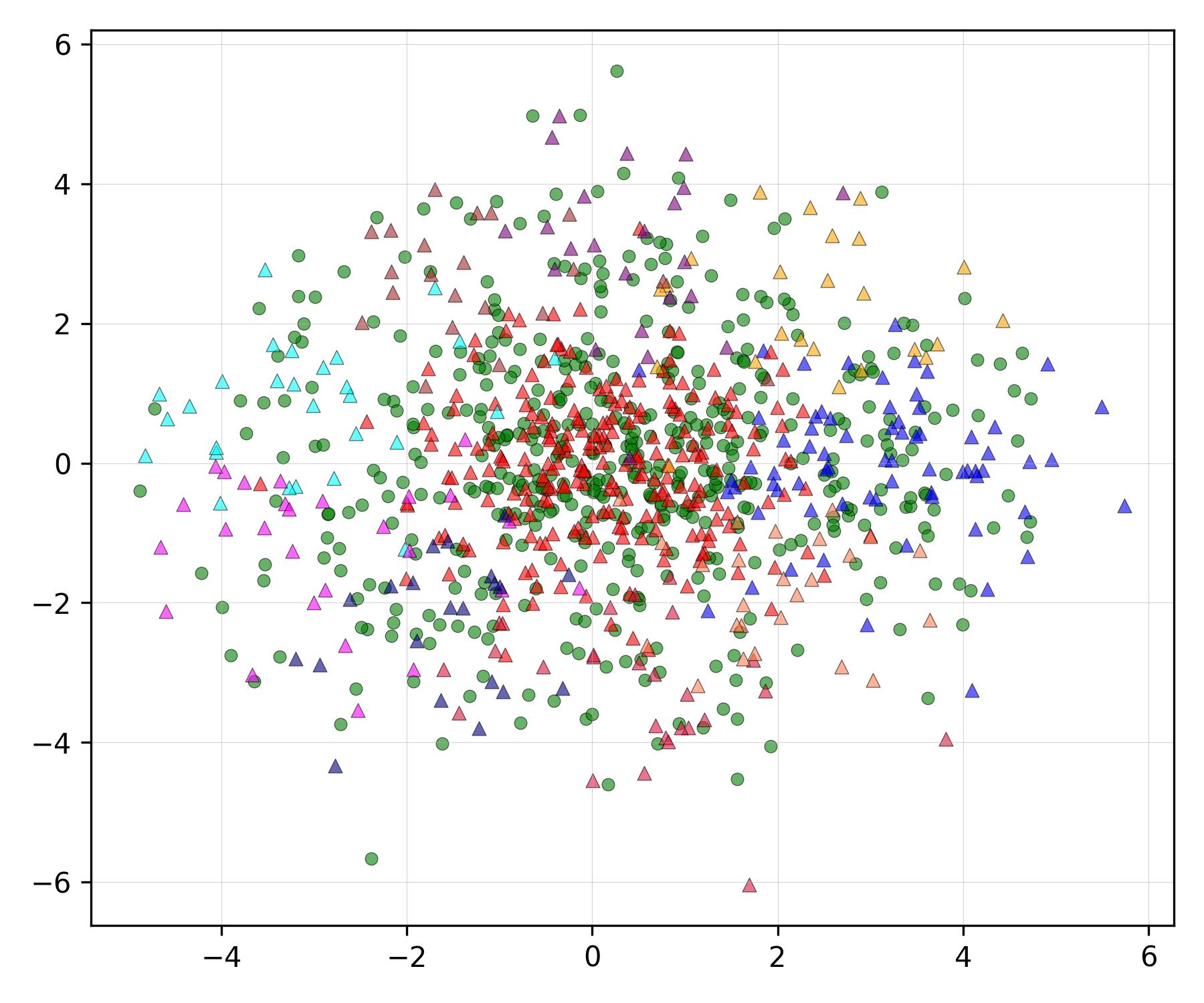}
  \hfill
  \includegraphics[width=0.48\linewidth]
  {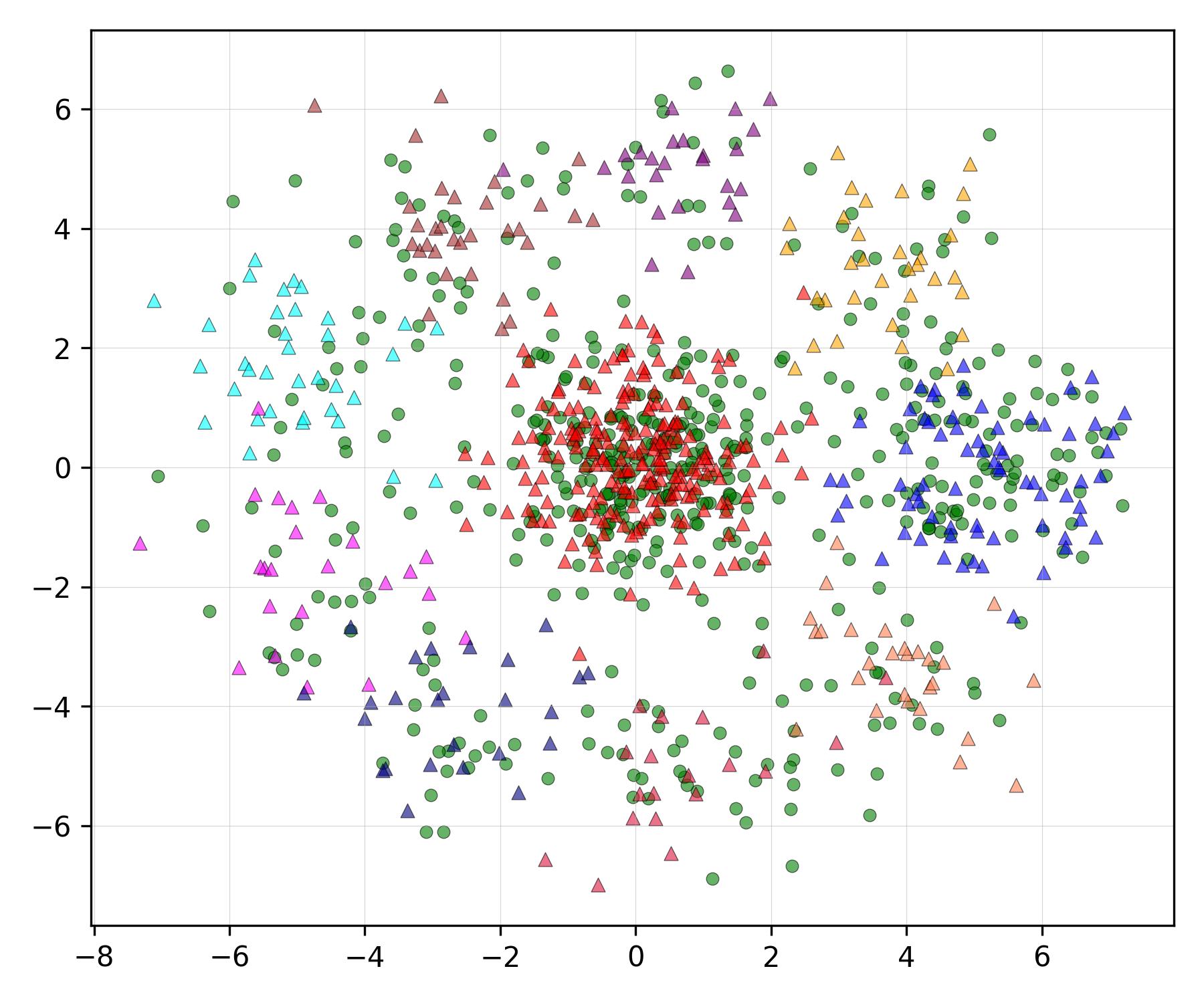}
  \caption{A sample of size 500 calibration points (triangles) and 500 test points (green circles) from the $K=10$ bivariate normal mixture introduced in \S~\ref{sim-Gaussian}, where the class probabilities are 0.45 for the central class, 0.15 for one peripheral class, and 0.05 for each of the remaining eight classes. The distance between the center of the central bivariate distribution and the other centers is SNR = 3  in the left panel and SNR=5 in the right panel. }
  \label{app-sim-Gaussian-K10-datagen}
\end{figure}

Next, we address  the setting of a label shift from the training distribution to the calibration and test distributions. We consider both directions of shift between the even and uneven class distributions:
even-to-uneven denotes even class probabilities in training and uneven
class probabilities in calibration and testing, while uneven-to-even
denotes the reverse. For each SNR, logistic regression is trained on
$10{,}000$ observations from the corresponding training distribution.
For vector scaling, $20\%$ of the calibration sample is used to estimate
the shift as detailed in \S~\ref{app-vector-scaling}, and the remaining $80\%$ is used for conformal calibration.

Figure~\ref{app-fig-bvndatageneration-results-K10-VS} shows the results.
Without correction, the power of \infoOSP\ and \infoOSPCal\ can be low
because the estimated conditional class probabilities target the
training distribution rather than the calibration and test distribution.
Vector scaling substantially improves
the power of both procedures. We refero to the bias-only vector-scaling correction derived in \S~\ref{app-vector-scaling} as Bias VS. Among the methods that control the FCR,
the corrected procedures, \infoOSP\ (Bias VS) and
\infoOSPCal\ (Bias VS), attain the highest resolution-adjusted power
across both shift directions and both notions of informativeness, often
by a substantial margin.


\begin{figure}[!htbp]
  \centering
  \SMlegend{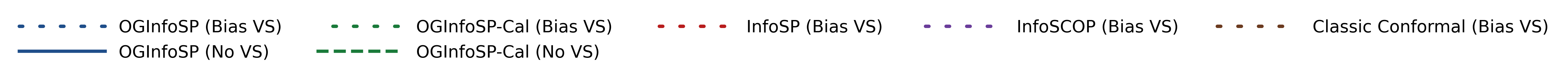}

  \includegraphics[width=0.24\linewidth]
  {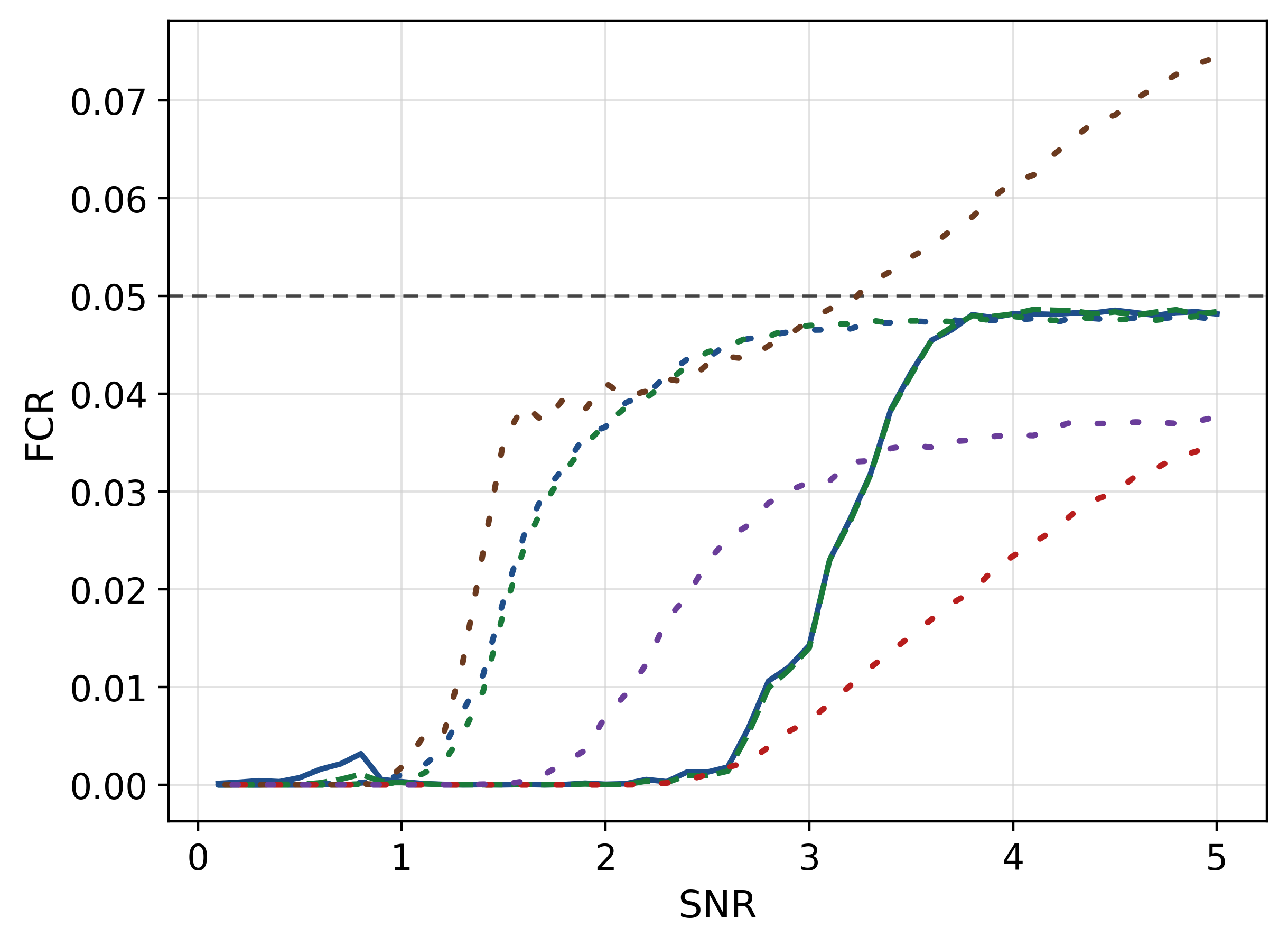}
  \hfill
  \includegraphics[width=0.24\linewidth]
  {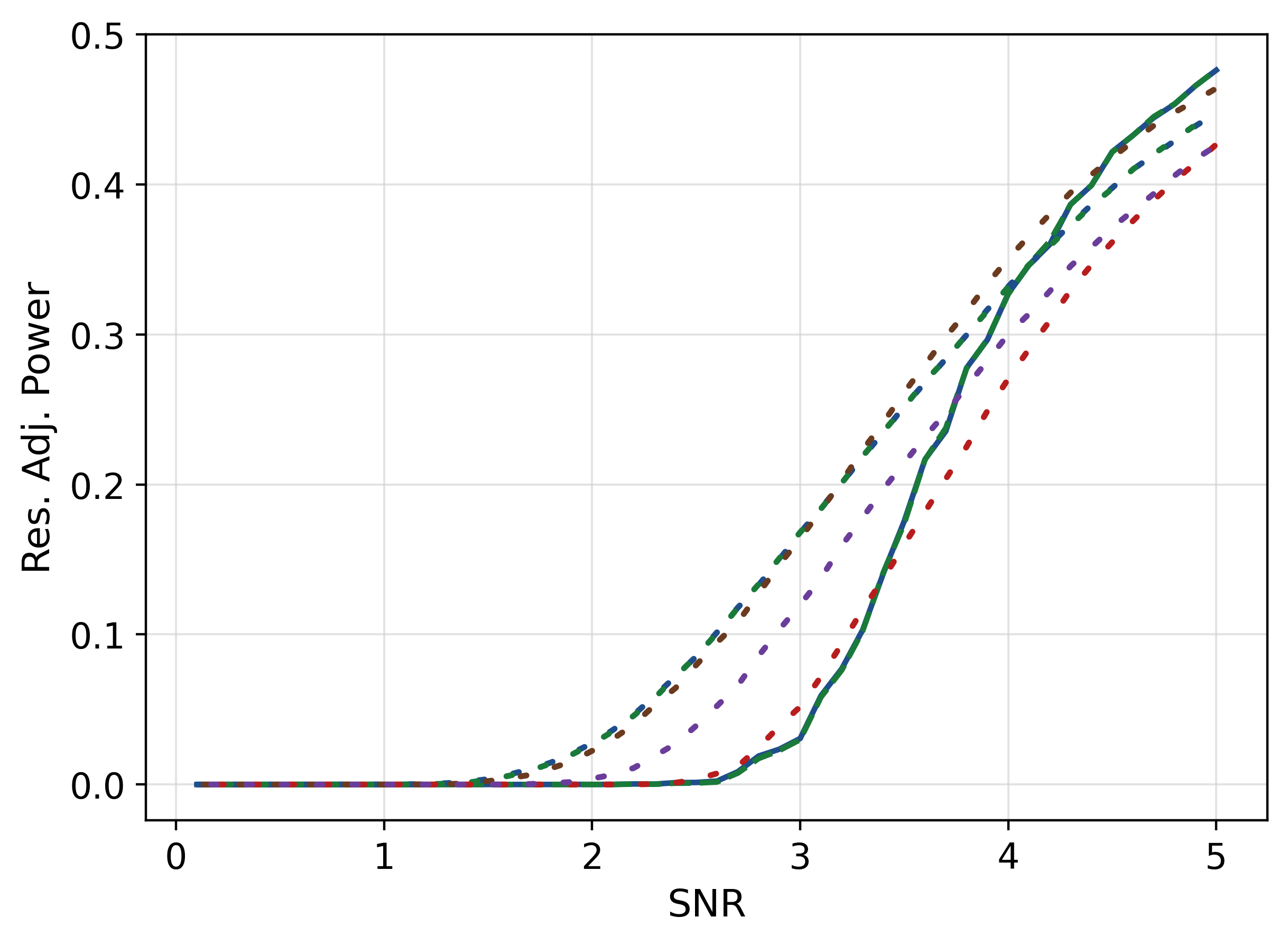}
  \hfill
  \includegraphics[width=0.24\linewidth]
  {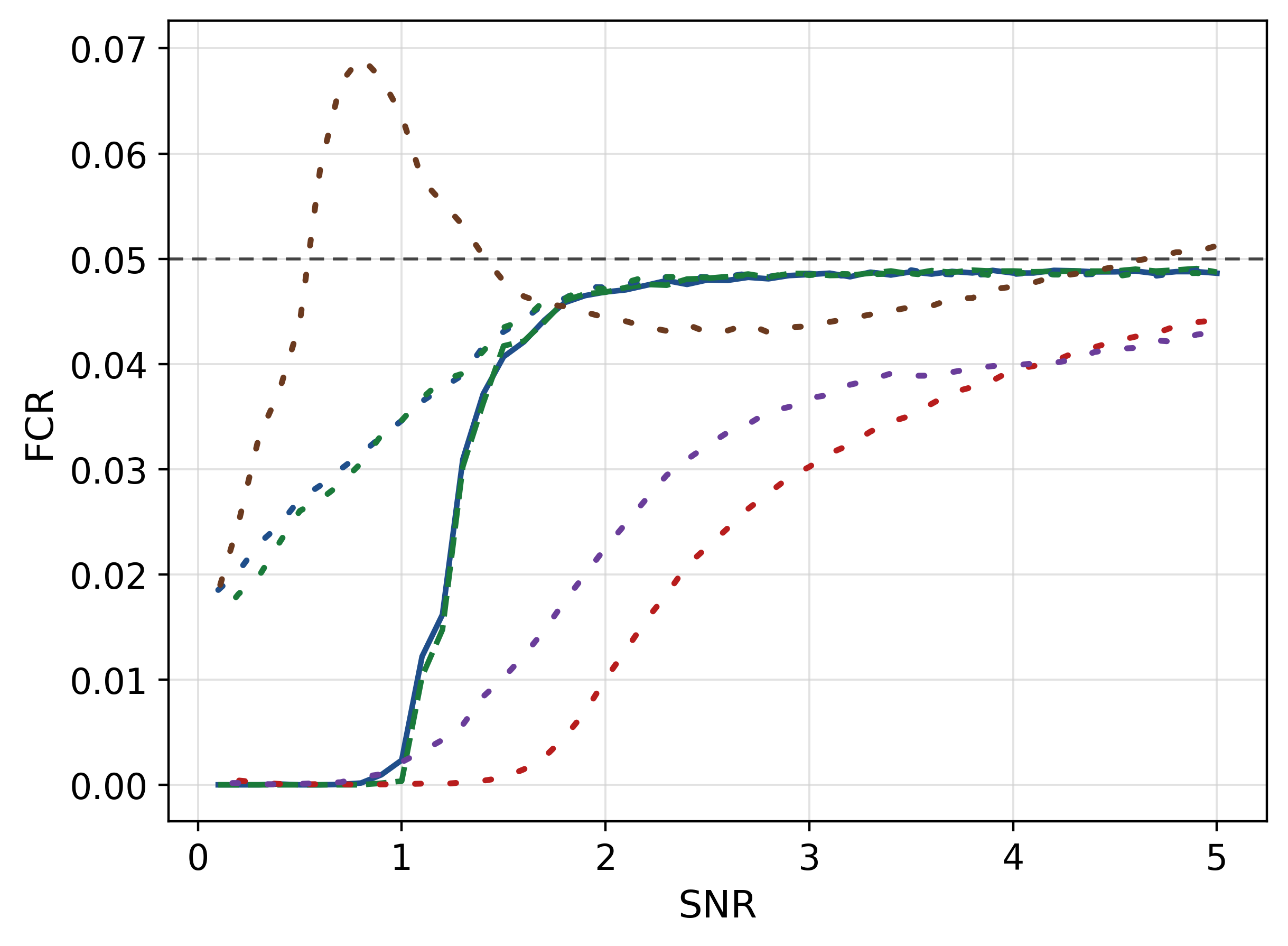}
  \hfill
  \includegraphics[width=0.24\linewidth]
  {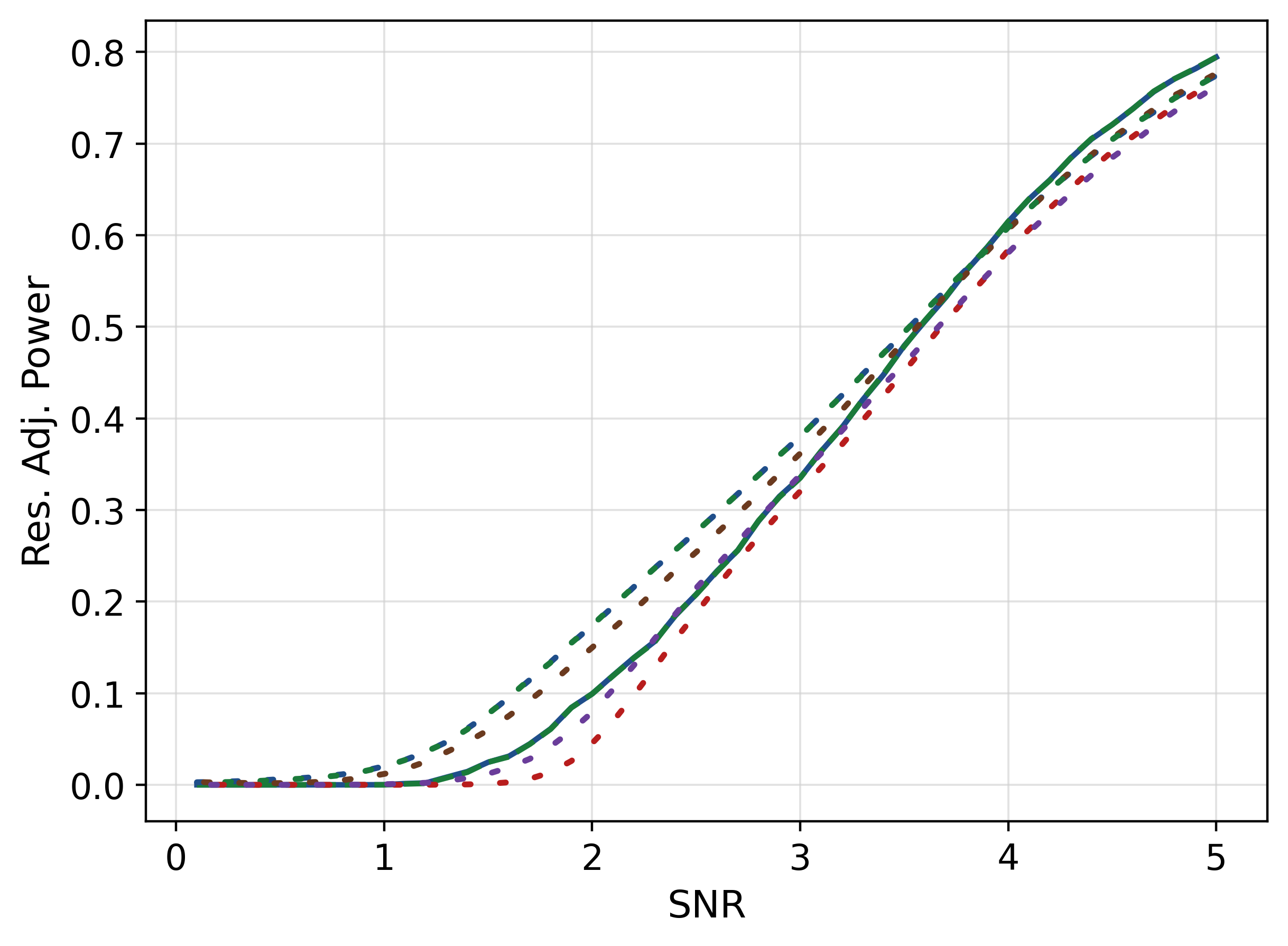}

  \par\smallskip

  \includegraphics[width=0.24\linewidth]
  {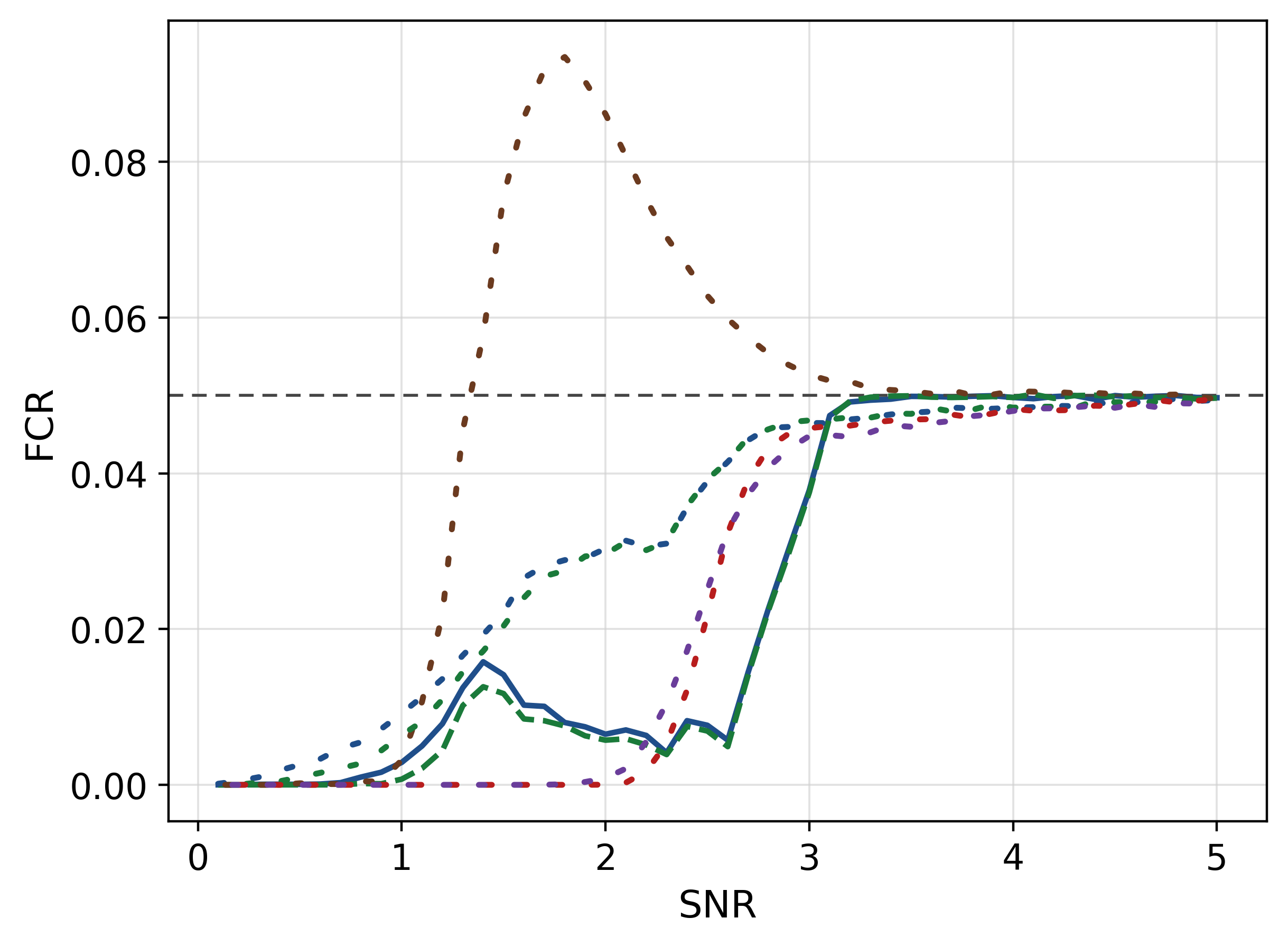}
  \hfill
  \includegraphics[width=0.24\linewidth]
  {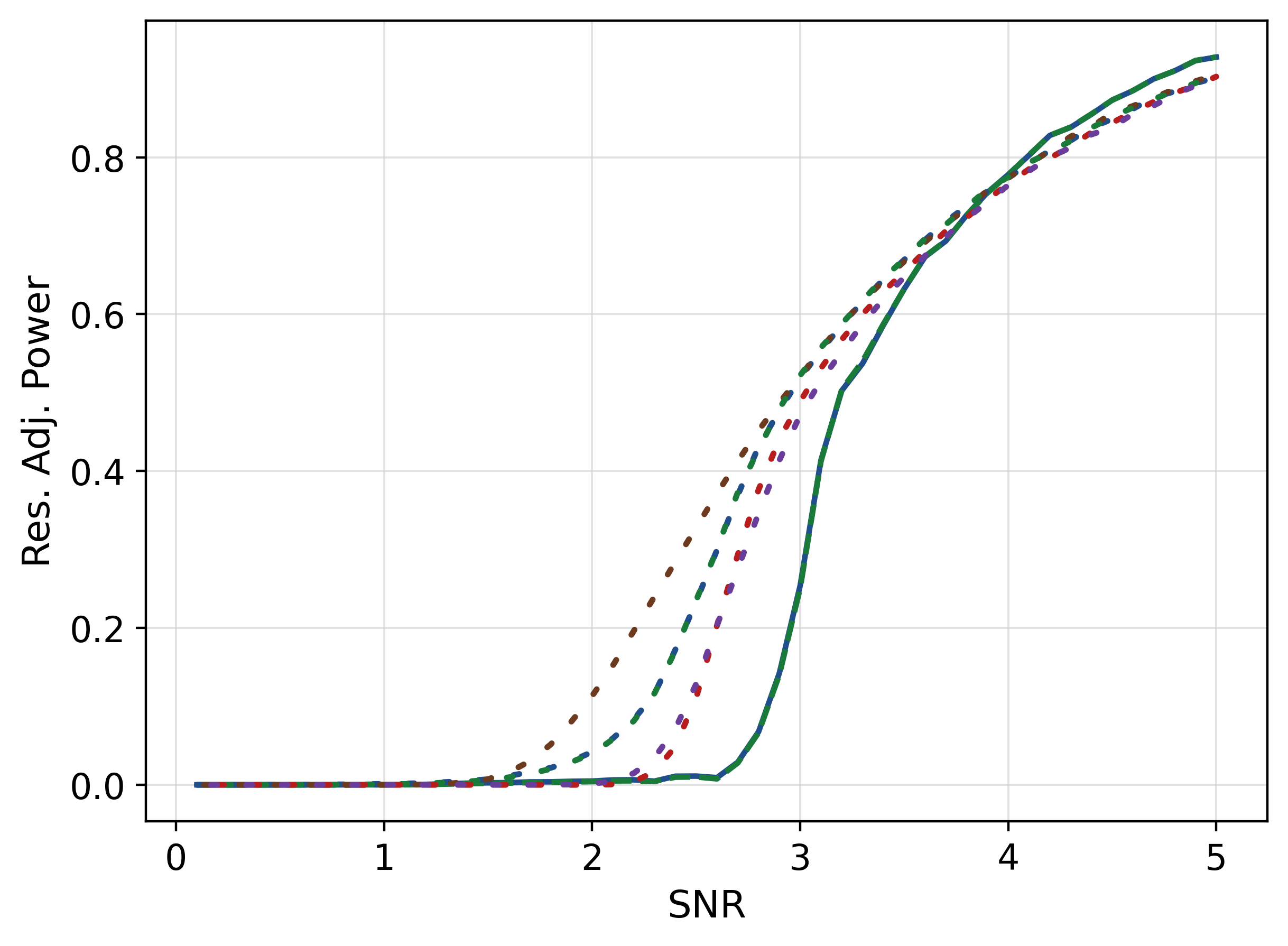}
  \hfill
  \includegraphics[width=0.24\linewidth]
  {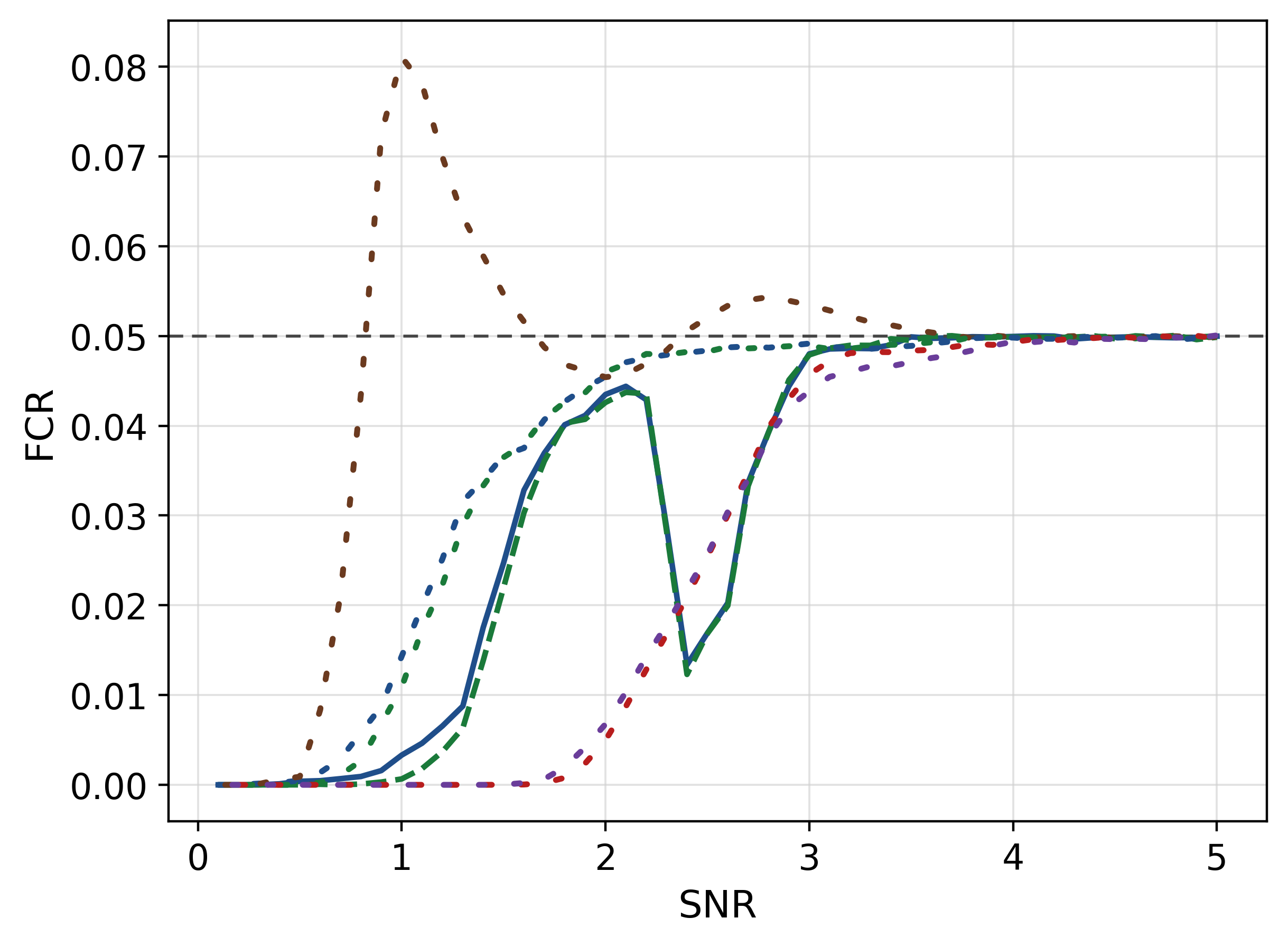}
  \hfill
  \includegraphics[width=0.24\linewidth]
  {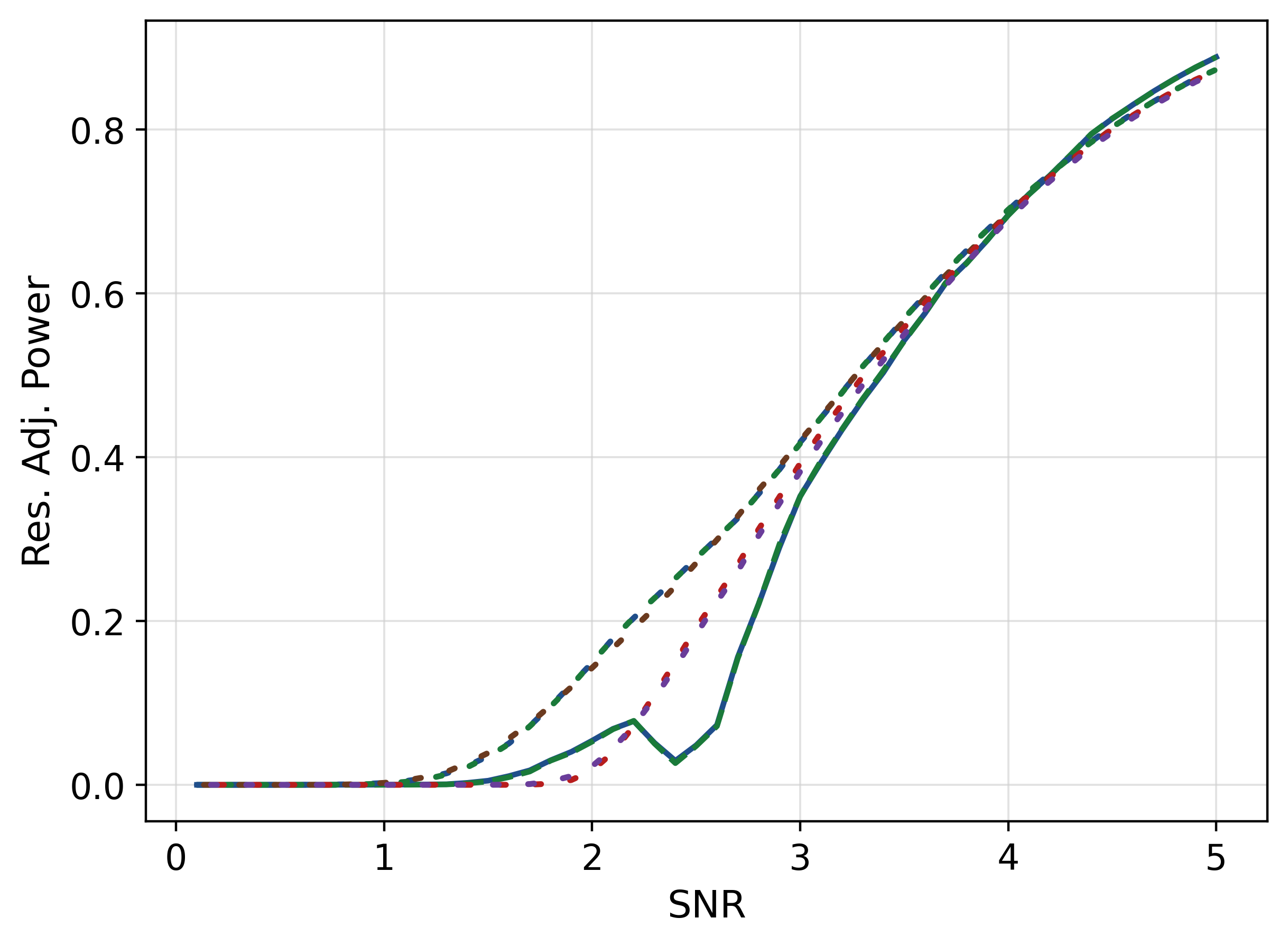}

  \caption{Label-shift results for the $K=10$ bivariate normal mixture.
Columns 1 and 2 correspond to the even-to-uneven setting, and columns
3 and 4 to the uneven-to-even setting. Within each pair, the first
column shows the FCR and the second shows resolution-adjusted power.
Informativeness is defined by exclusion of the central class in the top row
and prediction-set size at most three in the bottom row. We report the
vector-scaling-corrected versions of all methods, together with the
uncorrected versions of \infoOSP\ and \infoOSPCal. The initial
calibration and test sample sizes are $n=m=500$, $\alpha=0.05$, and
results are averaged over $10{,}000$ iterations. The horizontal line
in the FCR panels marks the nominal level.}
  \label{app-fig-bvndatageneration-results-K10-VS}
\end{figure}

The effect of vector scaling on the competing procedures is less systematic. Figure \ref{app-fig-bvndatageneration-results-K10-competitors-VS} shows that it generally improves their power, but not in every setting; the corrected novel procedures remain better or comparable. The main exception is the even-to-uneven setting when informativeness requires exclusion of the central class. Here, the central-class prior increases from 0.10 in training to 0.45 in calibration and testing, so vector scaling upweights its fitted conditional probability by a factor of 4.5 before renormalization. Because the competitors use one minus the fitted class probability as the nonconformity score, this correction lowers the score for the central class and tends to increase its conformal p-value. The central class is therefore harder to exclude, substantially fewer examples are selected, and resolution-adjusted power decreases. Thus, in this particular setting, the uncorrected scores are more powerful for the competing procedures, even though the corrected probabilities better reflect the target distribution.

\begin{figure}[!htbp]
  \centering
  \SMlegend{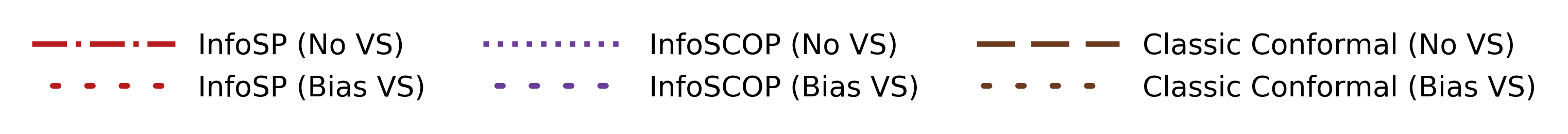}

  \includegraphics[width=0.32\linewidth]
  {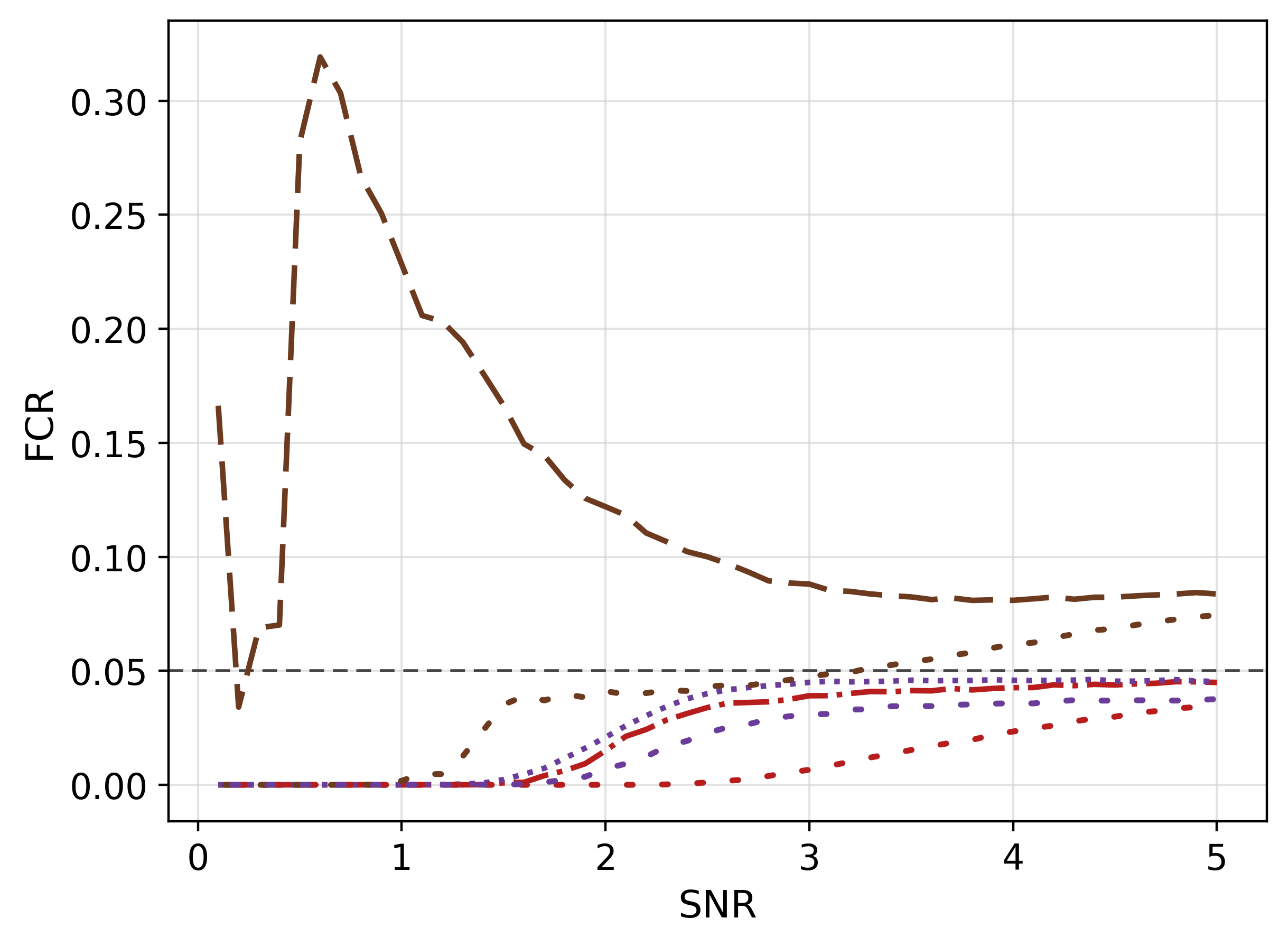}
  \hfill
  \includegraphics[width=0.32\linewidth]
  {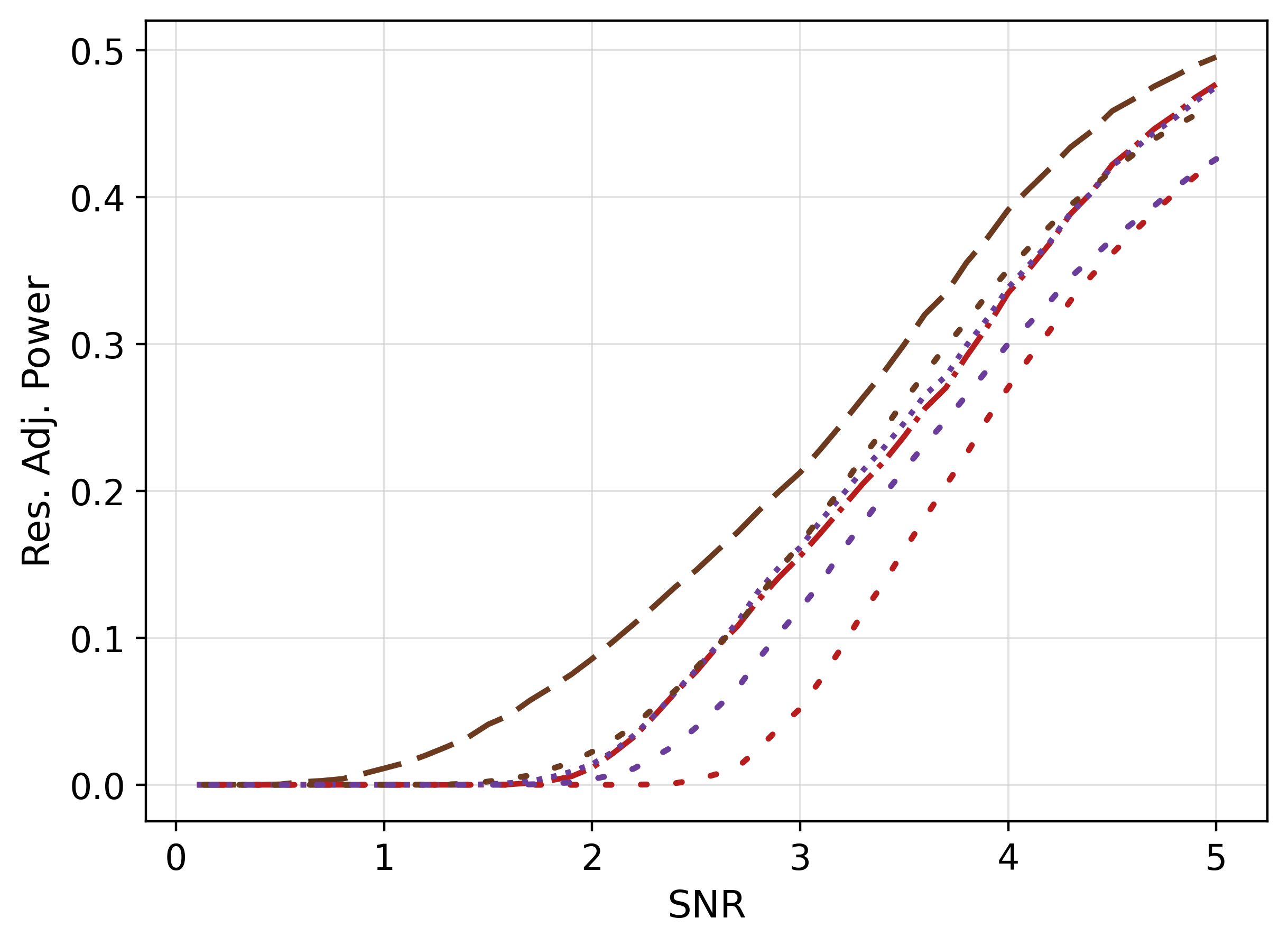}
  \hfill
  \includegraphics[width=0.32\linewidth]
  {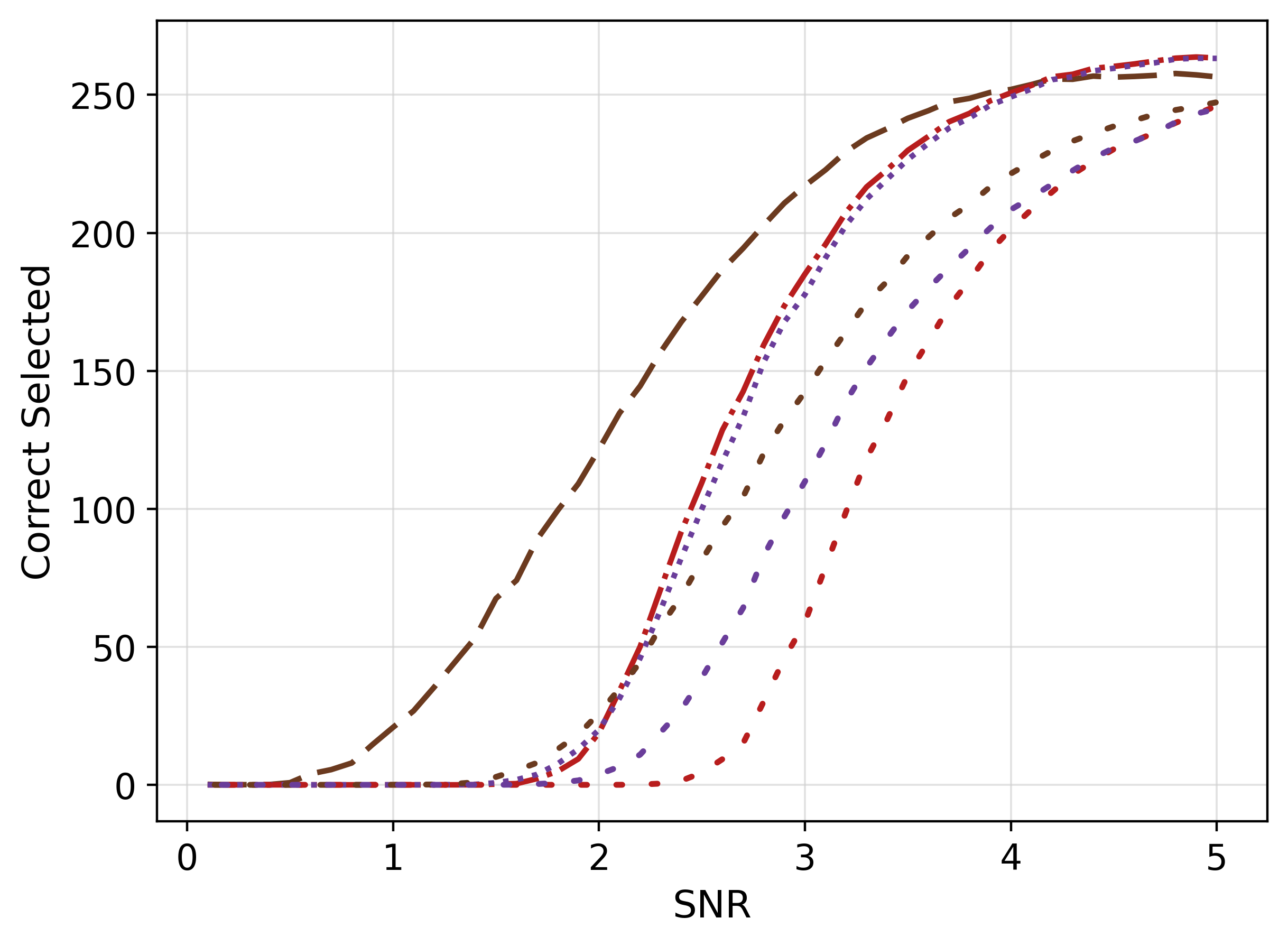}

  \par\smallskip

  \includegraphics[width=0.32\linewidth]
  {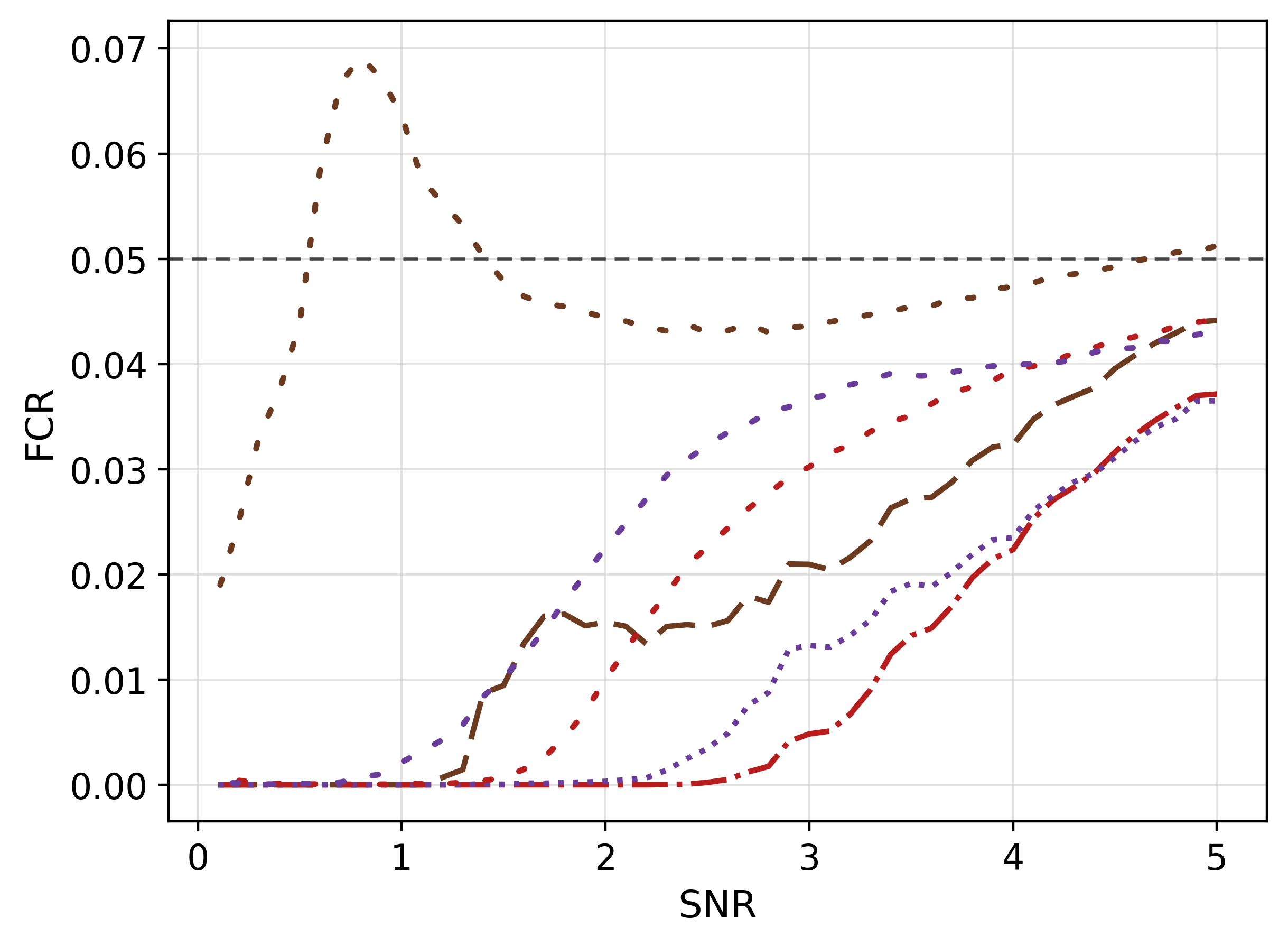}
  \hfill
  \includegraphics[width=0.32\linewidth]
  {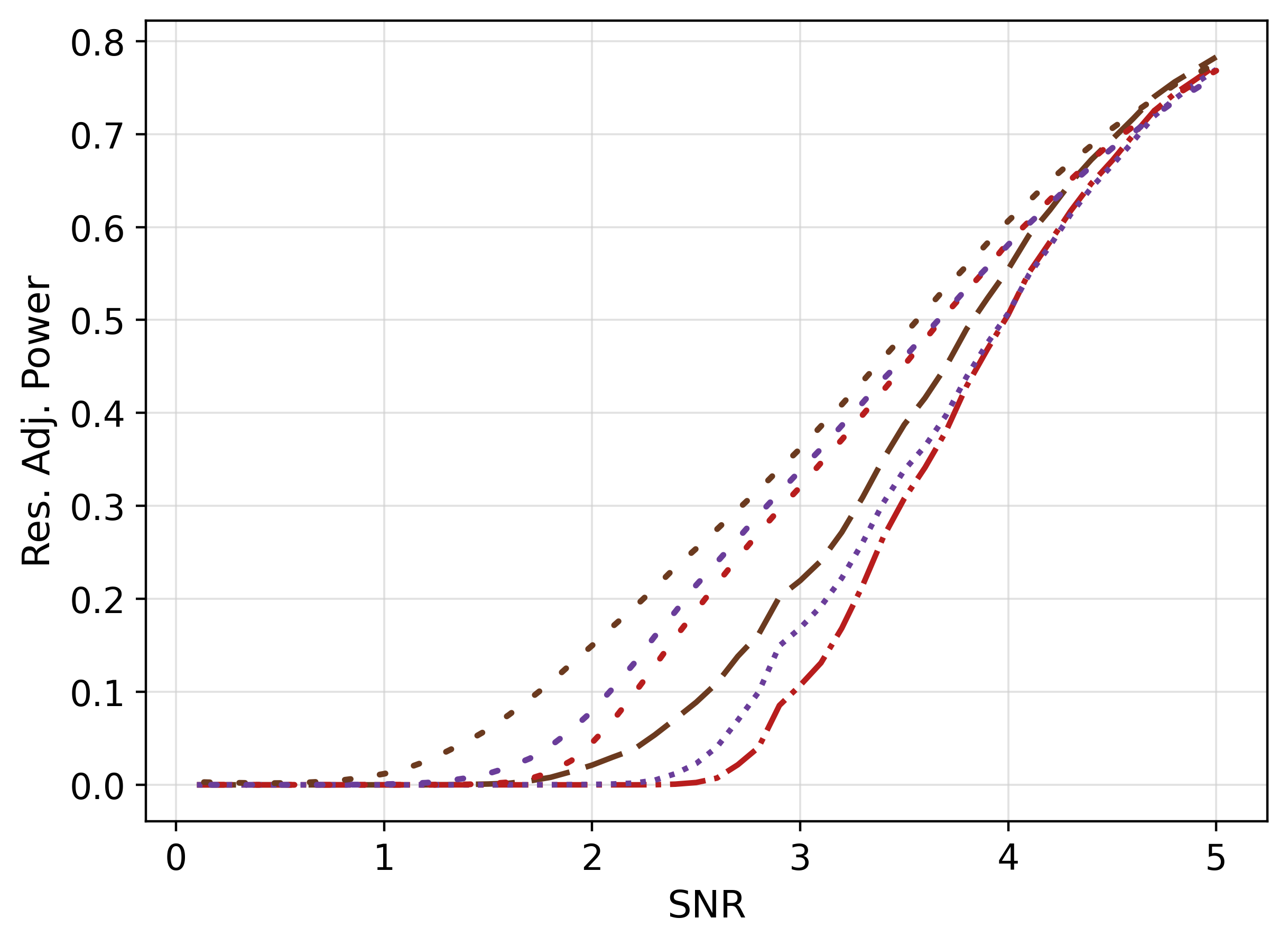}
  \hfill
  \includegraphics[width=0.32\linewidth]
  {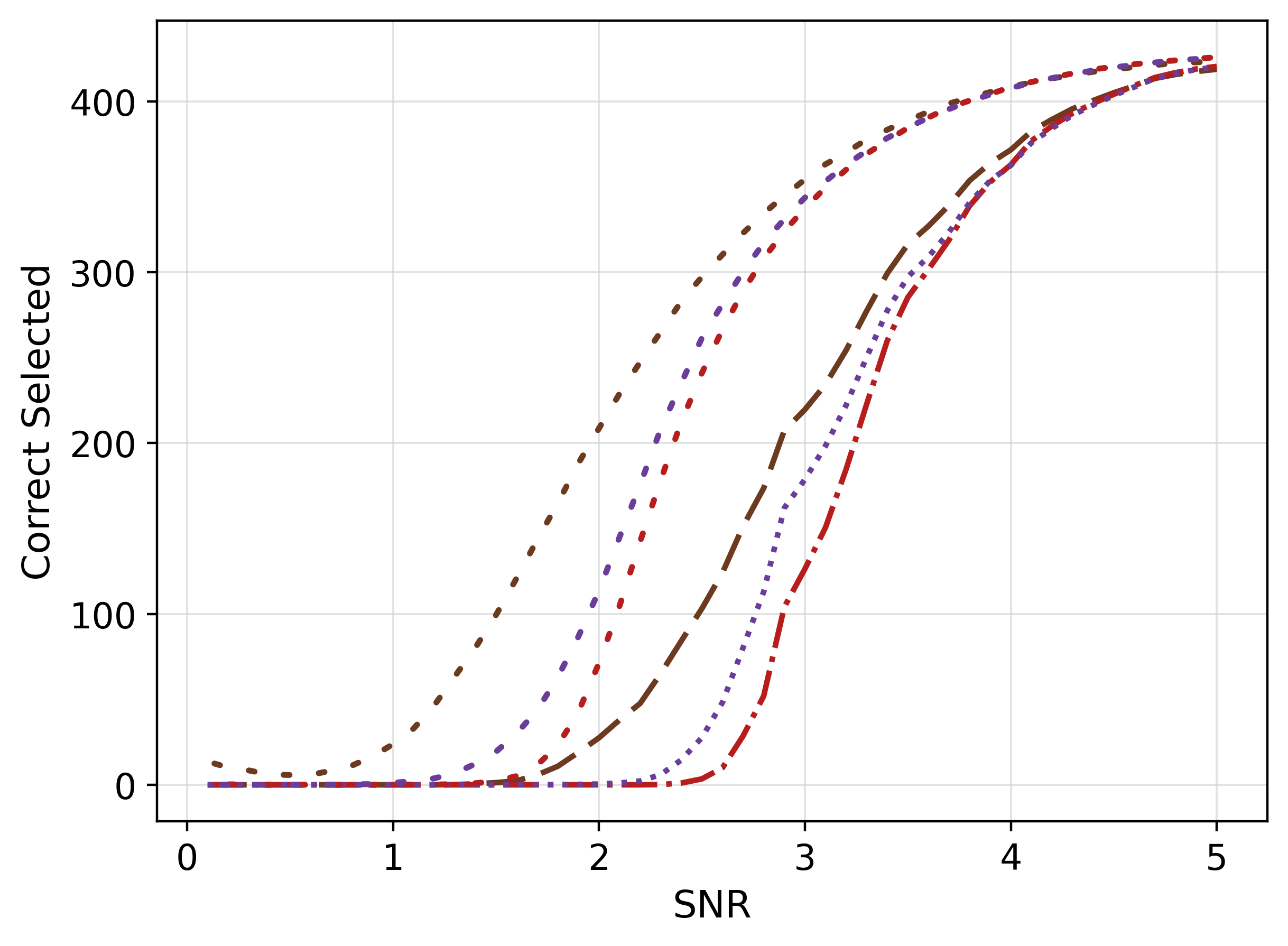}

  \par\smallskip

  \includegraphics[width=0.32\linewidth]
  {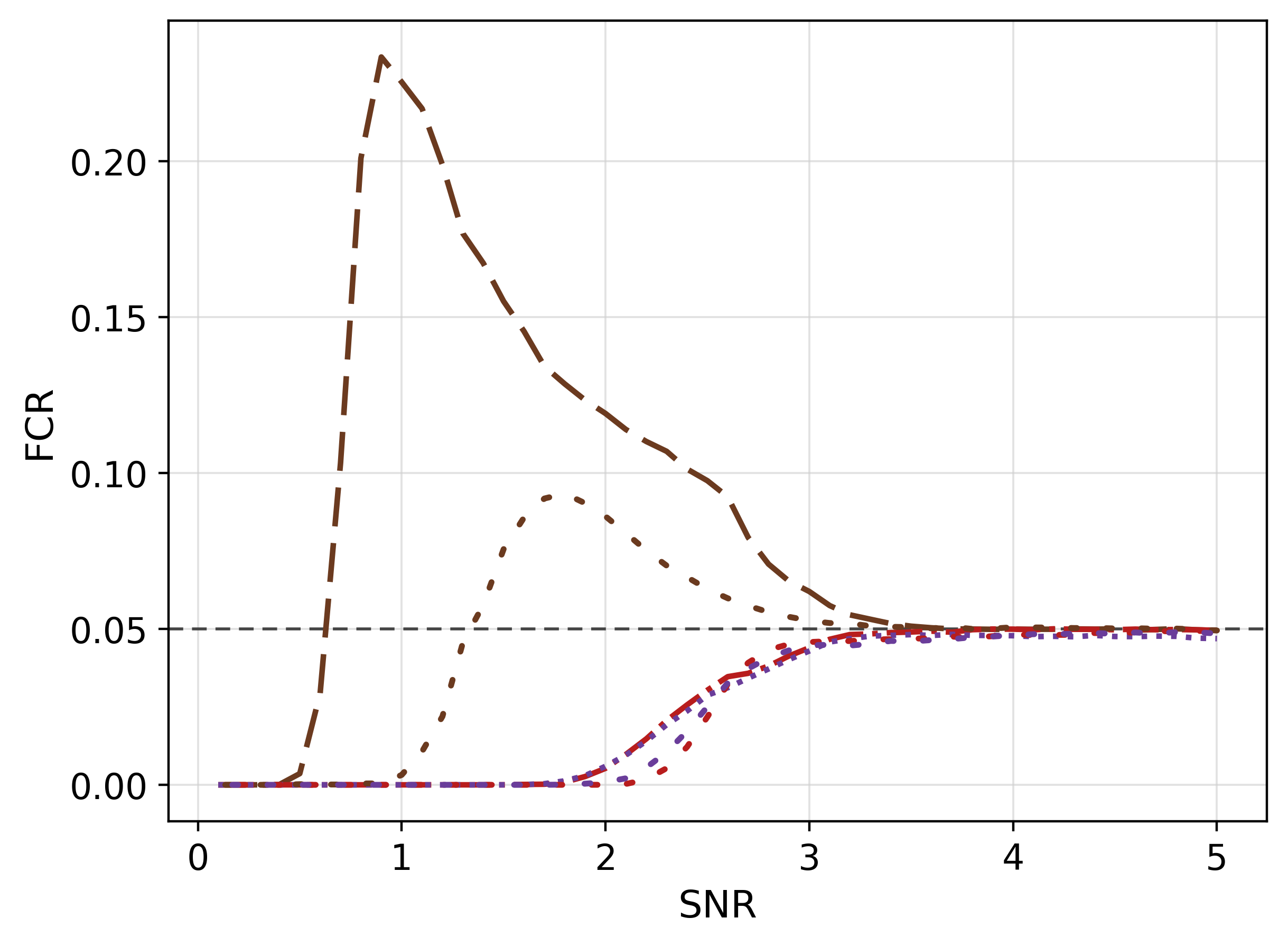}
  \hfill
  \includegraphics[width=0.32\linewidth]
  {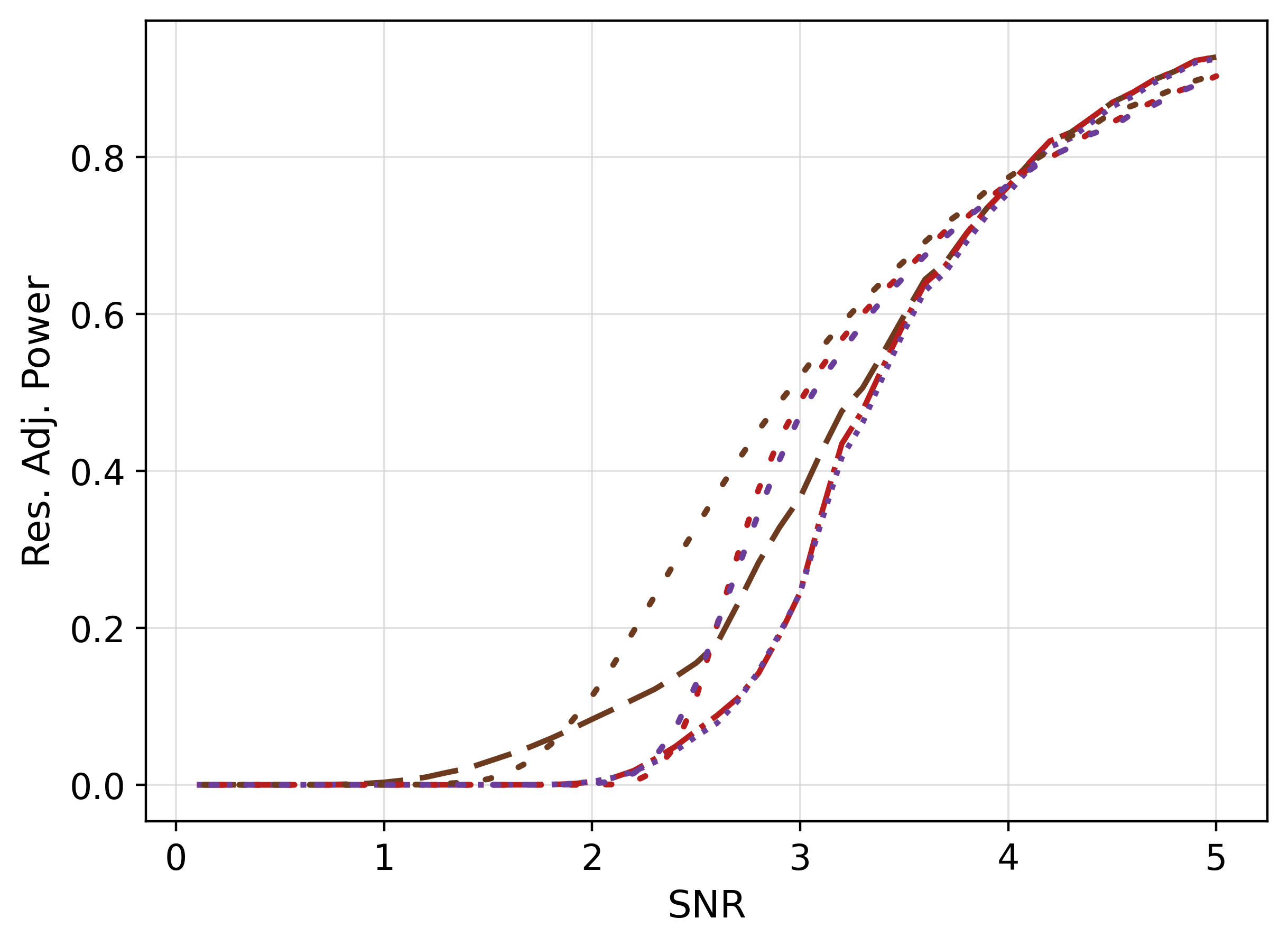}
  \hfill
  \includegraphics[width=0.32\linewidth]
  {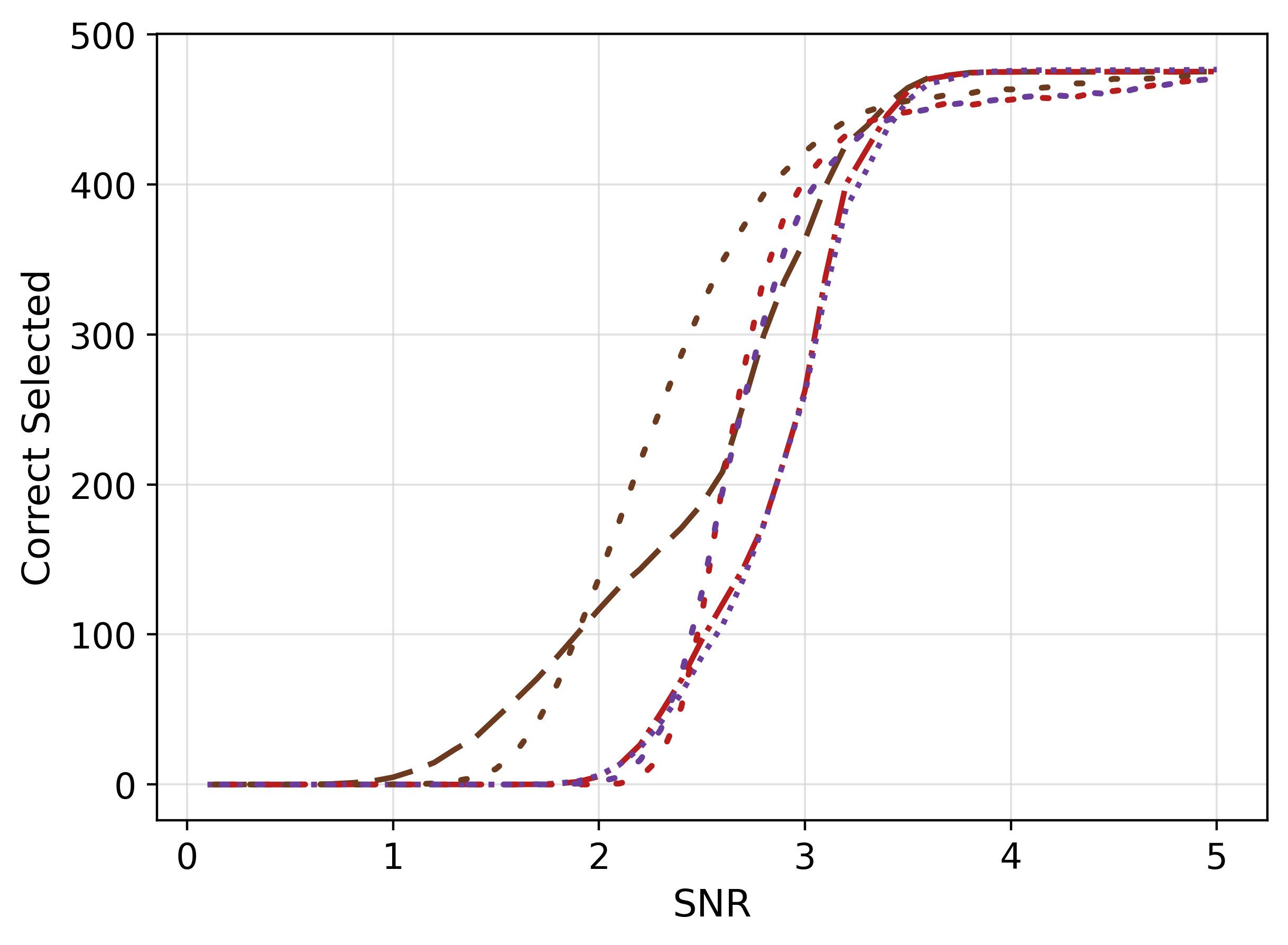}

  \par\smallskip

  \includegraphics[width=0.32\linewidth]
  {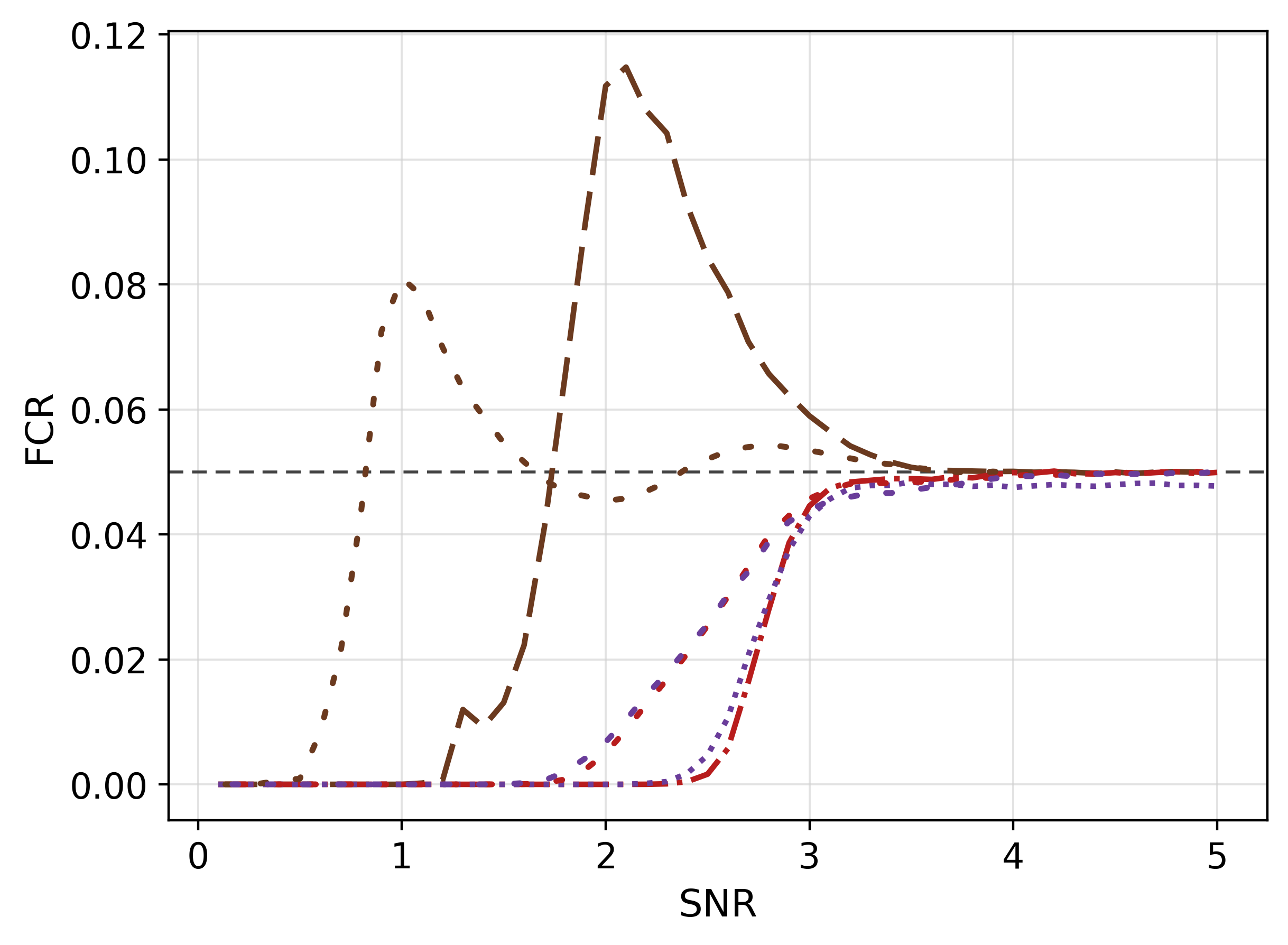}
  \hfill
  \includegraphics[width=0.32\linewidth]
  {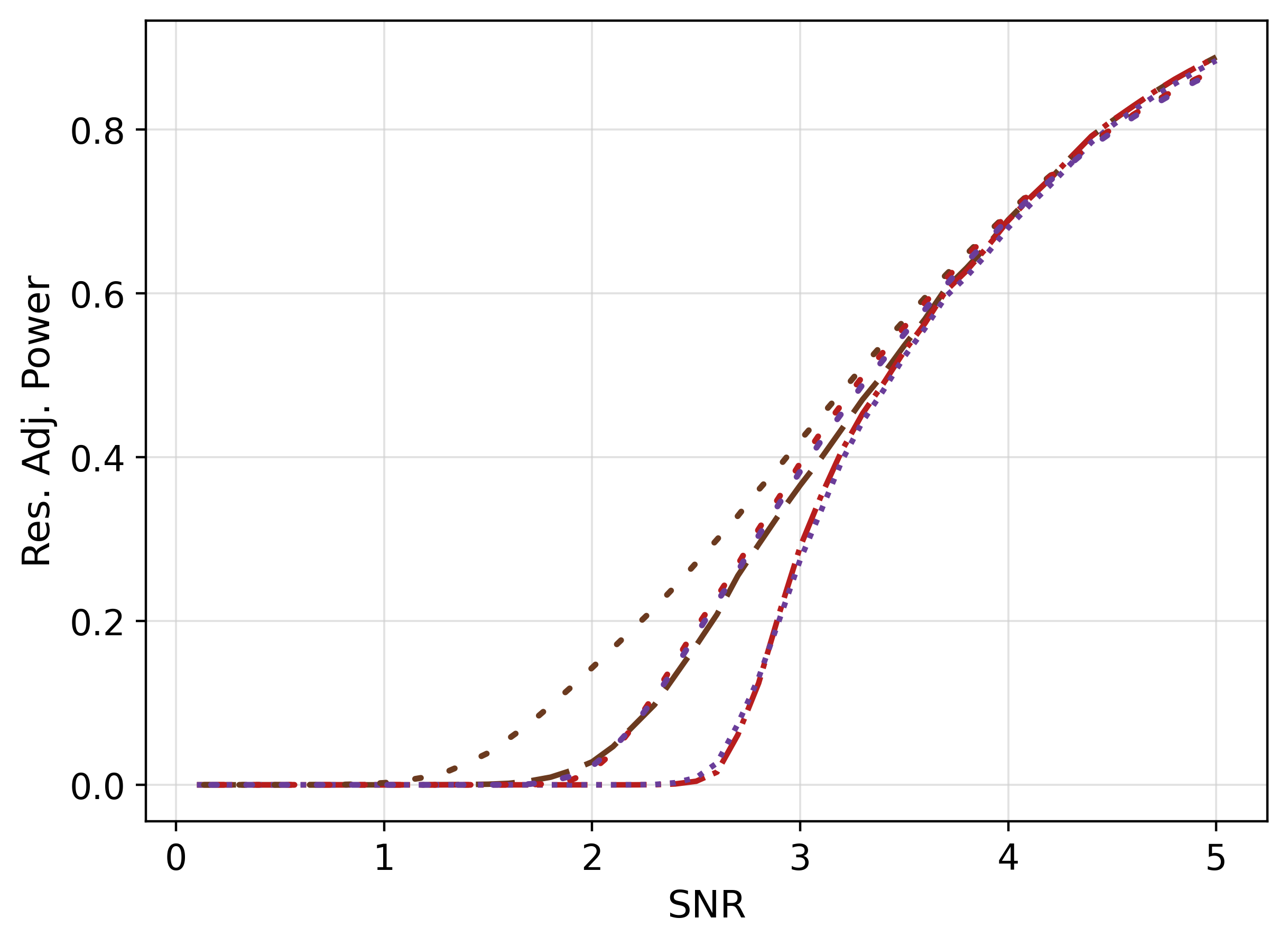}
  \hfill
  \includegraphics[width=0.32\linewidth]
  {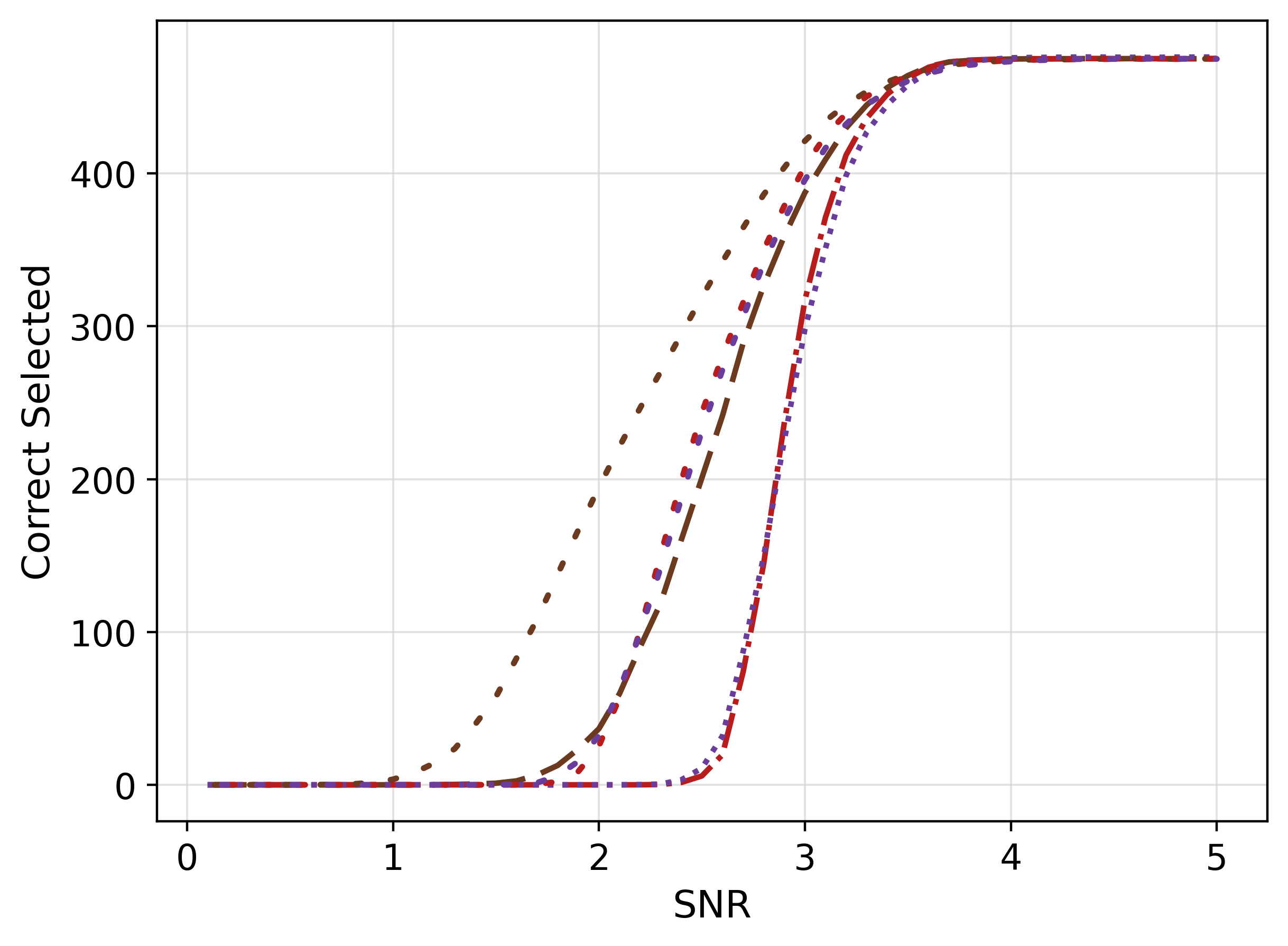}

  \caption{Effect of vector scaling on the competing procedures for the
  $K=10$ bivariate normal mixture under label shift. The three columns
  show the FCR, resolution-adjusted power, and the mean number of
  selected prediction sets containing the true outcome. The four rows
  correspond, respectively, to exclusion of the central class in the
  even-to-uneven setting, exclusion of the central class in the
  uneven-to-even setting, prediction-set size at most three in the
  even-to-uneven setting, and prediction-set size at most three in the
  uneven-to-even setting. We compare InfoSP, InfoSCOP, and classic
  conformal with and without vector scaling. Here $n=m=500$ and
  $\alpha=0.05$. For vector scaling, $100$ calibration observations are
  used to estimate the shift, leaving $400$ for conformal calibration.
  The horizontal line in the FCR panels marks the nominal level.}
  \label{app-fig-bvndatageneration-results-K10-competitors-VS}
\end{figure}

\subsection{Image-classification experiments on CIFAR-10 with   $K=3$ classes}\label{app-sim-CIFAR-K3}
We consider here the setting in \cite{GazinHellerMarandonRoquain2025} with  $K=3$ classes (birds, cats, and dogs) from the  CIFAR-10 image datasets \citep{krizhevsky2009learning}. The  prediction algorithm construction is detailed in \S~\ref{sim-CIFAR}.  The training sample has even classes. 

We added another competitor, that uses the non-conformity score in \cite{Romano2020} for prediction set construction. We denote the procedure using this score by Classic Conformal Adaptive Score. This non-conformity score is expected to have better conditional coverage, at the cost of wider prediction sets. 

Figure \ref{app-fig-cifar10-K3-exclude} shows a comparison of the different methods when informativeness means excluding the class of cats. The conclusions are similar to the qualitative conclusions in \S~\ref{sim-CIFAR}: there is a  power improvement of \infoOSP\ and \infoOSPCal\ over the methods that have an FCR control guarantee, \texttt{InfoSCOP} and \texttt{InfoSP}. For Classic Conformal,  the average FCP happens to be below the nominal level only when the distribution of classes is even. The Classic Conformal Adaptive Score method has average FCP that does not exceed the nominal level, but its resolution adjusted power is lower than that of \infoOSP\ and \infoOSPCal. The method Top Class, that reports the class with largest estimated conditional probability if it is informative, has an average FCP far above the nominal level.

Figure \ref{app-fig-cifar10-K3-nontrivial} shows a comparison of the different methods when informativeness means nontrivial selection. In this setting, all methods with a theoretical FCR control guarantee perform similarly. This is not very surprising, since \texttt{InfoSP} and \infoOSP\ coincide for nontrivial selection with $K=2$.


\begin{figure}[!htbp]
  \centering
  \SMpair{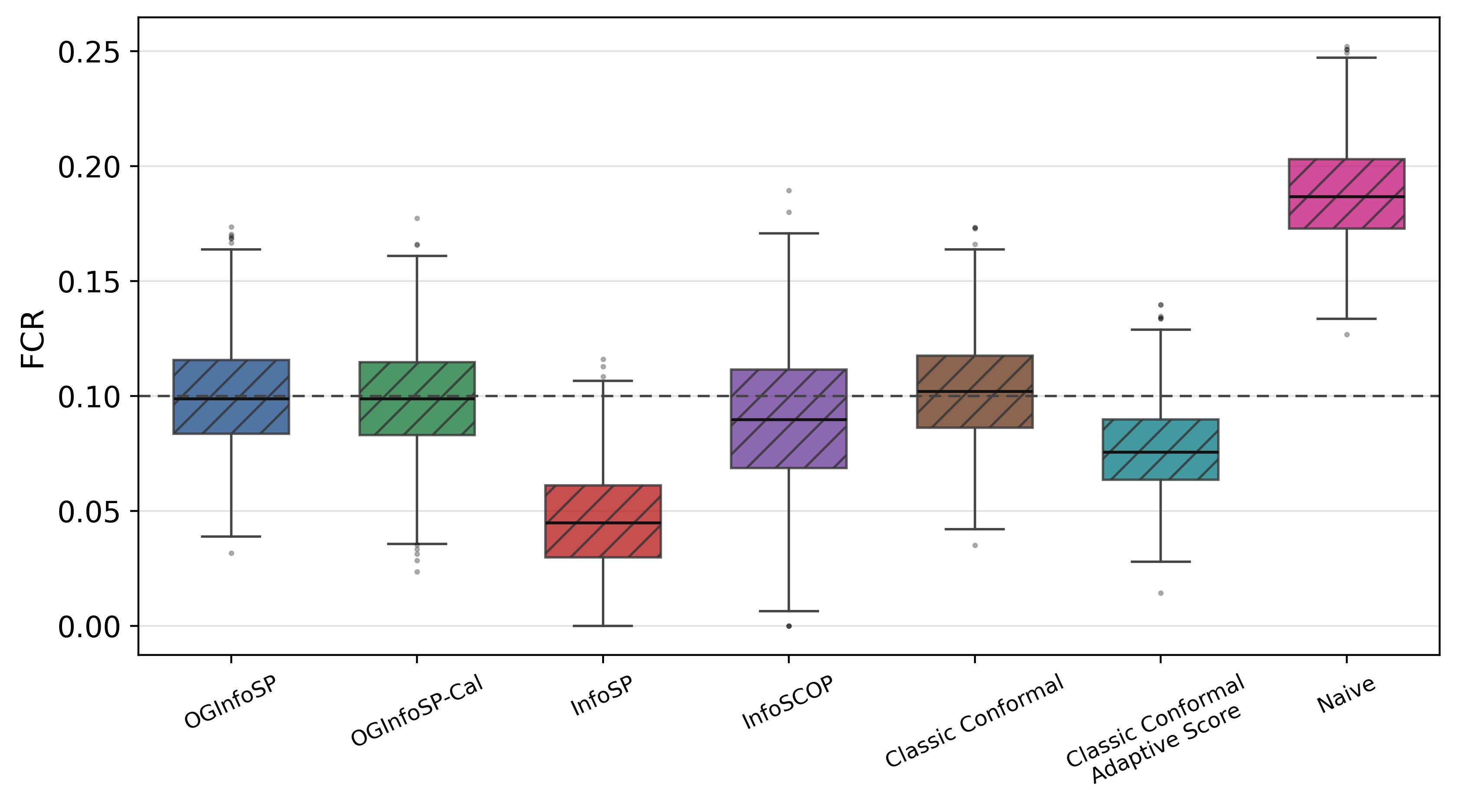}{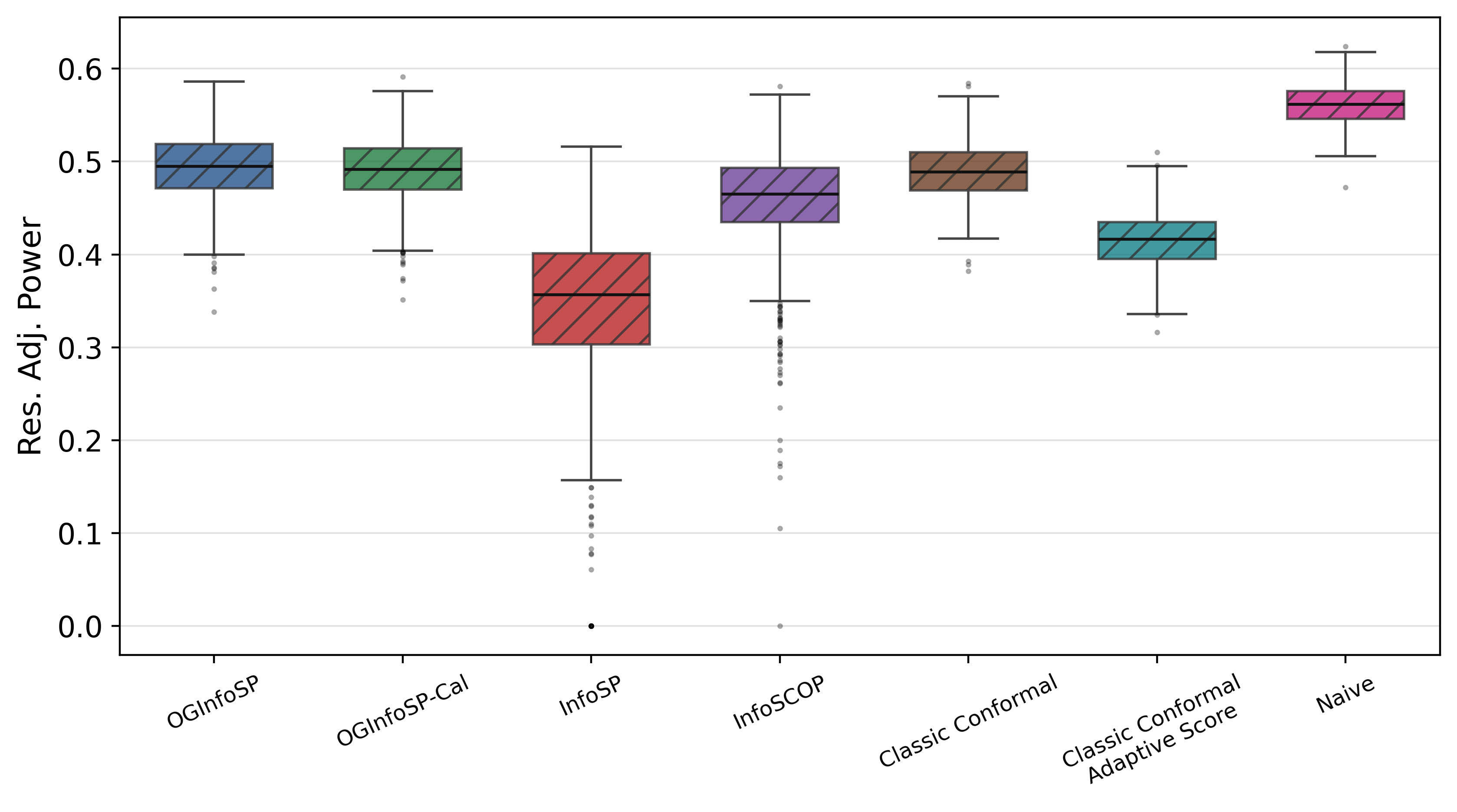}
  \SMpair{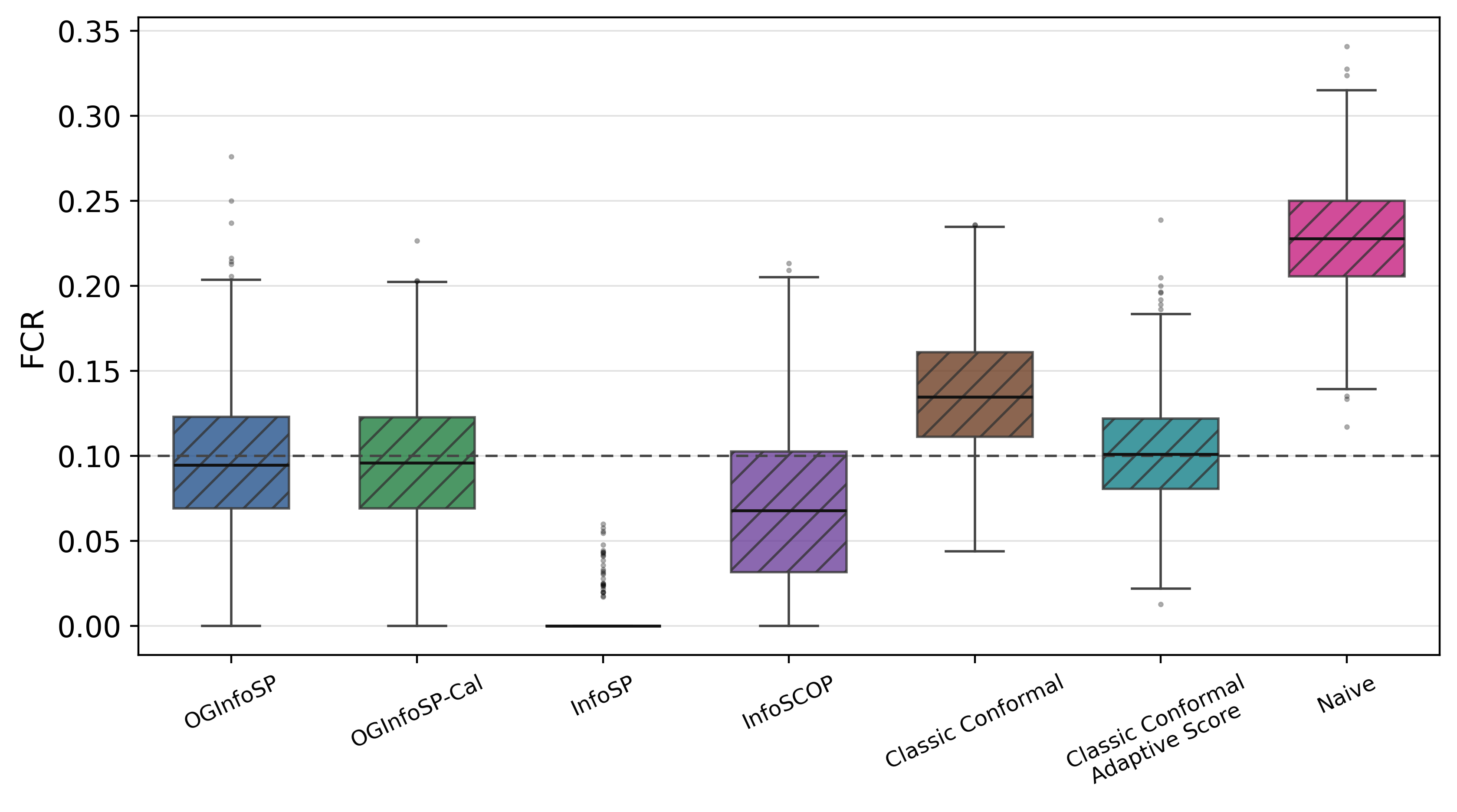}{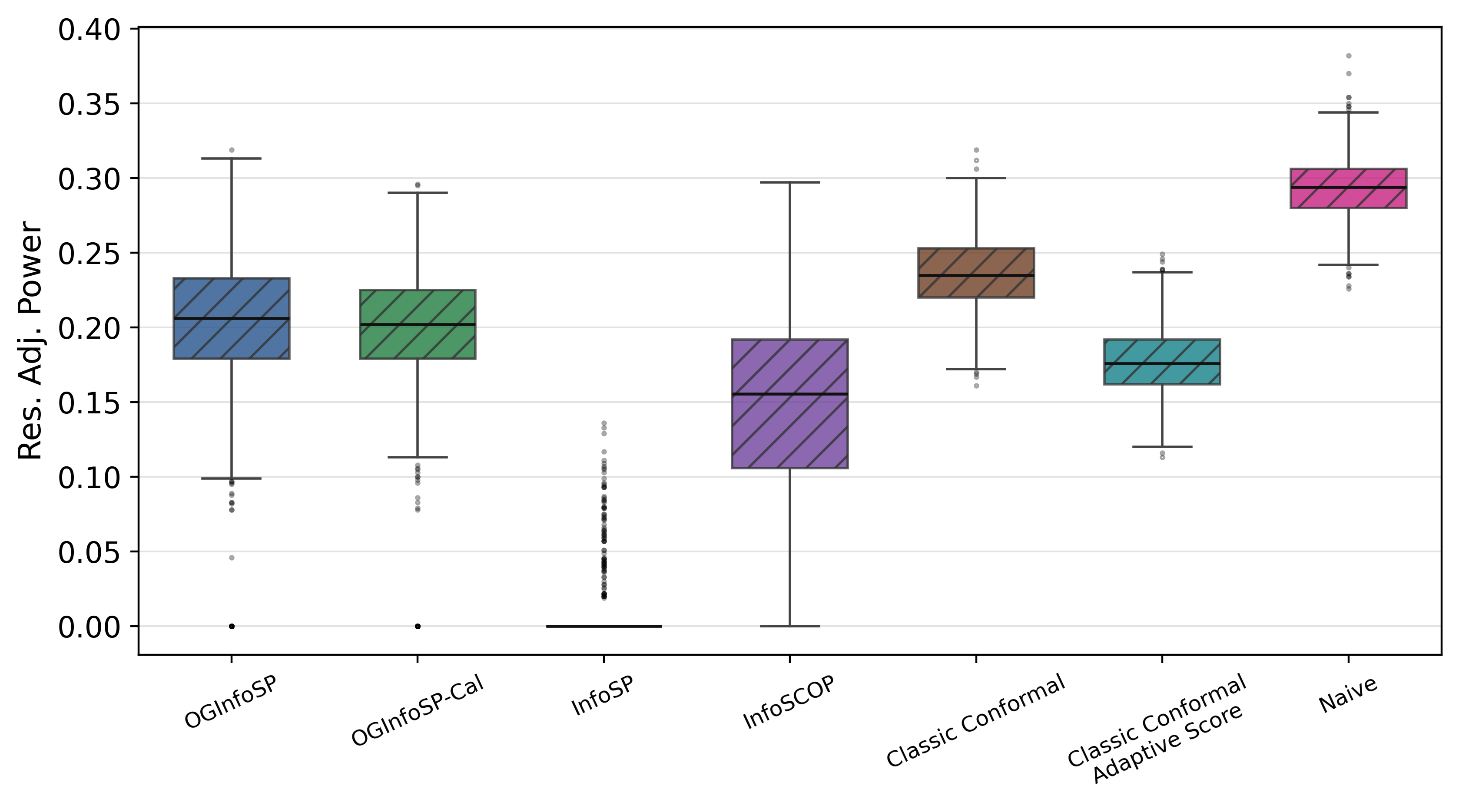}
  \caption{For the CIFAR-10 data described in \S~\ref{sim-CIFAR} with $K=3$ classes, the FCR (column 1) and resolution-adjusted power (column 2) for  prediction sets that  exclude the class of cats. The training sample had an even number of classes. The calibration and test samples come from a mixture distribution with  even class probabilities (row 1);  or uneven class probabilities (row 2), where the class of cats has probability $0.6$ and each of the other two classes has probability $0.2$. The sample sizes are $m= 500, n=1{,}000$, $\alpha = 0.1$. Half of the calibration samples are used for estimated the conditional probabilities using vector scaling, and the remaining half are used by the methods. The results are based on $1{,}000$ iterations.}\label{app-fig-cifar10-K3-exclude}
\end{figure}

\begin{figure}[!htbp]
  \centering
  \SMpair{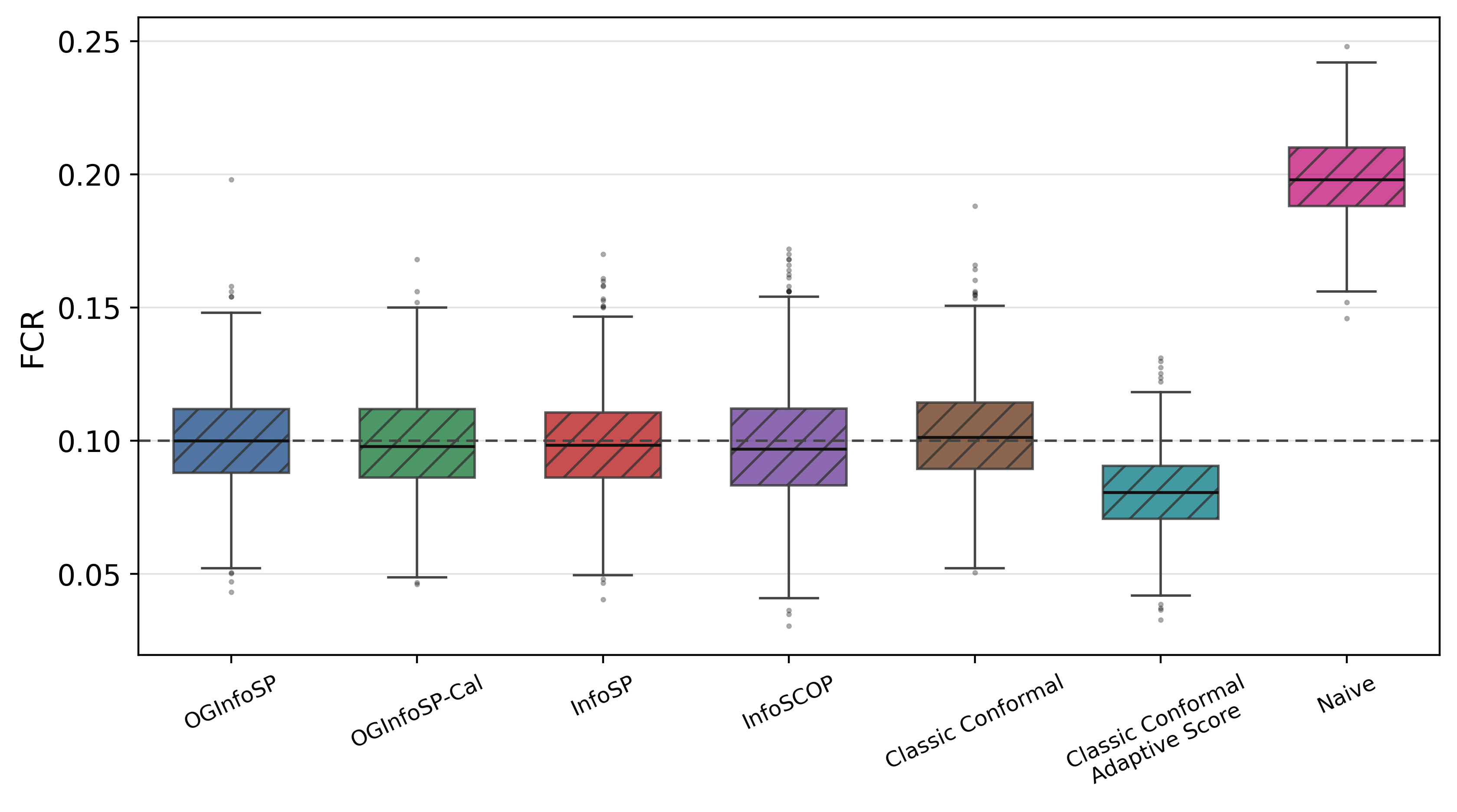}{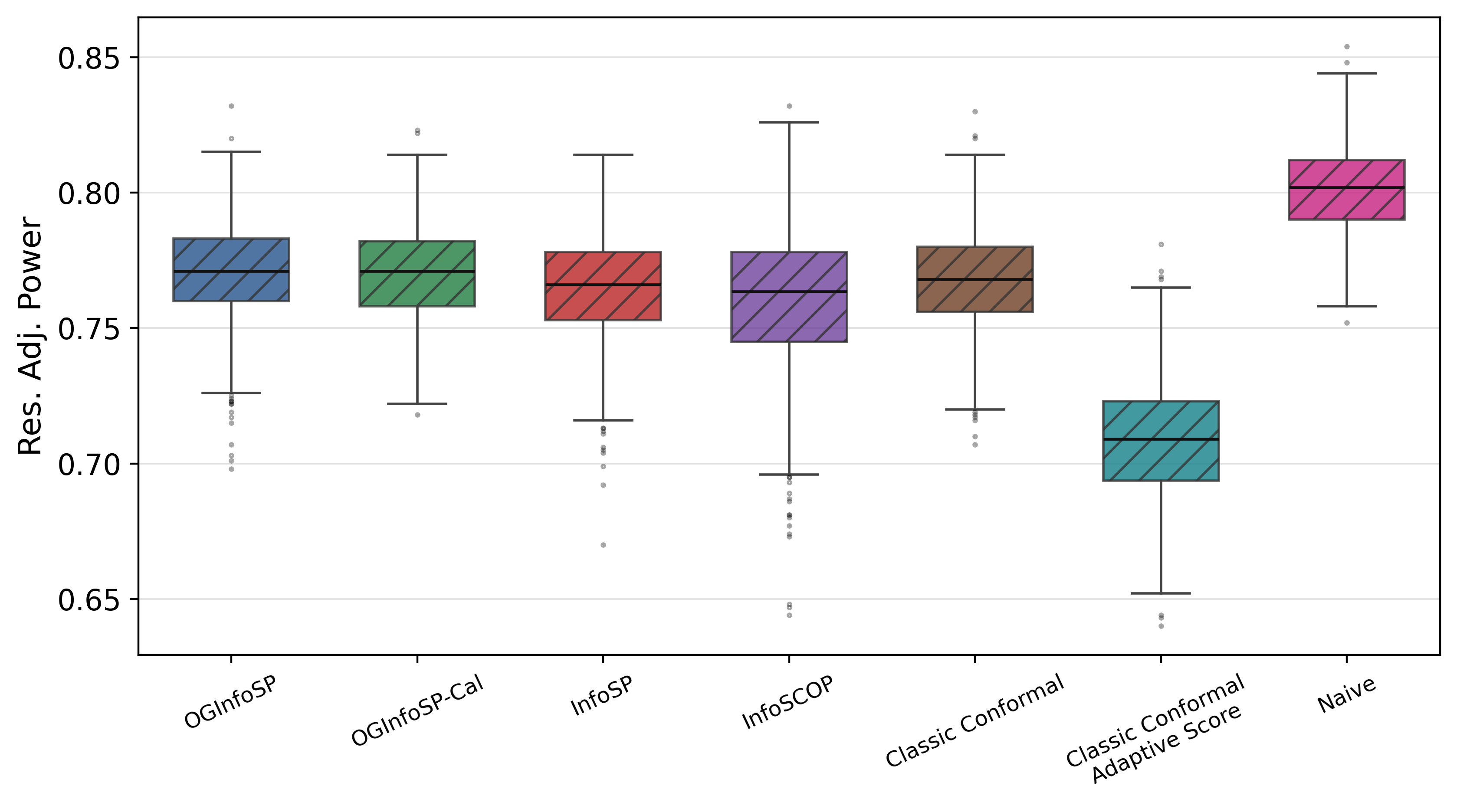}
 
  \SMpair{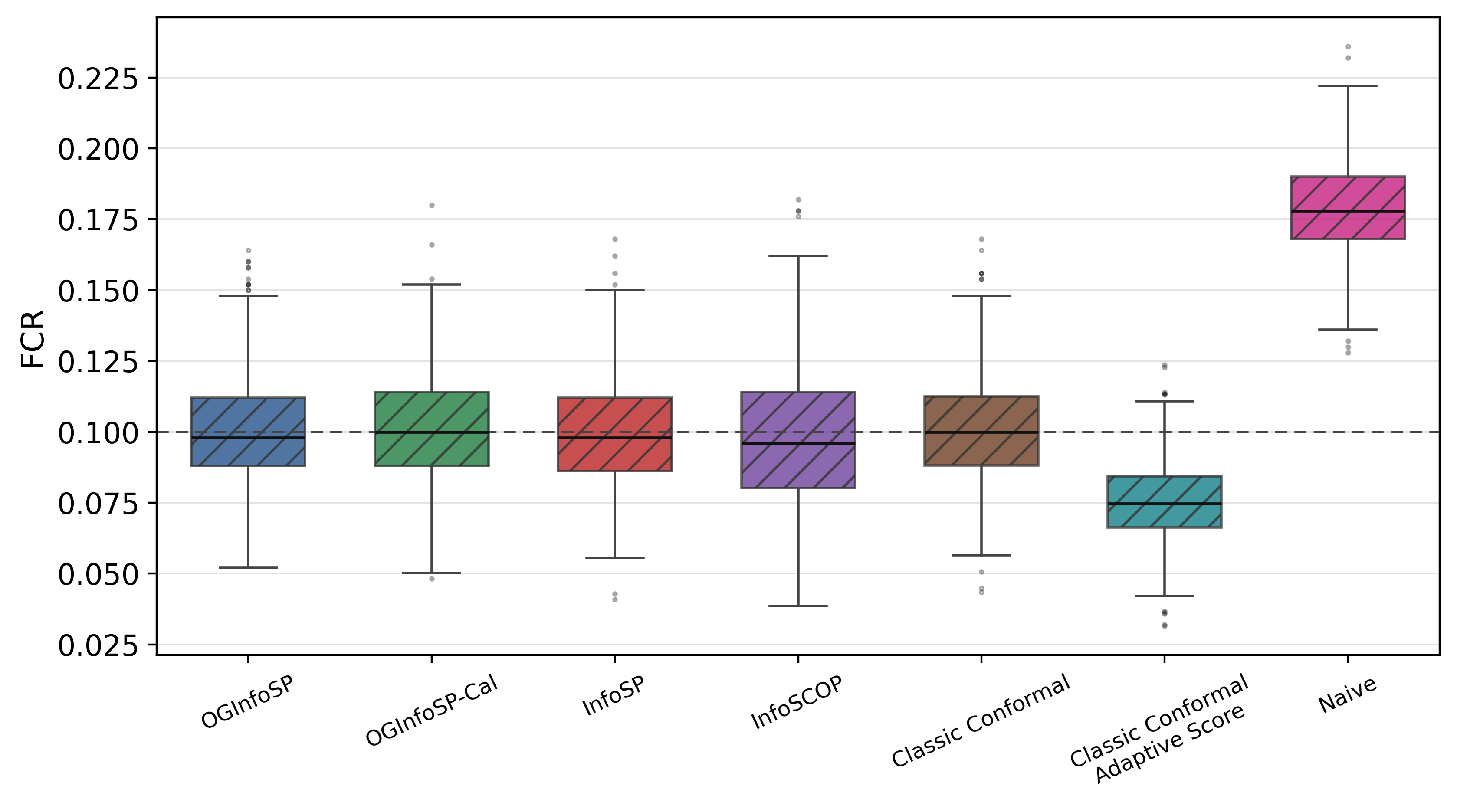}{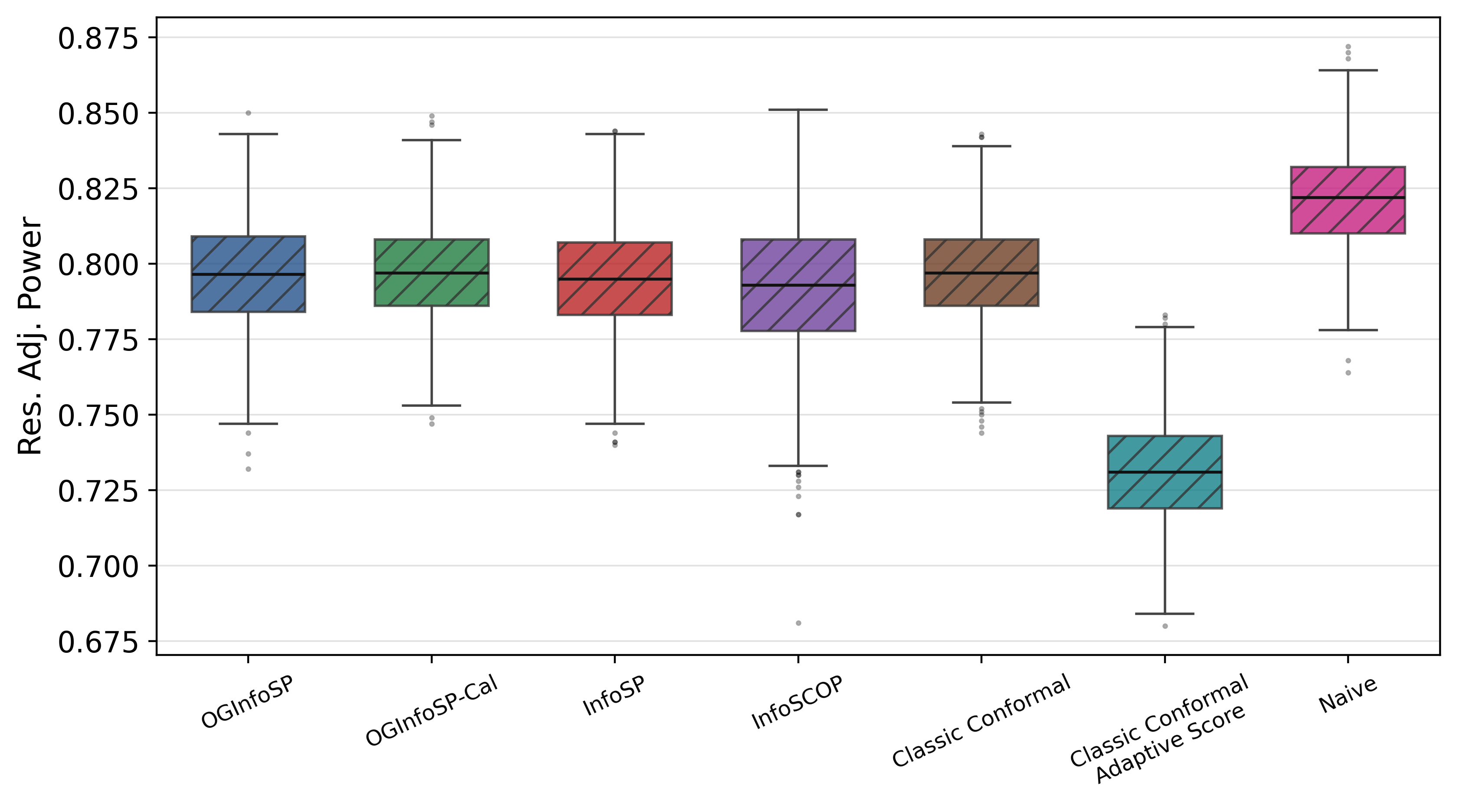}
  \caption{For the CIFAR-10 data described in \S~\ref{sim-CIFAR} with $K=3$ classes, the FCR (column 1) and resolution-adjusted power (column 2) for non-trivial prediction sets. The training sample had an even number of classes. The calibration and test samples come from a mixture distribution with  even class probabilities (row 1);  or uneven class probabilities (row 2), where the class of cats has probability $0.6$ and each of the other two classes has probability $0.2$. The sample sizes are $m= 500, n=1{,}000$, $\alpha = 0.1$. Half of the calibration samples are used for estimated the conditional probabilities using vector scaling, and the remaining half are used by the methods. The results are based on $1{,}000$ iterations.}\label{app-fig-cifar10-K3-nontrivial}
\end{figure}

\subsection{Image-classification experiments on CIFAR-10 with   $K=10$ classes}\label{app-sim-CIFAR-K10}

We consider here the setting in \cite{GazinHellerMarandonRoquain2025} with  $K=10$ classes from the  CIFAR-10 image datasets \citep{krizhevsky2009learning}. The  prediction algorithm construction is detailed in \S~\ref{sim-CIFAR}.

Figure~\ref{fig:CIFAR10-uneven} differs from Figure~\ref{fig:CIFAR10-even} in that the calibration and test samples come from a mixture with uneven class probabilities, where the relative frequency of cats is $0.45$, of birds is $0.15$, and of each remaining class is $0.05$.  The conclusions are similar to the qualitative conclusions in \S~\ref{sim-CIFAR}. 

Figure~\ref{fig:CIFAR10-even-vs-noVS} differs from Figure~\ref{fig:CIFAR10-even} in that it compares the recommended procedures following vector scaling with their counterparts that use the estimated conditional distributions as is. In the two settings displayed, vector scaling improved power. For \infoOSP\ and \infoOSPCal, in the setting that informativeness means to predict a single animal class while excluding all non-animal classes, the difference is very large: from no power to excellent power.  


\begin{figure}[!htbp]
  \centering
  \SMpair{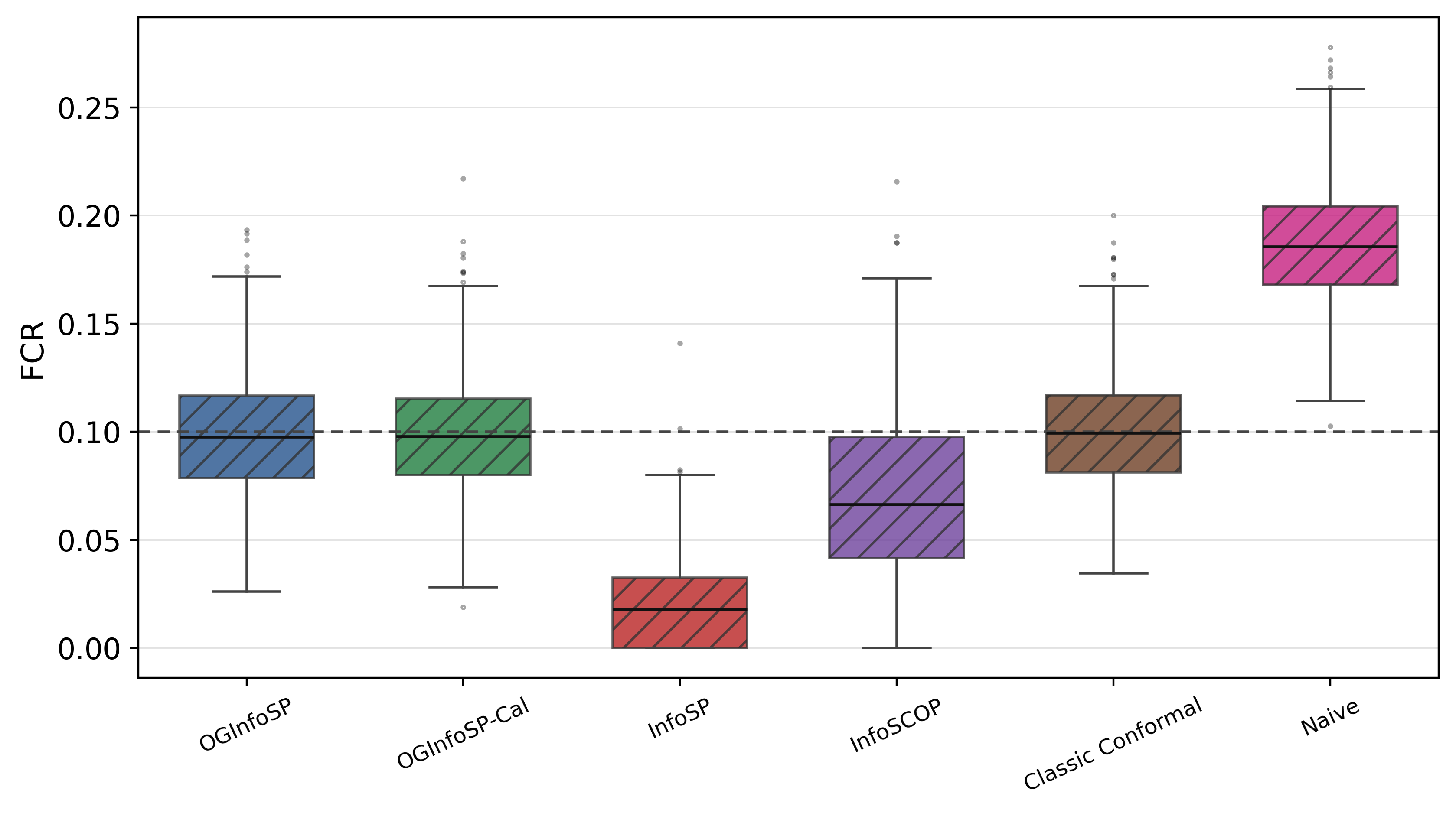}{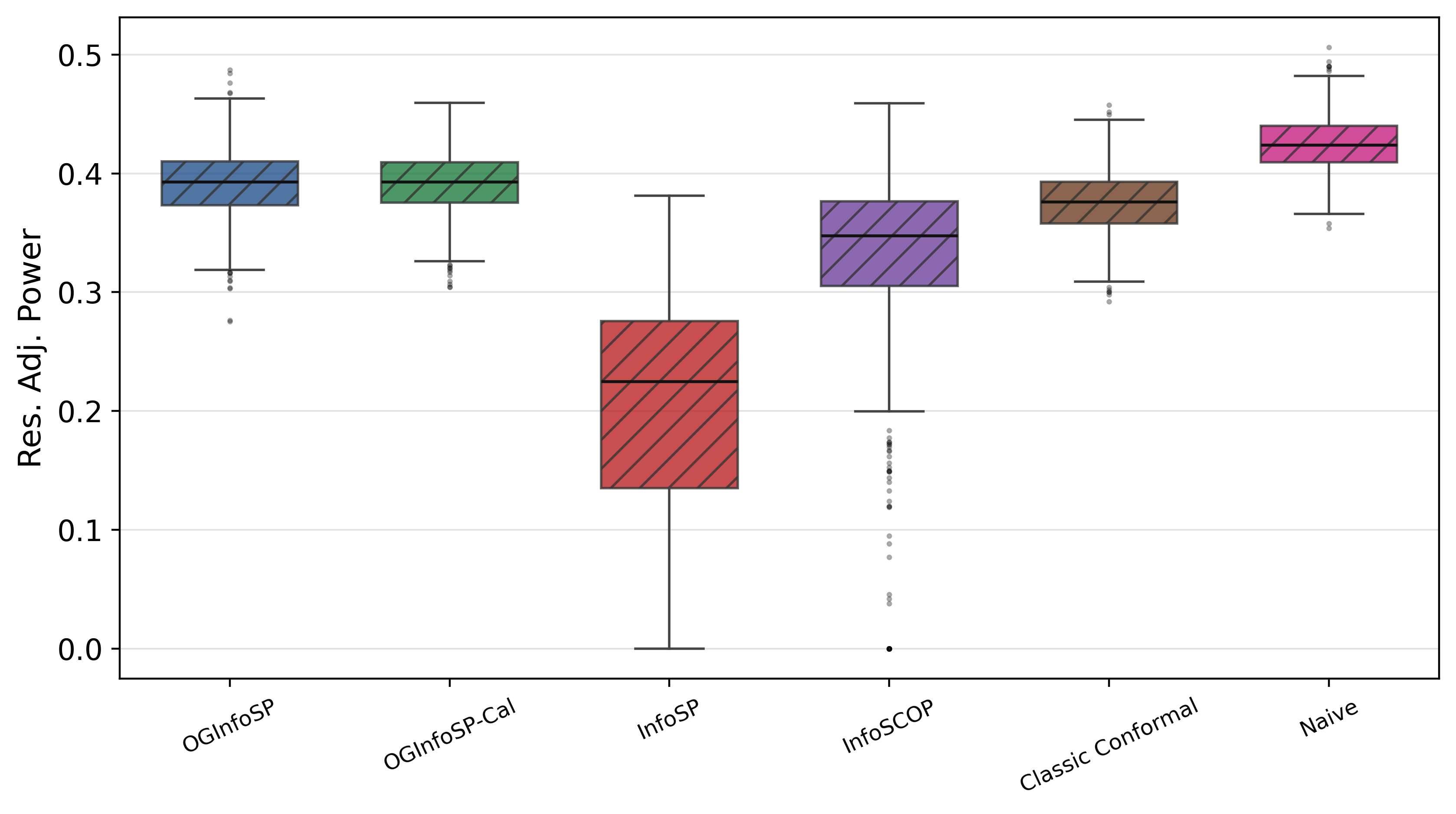}

  \SMpair{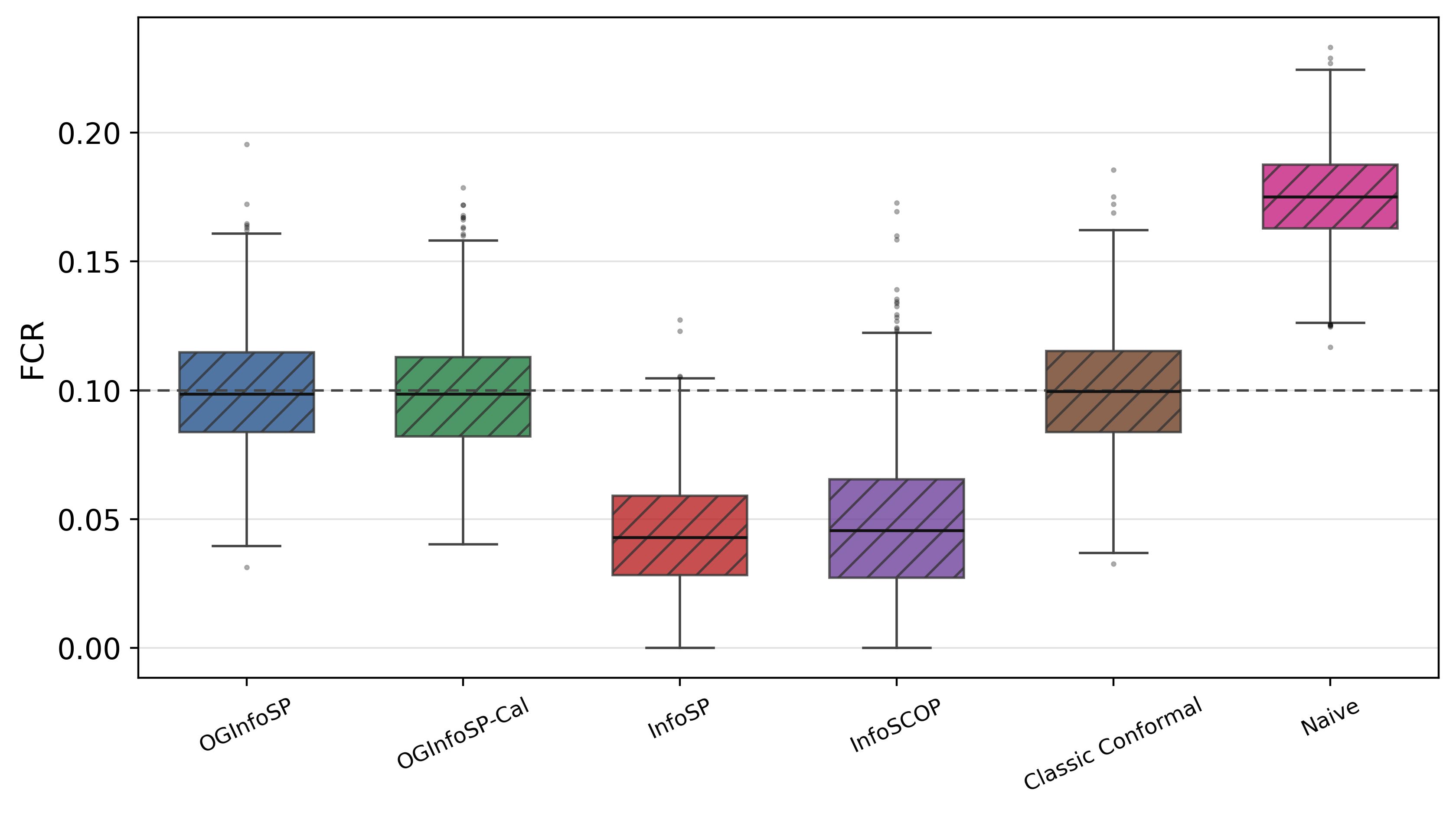}{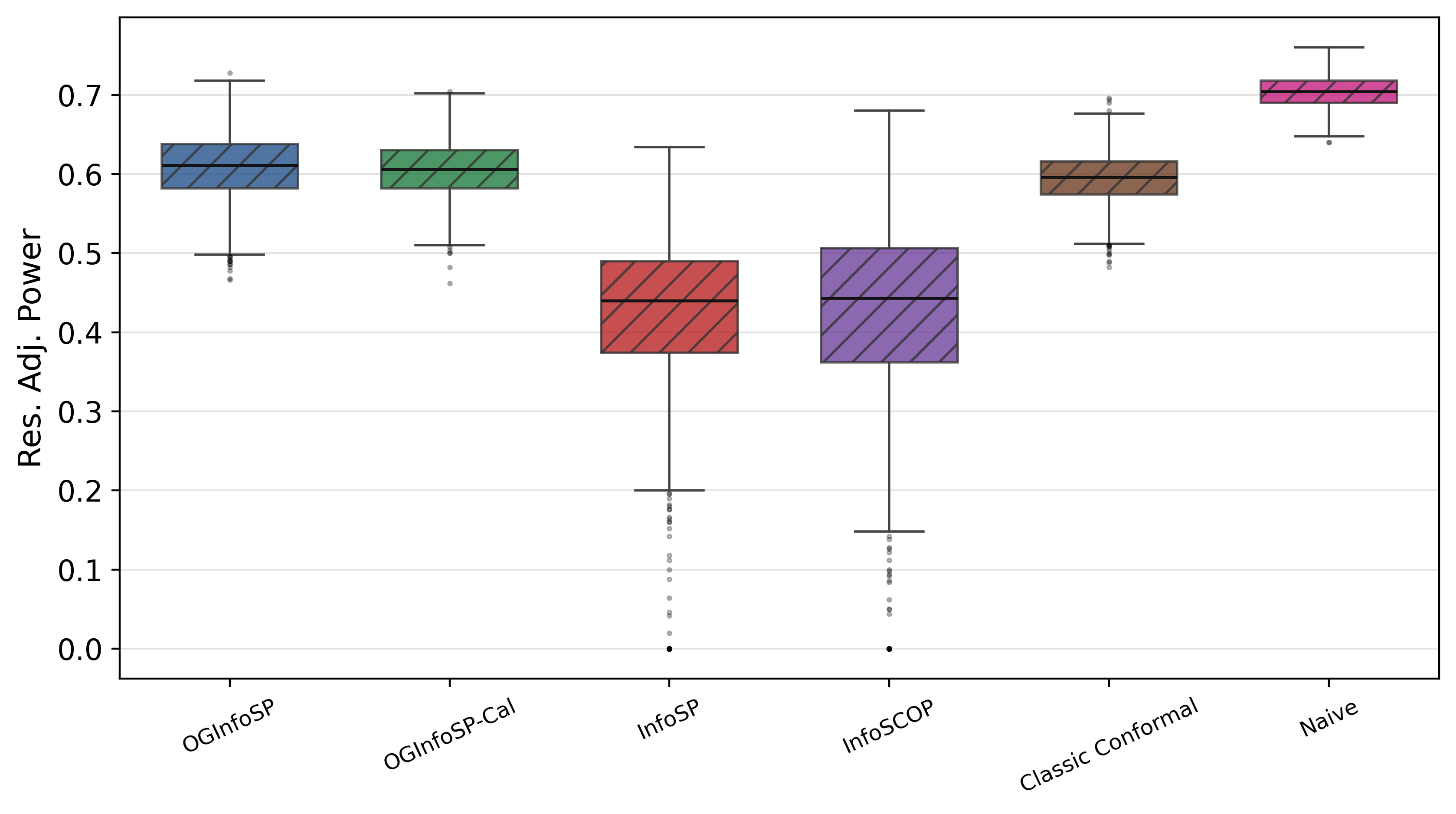}
  \caption{For the CIFAR-10 data described in \S~\ref{sim-CIFAR} with $K=10$ classes unevenly distributed, the FCR (column 1) and resolution-adjusted power (column 2) for prediction sets that:  exclude the class of cats (row 1);  are of size one and exclude the three non-animal classes (row 2). The sample sizes are $m= 500, n=1{,}000$, $\alpha = 0.1$. Half of the calibration samples are used for estimated the conditional probabilities using vector scaling, and the remaining half are used by the methods. The results are based on $1{,}000$ iterations.}
    \label{fig:CIFAR10-uneven}
\end{figure}


\begin{figure}[!htbp]
  \centering
  \SMpair{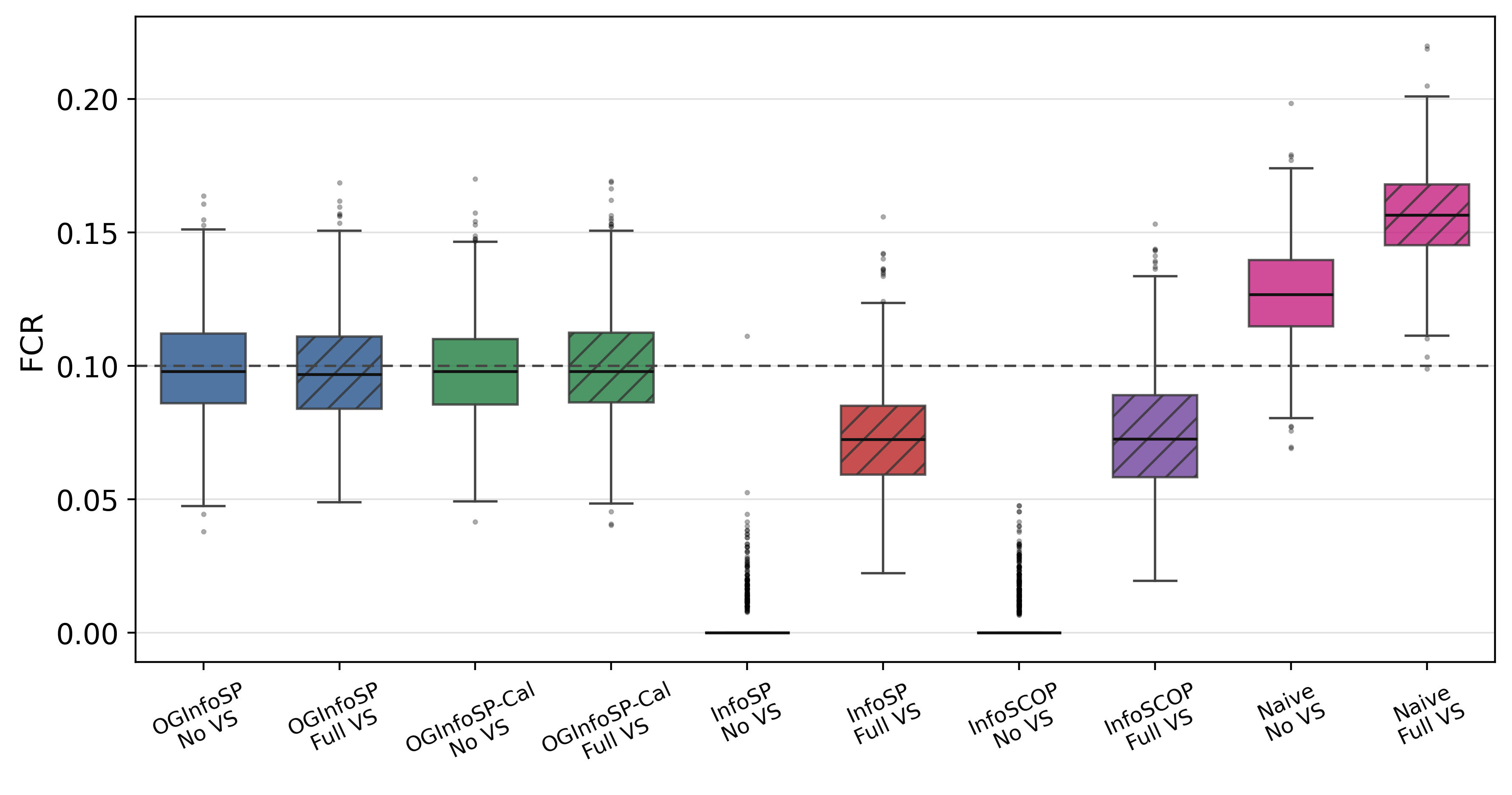}{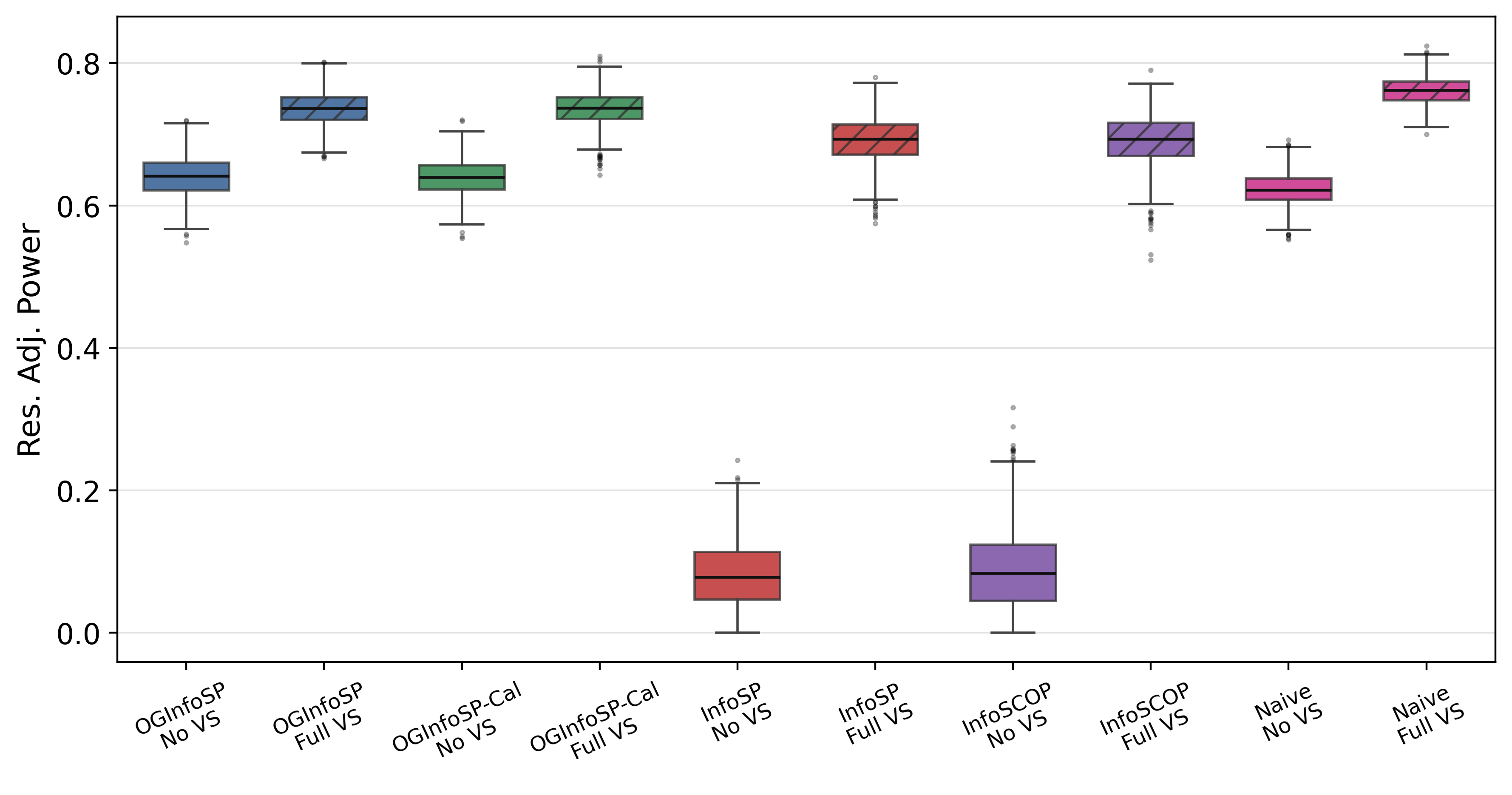}
\SMpair{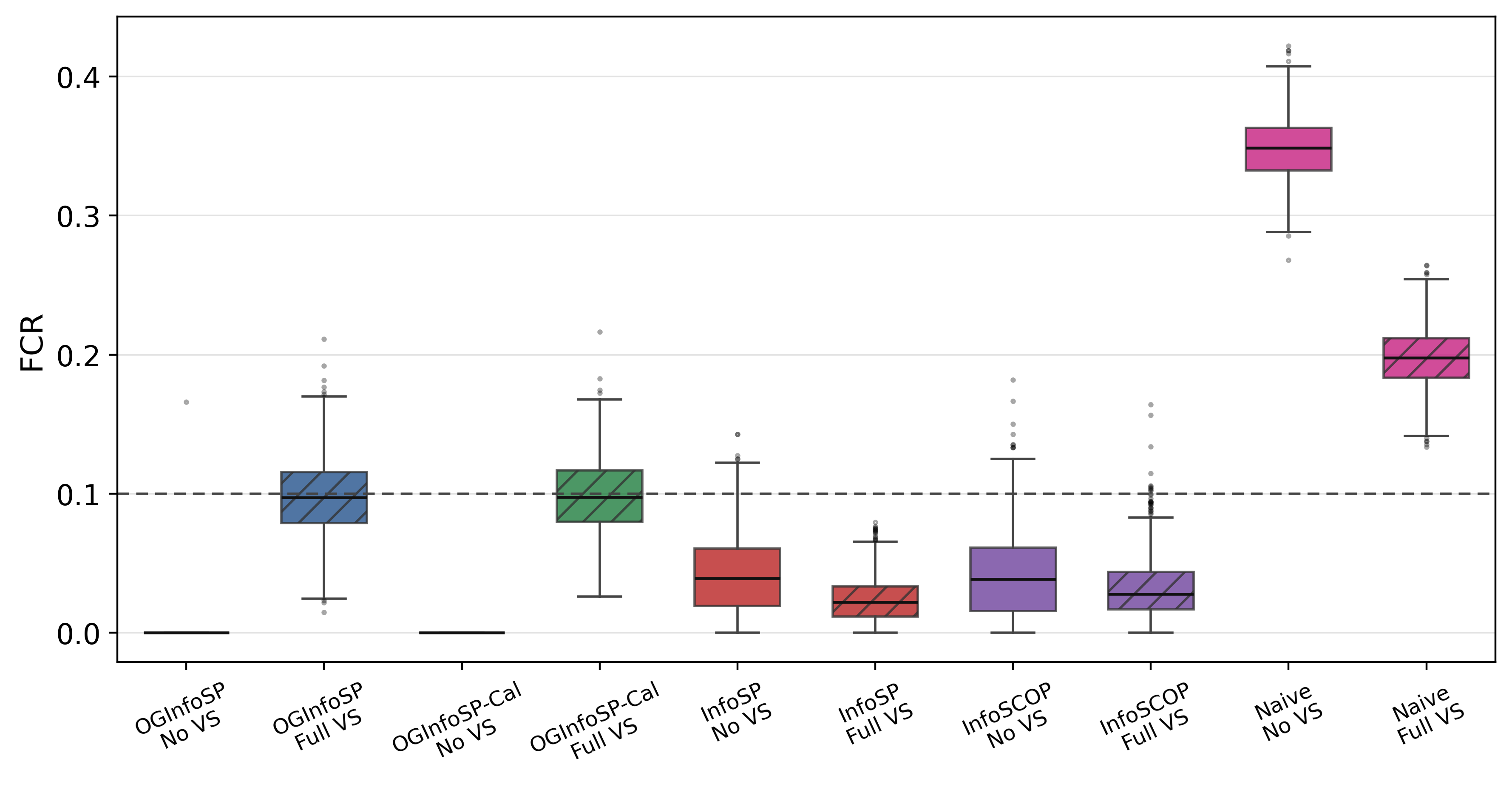}{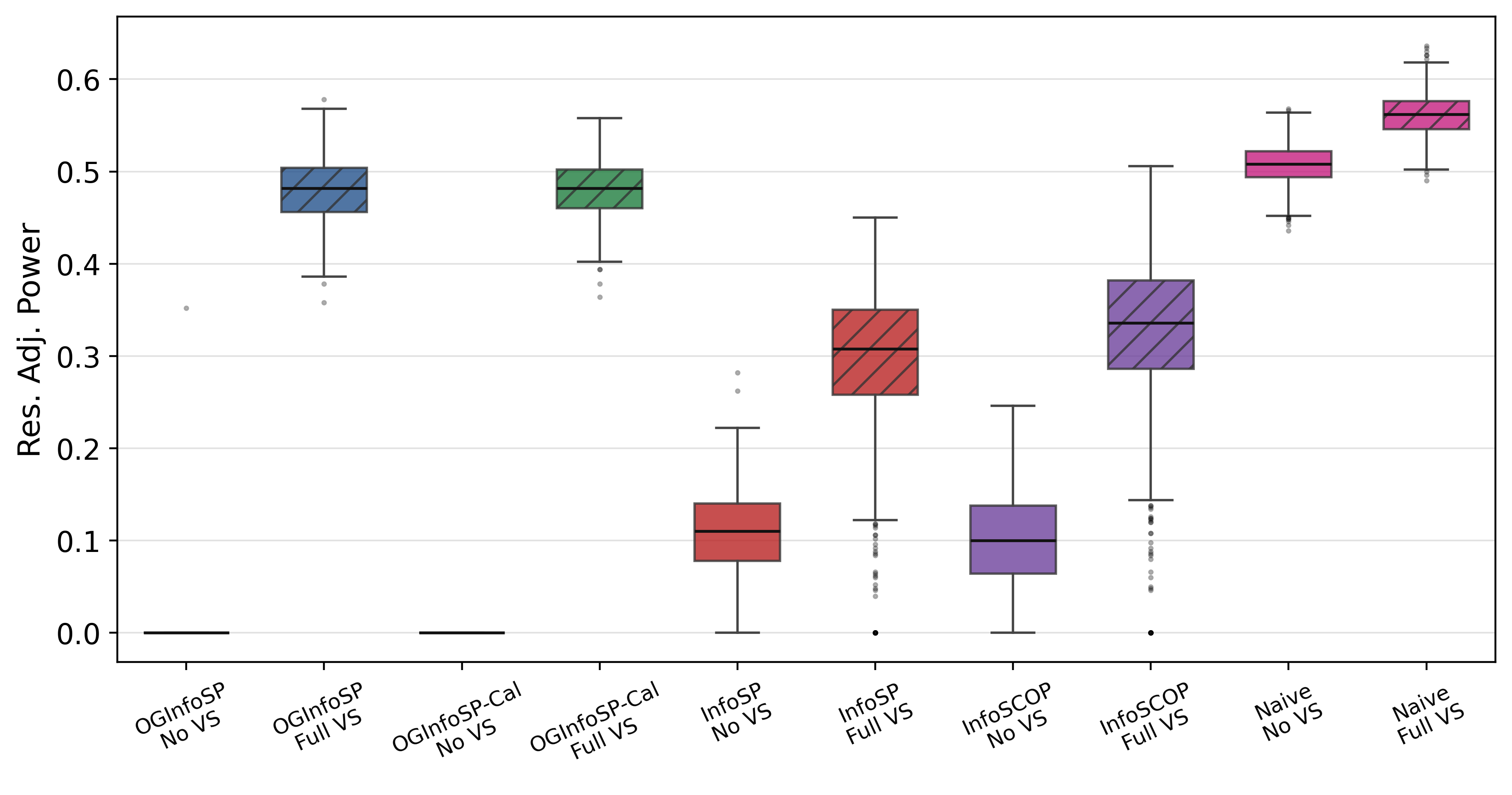}
  \caption{For the CIFAR-10 data described in \S~\ref{sim-CIFAR} with $K=10$ classes evenly distributed, the FCR (column 1) and resolution-adjusted power (column 2) for prediction sets that:  exclude the class of cats (row 1);  are of size one and exclude the three non-animal classes (row 2). The sample sizes are $m= 500, n=1{,}000$, $\alpha = 0.1$. The methods denoted with No VS do not carry out the vector scaling, so they use all $n=1{,}000$ calibration samples for the inference. The remaining procedures perform the recommended vector scaling, so half of the calibration samples are used for estimated the conditional probabilities using vector scaling, and the remaining half are used by the methods.   The results are based on $1{,}000$ iterations.}
    \label{fig:CIFAR10-even-vs-noVS}
\end{figure}

\section{Asymptotic connection between results of Sections \ref{sec-optimal} and \ref{sec-proc}}\label{app-asymptotic-oracle}

In this section, we analyze the \infoOSP\ and \infoOSPCal\ procedures in an idealized setting with classification outcomes. We make the following assumptions:
\begin{enumerate}
    \item Assumptions~\ref{ass-probratio} and~\ref{ass-martingale} hold.
    \item The estimated probabilities equal the oracle probabilities:
    \[
    \hat{\mP}(Y=k\mid X=x)=\mP(Y=k\mid X=x),
    \qquad k\in[K],\ x\in\mathcal X.
    \]
    \item The calibration sample size $n$ and the test sample size $m$, where relevant, tend to infinity. For simplicity, we take $m=n$.
    \item The optimal multiplier $\mu^*$ in Theorem~\ref{thm-optimal-nonrandom} is finite and is the unique zero of
    \[
    G(\mu):=G(D^\mu,C^\mu).
    \]
    Since $G$ is non-increasing by Corollary~\ref{coro-prop-monotoneD}, this assumption implies
    \[
    G(\mu)>0\quad\text{for }0\leq\mu<\mu^*,
    \qquad
    G(\mu)<0\quad\text{for }\mu>\mu^*.
    \]
\end{enumerate}

We write $\mu_{\alpha,n}$ for the \infoOSP\ threshold and $\mu_{\alpha,n}^{\mathrm{cal}}$ for the \infoOSPCal\ threshold when $m=n$. We next express both thresholds through empirical versions of the oracle constraint.

Recall the quantities $\lambda$ and $\Lambda$ defined in Eq.~\eqref{eq-keystatistics}. Under the oracle-probability assumption,
\[
\mI\{\lambda(X,Y)>\mu\}
=
\mI\{Y\notin \hat C^\mu(X)\}\hat D^\mu(X),
\qquad
\mI\{\Lambda(X)>\mu\}
=
\hat D^\mu(X).
\]
Assumption~\ref{ass-probratio} ensures that, for every fixed $\mu$, the different tie-breaking and boundary conventions used for the oracle and fitted policies affect only events of probability zero. Consequently,
\begin{align}
\mP(\lambda(X,Y)>\mu)
-
\alpha\mP(\Lambda(X)>\mu)
&=
\mE\left[
\left\{\mI\{Y\notin \hat C^\mu(X)\}-\alpha\right\}
\hat D^\mu(X)
\right]
\nonumber\\
&=G(D^\mu,C^\mu)
=G(\mu).
\label{eq:S8-population-G}
\end{align}

For \infoOSP, define
\begin{align}
G_n^{\mathrm{OG}}(\mu)
:=&\frac{1}{n+1}
\left(1+\sum_{i=1}^n \mI\{\lambda(X_i,Y_i)>\mu\}\right)
\nonumber\\
&-
\alpha\frac{1}{n}
\left(1\lor\sum_{i=1}^n \mI\{\Lambda(X_{n+i})>\mu\}\right).
\label{eq:S8-Gn-OG}
\end{align}
For \infoOSPCal, define
\begin{align}
G_n^{\mathrm{cal}}(\mu)
:=&\frac{1}{n+1}
\left(1+\sum_{i=1}^n \mI\{\lambda(X_i,Y_i)>\mu\}\right)
\nonumber\\
&-
\alpha\frac{1}{n+1}
\left(1+\sum_{i=1}^n \mI\{\Lambda(X_i)>\mu\}\right).
\label{eq:S8-Gn-cal}
\end{align}
By the definitions of the two procedures,
\begin{align*}
\mu_{\alpha,n}
&=
\inf\{\mu\geq0:G_n^{\mathrm{OG}}(\mu)\leq0\},\\
\mu_{\alpha,n}^{\mathrm{cal}}
&=
\inf\{\mu\geq0:G_n^{\mathrm{cal}}(\mu)\leq0\},
\end{align*}
where the corresponding threshold is set to $\infty$ if the relevant set is empty.

\begin{thm}[Consistency of the calibrated thresholds]\label{thm:S8-threshold-consistency}
Under the assumptions above,
\[
\mu_{\alpha,n}\xrightarrow[n\to\infty]{\mathrm{a.s.}}\mu^*,
\qquad
\mu_{\alpha,n}^{\mathrm{cal}}
\xrightarrow[n\to\infty]{\mathrm{a.s.}}\mu^*.
\]
\end{thm}

\begin{proof}
By the Glivenko--Cantelli theorem,
\[
\sup_{\mu\geq0}
\left|
\frac{1}{n}\sum_{i=1}^n\mI\{\lambda(X_i,Y_i)>\mu\}
-
\mP(\lambda(X,Y)>\mu)
\right|
\xrightarrow{\mathrm{a.s.}}0.
\]
The same conclusion holds for the empirical averages of $\mI\{\Lambda(X_i)>\mu\}$ over the calibration sample and of $\mI\{\Lambda(X_{n+i})>\mu\}$ over the test sample. For the calibration-only procedure, the two empirical processes are based on the same observations, but independence between them is not required: each converges uniformly almost surely, and hence the two convergences hold simultaneously almost surely.

Moreover, uniformly over $\mu\geq0$,
\begin{align*}
\left|
\frac{1}{n+1}
\left(1+\sum_{i=1}^n\mI\{\lambda(X_i,Y_i)>\mu\}\right)
-
\frac{1}{n}\sum_{i=1}^n\mI\{\lambda(X_i,Y_i)>\mu\}
\right|
&\leq \frac{1}{n+1},\\
\left|
\frac{1}{n}
\left(1\lor\sum_{i=1}^n\mI\{\Lambda(X_{n+i})>\mu\}\right)
-
\frac{1}{n}\sum_{i=1}^n\mI\{\Lambda(X_{n+i})>\mu\}
\right|
&\leq \frac{1}{n},\\
\left|
\frac{1}{n+1}
\left(1+\sum_{i=1}^n\mI\{\Lambda(X_i)>\mu\}\right)
-
\frac{1}{n}\sum_{i=1}^n\mI\{\Lambda(X_i)>\mu\}
\right|
&\leq \frac{1}{n+1}.
\end{align*}
Combining these bounds with Eq.~\eqref{eq:S8-population-G} gives
\begin{equation}
\sup_{\mu\geq0}
\left|G_n^{\mathrm{OG}}(\mu)-G(\mu)\right|
\xrightarrow{\mathrm{a.s.}}0,
\qquad
\sup_{\mu\geq0}
\left|G_n^{\mathrm{cal}}(\mu)-G(\mu)\right|
\xrightarrow{\mathrm{a.s.}}0.
\label{eq:S8-uniform-G}
\end{equation}

It remains to show that uniform convergence of either empirical criterion implies consistency of its threshold. Let $G_n^\circ$ denote either $G_n^{\mathrm{OG}}$ or $G_n^{\mathrm{cal}}$, and let $\mu_{\alpha,n}^\circ$ denote the corresponding threshold.

Fix $\delta>0$. If $\mu^*>0$ and $0<\delta<\mu^*$, then
\[
c_-(\delta):=G(\mu^*-\delta)>0.
\]
Since $G$ is non-increasing,
\[
\inf_{0\leq\mu\leq\mu^*-\delta}G(\mu)
=G(\mu^*-\delta)
=c_-(\delta).
\]
Thus, Eq.~\eqref{eq:S8-uniform-G} implies that, almost surely for all sufficiently large $n$,
\[
G_n^\circ(\mu)>0
\qquad\text{for every }0\leq\mu\leq\mu^*-\delta.
\]
Hence $\mu_{\alpha,n}^\circ>\mu^*-\delta$ eventually almost surely. If $\delta\geq\mu^*$, this lower bound is immediate from $\mu_{\alpha,n}^\circ\geq0$.

For the upper bound, the unique-zero assumption gives
\[
c_+(\delta):=-G(\mu^*+\delta)>0.
\]
Equation~\eqref{eq:S8-uniform-G} therefore implies that, almost surely for all sufficiently large $n$,
\[
G_n^\circ(\mu^*+\delta)<0.
\]
The feasible set defining $\mu_{\alpha,n}^\circ$ is then nonempty and contains $\mu^*+\delta$, so
\[
\mu_{\alpha,n}^\circ\leq\mu^*+\delta
\]
eventually almost surely. Since $\delta>0$ is arbitrary,
\[
\mu_{\alpha,n}^\circ\xrightarrow{\mathrm{a.s.}}\mu^*.
\]
Applying this argument to both empirical criteria proves the result.
\end{proof}

\section{Asymptotic optimality of \infoOSP\ and  \infoOSPCal}\label{app-rates-and-asymptotics}

In \S~\ref{app-asymptotic-oracle} we established almost-sure convergence of the calibrated thresholds to the oracle threshold under the idealized assumption that the fitted conditional class probabilities equal their oracle counterparts. Here we allow the class-probability estimates to be learned from an independent training sample of size $t$ and derive nonasymptotic oracle-power bounds. For \infoOSP, the remainder depends on $t$,$n$, and $m$; for the calibration-only variant \infoOSPCal, it depends only on $t$ and $n$, so the resulting guarantee applies to any finite test-sample size $m$. We also derive an approximate error-control bound for \infoOSPCal.

Throughout this section, $\mathcal Y = [K].$

\subsection{Assumptions and statements}

We first introduce additional assumptions that will be needed for our analysis.
The quality of the class conditional probability estimation is assessed as follows.
\begin{ass}\label{ass-approxproba}
For  $t\geq 2$, we have
\begin{align}
\mE\left[\max_{y\in [K]} \left| \hat{\mP}(Y=y\mid X)- {\mP}(Y=y\mid X)\right|\right]\leq \eta_t, \label{eq-ass-approxproba}
\end{align}
where $\eta_t := c_1(\log t)^{\gamma_0} t^{-\gamma}$ for some rates $\gamma_0>0$, $\gamma\in (0,1/2]$ and some constant $c_1>0$.
\end{ass}
The expectation in \eqref{eq-ass-approxproba} is taken jointly over the training sample used to construct $\hat \mP$ and an independent draw $X$ from the target covariate distribution. 
This property is satisfied by standard class probability estimators under suitable regularity assumptions \citep{tsybakov2007fast,bos2022convergence}.

We strengthen Assumption~\ref{ass-probratio} slightly  by assuming that the maximum of $C\in \mathcal{I}\mapsto \ell_{X,C}(\mu)$ is 'isolated' in a stronger sense than Assumption~\ref{ass-probratio}.
\begin{ass}\label{ass-isolatedmax}
For any  $\mu\geq 0$, the random variables 
\begin{equation}
    \label{equDelta}
    \Delta_{\mu}(X)= \ell_{X,{C}^{\mu}(X)}(\mu) - \max_{C\in \mathcal{I}\backslash\{{C}^{\mu}(X)\}} \ell_{X,C}(\mu)\geq 0
\end{equation}
and  $|\max_{C \in \mathcal{I}}\ell_{X,C}(\mu)|$  have densities on $[0,+\infty)$ which are bounded by $c_{\mu}>0$, where $\mu\mapsto c_{\mu}$ is continuous.
\end{ass}
The first bound is a margin condition preventing near ties between the best and second-best candidate sets. The second prevents the upper envelope from lying too close to zero, where small conditional probability-estimation errors can change the selection decision.
Hence, the oracle prediction set and decision are locally constant in $\mu$, and dominated convergence yields that $G$ is continuous in $\mu$. Therefore, $\inf\{\mu\geq 0: G(D^{\mu}, C^{\mu})\leq 0\}$ is attained if the set is non-empty, i.e., $\mu^* = \min\{\mu\geq 0: G(D^{\mu}, C^{\mu})\leq 0\}$  for $\mu^*<\infty$. 

 We make a regularity assumption of the non-increasing functionals $\mu\mapsto G(D^{\mu}, C^{\mu})$ and power $\mu\mapsto \Pi(D^{{\mu}},C^{{\mu}})$ around the oracle threshold $\mu^*$.

\begin{ass}
    \label{ass-derivativeGalpha-together}
   There exists $a_0,b_0, d_0, \epsilon_0>0$ such that for any $0<\epsilon\leq \epsilon_0$ and $0\leq\mu\leq \mu^*+\epsilon$, 
\begin{align*}
    {G}(D^{\mu}, C^{\mu})-{G}(D^{\mu+\epsilon}, C^{\mu+\epsilon}) &\geq a_0 \epsilon \\
    \Pi(D^{{\mu}^*},C^{{\mu}^*}) - 
\Pi(D^{{\mu}^*+\epsilon},C^{{\mu}^*+\epsilon}) &\leq b_0\epsilon\\
 \mE[D^{\mu^*}(X)]-\mE[D^{\mu^*+\epsilon}(X)]&\leq d_0\epsilon
\end{align*}
\end{ass}
In Assumption~\ref{ass-derivativeGalpha-together}, the first requirement ensures that ${G}(D^{\mu}, C^{\mu})$, as a function of $\mu$, strictly crosses zero rather than flattening. This is a weak assumption, which requires  that  when shifting $\mu$ slightly  the selection status or conditional coverage probability of some of the population must change  (a scenario that is highly  likely with continuous covariates). The remaining conditions with $b_0, d_0$ prevent the power and expected decision to select, respectively, from changing too abruptly near $\mu^*$.

From Theorem~\ref{thm-FCR-control} in the main document, we know that the \infoOSP\ procedure $(\widehat{D}^{\thone}(X_{n+i}),\widehat{C}^{\thone}(X_{n+i}))_{i\in [m]}$ with $\thone$ defined in \eqref{eq-thone}  
controls the FCR at level $\alpha$ for any $t,n,m\geq 1$. 
We study here its power \eqref{equPower} denoted by 
\begin{align}
\Pi
&:=
\mE\left[
\frac{1}{m}\sum_{i\in[m]}
w\!\left(\widehat{C}^{\thone}(X_{n+i})\right)
\mI\{Y_{n+i}\in \widehat{C}^{\thone}(X_{n+i})\}
\hat D^{\thone}(X_{n+i})
\right],\label{equ:powerOGInfoSP}
\end{align}
for which the expectation is taken wrt the three sources of randomness (training sample, calibration sample and test covariates). 
For procedure \infoOSPCal, 
 given by $(\hat D^{\thtwo}(X_{n+i}),\widehat{C}^{\thtwo}(X_{n+i}))_{i\in [m]}$ with $\thtwo$ given by \eqref{eq-thtwo},  the FCR, marginal FCR and  power are denoted by 
\begin{align}
\FCR_{cal} &:= 
\mE\left[
\frac{
\sum_{i\in[m]}\mI\{Y_{n+i}\notin \widehat{C}^{\thtwo}(X_{n+i})\}\hat D^{\thtwo}(X_{n+i})
}{
1\lor \sum_{i\in[m]}\hat D^{\thtwo}(X_{n+i})
}
\right];\label{equ:FCROGInfoSPOnlyCal}\\
\mFCR_{cal} &:= 
\frac{\mE\left[
\sum_{i\in[m]}\mI\{Y_{n+i}\notin \widehat{C}^{\thtwo}(X_{n+i})\}\hat D^{\thtwo}(X_{n+i})
\right]}{
\mE\left[\sum_{i\in[m]}\hat D^{\thtwo}(X_{n+i})
\right]}
;\label{equ:mFCROGInfoSPOnlyCal}\\
    \Pi_{cal} &:= 
\mE\left[
\frac{1}{m}\sum_{i\in[m]}
w\!\left(\widehat{C}^{\thtwo}(X_{n+i})\right)
\mI\{Y_{n+i}\in \widehat{C}^{\thtwo}(X_{n+i})\}
\hat D^{\thtwo}(X_{n+i})
\right].\label{equ:powerOGInfoSPOnlyCal}
\end{align}

The next result shows that $\FCR_{cal}$ and $\mFCR_{cal}$ are close to $\alpha$, and that  $\Pi$ and $\Pi_{cal}$ are close to exceeding the  power \eqref{equPower} of the oracle, 
\begin{align*}
\Pi(D^{{\mu}^*},C^{{\mu}^*}) = \mE\left[
 w\!\left(C^{{\mu}^*}(X)\right)
 \mI\{Y\in C^{{\mu}^*}(X)\}
 D^{\mu^*}(X)
 \right].
\end{align*}

\begin{thm}[Optimality of \infoOSP\: and \infoOSPCal]\label{opt-theorem-merged}
Let $\alpha\in (0,1)$ be fixed with ${\mu}^*<+\infty$ and $s_0:=\mE\left[{D}^{\mu^*}(X)\right]>0$. Suppose 
Assumption~\ref{ass-martingale} holds (both for the true and estimated probabilities), as well as  Assumptions~\ref{ass-approxproba},~\ref{ass-isolatedmax}~and~\ref{ass-derivativeGalpha-together}.
Let the training, calibration, and test sample sizes be $t,n,$ and $m$, respectively. Recall that $w(C)\leq B$ for all $C\in \mathcal I.$

With $\eta_t=c_1(\log t)^{\gamma_0} t^{-\gamma}$, define $c_0=\sup_{0\leq \mu \leq \mu^*+1}c_{\mu}<+\infty$,
\begin{align*}
    \nu(t,n,m)&:=
       2/(n\wedge m)+8(c_{0}+  K (B+\mu^*+1)) \eta^{1/4}_t
      \\
      \epsilon(t,n,m)&:=  a_0^{-1} \bigg(\frac{2+ \sqrt{2\log(2 (n\wedge m))}}{\sqrt{n\wedge m}} + \eta^{1/4}_t\bigg).
\end{align*}
  
 Then the following inequalities hold for any sample sizes $t,n,m$:
 
\begin{itemize}
\item[(i)] $\Pi \geq  \Pi(D^{{\mu}^*},C^{{\mu}^*})-b_0 \epsilon - B \nu $,
    for $\nu=\nu(t,n,m), \epsilon=\epsilon(t,n,m)$,  provided that $\epsilon\leq 1\wedge \epsilon_0$ and $\eta_t\leq 1$;
\item[(ii)]  $\Pi_{cal} \geq  \Pi(D^{{\mu}^*},C^{{\mu}^*})-b_0 \epsilon- B \nu$ for $\nu=\nu(t,n,n), \epsilon=\epsilon(t,n,n)$, provided that $\epsilon\leq 1\wedge \epsilon_0$ and $\eta_t\leq 1$. If in addition $(d_0+a_0)\epsilon+\nu \leq s_0/2$ then $mFCR_{cal}\leq \alpha+2s_0^{-1}(a_0\epsilon+\nu)$ and $FCR_{cal}\leq \alpha+2a_0s_0^{-1}\epsilon+\nu.$
\end{itemize}
\end{thm}

Theorem~\ref{opt-theorem-merged} implies that, under appropriate regularity assumptions, when $t,n,m$ all tend to infinity, we have for the power $\Pi$ of \infoOSP\ that 
$$
\Pi(D^{{\mu}^*},C^{{\mu}^*}) \leq \Pi + O\bigg(\sqrt{\frac{\log(n\wedge m)}{n\wedge m}} + \eta_t^{1/4}\bigg),
$$
that is, the power of \infoOSP\ is optimal up to a remainder term of order $\sqrt{\frac{\log(n\wedge m)}{n\wedge m}} + \eta_t^{1/4}$.

It also implies both optimality and validity of \infoOSPCal: for all $m\geq  1$ and when $t,n$ are large enough, the power loss and excess error rates are $O\bigg(\sqrt{\frac{\log n}{n}} + \eta_t^{1/4}\bigg).$

\begin{remark}
In Theorem~\ref{opt-theorem-merged}, the conditions ${\mu}^*<+\infty$ and $s_0=\mE\left[ D^{\mu^*}(X)\right]>0$ are both satisfied provided that $\mP(\max_{C\in\mathcal{I}} \mP(Y\in C\mid X) \geq 1-\alpha) > 0$ (see  Theorem~\ref{thm-optimal-nonrandom} and Proposition~\ref{prop-positive-selection}).
\end{remark}

\subsection{Notation} 
Throughout this section, $\ell_{x,C}, C^{\mu}, D^{\mu}, G(\mu) = G(D^{\mu}, C^{\mu}),$ and $\mu^*$ are as defined in \S~\ref{sec-optimal}. The fitted counterparts $\hell_{x,C}, \hat C^{\mu}, \hat D^{\mu}, $  are as defined in \S~\ref{sec-proc}. They depend on the independent training sample $\mathcal D_{train}$ but not on the calibration sample $\mathcal{D}^{cal}$ and the test sample $\mathcal{D}^{test}$.

The data driven thresholds $\thone$ and $\thtwo$ are also defined in \S~\ref{sec-proc}. The threshold for \infoOSP, $\thone$, depends on the test covariates as well as on  $\mathcal{D}^{train}$, $\mathcal{D}^{cal}$. The threshold for \infoOSPCal\ depends only on $\mathcal{D}^{train}$, $\mathcal{D}^{cal}$.   From their definitions, clearly  $\thone = \min \left\{\mu\geq 0 : \tilde{G}(\hat D^{\mu}, \hat C^{\mu})\leq 0 \right\}$ and  $\thtwo = \min \left\{\mu\geq 0 : \tilde{G}_{cal}(\hat D^{\mu}, \hat C^{\mu})\leq 0 \right\}$, for the empirical constraint functionals 
\begin{eqnarray}
    \tilde{G}(D,C) &= \frac{1}{1+n} \left(1 + \sum_{j=1}^n \mI\{Y_j\notin C(X_j)\} D(X_j)\right)  - \alpha \left(\frac 1m\lor \frac 1m\sum_{i=1}^m  D(X_{n+i}) \right) \label{app-eq-Gtilde} \\
    \tilde{G}_{cal}(D,C) &= \frac{1}{1+n} \left(1 + \sum_{j=1}^n \mI\{Y_j\notin C(X_j)\} D(X_j)\right)  - \alpha \frac{1}{1+n} \left(1 + \sum_{j=1}^n  D(X_j)\right).  \label{app-eq-Gtilde-cal}  
\end{eqnarray}
  
For any possibly data-dependent policy pair $(D,C)$ that is measurable
with respect to the training and calibration samples, define
$\widehat G(D,C)$ as the constraint for a new observation, conditional
on $\mathcal D^{train}$ and $\mathcal D^{cal}$:
\begin{align}
\widehat G(D,C)
&:=
\mE\left[
\left(
\mI\left\{Y\notin C(X)\right\}-\alpha
\right)
D(X)
\,\middle|\,
\mathcal D^{train},\mathcal D^{cal}
\right]
\notag\\
&=
\mE\left[
\left(
\mP\left(
Y\notin C(X)
\mid X,\mathcal D^{train},\mathcal D^{cal}
\right)-\alpha
\right)
D(X)
\,\middle|\,
\mathcal D^{train},\mathcal D^{cal}
\right].
\label{app-eq-Ghat}
\end{align}
The second equality follows by conditioning additionally on $X$. 

When $(D,C)=(\hat D^\mu,\hat C^\mu)$ for a fixed $\mu$, the policy
depends only on $\mathcal D^{train}$. Since the new observation is
independent of $\mathcal D^{cal}$ conditionally on
$\mathcal D^{train}$,
\begin{align*}
\widehat G(\hat D^\mu,\hat C^\mu)
&=
\mE\left[
\left(
\mI\left\{Y\notin\hat C^\mu(X)\right\}-\alpha
\right)
\hat D^\mu(X)
\,\middle|\,
\mathcal D^{train}
\right]\\
&=
\mE\left[
\left(
\mP\left(
Y\notin\hat C^\mu(X)
\mid X,\mathcal D^{train}
\right)-\alpha
\right)
\hat D^\mu(X)
\,\middle|\,
\mathcal D^{train}
\right].
\end{align*}

\subsection{Proof of Theorem~\ref{opt-theorem-merged}}

\begin{proof}
The finite-sample guarantee $\FCR\leq\alpha$ for \infoOSP\ follows
directly from Theorem~\ref{thm-FCR-control}. We prove the power bounds
and the approximate error-rate guarantees for \infoOSPCal. Throughout the proof, $(X,Y)$ denotes an observation from the joint distribution, independent of all other (training, calibration, test) observations. Expectations and probabilities are taken over all random quantities appearing in the expression (e.g., the expectation of $\hat D^{\mu}(X))$ for a fixed $\mu$ is over the observation $X$ and $\mathcal D_{train}$, but the expectation of   $\hat D^{\thtwo}(X))$ is also over $\mathcal D_{cal}$); conditional expectations specify which samples are held fixed.

We begin by comparing the fitted and oracle policies at a fixed
multiplier $\mu$. 
For any deterministic $\mu\leq\mu^*+1$ and $C\in\mathcal I$,
\begin{align*}
\left|
\hell_{X,C}(\mu)-\ell_{X,C}(\mu)
\right|
&=
\left(w(C)+\mu\right)
\left|
\hat{\mP}(Y\in C\mid X)
-
\mP(Y\in C\mid X)
\right|\\
&\leq
K(B+\mu^*+1)
\max_{y\in[K]}
\left|
\hat{\mP}(Y=y\mid X)
-
\mP(Y=y\mid X)
\right|, 
\end{align*}
where the inequality follows since $w(C)\leq B$, $\mu\leq\mu^*+1$, and 
\[
\left|
\hat{\mP}(Y\in C\mid X)-\mP(Y\in C\mid X)
\right|
\leq
\sum_{y\in C}
\left|
\hat{\mP}(Y=y\mid X)-\mP(Y=y\mid X)
\right|
\]
 with $|C|\leq K$ so $\sum_{y\in C}
\left|
\hat{\mP}(Y=y\mid X)-\mP(Y=y\mid X)
\right|
\leq K\max_{y\in[K]}
\left|
\hat{\mP}(Y=y\mid X)
-
\mP(Y=y\mid X)
\right|$. Thus, Assumption~\ref{ass-approxproba} gives
\begin{align}
\mE\left[
\max_{C\in\mathcal I}
\left|
\hell_{X,C}(\mu)-\ell_{X,C}(\mu)
\right|
\right]
\leq
K(B+\mu^*+1)\eta_t.
\label{eq:S9-line-error}
\end{align}

Suppose that
$C^\mu(X)\neq\hat C^\mu(X)$. Sincethe
$\hat C^\mu(X)$ is one of the competitors to $C^\mu(X)$  we maximize over
 in equation \eqref{equDelta},
\begin{align*}
\Delta_\mu(X)
&\leq
\ell_{X,C^\mu(X)}(\mu)
-
\ell_{X,\hat C^\mu(X)}(\mu)\\
&=
\left(
\ell_{X,C^\mu(X)}(\mu)
-
\hell_{X,C^\mu(X)}(\mu)
\right)+
\left(
\hell_{X,C^\mu(X)}(\mu)
-
\hell_{X,\hat C^\mu(X)}(\mu)
\right)+
\left(
\hell_{X,\hat C^\mu(X)}(\mu)
-
\ell_{X,\hat C^\mu(X)}(\mu)
\right)\\
&\leq
2\max_{C\in\mathcal I}
\left|
\hell_{X,C}(\mu)-\ell_{X,C}(\mu)
\right|, 
\end{align*}
where the last inequality follows because
$\hat C^\mu(X)$ maximizes the fitted candidate lines, and hence
\[
\hell_{X,C^\mu(X)}(\mu)
-
\hell_{X,\hat C^\mu(X)}(\mu)
\leq0.
\]
It follows that $\mP\left(
C^\mu(X)\neq\hat C^\mu(X)
\right)\leq \mP\left(\Delta_{\mu}(X) \leq  2\max_{C\in\mathcal I}
\left|
\hell_{X,C}(\mu)-\ell_{X,C}(\mu)
\right|\right)$. So for every $r>0$,
\begin{align*}
&
\mP\left(
C^\mu(X)\neq\hat C^\mu(X)
\right) \leq \mP\left(\Delta_{\mu}(X) \leq 2\max_{C\in\mathcal I}
\left|
\hell_{X,C}(\mu)-\ell_{X,C}(\mu)
\right|,  2\max_{C\in\mathcal I}
\left|
\hell_{X,C}(\mu)-\ell_{X,C}(\mu)
\right|\leq 2r\right) \nonumber \\ & +\mP\left(\Delta_{\mu}(X)  \leq 2\max_{C\in\mathcal I}
\left|
\hell_{X,C}(\mu)-\ell_{X,C}(\mu)
\right|,  2\max_{C\in\mathcal I}
\left|
\hell_{X,C}(\mu)-\ell_{X,C}(\mu)
\right|> 2r\right)  \\
&\quad\leq
\mP\left(\Delta_\mu(X)\leq2r\right)
+
\mP\left(
\max_{C\in\mathcal I}
\left|
\hell_{X,C}(\mu)-\ell_{X,C}(\mu)
\right|>r
\right)\leq
2c_\mu r+
\frac{K(B+\mu^*+1)\eta_t}{r},
\end{align*}
where the final inequality follows from
Assumption~\ref{ass-isolatedmax}, Markov's inequality, and
\eqref{eq:S9-line-error}. 

We assume $0<\eta_t<1$, since  $\eta_t=0$ implies by Assumption~\ref{ass-isolatedmax} that the fitted and true conditional class probabilities agree almost surely for an independent $X$. 
Taking $r=\eta_t^{1/2}$ and using
$c_\mu\leq c_0$ yields
\begin{align}
\mP\left(
C^\mu(X)\neq\hat C^\mu(X)
\right)
\leq
2\left(
c_0+K(B+\mu^*+1)
\right)\eta_t^{1/2}.
\label{eq:S9-set-stability}
\end{align}

We next compare the selection decisions. If
$D^\mu(X)\neq\hat D^\mu(X)$, then zero lies between the oracle and fitted upper
envelopes, and hence
\begin{align*}
\left|
\max_{C\in\mathcal I}\ell_{X,C}(\mu)
\right|
&\leq
\left|
\max_{C\in\mathcal I}\hell_{X,C}(\mu)
-
\max_{C\in\mathcal I}\ell_{X,C}(\mu)
\right|\\
&\leq
\max_{C\in\mathcal I}
\left|
\hell_{X,C}(\mu)-\ell_{X,C}(\mu)
\right|,
\end{align*}
where the last inequality follows since the difference between two maxima cannot exceed the largest paired difference in absolute value. 

Therefore, 
$$\mE\left[
\left|
D^\mu(X)-\hat D^\mu(X)
\right|
\right]=
\mP\left(
D^\mu(X)\neq\hat D^\mu(X)
\right)\leq \mP\left( \left|
\max_{C\in\mathcal I}\ell_{X,C}(\mu)
\right|\leq \max_{C\in\mathcal I}
\left|
\hell_{X,C}(\mu)-\ell_{X,C}(\mu)
\right| \right) $$
So for every
$r>0$,
\begin{align*}
&
\mE\left[
\left|
D^\mu(X)-\hat D^\mu(X)
\right|
\right]
\leq
\mP\left(
\left|
\max_{C\in\mathcal I}\ell_{X,C}(\mu)
\right|\leq r
\right)
+
\mP\left(
\max_{C\in\mathcal I}
\left|
\hell_{X,C}(\mu)-\ell_{X,C}(\mu)
\right|>r
\right) \\ & \leq
c_\mu r+
\frac{K(B+\mu^*+1)\eta_t}{r}.
\end{align*}
Taking $r=\eta_t^{1/2}$ gives
\begin{align}
\mE\left[
\left|
D^\mu(X)-\hat D^\mu(X)
\right|
\right]
\leq
\left(
c_0+K(B+\mu^*+1)
\right)\eta_t^{1/2}.
\label{eq:S9-decision-stability}
\end{align}

For a fixed $\mu$, the fitted policy depends only on $\mathcal{D}^{train}$, so $\widehat G$ in \eqref{app-eq-Ghat} depends only on $\mathcal{D}^{train}$. Moreover, because the oracle policy does not depend on any sample, $G(D^\mu, C^{\mu}) = \hat G(D^{\mu}, C^{\mu})$. Therefore, 
\begin{align*}
&\left|
\widehat G(\hat D^\mu, \hat C^{\mu})-G(D^\mu, C^{\mu})
\right|
 = \left| \mE\left[\left(\mP\left(Y\notin \hat C^{\mu} \mid X, \mathcal{D}^{train}\right)-\alpha\right)\hat D^{\mu} - \left(\mP\left(Y\notin  C^{\mu} \mid X\right)-\alpha\right) D^{\mu}\mid  \mathcal{D}^{train}\right]\right| \\
&\leq
\mE\left[
\left|
\left(
\mP\left(Y\notin\hat C^\mu(X)\mid X, \mathcal{D}^{train}\right)-\alpha
\right)
\hat D^\mu(X)-
\left(
\mP\left(Y\notin C^\mu(X)\mid X\right)-\alpha
\right)
D^\mu(X)
\right|
\mid  \mathcal{D}^{train} \right],
\end{align*}
where the last inequality follows from Jensen's inequality,

The integrand is bounded pointwise by
\[
\mI\left\{
C^\mu(X)\neq\hat C^\mu(X)
\right\}
+
\left|
D^\mu(X)-\hat D^\mu(X)
\right|.
\]
Therefore,
\[
\left|
\widehat G(\hat D^\mu, \hat C^{\mu})-G(D^\mu, C^{\mu})
\right|
\leq
\mE\left[
\mI\left\{
C^\mu(X)\neq\hat C^\mu(X)
\right\}
+
\left|
D^\mu(X)-\hat D^\mu(X)
\right|
\,\middle|\,
\mathcal D^{\mathrm{train}}
\right].
\]
Taking expectations over $\mathcal D^{\mathrm{train}}$ and applying
\eqref{eq:S9-set-stability} and
\eqref{eq:S9-decision-stability}  gives

\begin{align}
\mE\left[
\left|
\widehat G(\hat D^\mu, \hat C^{\mu})-G(D^\mu, C^{\mu})
\right|
\right]
\leq
3\left(
c_0+K(B+\mu^*+1)
\right)\eta_t^{1/2}.
\label{eq:S9-G-expectation}
\end{align}

Consequently, Markov's inequality gives
\begin{align}
\mP\left(
\left|
\widehat G(\hat D^\mu, \hat C^{\mu})-G(D^\mu, C^{\mu})
\right|>\eta_t^{1/4}
\right)
\leq
3\left(
c_0+K(B+\mu^*+1)
\right)\eta_t^{1/4}.
\label{eq:S9-G-probability}
\end{align}

The same argument gives a fixed-multiplier power comparison. Since each
weighted power contribution lies in $[0,B]$,
\begin{align}
\Bigg|
&
\mE\left[
w\left(\hat C^\mu(X)\right)
\mP\left(
Y\in\hat C^\mu(X)\mid X, \mathcal{D}^{train}
\right)
\hat D^\mu(X)
\right]
-
\Pi(D^\mu,C^\mu)
\Bigg|\nonumber\\
&\leq B\left(\mP\left(
C^\mu(X)\neq\hat C^\mu(X)
\right)
+\mE\left[
\left|
D^\mu(X)-\hat D^\mu(X)
\right|\right]\right)\leq
3B\left(
c_0+K(B+\mu^*+1)
\right)\eta_t^{1/2}.
\label{eq:S9-power-stability}
\end{align}

We next control the empirical criteria uniformly in $\mu$. This is the
only part of the proof in which the fitted version of
Assumption~\ref{ass-martingale} is used. By Proposition~3.1 in the
main paper,
\begin{align*}
\mI\left\{
Y\notin\hat C^\mu(X)
\right\}\hat D^\mu(X)
&=
\mI\{\lambda(X,Y)>\mu\},\\
\hat D^\mu(X)
&=
\mI\{\Lambda(X)>\mu\}.
\end{align*}

Conditional on $\mathcal D^{train}$,
$\{\lambda(X_j,Y_j)\}_{j=1}^n$ are iid, and
$\{\Lambda(X_{n+i})\}_{i=1}^m$ are iid. The two samples are also
conditionally independent. By the Dvoretzky--Kiefer--Wolfowitz--Massart
inequality (DKWM,  \citealt{Mass1990} ), with conditional probability at least
$1-2/(n\wedge m)$,
\begin{align}
&
\sup_{\mu\geq0}
\left|
\frac1n\sum_{j=1}^n
\mI\left\{
Y_j\notin\hat C^\mu(X_j)
\right\}\hat D^\mu(X_j)
-
\mE\left[
\mI\left\{
Y\notin\hat C^\mu(X)
\right\}\hat D^\mu(X)
\mid\mathcal D^{train}
\right]
\right|\nonumber\\
&\hspace{45mm}\leq
\sqrt{
\frac{\log\left(2(n\wedge m)\right)}{2n}
},
\label{eq:S9-DKW-cal}
\end{align}
and
\begin{align}
&
\sup_{\mu\geq0}
\left|
\frac1m\sum_{i=1}^m
\hat D^\mu(X_{n+i})
-
\mE\left[
\hat D^\mu(X)
\mid\mathcal D^{train}
\right]
\right|
\leq
\sqrt{
\frac{\log\left(2(n\wedge m)\right)}{2m}
}.
\label{eq:S9-DKW-test}
\end{align}

For every $\mu\geq0$, since $\frac{1}{1+n} \left(1 + \sum_{j=1}^n \mI\left\{Y_j\notin \hat{C}^{\mu}(X_j)\right\}\hat{D}^{\mu}(X_j)\right)-\frac1n\sum_{j=1}^n
\mI\left\{
Y_j\notin\hat C^\mu(X_j)
\right\}\hat D^\mu(X_j)\leq \frac 1{n+1},$ we have
\begin{align*}
&
\left|
\frac{1}{1+n} \left(1 + \sum_{j=1}^n \mI\left\{Y_j\notin \hat{C}^{\mu}(X_j)\right\}\hat{D}^{\mu}(X_j)\right)
-
\mE\left[
\mI\left\{
Y\notin\hat C^\mu(X)
\right\}\hat D^\mu(X)
\mid\mathcal D^{train}
\right]
\right|\\
&\quad\leq
\left|
\frac1n\sum_{j=1}^n
\mI\left\{
Y_j\notin\hat C^\mu(X_j)
\right\}\hat D^\mu(X_j)
-
\mE\left[
\mI\left\{
Y\notin\hat C^\mu(X)
\right\}\hat D^\mu(X)
\mid\mathcal D^{train}
\right]
\right|
+
\frac1{n+1}.
\end{align*}
Similarly,
\begin{align*}
&
\left|
\left(
\frac1m\lor\frac 1m\sum_{i=1}^m \hat{D}^{\mu}(X_{n+i})
\right)
-
\mE\left[
\hat D^\mu(X)
\mid\mathcal D^{train}
\right]
\right|\\
&\quad\leq
\left|
\frac1m\sum_{i=1}^m
\hat D^\mu(X_{n+i})
-
\mE\left[
\hat D^\mu(X)
\mid\mathcal D^{train}
\right]
\right|
+
\frac1m,
\end{align*}
where the inequality follows since 
$\frac{1}{m}\lor \frac{1}{m}\sum_{i=1}^m \hat{D}^{\mu}(X_{n+i})$ differs from
$\frac 1m\sum_{i=1}^m \hat{D}^{\mu}(X_{n+i})$ by at most $1/m$.

Combining the preceding inequalities with
\eqref{eq:S9-DKW-cal}--\eqref{eq:S9-DKW-test} gives, with conditional probability at least
$1-2/(n\wedge m)$,
\begin{align}
\sup_{\mu\geq0}
\left|
\tilde G(\hat D^\mu, \hat C^{\mu})-\widehat G(\hat D^\mu, \hat C^{\mu})
\right|
&\leq
\frac1{n+1}+\frac1m
+
\sqrt{
\frac{\log\left(2(n\wedge m)\right)}{2n}
}
+
\sqrt{
\frac{\log\left(2(n\wedge m)\right)}{2m}
}\nonumber\\
&\leq
\frac{
2+\sqrt{2\log\left(2(n\wedge m)\right)}
}{
\sqrt{n\wedge m}
}.
\label{eq:S9-empirical-G}
\end{align}

For the calibration-only criterion, apply the same argument to the two
calibration processes
\[
\mI\left\{
Y_j\notin\hat C^\mu(X_j)
\right\}\hat D^\mu(X_j)
\quad\text{and}\quad
\hat D^\mu(X_j).
\]
They need not be independent of each other; each process is iid across
$j$, and the two DKW events are combined by a union bound. With
conditional probability at least $1-2/n$,
\begin{align}
\sup_{\mu\geq0}
\left|
\tilde G_{cal}(\hat D^\mu, \hat C^{\mu})-\widehat G(\hat D^\mu, \hat C^{\mu})
\right|
&\leq
\frac{2}{n+1}
+
2\sqrt{\frac{\log(2n)}{2n}}\nonumber\\
&\leq
\frac{
2+\sqrt{2\log(2n)}
}{
\sqrt n}.
\label{eq:S9-empirical-G-cal}
\end{align}

We now prove item~(i). 
Let
$\epsilon=\epsilon(t,n,m)$ and $\nu=\nu(t,n,m)$. Applying
\eqref{eq:S9-G-probability} at the deterministic multiplier
$\mu^*+\epsilon$ and combining it with
\eqref{eq:S9-empirical-G}, we obtain, except on an event of probability
at most
\[
\frac{2}{n\wedge m}
+
3\left(
c_0+K(B+\mu^*+1)
\right)\eta_t^{1/4},
\]
that
\begin{align*}
\tilde G(\hat D^{\mu^*+\epsilon}, \hat C^{\mu^*+\epsilon})
&\leq
G(D^{\mu^*+\epsilon}, C^{\mu^*+\epsilon})
+
\eta_t^{1/4}
+
\frac{
2+\sqrt{2\log\left(2(n\wedge m)\right)}
}{
\sqrt{n\wedge m}
}.
\end{align*}
Since $\epsilon\leq\epsilon_0$, the first condition in
Assumption~\ref{ass-derivativeGalpha-together} gives
\[
G(D^{\mu^*+\epsilon}, C^{\mu^*+\epsilon})
\leq
G(D^{\mu^*}, C^{\mu^*})-a_0\epsilon
\leq 
-a_0\epsilon,
\]
where the second inequality follows from the fact that  Assumption \ref{ass-isolatedmax} implies that $G$ is continuous at $\mu^*$  and  so $G(D^{\mu^*}, C^{\mu^*})\leq 0$ for  $\mu^*<\infty$ (if $G(D^{0}, C^{0})>0$, then $G(D^{\mu^*}, C^{\mu^*})=0$ and the second inequality is in fact an equality). By the
definition of $\epsilon(t,n,m)$ it thus follows that $\tilde G(D^{\mu^*+\epsilon}, C^{\mu^*+\epsilon})\leq0.$
Therefore, by the definition of $\thone$,
\begin{align}
\mP\left(
\thone\leq\mu^*+\epsilon
\right)
\geq
1-\nu.
\label{eq:S9-localization-OG}
\end{align}

We next use this one-sided localization to lower-bound the power.
Since $\thone$ depends on the training sample, the calibration sample,
and the test covariates, but not on the test labels, conditioning on
these quantities gives
\begin{align}
\Pi
=
\mE\left[
\frac1m\sum_{i=1}^m
w\left(\hat C^{\thone}(X_{n+i})\right)
\mP\left(
Y\in\hat C^{\thone}(X_{n+i})
\mid X=X_{n+i}, \mathcal{D}^{train}
\right)
\hat D^{\thone}(X_{n+i})
\right].
\label{eq:S9-power-conditional}
\end{align}

On the event
$\{\thone\leq\mu^*+\epsilon\}$,  for every $i\in[m]$,
\begin{align*}
&
w\left(\hat C^{\thone}(X_{n+i})\right)
\mP\left(
Y\in\hat C^{\thone}(X_{n+i})
\mid X=X_{n+i}
\right)
\hat D^{\thone}(X_{n+i})\\
&\quad\geq
w\left(\hat C^{\thone}(X_{n+i})\right)
\hat{\mP}\left(
Y\in\hat C^{\thone}(X_{n+i})
\mid X=X_{n+i}
\right)
\hat D^{\thone}(X_{n+i})\\
&\qquad-
BK\max_{y\in[K]}
\left|
\hat{\mP}(Y=y\mid X=X_{n+i})
-
\mP(Y=y\mid X=X_{n+i})
\right|\\
&\quad\geq
w\left(\hat C^{\mu^*+\epsilon}(X_{n+i})\right)
\hat{\mP}\left(
Y\in\hat C^{\mu^*+\epsilon}(X_{n+i})
\mid X=X_{n+i}
\right)
\hat D^{\mu^*+\epsilon}(X_{n+i})\\
&\qquad-
BK\max_{y\in[K]}
\left|
\hat{\mP}(Y=y\mid X=X_{n+i})
-
\mP(Y=y\mid X=X_{n+i})
\right|\\
&\quad\geq
w\left(\hat C^{\mu^*+\epsilon}(X_{n+i})\right)
\mP\left(
Y\in\hat C^{\mu^*+\epsilon}(X_{n+i})
\mid X=X_{n+i}
\right)
\hat D^{\mu^*+\epsilon}(X_{n+i})\\
&\qquad-
2BK\max_{y\in[K]}
\left|
\hat{\mP}(Y=y\mid X=X_{n+i})
-
\mP(Y=y\mid X=X_{n+i})
\right|.
\end{align*}
The first and third inequalities replace a true set probability by its
fitted counterpart and then back again. The middle inequality follows from the monotonicty of $ \mu\longmapsto
w\left(\hat C^\mu(x)\right)
\hat{\mP}\left(
Y\in\hat C^\mu(x)\mid X=x
\right)
\hat D^\mu(x)$, see  item~4 of Proposition~\ref{prop-monotoneD-procedure}. 

Using \eqref{eq:S9-power-conditional}, the non-negativity of all power
contributions, Assumption~\ref{ass-approxproba}, and
\eqref{eq:S9-localization-OG}, we obtain
\begin{align*}
\Pi
&\geq
\mE\left[
w\left(\hat C^{\mu^*+\epsilon}(X)\right)
\mP\left(
Y\in\hat C^{\mu^*+\epsilon}(X)\mid X
\right)
\hat D^{\mu^*+\epsilon}(X)
\right]\\
&\quad-
B\mP\left(
\thone>\mu^*+\epsilon
\right)
-
2BK\eta_t.
\end{align*}
Applying \eqref{eq:S9-power-stability} at
$\mu=\mu^*+\epsilon$ gives
\begin{align*}
\Pi
&\geq
\Pi(D^{\mu^*+\epsilon},C^{\mu^*+\epsilon})
-
B\mP\left(
\thone>\mu^*+\epsilon
\right)\\
&\quad-
2BK\eta_t
-
3B\left(
c_0+K(B+\mu^*+1)
\right)\eta_t^{1/2}.
\end{align*}
Since $\eta_t\leq1$ and
$c_0+K(B+\mu^*+1)\geq K$, the definition of $\nu$ implies
\[
B\mP\left(
\thone>\mu^*+\epsilon
\right)
+
2BK\eta_t
+
3B\left(
c_0+K(B+\mu^*+1)
\right)\eta_t^{1/2}
\leq
B\nu.
\]
Consequently,
\[
\Pi
\geq
\Pi(D^{\mu^*+\epsilon},C^{\mu^*+\epsilon})
-
B\nu.
\]
Finally, Assumption~\ref{ass-derivativeGalpha-together} gives
\[
\Pi(D^{\mu^*+\epsilon},C^{\mu^*+\epsilon})
\geq
\Pi(D^{\mu^*},C^{\mu^*})
-
b_0\epsilon,
\]
which proves item~(i).

We now prove item~(ii). Let
$\epsilon=\epsilon(t,n,n)$ and $\nu=\nu(t,n,n)$. Repeating the
preceding localization argument, now using
\eqref{eq:S9-empirical-G-cal}, gives
\begin{align}
\mP\left(
\thtwo\leq\mu^*+\epsilon
\right)
\geq
1-\nu.
\label{eq:S9-localization-cal}
\end{align}
The same power argument as above, with $\thtwo$ in place of $\thone$,
yields
\begin{align}
\Pi_{cal}
\geq
\Pi(D^{\mu^*},C^{\mu^*})
-
b_0\epsilon
-
B\nu.
\label{eq:S9-power-cal}
\end{align}

Next, we bound the conditional, or instantaneous, selected noncoverage
rate. Conditional on $\mathcal D^{train}$ and $\mathcal D^{cal}$, the
threshold $\thtwo$ and hence the final fitted policy are fixed.
Therefore, by the definition of $\widehat G$,
\begin{align*}
\widehat G\left(\hat D^{\thtwo},\hat C^{\thtwo}\right)
=
\mE\left[
\mI\left\{
Y\notin\hat C^{\thtwo}(X)
\right\}
\hat D^{\thtwo}(X)
\,\middle|\,
\mathcal D^{train},\mathcal D^{cal}
\right]-
\alpha
\mE\left[
\hat D^{\thtwo}(X)
\,\middle|\,
\mathcal D^{train},\mathcal D^{cal}
\right].
\end{align*}
Consequently, whenever the conditional selection probability is
positive, the conditional selected noncoverage rate equals
\begin{align*}
\frac{
\mE\left[
\mI\left\{
Y\notin\hat C^{\thtwo}(X)
\right\}
\hat D^{\thtwo}(X)
\,\middle|\,
\mathcal D^{train},\mathcal D^{cal}
\right]
}{
\mE\left[
\hat D^{\thtwo}(X)
\,\middle|\,
\mathcal D^{train},\mathcal D^{cal}
\right]
}=
\alpha+
\frac{
\widehat G\left(\hat D^{\thtwo},\hat C^{\thtwo}\right)
}{
\mE\left[
\hat D^{\thtwo}(X)
\,\middle|\,
\mathcal D^{train},\mathcal D^{cal}
\right]
}.
\end{align*}
Here and below, a ratio with a zero denominator is defined to be zero.

Consider the intersection of the following three events:
\begin{align}
&
\sup_{\mu\geq0}
\left|
\tilde G_{cal}\left(\hat D^\mu,\hat C^\mu\right)
-
\widehat G\left(\hat D^\mu,\hat C^\mu\right)
\right|
\leq
\frac{
2+\sqrt{2\log(2n)}
}{
\sqrt n
},
\label{eq:S9-FCR-event1}\\
&
\left|
\widehat G\left(
\hat D^{\mu^*+\epsilon},
\hat C^{\mu^*+\epsilon}
\right)
-
G\left(
D^{\mu^*+\epsilon},
C^{\mu^*+\epsilon}
\right)
\right|
\leq
\eta_t^{1/4},
\label{eq:S9-FCR-event2}\\
&
\mE\left[
\left|
\hat D^{\mu^*+\epsilon}(X)
-
D^{\mu^*+\epsilon}(X)
\right|
\,\middle|\,
\mathcal D^{train}
\right]
\leq
\eta_t^{1/4}.
\label{eq:S9-FCR-event3}
\end{align}
By \eqref{eq:S9-empirical-G-cal}, the failure probability of
\eqref{eq:S9-FCR-event1} is at most $2/n$. By
\eqref{eq:S9-G-probability}, the failure probability of
\eqref{eq:S9-FCR-event2} is at most
\[
3\left(
c_0+K(B+\mu^*+1)
\right)\eta_t^{1/4}.
\]

Finally, by Markov's inequality and
\eqref{eq:S9-decision-stability}, the failure probability of
\eqref{eq:S9-FCR-event3} is at most
\begin{align*}
&
\mP\left(
\mE\left[
\left|
\hat D^{\mu^*+\epsilon}(X)
-
D^{\mu^*+\epsilon}(X)
\right|
\,\middle|\,
\mathcal D^{train}
\right]
>
\eta_t^{1/4}
\right)\\
&\quad\leq
\eta_t^{-1/4}
\mE\left[
\mE\left[
\left|
\hat D^{\mu^*+\epsilon}(X)
-
D^{\mu^*+\epsilon}(X)
\right|
\,\middle|\,
\mathcal D^{train}
\right]
\right]\\
&\quad=
\eta_t^{-1/4}
\mE\left[
\left|
\hat D^{\mu^*+\epsilon}(X)
-
D^{\mu^*+\epsilon}(X)
\right|
\right]\\
&\quad\leq
\left(
c_0+K(B+\mu^*+1)
\right)\eta_t^{1/4}.
\end{align*}

Consequently,the three events hold simultaneously except on an event
of probability at most $\nu$.
On their intersection,
\begin{align*}
&\tilde G_{cal}\left(
\hat D^{\mu^*+\epsilon},
\hat C^{\mu^*+\epsilon}
\right)
\leq
\widehat G\left(
\hat D^{\mu^*+\epsilon},
\hat C^{\mu^*+\epsilon}
\right)
+
\frac{
2+\sqrt{2\log(2n)}
}{
\sqrt n
}\\
&\quad\leq
G\left(
D^{\mu^*+\epsilon},
C^{\mu^*+\epsilon}
\right)
+
\eta_t^{1/4}
+
\frac{
2+\sqrt{2\log(2n)}
}{
\sqrt n
}
\leq0,
\end{align*}
where the last inequality follows from the first condition in
Assumption~\ref{ass-derivativeGalpha-together},
$G(D^{\mu^*},C^{\mu^*})\leq0$, and the definition of $\epsilon$.
Therefore,
\[
\thtwo\leq\mu^*+\epsilon.
\]
Moreover, \eqref{eq:S9-FCR-event1} and the definition of $\thtwo$ give
\begin{equation}
 \widehat G\left(\hat D^{\thtwo}, \hat C^{\thtwo}\right)
\leq
\tilde G_{cal}\left(
\hat D^{\thtwo},
\hat C^{\thtwo}
\right)
+
\frac{
2+\sqrt{2\log(2n)}
}{
\sqrt n
}\leq
\frac{
2+\sqrt{2\log(2n)}
}{
\sqrt n
}
\leq
a_0\epsilon.
\label{eq:S9-FCR-G}
\end{equation}
The monotonicity of $\hat D^\mu$ and
\eqref{eq:S9-FCR-event3} imply
\begin{align*}
&\mE\left[
\hat D^{\thtwo}(X)
\,\middle|\,
\mathcal D^{train},\mathcal D^{cal}
\right]\\
&\quad\geq
\mE\left[
\hat D^{\mu^*+\epsilon}(X)
\,\middle|\,
\mathcal D^{train},\mathcal D^{cal}
\right]\\
&\quad=
\mE\left[
\hat D^{\mu^*+\epsilon}(X)
\,\middle|\,
\mathcal D^{train}
\right]\\
&\quad\geq
\mE\left[
D^{\mu^*+\epsilon}(X)
\right]
-
\eta_t^{1/4}\\
&\quad\geq
s_0-d_0\epsilon-\eta_t^{1/4}
\geq
s_0/2.
\end{align*}
The equality follows because, at the deterministic multiplier
$\mu^*+\epsilon$, the fitted decision depends only on
$\mathcal D^{train}$ and the fresh $X$ is independent of
$\mathcal D^{cal}$. The last inequality uses
$\eta_t^{1/4}\leq a_0\epsilon$ and the condition
$(d_0+a_0)\epsilon+\nu\leq s_0/2$.
Combining this bound with the definition of $\widehat G$ and
\eqref{eq:S9-FCR-G} yields, on the same event,
\begin{align}
&\frac{
\mE\left[
\mI\left\{
Y\notin\hat C^{\thtwo}(X)
\right\}
\hat D^{\thtwo}(X)
\,\middle|\,
\mathcal D^{train},\mathcal D^{cal}
\right]
}{
\mE\left[
\hat D^{\thtwo}(X)
\,\middle|\,
\mathcal D^{train},\mathcal D^{cal}
\right]
}\nonumber\\
&\quad=
\alpha+
\frac{
\widehat G\left(\hat D^{\thtwo},\hat C^{\thtwo}\right)
}{
\mE\left[
\hat D^{\thtwo}(X)
\,\middle|\,
\mathcal D^{train},\mathcal D^{cal}
\right]
}\nonumber\\
&\quad\leq
\alpha+
\frac{
a_0\epsilon
}{
s_0-d_0\epsilon-\eta_t^{1/4}
}
\leq
\alpha+2a_0s_0^{-1}\epsilon.
\label{eq:S9-instantaneous-error}
\end{align}

We now bound $\FCR_{cal}$. Conditional on $\mathcal D^{train}$ and
$\mathcal D^{cal}$, the threshold $\thtwo$ and the fitted policy are
fixed, while the test observations are iid. The conditional version of
the calculation in the proof of Proposition~2.1 therefore gives
\begin{align}
\FCR_{cal}
=
\mE\Bigg[
&
\frac{
\mE\left[
\mI\left\{
Y\notin\hat C^{\thtwo}(X)
\right\}
\hat D^{\thtwo}(X)
\,\middle|\,
\mathcal D^{train},\mathcal D^{cal}
\right]
}{
\mE\left[
\hat D^{\thtwo}(X)
\,\middle|\,
\mathcal D^{train},\mathcal D^{cal}
\right]
}\nonumber\\
&\quad\times
\mP\left(
\sum_{i=1}^m
\hat D^{\thtwo}(X_{n+i})>0
\,\middle|\,
\mathcal D^{train},\mathcal D^{cal}
\right)
\Bigg].
\label{eq:S9-FCR-conditional}
\end{align}
The second factor is at most one. On the intersection of
\eqref{eq:S9-FCR-event1}--\eqref{eq:S9-FCR-event3}, the first factor is
bounded by \eqref{eq:S9-instantaneous-error}; on the complementary
event, it is at most one. It follows that
\begin{align}
\FCR_{cal}
&\leq
\alpha+2a_0s_0^{-1}\epsilon+\nu.
\label{eq:S9-FCR-final}
\end{align}

Finally, we bound $\mFCR_{cal}$. If $\mE\left[
\hat D^{\thtwo}(X)
\right]=0,$ 
then $\mFCR_{cal}=0$ by convention. Otherwise, taking expectations in
the definition of $\widehat G$ gives
\begin{align*}
\mE\left[
\widehat G\left(\hat D^{\thtwo},\hat C^{\thtwo}\right)
\right]
&=
\mE\left[
\mI\left\{
Y\notin\hat C^{\thtwo}(X)
\right\}
\hat D^{\thtwo}(X)
\right]-
\alpha
\mE\left[
\hat D^{\thtwo}(X)
\right].
\end{align*}
Therefore,
\begin{align}
\mFCR_{cal}-\alpha
=
\frac{
\mE\left[
\widehat G\left(\hat D^{\thtwo},\hat C^{\thtwo}\right)
\right]
}{
\mE\left[
\hat D^{\thtwo}(X)
\right]
}.
\label{eq:S9-mFCR-ratio}
\end{align}
On the event in \eqref{eq:S9-FCR-event1}, the definition of $\thtwo$
gives
\[
\widehat G\left(\hat D^{\thtwo},\hat C^{\thtwo}\right)
\leq
\frac{
2+\sqrt{2\log(2n)}
}{
\sqrt n
}
\]
when $\thtwo<\infty$. If $\thtwo=\infty$, no observations are selected
and this evaluation equals zero. Since $\widehat G(D,C)\leq1$, and the
failure probability of \eqref{eq:S9-FCR-event1} is at most $2/n$,
\begin{align}
\mE\left[
\widehat G\left(\hat D^{\thtwo},\hat C^{\thtwo}\right)
\right]
&\leq
\frac{
2+\sqrt{2\log(2n)}
}{
\sqrt n
}
+
\frac2n\nonumber\\
&\leq
a_0\epsilon+\nu.
\label{eq:S9-mFCR-numerator}
\end{align}
For the denominator, the monotonicity of $\hat D^\mu$ gives
\begin{align*}
\mE\left[
\hat D^{\thtwo}(X)
\right]
&\geq
\mE\left[
\mI\left\{
\thtwo\leq\mu^*+\epsilon
\right\}
\hat D^{\mu^*+\epsilon}(X)
\right]\\
&\geq
\mE\left[
\hat D^{\mu^*+\epsilon}(X)
\right]
-
\mP\left(
\thtwo>\mu^*+\epsilon
\right)\\
&\geq
\mE\left[
D^{\mu^*+\epsilon}(X)
\right]
-
\left(
c_0+K(B+\mu^*+1)
\right)\eta_t^{1/2}
-
\mP\left(
\thtwo>\mu^*+\epsilon
\right)\\
&\geq
s_0-d_0\epsilon-\nu.
\end{align*}
For the last inequality, the proof of
\eqref{eq:S9-localization-cal} gives
\[
\mP\left(
\thtwo>\mu^*+\epsilon
\right)
\leq
\frac2n
+
3\left(
c_0+K(B+\mu^*+1)
\right)\eta_t^{1/4},
\]
and $\eta_t\leq1$, so the two approximation terms are bounded by
$\nu$.
Combining the denominator bound with
\eqref{eq:S9-mFCR-ratio}--\eqref{eq:S9-mFCR-numerator} gives
\[
\mFCR_{cal}
\leq
\alpha+
\frac{
a_0\epsilon+\nu
}{
s_0-d_0\epsilon-\nu
}.
\]
Under the condition
\[
(d_0+a_0)\epsilon+\nu\leq s_0/2,
\]
the denominator is at least $s_0/2$, and hence
\begin{align}
\mFCR_{cal}
&\leq
\alpha+2s_0^{-1}(a_0\epsilon+\nu).
\label{eq:S9-mFCR-final}
\end{align}
Together, \eqref{eq:S9-power-cal}, \eqref{eq:S9-FCR-final}, and
\eqref{eq:S9-mFCR-final} prove item~(ii).

\end{proof}

\end{document}